\documentclass[mnsc,nonblindrev]{informs3_hide} 

\OneAndAHalfSpacedXI 

\usepackage{natbib, multirow, multicol, xspace, enumitem, amsmath, amssymb, pifont, subcaption}
\usepackage{hyperref}
\hypersetup{hidelinks}
\usepackage{accents}
\bibpunct[, ]{(}{)}{,}{a}{}{,}%
\def\bibfont{\small}%
\usepackage{tikz}
\usepackage{bm}
\usepackage[normalem]{ulem}

\usepackage{algorithm}
\usepackage[noend]{algpseudocode}

\algnewcommand{\algorithmicgoto}{\phantom{for} \textbf{go to} Line}%
\algnewcommand{\Goto}[1]{\State \algorithmicgoto~\ref{#1}}%

\newcommand{\mathbbm}[1]{\text{\usefont{U}{bbm}{m}{n}#1}} 

\DeclareMathOperator{\Tr}{Tr}
\newcommand{\bI}{\mathbbm{1}}

\newcommand{\bE}{\mathrm{E}}
\newcommand{\Cov}{\mathrm{Cov}}
\newcommand{\Var}{\mathrm{Var}}
\newcommand{\bR}{\mathbb{R}}

\newcommand{\cV}{\mathcal{V}}
\newcommand{\cY}{\mathcal{Y}}
\newcommand{\cW}{\mathcal{W}}
\newcommand{\cN}{\mathcal{N}}
\newcommand{\cA}{\mathcal{A}}
\newcommand{\cB}{\mathcal{B}}
\newcommand{\cE}{\mathcal{E}}

\newcommand{\HT}{\mathrm{HT}}
\newcommand{\CV}{\mathrm{CV}}
\newcommand{\AIPW}{\mathrm{AIPW}}
\newcommand{\TMLE}{\mathrm{TMLE}}
\newcommand{\Hajek}{\mathrm{Hajek}}
\newcommand{\Naive}{\mathrm{Naive}}
\newcommand{\Normalized}{\mathrm{Norm}}
\newcommand{\PO}{\mathrm{PO}}
\newcommand{\GATE}{\mathrm{GATE}}

\newcommand{\eT}{\delta(1)}
\newcommand{\eC}{\delta(0)}

\newcommand\blfootnote[1]{%
  \begingroup
  \renewcommand\thefootnote{}\footnote{#1}%
  \addtocounter{footnote}{-1}%
  \endgroup
}

\newtheorem{theorem}{Theorem}
\newtheorem{lemma}{Lemma}
\newtheorem{proposition}{Proposition}
\newtheorem{corollary}{Corollary}
\newtheorem{definition}{Definition}
\newtheorem{example}{Example}
\newtheorem{assumption}{Assumption}
\newtheorem{subassumption}{Assumption}

\newenvironment{assumption+}[1]
 {\renewcommand{\thesubassumption}{#1}\subassumption}
 {\endsubassumption}

\EquationsNumberedThrough    

\MANUSCRIPTNO{}

\begin{document}


\RUNAUTHOR{}

\RUNTITLE{}

\TITLE{Optimal Control Variates for Survey Sampling and Causal Inference}

\ARTICLEAUTHORS{Jinglong Zhao}

\ABSTRACT{
We propose a family of control variate estimators for variance reduction in design-based survey sampling and causal inference, with and without interference. 
In these settings, inverse probability weighting (IPW) estimators are widely used, but may have large variance when sampling, treatment, or exposure probabilities are small. 
Building on the observation that several common estimators, including the Hajek estimator, the normalized estimator, the augmented inverse probability weighting (AIPW) estimator, and the targeted maximum likelihood estimator (TMLE), all correct the Horvitz-Thompson estimator by canceling part of its randomness, we provide a unified interpretation of these estimators as special cases of a general control variate estimator. 
We then construct optimal control variates that can reduce the finite sample variance compared to these common estimators. 
We parameterize the proposed control variates by their bases and characterize the optimal bases through a stochastic optimization formulation. 
In survey sampling and causal inference without interference, the optimal bases are characterized by leading eigenvectors of matrices that depend on both the design-based sampling structure and the model-based outcome uncertainty.  
In causal inference under network interference, the optimal bases solve a nonconvex quadratic optimization problem; we provide a $\frac{1}{2}$-approximate solution and an alternating local search heuristic. 
We apply the control variate estimators to the Swiss Environmental Panel survey data and the Chinese social network data, and conduct extensive simulations to show that the proposed control variate estimators can achieve substantial variance reduction.
}



\begin{center}
{\Large\bf Optimal Control Variates for Survey Sampling and Causal Inference} 
\vskip 1.5em
Jinglong Zhao 
\blfootnote{Jinglong Zhao, Questrom School of Business, Boston University, \url{jinglong@bu.edu}. 
The author thanks David Bruns-Smith, Peng Ding, Han Hong, Lihua Lei, Shuangning Li, Tu Ni, Erol Pekoz, Nian Si, Carl Sun, George Shanthikumar, Johan Ugander, Stefan Wager, Yunzong Xu, and seminar and workshop participants at Harvard University and Stanford University for their insightful comments that have greatly improved this manuscript.
}
\\
Boston University\\
\today \par
\end{center}\par
\vskip 1.5em%

\begin{center} \normalsize\bf\text{Abstract} \end{center}

\begin{quote} \small
We propose a family of control variate estimators for variance reduction in design-based survey sampling and causal inference, with and without interference. 
In these settings, inverse probability weighting (IPW) estimators are widely used, but may have large variance when sampling, treatment, or exposure probabilities are small. 
Building on the observation that several common estimators, including the Hajek estimator, the normalized estimator, the augmented inverse probability weighting (AIPW) estimator, and the targeted maximum likelihood estimator (TMLE), all correct the Horvitz-Thompson estimator by canceling part of its randomness, we provide a unified interpretation of these estimators as special cases of a general control variate estimator. 
We then construct optimal control variates that can reduce the finite sample variance compared to these common estimators. 
We parameterize the proposed control variates by their bases and characterize the optimal bases through a stochastic optimization formulation. 
In survey sampling and causal inference without interference, the optimal bases are characterized by leading eigenvectors of matrices that depend on both the design-based sampling structure and the model-based outcome uncertainty.  
In causal inference under network interference, the optimal bases solve a nonconvex quadratic optimization problem; we provide a $\frac{1}{2}$-approximate solution and an alternating local search heuristic. 
We apply the control variate estimators to the Swiss Environmental Panel survey data and the Chinese social network data, and conduct extensive simulations to show that the proposed control variate estimators can achieve substantial variance reduction.
\end{quote}

\vskip 2.5em%

\section{Introduction}

Inverse probability weighting (IPW) estimators are widely used in survey sampling when units are observed with unequal probabilities.
The basic idea dates back to the early work of \citet{hansen1943theory}, who suggested that when a certain type of unit is less likely to be observed, then each observed unit of this type should count more. 
\citet{horvitz1952generalization} formalized this idea by showing that directly weighting each observed unit by its inverse probability yields an unbiased estimator of the target population. 
This simple form of IPW estimator is commonly known as the Horvitz-Thompson estimator.

Horvitz-Thompson estimators have broad applications beyond survey sampling. 
They are widely used in causal inference \citep{rosenbaum1983central, rosenbaum1987model}, including network interference \citep{aronow2017estimating, tchetgen2012causal}.
In these settings, the Horvitz-Thompson estimator is well known to have a large variance, especially when the sampling, treatment, or exposure probabilities are small.
These small probabilities are common in applications such as social experiments \citep{bond201261, crepon2013labor}, marketing research \citep{feit2019test, lewis2015unfavorable}, and online A/B testing \citep{karrer2021network, kohavi2007practical}. 

This paper considers two empirical applications. 
In the first application, the Swiss Environmental Panel conducts surveys among Swiss residents to measure public support for food waste regulation \citep{fesenfeld2022policy, quoss2021swiss}. 
See Figure~\ref{fig:SwissResponseRates}.
Because the survey has unequal sampling probabilities, the Horvitz-Thompson estimator can have a large variance, even though no region is sampled with a particularly small probability.
In the second application, the People Insurance Company of China conducts a randomized field experiment among rice producing households in rural China to study how financial education about a new weather insurance product diffuses through social networks \citep{cai2015social}. 
See Figure~\ref{fig:PICCInsuranceNetwork}.
Because the exposure probabilities depend on the number of eligible friends in a household's network and are thus heterogeneous and small, the Horvitz-Thompson estimator can again have a large variance. 

\begin{figure}[!tb]
\centering
\includegraphics[width=0.6\linewidth]{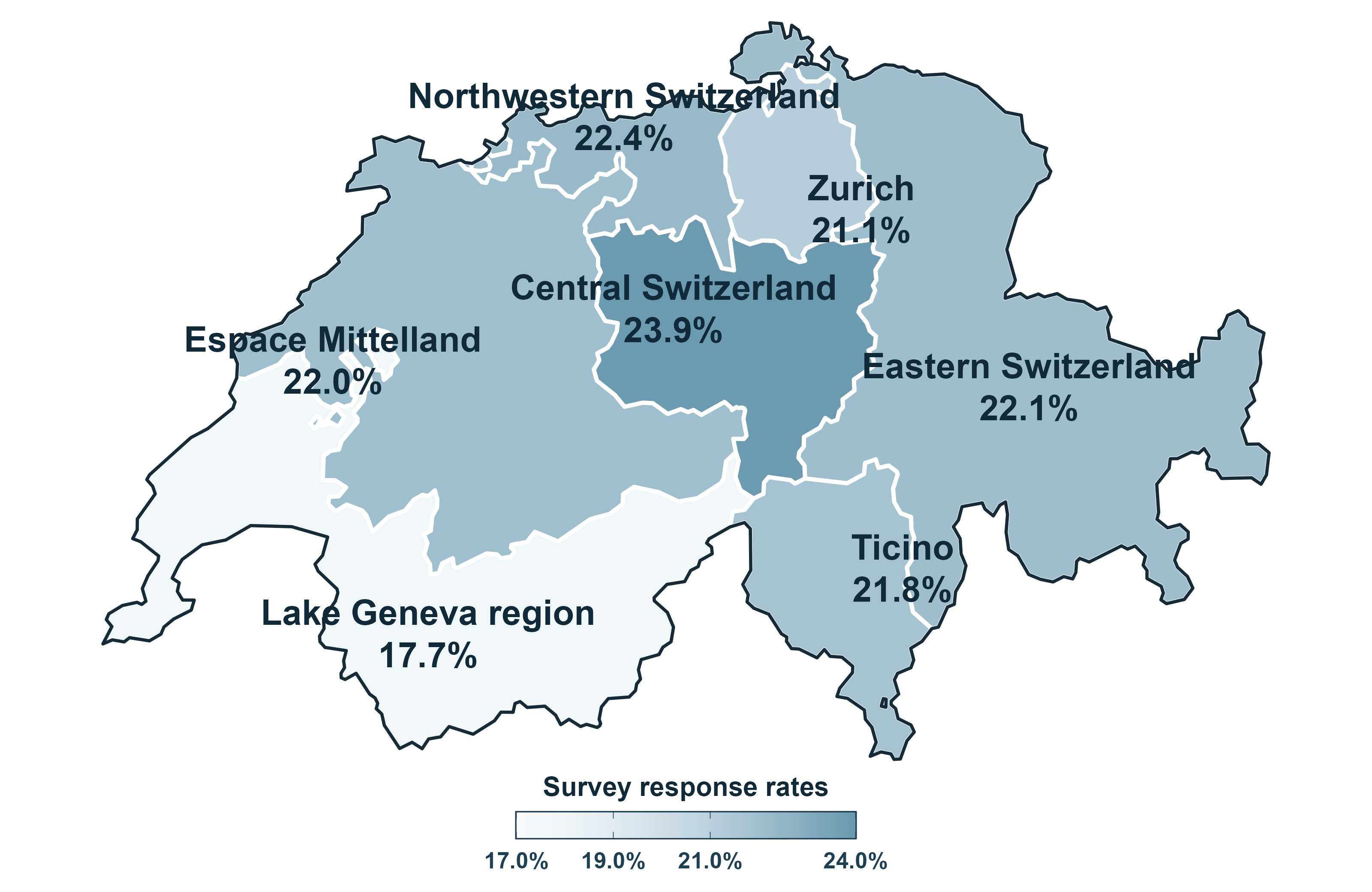}
\caption{Regional cumulative response rates in the Swiss Environmental Panel survey \citep{quoss2021swiss}}
\label{fig:SwissResponseRates}
\end{figure}

\begin{figure}[!tb]
\centering
\includegraphics[width=0.6\linewidth]{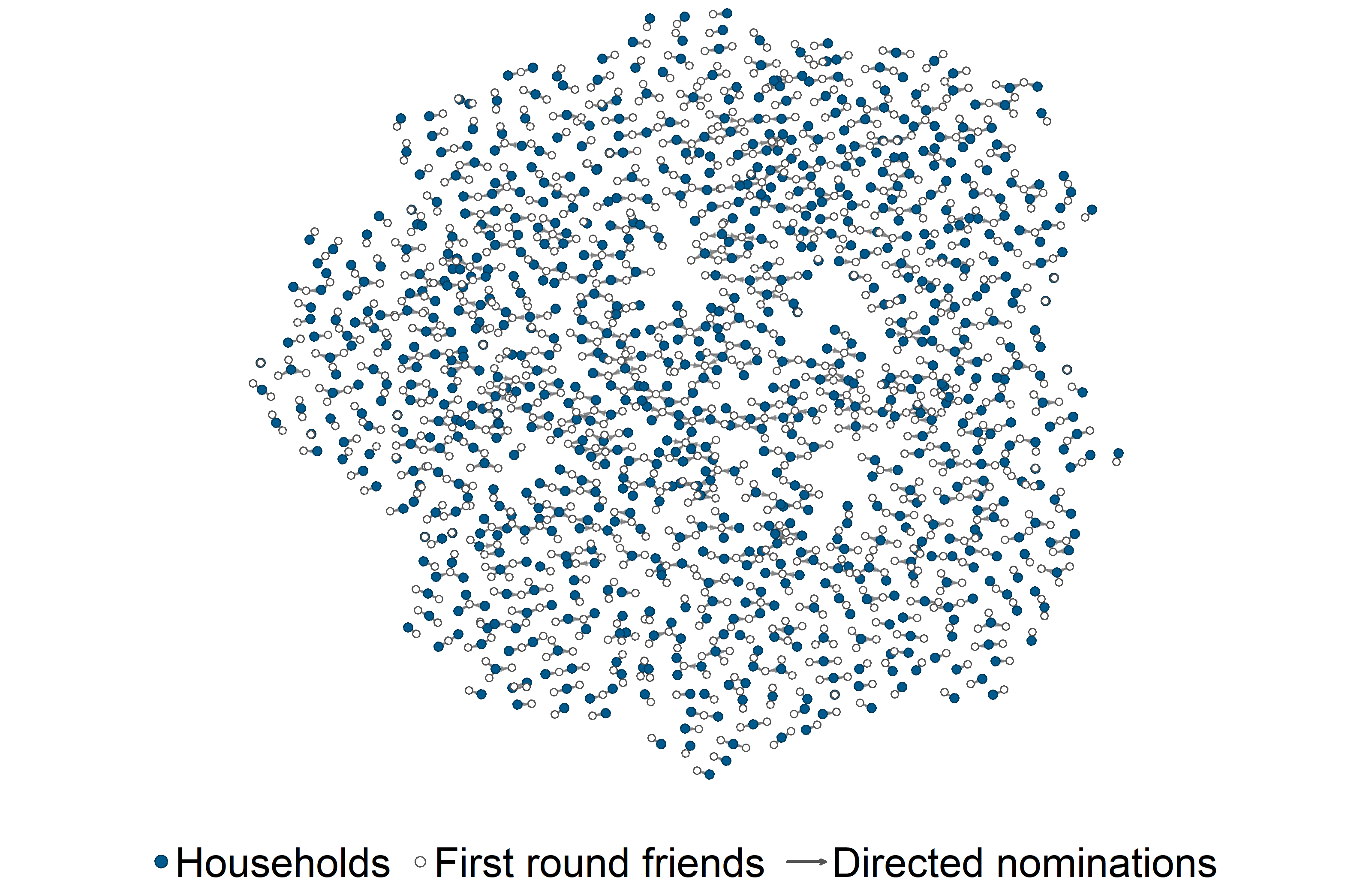}
\caption{Illustration of the heterogeneous degrees in the social network \citep{cai2015social}}
\label{fig:PICCInsuranceNetwork}
\end{figure}

When the variance is large, a survey or experiment will not have enough power to estimate the effect of interest. 
The literature has therefore studied numerous alternative estimators and variance reduction techniques to reduce the variance of the Horvitz-Thompson estimator, including the Hajek estimator \citep{basu1971essay}, the normalized IPW estimator \citep{tukey1956conditional}, the augmented IPW (AIPW) estimator \citep{cassel1976some, robins1994estimation}, the targeted maximum likelihood estimator \citep{van2006targeted}, and the use of inverse probabilities as control variates \citep{hesterberg1988advances, hesterberg1995weighted}. 
These different variance reduction techniques share the same idea: they all correct the Horvitz-Thompson estimator by canceling part of its randomness.
This is exactly what a control variate does. 

In this paper, we provide a unified interpretation of these variance reduction techniques through the lens of control variates. 
Furthermore, we construct optimal control variates that further reduce the variance compared to the above estimators in finite sample. 

A control variate is a random variable that is correlated with the primary variable, and it reduces the variance of the primary variable by canceling part of its randomness. 
In the context of survey sampling, the conventional use of control variates is to construct them from auxiliary random variables as in \citet{deville1992calibration, fieller1954sampling, hickernell2005control}. 
In contrast, we use the sampling indicators as the control variates, whose randomness is fully determined by the random sampling. 
More specifically, in a population of $n$ units, we consider control variates $\widehat{X}$ of the form $\widehat{X} = \frac{1}{n} \sum_{i=1}^n a_i \bI\{W_i=1\}$, where $a_i$ is a constant and $\bI{W_i=1}$ is the sampling indicator that equals to $1$ if unit $i$ is sampled.

Through the lens of control variates, we show that the Hajek estimator, the normalized IPW estimator, the AIPW estimator, and the TMLE estimator are all control variate estimators using specific constants $a_i = \frac{1}{\Pr\{W_i=1\}}$. 
In contrast, we go beyond these specific constants and derive the optimal control variates.
To do so, we combine a design-based perspective with a model-based perspective \citep{sarndal1978design, little2004model} and formulate the choice of optimal control variates as a decision-making problem under uncertainty \citep{wald1949statistical, wu1981robustness}.
We show that the optimal control variate is characterized by the principal eigenvector of a matrix that depends on both the design-based sampling structure and the model-based outcome uncertainty, and the resulting expected variance reduction is characterized by the largest eigenvalue of this matrix. 

We then generalize the control variate framework to causal inference under network interference where the exposure mapping is correctly specified.
Because the exposure probabilities depend on the heterogeneous degrees of the network, this is naturally a setting where the exposure probabilities are heterogeneous and small. 
We characterize the optimal control variates through a nonconvex quadratic optimization problem whose objective depends on a matrix that depends on the design-based sampling structure and the model-based outcome uncertainty.
We provide a $\frac{1}{2}$-approximate solution to the quadratic optimization problem and an alternating local search heuristic.
We use the $\frac{1}{2}$-approximate solution to initialize the alternating local search heuristic, which monotonically converges to a local maximum. 
In the special case without interference, the optimal control variates are characterized in closed form by the two leading eigenvectors of this matrix.

\subsection*{Roadmap}
The paper is structured as follows. 
In Section~\ref{sec:SurveySampling}, we introduce the survey sampling setup and characterize the optimal control variates for survey sampling. 
In particular, we provide a unified interpretation to Hajek, normalized, AIPW, and TMLE estimators through the lens of control variates in Section~\ref{sec:Connections}.
In Section~\ref{sec:NetworkInterference}, we introduce the network interference setup, characterize the optimal control variates for network interference, and discuss the computation of the optimal control variates. 
In Section~\ref{sec:RealData}, we apply our control variate estimators to survey data from the Swiss Environmental Panel and to experimental data from the Chinese social network.
In Section~\ref{sec:Simulations}, we use synthetic data to demonstrate the effectiveness of our control variate estimators.
In Section~\ref{sec:Conclusions} we conclude the paper and point out three limitations that each leads to a future research direction.
All mathematical details are deferred to the Appendix.

\section{Survey Sampling}
\label{sec:SurveySampling}

\subsection{Sampling Design}
\label{sec:Setup}

We start with the following survey sampling problem.
Let there be $n$ units in a finite population, where each unit is associated with an outcome $Y_1, Y_2, ..., Y_n$.
We collect them into vector form $\bm{Y} = (Y_1,Y_2,...,Y_n)^\top$.
We adopt a design-based perspective and condition on the outcomes \citep{sarndal1978design, little2004model}.
We are interested in estimating the mean of the finite population,
\begin{align}
\mu_n = \frac{1}{n} \sum_{i=1}^n Y_i. \label{eqn:Mean}
\end{align}

We follow a randomized sampling design to draw a random subset of the $n$ units to estimate $\mu_n$ the mean of the finite population. 
We use $W_i = 1$ to stand for a unit being sampled, and $W_i = 0$ otherwise. 
The randomized sampling design determines the random sampling indicators $\bm{W} = (W_1, W_2, ..., W_n)^\top$ through a known joint probability distribution $\cW$. 
We denote the marginal sampling probabilities as $\pi_i = \Pr(W_i=1)$ for any $i \in [n]:= \{1,2,...,n\}$, and denote the following $n \times n$ diagonal matrix $\bm{\Pi} = \bm{\mathrm{diag}}(\pi_1, \pi_2, ..., \pi_n)$.
Denote the covariance matrix $\bm{\Sigma} = \Var_{\cW}(\bm{W})$.

The joint probability distribution $\cW$ may have some dependency structure. 
To describe this dependency structure, we define the dependency neighborhood as follows.

\begin{definition}[Dependency Neighborhood]
\label{defn:DependencyNeighborhood}
A collection of random variables $\{W_i\}_{i\in[n]}$ has dependency neighborhoods $\cN_i \subseteq [n]$ if $W_i$ is independent of $\{W_j\}_{j \notin \cN_i}$.
\end{definition}

In this paper, we consider sampling designs $\cW$ that satisfy Assumption~\ref{asp:DesignProbabilities} below.

\begin{assumption}
\label{asp:DesignProbabilities}
The sampling design $\cW$ satisfies two conditions:
\begin{enumerate}[label=(\roman*)]
\item \textbf{(Positivity)} The marginal sampling probabilities converge to zero at a slow rate, that is, there exist constants $\underline{\pi} > 0$ and $0 \leq \beta < 1$ such that for $i \in [n]$,
\begin{align*}
\pi_i \geq \underline{\pi} n^{-\beta} > 0.
\end{align*}
\item \textbf{(Dependency Neighborhood)} The correlation between the $n$ sampling indicators $\bm{W}$ are restricted to some neighborhoods, that is, there exist constants $\overline{d} > 0$ and $0 \leq \alpha < 1$ such that for $i \in [n]$, $W_i$ is independent of $\{W_j\}_{j \notin \cN_i}$, and
\begin{align*}
|\cN_i| \leq \overline{d} n^\alpha.
\end{align*}
\end{enumerate}
\end{assumption}

Assumption~\ref{asp:DesignProbabilities}-(i) is weaker than the standard assumptions that usually appear in the design based sampling and causal inference literature \citep{ding2023first, imbens2015causal, wager2020stats}.
Assumption~\ref{asp:DesignProbabilities}-(ii) is a mild assumption to describe weakly dependent random variables \citep{ross2011fundamentals}.
We will explicitly specify the required upper bounds on $\alpha$ and $\beta$ when we present our main results. 
We provide Examples~\ref{exa:Bernoulli}~--~\ref{exa:Stratified} below to illustrate the above notations and assumptions.
Additional examples are provided in Section~\ref{sec:AdditionalExamples} in the Appendix.

\begin{example}[Bernoulli Sampling]
\label{exa:Bernoulli}
For each unit $i \in [n]$, we independently sample unit $i$ with probability $\pi_i$.
The covariance matrix is given by
\begin{align*}
\bm{\Sigma} = 
\begin{bmatrix}
\pi_1 (1 - \pi_1) & 0                 & \dots  & 0                 \\
0                 & \pi_2 (1 - \pi_2) & \dots  & 0                 \\
\vdots            & \vdots            & \ddots & \vdots            \\
0                 & 0                 & \dots  & \pi_n (1 - \pi_n)
\end{bmatrix}.
\end{align*}
Bernoulli sampling satisfies Assumption~\ref{asp:DesignProbabilities}-(i) if $\pi_i \geq \underline{\pi} > 0$. 
It also satisfies Assumption~\ref{asp:DesignProbabilities}-(ii) because $|\cN_i| = 1$ for $i \in [n]$.
\hfill \halmos
\end{example}

\begin{example}[Sampling without Replacement]
\label{exa:WithoutReplacement}
We randomly sample exactly $n_1$ of the units from a total of $n$ units.
The covariance matrix is given by
\begin{align*}
\bm{\Sigma} = 
\begin{bmatrix}
\frac{n_1(n - n_1)}{n^2} & - \frac{n_1(n - n_1)}{n^2 (n-1)} & \dots  & - \frac{n_1(n - n_1)}{n^2 (n-1)} \\
- \frac{n_1(n - n_1)}{n^2 (n-1)} & \frac{n_1(n - n_1)}{n^2} & \dots  & - \frac{n_1(n - n_1)}{n^2 (n-1)} \\
\vdots            & \vdots            & \ddots & \vdots            \\
- \frac{n_1(n - n_1)}{n^2 (n-1)} & - \frac{n_1(n - n_1)}{n^2 (n-1)} & \dots  & \frac{n_1(n - n_1)}{n^2} \\
\end{bmatrix}.
\end{align*}
Sampling without replacement does not satisfy Assumption~\ref{asp:DesignProbabilities}-(i) if the number of sampled units is small (e.g., $n_1 < n^{1-\beta}$);
it satisfies Assumption~\ref{asp:DesignProbabilities}-(i) if the number of sampled units is large (e.g., $n_1 \geq n^{1-\beta}$).
It does not satisfy Assumption~\ref{asp:DesignProbabilities}-(ii) because $|\cN_i| = n$ for $i \in [n]$. 
\hfill \halmos
\end{example}

\begin{example}[Stratified Sampling]
\label{exa:Stratified}
We partition $n$ units into equal size strata where the size of each stratum is $s$.
Within each stratum, we sample exactly $1$ unit from this stratum uniformly at random;
across different strata, sampling is independent.
The covariance matrix is given by
\begin{align*}
\bm{\Sigma} = 
\begin{bmatrix}
\bm{\Sigma}_{s} & 0               & \dots  & 0               \\
0               & \bm{\Sigma}_{s} & \dots  & 0               \\
\vdots          & \vdots          & \ddots & \vdots          \\
0               & 0               & \dots  & \bm{\Sigma}_{s}
\end{bmatrix},
\qquad \text{where} \qquad
\bm{\Sigma}_{s} = 
\begin{bmatrix}
\frac{(s - 1)}{s^2} & - \frac{1}{s^2}     & \dots  & - \frac{1}{s^2}     \\
- \frac{1}{s^2}     & \frac{(s - 1)}{s^2} & \dots  & - \frac{1}{s^2}     \\
\vdots              & \vdots              & \ddots & \vdots              \\
- \frac{1}{s^2}     & - \frac{1}{s^2}     & \dots  & \frac{(s - 1)}{s^2} \\
\end{bmatrix}.
\end{align*}
Stratified sampling satisfies Assumption~\ref{asp:DesignProbabilities}-(i) if the size of each stratum is small (e.g., $s \leq n^{1-\beta}$);
it does not satisfy Assumption~\ref{asp:DesignProbabilities}-(i) if the size of each stratum is large (e.g., $s > n^{1-\beta}$).
It satisfies Assumption~\ref{asp:DesignProbabilities}-(ii) if the size of each stratum is small (e.g., $s < n^\alpha$);
it does not satisfy Assumption~\ref{asp:DesignProbabilities}-(ii) if the size of each stratum is large (e.g., $s \geq n^\alpha$).
\hfill \halmos
\end{example}

\subsection{Control Variate Estimators for Survey Sampling}
\label{sec:CVEstimator}

One of the most popular ways of estimating $\mu_n$ is to consider the Horvitz-Thompson estimator, which is also referred to as the unnormalized IPW estimator \citep{horvitz1952generalization}.
The Horvitz-Thompson estimator is defined as
\begin{align*}
\widehat{\mu}^{\HT} = \frac{1}{n} \sum_{i=1}^n \frac{Y_i \bI\{W_i=1\}}{\pi_i}.
\end{align*}
It is easy to see that the Horvitz-Thompson estimator is unbiased, that is, 
\begin{align*}
\bE_{\cW}\big[\widehat{\mu}^{\HT}\big] = \mu_n.
\end{align*}
However, the Horvitz-Thompson estimator is recognized to have a large variance, especially when the marginal probabilities $\pi_i$ are small.
To address this challenge, we draw on the control variate techniques common in the simulation literature \citep{asmussen2007stochastic, glasserman2004monte, owen2013monte, ross2013simulation}.
We take the baseline estimator $\widehat{\mu}^{\HT}$ and propose a family of control variate estimators to reduce the variance.

We define a control variate $\widehat{X}$ to be any random variable that may be correlated with the baseline estimator $\widehat{\mu}^{\HT}$.
A control variate estimator parameterized by coefficient $\gamma$ can then be defined as
\begin{align*}
\widehat{\mu}^{\CV}(\gamma) = \widehat{\mu}^{\HT} - \gamma\Big(\widehat{X} - \bE\big[\widehat{X}\big] \Big). 
\end{align*}
When $\gamma$ is a constant coefficient, the control variate estimator is unbiased, that is, 
\begin{align*}
\bE_{\cW}\big[ \widehat{\mu}^{\CV}(\gamma) \big] = \bE_{\cW}\big[ \widehat{\mu}^{\HT} \big] - \gamma \bE_\cW\big[\widehat{X} - \bE[\widehat{X}]\big] = \mu_n.
\end{align*}
If the coefficient $\gamma$ is selected properly, the control variate estimator has a smaller variance than the baseline Horvitz-Thompson estimator.
It is well known that the variance minimizing coefficient $\gamma^*$ for any generic control variate $\widehat{X}$ is given by
\begin{align}
\gamma^* = \frac{\Cov\big(\widehat{\mu}^{\HT}, \widehat{X}\big)}{\Var\big(\widehat{X}\big)}. \label{eqn:gammaexpression}
\end{align}

In this paper, we propose a family of control variates given in the following form
\begin{align}
\widehat{X} = \frac{1}{n} \sum_{i=1}^n a_i \bI\{W_i=1\}, \label{eqn:ControlVariate}
\end{align}
where $a_1, a_2, ..., a_n$ are any arbitrary constants that we can choose.
We collect these constants into vector form $\bm{a} = (a_1,a_2,...,a_n)^\top$ and refer to them as the ``basis'' of the control variates.
We use the sampling indicators as the control variates, and for any given basis $\bm{a}$, the control variates are fully determined by the random sampling indicators $\bm{W}$.
This family of control variates is in contrast to the works that use auxiliary variables such as covariates as the control variates \citep{deville1992calibration, fieller1954sampling, hickernell2005control}.

For any control variate using basis $\bm{a}$, the mean value of this control variate is given by $\bE_{\cW}\big[\widehat{X}\big] = \frac{1}{n} \sum_{i=1}^n a_i \pi_i$ which is known in advance. 
So the control variate estimator is well defined, and the variance minimizing coefficient $\gamma^*$ is given by Lemma~\ref{lem:Optimalgamma} below.
We present Lemma~\ref{lem:Optimalgamma} using succinct matrix notations.

\begin{lemma}[Variance Minimizing Coefficient]
\label{lem:Optimalgamma}
In the survey sampling setting, the variance minimizing coefficient $\gamma^*$ is given by
\begin{align*}
\gamma^* = \frac{\bm{Y}^\top \bm{\Pi}^{-1} \bm{\Sigma} \bm{a}}{\bm{a}^\top \bm{\Sigma} \bm{a}},
\end{align*}
and the variance of the control variate estimator under this coefficient is given by
\begin{align*}
\Var\big( \widehat{\mu}^{\CV}(\gamma^*) \big) = \frac{1}{n^2} \bigg( \bm{Y}^\top \bm{\Pi}^{-1} \bm{\Sigma} \bm{\Pi}^{-1} \bm{Y} - \frac{\big(\bm{Y}^\top \bm{\Pi}^{-1} \bm{\Sigma} \bm{a} \big)^2}{\bm{a}^\top \bm{\Sigma} \bm{a}} \bigg).
\end{align*}
\end{lemma}

The proof of Lemma~\ref{lem:Optimalgamma} is given in Section~\ref{sec:MissingProofs}.
Intuitively, we try to use the randomness in $\widehat{X}$ to explain as much randomness in $\widehat{\mu}^{\HT}$ as possible.
The variance $\Var( \widehat{\mu}^{\CV}(\gamma^*))$ is given by the residual from projecting the vector $\bm{\Pi}^{-1} \bm{Y}$ onto the direction of $\bm{a}$.
Because the variance minimizing coefficient $\gamma^*$ depends on unknown outcomes $\bm{Y}$, we cannot directly observe $\gamma^*$; instead, we have to estimate $\gamma^*$ from the data.

\subsection{Estimating the Variance Minimizing Coefficient}
\label{sec:EstimatingCoefficient}

To estimate $\gamma^*$ from the data, we can replace the unknown outcomes $\bm{Y}$ by their Horvitz-Thompson estimates.
Let $\widehat{Y}_i = Y_i \frac{\bI\{W_i=1\}}{\pi_i}$ and collect $\widehat{\bm{Y}} = (\widehat{Y}_1, ..., \widehat{Y}_n)$.
Then the coefficient $\gamma^*$ can be estimated using
\begin{align}
\widehat{\gamma}^\HT = \frac{\widehat{\bm{Y}}^\top \bm{\Pi}^{-1} \bm{\Sigma} \bm{a}}{\bm{a}^\top \bm{\Sigma} \bm{a}}. \label{eqn:gamma-hat}
\end{align}
It is easy to see that $\widehat{Y}_i$ is an unbiased estimator of $Y_i$, that is, $\bE_{\cW}[\widehat{Y}_i] = Y_i$. 
So we have that $\widehat{\gamma}^\HT$ is an unbiased estimator of $\gamma^*$, that is, $\bE_{\cW}[\widehat{\gamma}^\HT] = \gamma^*$.

Note that, $\widehat{\gamma}^\HT$ obtained in this way is not a constant, but a random variable correlated with the control variate $\widehat{X}$.
Using the estimated coefficient $\widehat{\gamma}^\HT$ introduces a bias to the control variate estimator, that is, 
\begin{align*}
\bE\big[ \widehat{\mu}^\CV(\widehat{\gamma}^\HT) \big] \ne \bE\big[ \widehat{\mu}^\CV(\gamma^*) \big] = \mu_n.
\end{align*}
Nonetheless, we can show that the bias is small and that $\widehat{\mu}^\CV(\widehat{\gamma}^\HT)$ is asymptotically equivalent to $\widehat{\mu}^\CV(\gamma^*)$, as long as the basis $\bm{a}$ satisfies certain regularity conditions. 
See Assumption~\ref{asp:RegularBasis} below.

\begin{assumption}
\label{asp:RegularBasis}
The basis $\bm{a}$ satisfies two conditions:
\begin{enumerate}[label=(\roman*)]
\item \textbf{(Boundedness)} There exists a constant $\overline{a}$ such that for each $i \in [n]$,
\begin{align*}
\vert a_i \vert \leq \overline{a}.
\end{align*}
\item \textbf{(Non-degeneracy)} There exists a constant $\underline{c}_{\Sigma}$ such that 
\begin{align*}
\bm{a}^\top \bm{\Sigma} \bm{a} \geq \underline{c}_{\Sigma} n.
\end{align*}
\end{enumerate}
\end{assumption}

Assumption~\ref{asp:RegularBasis}-(i) is a standard boundedness assumption. 
Assumption~\ref{asp:RegularBasis}-(ii) is a standard non-degeneracy assumption on the control variate $\widehat{X}$.
Note that we have $\Var(\widehat{X}) = n^{-2} \bm{a}^\top \bm{\Sigma} \bm{a}$, and therefore Assumption~\ref{asp:RegularBasis}-(ii) is equivalent to requiring that $\Var(\widehat{X}) \geq \underline{c}_{\Sigma} n^{-1}.$
Both assumptions are commonly made in the asymptotic statistics literature \citep{van1996weak, wooldridge2010econometric}. 
Assumption~\ref{asp:RegularBasis}-(i) can be trivially satisfied, for example, by choosing $\bm{a} = (1, 1, ..., 1)^\top$. 
Assumption~\ref{asp:RegularBasis}-(ii) can be satisfied if $\bm{a}$ is not aligned with a null or near-null direction of $\bm{\Sigma}$. 
For example, Assumption~\ref{asp:RegularBasis}-(ii) holds under Bernoulli sampling (Example~\ref{exa:Bernoulli}) and by choosing $\bm{a} = (1, 1, ..., 1)^\top$.
Intuitively, putting both assumptions together, Assumption~\ref{asp:RegularBasis} ensures that the basis $\bm{a}$ does not contain extreme values that are disproportionately large, or contain too many extreme values that are disproportionately small.
Under Assumption~\ref{asp:RegularBasis}, we establish Theorem~\ref{thm:AsymptoticBehavior} below.

\begin{theorem}
\label{thm:AsymptoticBehavior}
Under Assumptions~\ref{asp:DesignProbabilities} and~\ref{asp:RegularBasis} where in Assumption~\ref{asp:DesignProbabilities} we assume $3\alpha + 4\beta < 1$, and assuming the outcomes $\vert Y_i \vert \leq \overline{y}$ are all bounded, the control variate estimator using the estimated coefficient $\widehat{\mu}^\CV(\widehat{\gamma}^\HT)$ is a consistent estimator of $\mu_n$, that is, 
\begin{align*}
\lim_{n \to +\infty} \widehat{\mu}^\CV(\widehat{\gamma}^\HT) \xrightarrow{p} \mu_n.
\end{align*}
Additionally, if in Assumption~\ref{asp:DesignProbabilities} we assume $4\alpha + 4\beta < 1$, then the asymptotic variance of the control variate estimator using the estimated coefficient $\Var(\widehat{\mu}^\CV(\widehat{\gamma}^\HT))$ is the same as the asymptotic variance of the control variate estimator using the optimal coefficient $\Var(\widehat{\mu}^\CV(\gamma^*))$, that is, 
\begin{align*}
\lim_{n \to +\infty} n \Big( \Var\big( \widehat{\mu}^\CV(\widehat{\gamma}^\HT) \big) - \Var\big( \widehat{\mu}^\CV(\gamma^*) \big) \Big) \to 0.
\end{align*}
Additionally, if in Assumption~\ref{asp:DesignProbabilities} we assume $7\alpha + 6\beta < 1$ and further assuming there exists a constant $\underline{c} > 0$ such that for sufficiently large $n$, $n \Var\big( \widehat{\mu}^{\CV}(\gamma^*) \big) \geq \underline{c}$, then the control variate estimator using the estimated coefficient $\widehat{\mu}^\CV(\widehat{\gamma}^\HT)$ is asymptotically normal, that is, 
\begin{align*}
\lim_{n \to +\infty} \frac{\widehat{\mu}^{\CV}(\widehat{\gamma}^{\HT}) - \mu_n}
{\sqrt{\Var\big( \widehat{\mu}^{\CV}(\widehat{\gamma}^{\HT}) \big)}}
\xrightarrow{d} \cN(0,1).
\end{align*}
\end{theorem}

The proof of Theorem~\ref{thm:AsymptoticBehavior} is given in Section~\ref{sec:MissingProofs}. 
As we have stronger results, from consistency to asymptotic variance equivalence and to asymptotic normality, the requirement on the dependency neighborhoods becomes correspondingly stronger, strengthening from $3\alpha + 4\beta < 1$ to $4\alpha + 4\beta < 1$ and to $7\alpha + 6\beta < 1$. 
Taken together, Theorem~\ref{thm:AsymptoticBehavior} shows that estimating $\gamma^*$ by $\widehat{\gamma}^{\HT}$ does not affect the first order asymptotic properties of the control variate estimator. 
In particular, estimating $\gamma^*$ by $\widehat{\gamma}^{\HT}$ leads to the same variance reduction asymptotically. 
The magnitude of variance reduction is determined by the alignment between the weighted outcomes $\bm{\Pi}^{-1}\bm{Y}$ and the basis $\bm{a}$. 
In Section~\ref{sec:OptimalControlVariate} below, we study how to choose a basis $\bm{a}$ for the optimal variance reduction.

\subsection{The Optimal Control Variate}
\label{sec:OptimalControlVariate}

Now we turn our attention to finding the optimal basis $\bm{a}$ that minimizes $\Var\big( \widehat{\mu}^{\CV}(\gamma^*) \big)$ the variance of the control variate estimator for fixed $n$. 
If the outcomes $\bm{Y}$ were known to take values $Y_i = y_i$ for any $i \in [n]$, the problem 
\begin{align*}
\min_{\bm{a}} \ \Var\big( \widehat{\mu}^{\CV}(\gamma^*) \big) \ = \ \min_{\bm{a}} \ \frac{1}{n^2} \bigg( \bm{y}^\top \bm{\Pi}^{-1} \bm{\Sigma} \bm{\Pi}^{-1} \bm{y} - \frac{\big(\bm{y}^\top \bm{\Pi}^{-1} \bm{\Sigma} \bm{a} \big)^2}{\bm{a}^\top \bm{\Sigma} \bm{a}} \bigg)
\end{align*}
can be easily solved by setting $\bm{a} = \bm{\Pi}^{-1} \bm{y}$.
In this case, the variance $\Var\big( \widehat{\mu}^{\CV}(\gamma^*) \big)$ can be reduced to exactly zero.

However, the above choice of $\bm{a}$ is infeasible because the outcomes $\bm{Y}$ are unknown. 
Therefore, we find the optimal basis $\bm{a}$ under uncertainty of the outcomes, which leads to a decision making problem under uncertainty.
There are various statistical decision rules for making decisions under uncertainty \citep{lindley1972bayesian, manski2004statistical, rubin1978bayesian, savage1951theory, stoye2009minimax, wald1949statistical, wu1981robustness, zhao2023adaptive}. 
We adopt a stochastic optimization perspective \citep{zhao2024experimental} and model each unknown outcome $Y_i$ to be independent and identically distributed (i.i.d.) samples from an unknown distribution $\cY$.
See Assumption~\ref{asp:iid} below. 

\begin{assumption}[i.i.d. Outcomes]
\label{asp:iid}
We assume that for each $i \in [n]$, the outcome $Y_i$ is i.i.d. sampled from an unknown distribution $\cY$. 
\end{assumption}

Denote the population mean as $\mu = \bE_{Y_i \sim \cY}[Y_i]$ and the population variance as $\sigma^2 = \Var_{Y_i \sim \cY}(Y_i)$.
Denote $\cY^n$ to be the joint probability distribution of the outcomes $\bm{Y}$.
Under Assumption~\ref{asp:iid}, we formulate the following optimization problem
\begin{align}
\bm{a}^* \in \argmin_{\bm{a}} \bE_{\bm{Y} \sim \cY^n} \Big[ \Var\big( \widehat{\mu}^{\CV}(\gamma^*) \big) \Big]. \label{eqn:OptFormulation}
\end{align}
The optimal control variate uses the basis that solves \eqref{eqn:OptFormulation}, which is given by Theorem~\ref{thm:OptimalControlVariate} below.
We present Theorem~\ref{thm:OptimalControlVariate} using succinct notations.
Recall that $\bm{\Sigma}$ is symmetric and positive semidefinite, we can define its square root matrix as $\bm{\Sigma}^{\frac{1}{2}}$ and its pseudo-inverse square root matrix as $\bm{\Sigma}^{-\frac{1}{2}}$.
See Section~\ref{sec:AdditionalDefn} for their detailed definitions and properties.

\begin{theorem}[Optimal Control Variate]
\label{thm:OptimalControlVariate}
Let $\bm{M}$ be a symmetric matrix defined as
\begin{align*}
\bm{M} = \bm{\Sigma}^{\frac{1}{2}} \bm{\Pi}^{-1} \big(\sigma^2 \bm{I}_n + \mu^2 \bm{1}_n \bm{1}_n^\top\big) \bm{\Pi}^{-1} \bm{\Sigma}^{\frac{1}{2}},
\end{align*}
where $\bm{I}_n$ is a $n \times n$ identity matrix and $\bm{1}_n$ is a $n$-dimensional vector with each element equal to $1$. 
Let $\lambda_1(\bm{M})$ be the largest eigenvalue of $\bm{M}$, and let $\bm{u}_1(\bm{M})$ be the eigenvector that corresponds to $\lambda_1(\bm{M})$.
Under Assumption~\ref{asp:iid}, the optimal basis is given by 
\begin{align*}
\bm{a}^* = c \bm{\Sigma}^{-\frac{1}{2}} \bm{u}_1(\bm{M}),
\end{align*}
where $c \ne 0$ is any constant. 
And in expectation, the optimal variance reduction is given by
\begin{align*}
\bE_{\bm{Y} \sim \cY^n} \Big[ \Var\big( \widehat{\mu}^{\HT} \big) - \Var\big( \widehat{\mu}^{\CV}(\gamma^*) \big) \Big] = \frac{\lambda_1(\bm{M})}{n^2}.
\end{align*}
\end{theorem}

The proof of Theorem~\ref{thm:OptimalControlVariate} is given in Section~\ref{sec:MissingProofs}.
Theorem~\ref{thm:OptimalControlVariate} combines a design-based perspective and a model-based perspective.
We assume a model for the unknown outcomes $\bm{Y}$ to guide the choice of basis.
But we rely on the sampling design as the only source of randomness when we estimate the sample mean.
We point out that Assumption~\ref{asp:iid} is only required to show optimality of the basis $\bm{a}^*$.
The basis $\bm{a}^*$ is still well defined without Assumption~\ref{asp:iid}.

In Theorem~\ref{thm:OptimalControlVariate}, the optimal basis $\bm{a}^*$ depends on the unknown quantities $\mu$ and $\sigma^2$.
But since $\bm{u}_1(\bm{M})$ is an eigenvector and thus scale-invariant to $\bm{M}$, we can always normalize $\bm{M}$ by $\sigma^2$ if $\sigma^2 > 0$ so the optimal basis $\bm{a}^*$ depends on the signal-to-noise ratio $\frac{\vert\mu\vert}{\sigma}$.
We can perform sample splitting to estimate the signal-to-noise ratio, and then use it to find the optimal basis $\bm{a}^*$.
We next show that the optimal basis $\bm{a}^*$ satisfies Assumption~\ref{asp:RegularBasis}-(ii).
The proof of Proposition~\ref{prop:RegularOpta} is given in Section~\ref{sec:MissingProofs}. 

\begin{proposition}
\label{prop:RegularOpta}
When $c = n^\frac{1}{2}$, the optimal basis $\bm{a}^*$ in Theorem~\ref{thm:OptimalControlVariate}, as given by $\bm{a}^* = n^\frac{1}{2} \bm{\Sigma}^{-\frac{1}{2}} \bm{u}_1(\bm{M})$, satisfies the following equality
\begin{align*}
\bm{a}^{*\top} \bm{\Sigma} \bm{a}^* = n.
\end{align*}
\end{proposition}

Theorem~\ref{thm:OptimalControlVariate} provides a refinement of a classical recommendation in the simulations literature: it is recommended to use the importance weights (i.e., the inverse probabilities $\frac{1}{\pi_i}$ for $i \in [n]$) as the control variates to reduce variance, whenever they are known (\citet{hesterberg1988advances, hesterberg1995weighted}, Section 9.2 of \citet{owen2013monte}).
In Theorem~\ref{thm:OptimalControlVariate}, we further extend this recommendation by studying what is the optimal form of the importance weights to use as the control variates, rather than using the importance weights themselves. 

In Section~\ref{sec:Connections} below, we study the quality of using the importance weights themselves as the control variates. 
This essentially leads to using $a_i = \frac{1}{\pi_i}$ for $i \in [n]$ as the basis, which we refer to as the inverse probability weighting (IPW) basis. 
This perspective of IPW basis allows us to connect the control variate estimators to other popular estimators in the literature, such as the Hajek estimator, the normalized estimator, and the AIPW estimator.

\subsection{Connections to Hajek, Normalized, and AIPW Estimators}
\label{sec:Connections}

\subsubsection*{Hajek Estimator.}

One popular alternative to the Horvitz-Thompson estimator is the Hajek estimator, which is also referred to as the self-normalized IPW estimator \citep{basu1971essay}.
The Hajek estimator is defined as
\begin{align*}
\widehat{\mu}^\Hajek = \frac{\widehat{\mu}^\HT}{\widehat{1}}, \quad \text{where} \quad \widehat{1} = \frac{1}{n} \sum_{i=1}^n \frac{\bI\{W_i=1\}}{\pi_i}. 
\end{align*}
The Hajek estimator normalizes the Horvitz-Thompson estimator by $\widehat{1}$, which serves as an estimate of one \citep{gao2025causal}. 
This self-normalization reduces the impact of observations with extreme inverse probabilities and is therefore generally recognized to have a smaller variance than the Horvitz-Thompson estimator, although such an improvement is not always true.

There are two connections between the Hajek estimator and the control variate estimator.
First, we can directly rewrite the Hajek estimator as
\begin{align*}
\widehat{\mu}^\Hajek = \widehat{\mu}^\HT - \widehat{\mu}^\Hajek \Big( \frac{1}{n} \sum_{i=1}^n \frac{\bI\{W_i=1\}}{\pi_i} - 1 \Big).
\end{align*}
Therefore, the Hajek estimator can be viewed as a control variate estimator using the IPW basis $a_i = \frac{1}{\pi_i}$ for $i \in [n]$ and coefficient $\widehat{\mu}^\Hajek$.

Second, if we view the denominator $\widehat{1}$, the estimated one, as a variable, we can make a first order approximation to the Hajek estimator by using Taylor expansion around $\widehat{1} \approx 1$. 
It is given as
\begin{align*}
\widehat{\mu}^\Hajek \approx \widehat{\mu}^\HT - \widehat{\mu}^\HT \Big( \frac{1}{n} \sum_{i=1}^n \frac{\bI\{W_i=1\}}{\pi_i} - 1 \Big).
\end{align*}
Therefore, the first order approximation to the Hajek estimator can be viewed as a control variate estimator using the IPW basis $a_i = \frac{1}{\pi_i}$ for $i \in [n]$ and coefficient $\widehat{\mu}^\HT$.

Using the IPW basis, the control variate estimator using the estimated coefficient $\widehat{\gamma}^\HT$ is 
\begin{align*}
\widehat{\mu}^\CV = \widehat{\mu}^\HT - \widehat{\gamma}^\HT \Big( \frac{1}{n} \sum_{i=1}^n \frac{\bI\{W_i=1\}}{\pi_i} - 1 \Big).
\end{align*}
Under Bernoulli sampling and when all the sampling probabilities are equal $\pi_i = \pi$, we have
\begin{align}
\widehat{\gamma}^\HT = \frac{1}{n} \sum_{i=1}^n \frac{Y_i \bI\{W_i=1\}}{\pi} = \widehat{\mu}^\HT. \label{eqn:CoefficientsEqual}
\end{align}
This shows that, the first order approximation of the Hajek estimator is the same as the control variate estimator using the IPW basis and the estimated coefficient $\widehat{\gamma}^{\HT}$, under Bernoulli sampling and when all the sampling probabilities are equal.

\subsubsection*{Normalized Estimators.}
Extending the idea of Hajek estimator, \citet{tukey1956conditional} proposed a family of normalized estimators in the form of 
\begin{align*}
\widehat{\mu}^{\Normalized} = \frac{\widehat{\mu}^\HT}{\lambda \widehat{1} + (1 - \lambda)}, 
\end{align*}
where $\lambda \in \bR$ is a parameter that can take values outside of $[0,1]$.
Two special cases are $\widehat{\mu}^{\Normalized} = \widehat{\mu}^{\HT}$ when $\lambda = 0$ and $\widehat{\mu}^{\Normalized} = \widehat{\mu}^{\Hajek}$ when $\lambda = 1$.
For this family of normalized estimators, we can make a first order approximation by using Taylor expansion around $\widehat{1} \approx 1$, which gives
\begin{align*}
\widehat{\mu}^{\Normalized} \approx \widehat{\mu}^{\HT} - \lambda \widehat{\mu}^{\HT} \Big( \frac{1}{n} \sum_{i=1}^n \frac{\bI\{W_i=1\}}{\pi_i} - 1 \Big).
\end{align*}
Therefore, the first order approximations to this family of normalized estimators can also be viewed as a control variate estimator using the IPW basis $a_i = \frac{1}{\pi_i}$ for $i \in [n]$. 

Building on this family of normalized estimators, \citet{khan2023adaptive} studied the optimal choice of $\lambda$ under Bernoulli sampling and from a sampling-based perspective where both the outcomes $Y_i$ and the sampling probabilities $\pi_i$ are i.i.d. sampled. 
They further proposed an adaptively normalized estimator to estimate $\lambda$, which yields
\begin{align*}
\widehat{\mu}^{\Normalized} \approx \widehat{\mu}^{\HT} - \widehat{\gamma}^{\Normalized} \Big( \frac{1}{n} \sum_{i=1}^n \frac{\bI\{W_i=1\}}{\pi_i} - 1 \Big),
\qquad \text{where} \qquad 
\widehat{\gamma}^{\Normalized} = \frac{\sum_{i=1}^n Y_i \bI\{W_i=1\} \frac{1-\pi_i}{\pi_i^2}}{\sum_{i=1}^n \bI\{W_i=1\} \frac{1-\pi_i}{\pi_i^2}}.
\end{align*}
On the other hand, under Bernoulli sampling, the estimated coefficient $\widehat{\gamma}^\HT$ in \eqref{eqn:gamma-hat} is given as
\begin{align*}
\widehat{\gamma}^{\HT} = \frac{\sum_{i=1}^n Y_i \bI\{W_i=1\} \frac{1-\pi_i}{\pi_i^2}}{\sum_{i=1}^n \frac{1-\pi_i}{\pi_i}}.
\end{align*}
So the coefficient $\widehat{\gamma}^{\Normalized}$ proposed in \citet{khan2023adaptive} can be viewed as a self-normalized version of the estimated coefficient $\widehat{\gamma}^\HT$, which has more stable numerical performance. 
The line of work on which \citet{khan2023adaptive} builds, including \citet{swaminathan2015self}, pointed out the connections between the family of normalized estimators and the control variate literature. 
Yet none of them studied the question of how to choose the optimal control variates.

\subsubsection*{AIPW Estimators.}
The augmented IPW estimator is an efficient estimator in survey sampling \citep{cassel1976some, little1983estimating} and causal inference \citep{robins1994estimation, robins1995semiparametric, scharfstein1999adjusting, tan2010bounded}.
In survey sampling, the AIPW estimator is given by
\begin{align*}
\widehat{\mu}^{\AIPW} = \frac{1}{n} \sum_{i=1}^n \Big( m_i + \frac{\bI\{W_i=1\}}{\pi_i} (Y_i - m_i) \Big),
\end{align*}
where $m_i$ is an outcome model for unit $i \in [n]$.
Since we do not consider covariates, we can use the same outcome model $m_i = \widehat{m}$ across all units $i \in [n]$.
Using this outcome model, the AIPW estimator can be written as
\begin{align*}
\widehat{\mu}^{\AIPW} = \widehat{\mu}^{\HT} - \widehat{m} \Big( \frac{1}{n} \sum_{i=1}^n \frac{\bI\{W_i=1\}}{\pi_i} - 1 \Big).
\end{align*}
Therefore, the AIPW estimator using the outcome model $m_i = \widehat{m}$ can be viewed as a control variate estimator using the IPW basis $a_i = \frac{1}{\pi_i}$ for $i \in [n]$ and coefficient $\widehat{m}$.

\subsubsection*{TMLE Estimators.}
The targeted maximum likelihood estimator (TMLE) is another efficient estimator in survey sampling and causal inference \citep{van2006targeted}. 
In survey sampling, consider a linear fluctuation version of the TMLE estimator.
Let $m_i$ be an initial outcome model for unit $i \in [n]$.
Let $\frac{1}{\pi_i}$ be the clever covariate that specifies the targeting direction.
The targeted outcome model is then given by
\begin{align*}
m_i^* = m_i + \frac{\widehat{\varepsilon}}{\pi_i},
\end{align*}
where $\widehat{\varepsilon}$ is chosen to set the following equation to zero
\begin{align*}
\frac{1}{n} \sum_{i=1}^n \frac{\bI\{W_i=1\}}{\pi_i} \big(Y_i - m_i^*\big) = \frac{1}{n} \sum_{i=1}^n \frac{\bI\{W_i=1\}}{\pi_i} \big(Y_i - m_i - \frac{\widehat{\varepsilon}}{\pi_i}\big) = 0.
\end{align*}
Then, the TMLE estimator is a plug-in estimator given in the form of
\begin{align*}
\widehat{\mu}^{\TMLE} = \frac{1}{n} \sum_{i=1}^n m_i^* = \frac{1}{n} \sum_{i=1}^n \Big( m_i + \frac{\widehat{\varepsilon}}{\pi_i} \Big).
\end{align*}
Since we do not consider covariates, we can naturally use the same initial outcome $m_i = \widehat{\mu}^{\Hajek}$ across all units $i \in [n]$.
Using this initial outcome model, we have $\widehat{\varepsilon} = 0$, and the TMLE estimator can be written as
\begin{align*}
\widehat{\mu}^{\TMLE} = \widehat{\mu}^{\Hajek} = \widehat{\mu}^{\HT} - \widehat{\mu}^{\Hajek} \Big( \frac{1}{n} \sum_{i=1}^n \frac{\bI\{W_i=1\}}{\pi_i} - 1 \Big).
\end{align*}
Therefore, the TMLE estimator using the initial outcome $m_i = \widehat{\mu}^{\Hajek}$ is the same as the Hajek estimator, and both can be viewed as a control variate estimator using the IPW basis $a_i = \frac{1}{\pi_i}$ for $i \in [n]$ and coefficient $\widehat{\mu}^{\Hajek}$.

\subsubsection*{Efficiency Gains.}

All the above estimators, including the Hajek estimator, the first order approximation to the normalized estimator, the AIPW estimator, and the TMLE estimator, can all be viewed as a control variate estimator using the IPW basis $a_i = \frac{1}{\pi_i}$ for $i \in [n]$.
Now we discuss the quality of the IPW basis and compare the optimal basis and the IPW basis.
We provide two conditions under which the optimal basis and the IPW basis are the same.

First, we show in Proposition~\ref{prop:OptBasisIPWBasis} below that, under Bernoulli sampling, the optimal basis converges to the IPW basis as the sample size increases.

\begin{proposition}
\label{prop:OptBasisIPWBasis}
Let $c$ be a normalizing constant such that
\begin{align*}
c^2 = \sum_{i=1}^n \frac{1-\pi_i}{\pi_i}, 
\qquad \text{and} \qquad
c \big( \bm{u}_1(\bm{M}) \big)^\top \bm{\Sigma}^{\frac{1}{2}} \bm{\Pi}^{-1} \bm{1}_n \geq 0.
\end{align*}
Assume $\mu \ne 0$.
Assume Assumption~\ref{asp:DesignProbabilities}-(i) where we assume $\beta < \frac{1}{4}$, and additionally assume there exist constants $\underline{\pi}>0$ and $0 \leq \beta < \frac{1}{4}$ such that for $i \in [n]$, $1-\pi_i \geq \underline{\pi} n^{-\beta} > 0$.
Under Bernoulli sampling, the optimal basis $\bm{a}^* = c \bm{\Sigma}^{-\frac{1}{2}} \bm{u}_1(\bm{M})$ converges to the IPW basis $\bm{a}^\circ = \bm{\Pi}^{-1} \bm{1}_n$, where $\bm{1}_n$ is a $n$-dimensional vector with each element equal to $1$. 
That is, as $n \to +\infty$,
\begin{align*}
\| \bm{a}^* - \bm{a}^\circ \|_2 \to 0.
\end{align*}
\end{proposition}

The proof of Proposition~\ref{prop:OptBasisIPWBasis} is given in Section~\ref{sec:MissingProofs}.
The intuition is that matrix $\bm{M}$ contains two components: a diagonal component $\sigma^2 \bm{\Sigma}^{\frac{1}{2}} \bm{\Pi}^{-2} \bm{\Sigma}^{\frac{1}{2}}$, and a rank-one component $\mu^2 \big( \bm{\Sigma}^{\frac{1}{2}} \bm{\Pi}^{-1} \bm{1}_n \big) \big( \bm{\Sigma}^{\frac{1}{2}} \bm{\Pi}^{-1} \bm{1}_n \big)^\top$.
The rank-one component points exactly in the direction of the IPW basis, and it dominates the diagonal component as $n$ grows. 

Proposition~\ref{prop:OptBasisIPWBasis} provides an asymptotic justification for the IPW basis under Bernoulli sampling. 
Proposition~\ref{prop:OptBasisIPWBasis} explains why the IPW basis, and hence the Hajek, normalized, AIPW, and TMLE estimators discussed above, is asymptotically efficient under Bernoulli sampling. 
At the same time, when $n$ is moderate, when $\mu$ is relatively small compared to $\sigma$, or when the sampling design is not Bernoulli sampling, the optimal basis differs meaningfully from the IPW basis.

Second, we show in Corollary~\ref{coro:NoiselessOutcomes} below that, when the outcomes have near-zero variance, the optimal basis reduces to the IPW basis. 

\begin{corollary}[Noiseless Outcomes]
\label{coro:NoiselessOutcomes}
Assume that the outcomes $Y_i = y$ take the same unknown constant, the optimal basis is given by $a_i = \frac{c}{\pi_i}$ for any $i \in [n]$, where $c \ne 0$ is any constant.
\end{corollary}

The proof of Corollary~\ref{coro:NoiselessOutcomes} is given in Section~\ref{sec:MissingProofs}. 
In this special case when all outcomes take the same constant, the only fluctuation of the Horvitz-Thompson estimator comes from the random component $\frac{1}{n} \sum_{i=1}^n \frac{\bI\{W_i=1\}}{\pi_i}$. 
Corollary~\ref{coro:NoiselessOutcomes} suggests to correct it by using the exact same term $\widehat{X} = \frac{1}{n} \sum_{i=1}^n \frac{\bI\{W_i=1\}}{\pi_i}$ as the control variate. 

Combining with equation \eqref{eqn:CoefficientsEqual}, it shows that the first order approximation to the Hajek estimator is asymptotically optimal when the outcomes have near-zero variance, when data is collected under Bernoulli sampling, and when all the sampling probabilities are equal. 
This extends the conventional wisdom in \citet{sarndal1978design} that Hajek estimator has smaller variance than Horvitz-Thompson estimator when ``$Y_i - \mu_n$ is small'' (i.e., when the outcomes have near-zero variance) and when ``the sample size is not fixed'' (e.g., under Bernoulli sampling).

\section{Causal Inference under Network Interference}
\label{sec:NetworkInterference}

\subsection{Experimental Design under Network Interference}
\label{sec:ExpDesign:general}

We consider the following causal inference problem under network interference \citep{hudgens2008toward, tchetgen2012causal}. 
The classical causal inference problem under the Stable Unit Treatment Value Assumption (SUTVA, \citealt{rubin1980discussion, holland1986statistics}) can be viewed as a special case of the network interference setup; see Example~\ref{exa:SUTVA} below for an illustration.
In Section~\ref{sec:CausalInference}, we present results for causal inference under SUTVA as an important special case.

Let there be $n$ units in a finite population that are connected by an undirected network.
Each vertex represents an experimental unit and each edge represents the connection between the two units.
Let there be two versions of treatments. 
We use ``treatment'' and ``control,'' or $1$ and $0$, respectively, to stand for these two versions of treatments. 
For each $i \in [n]$, let $W_i \in \{0,1\}$ stand for the treatment assignment that unit $i$ receives. 
We consider an experimental design setting where we determine the random treatment assignments $\bm{W} = (W_1, ..., W_n)^\top$ through a known joint probability distribution $\cW$. 
Following convention, we use $W_i$ for a random treatment assignment, and $w_i$ for one realization. 

We adopt a potential outcomes framework in describing the causal parameters \citep{neyman1923application, rubin1974estimating}.
In the full generality, for each unit $i \in [n]$, let $Y_i(\bm{w})$ be the potential outcome for unit $i$ when the treatment assignment vector is $\bm{w}$. 
The observed outcomes are random variables of the treatment assignments connected through the potential outcomes, that is, if $\bm{W} = \bm{w}$ then $Y_i = Y_i(\bm{w})$.
To simplify the connection between observed outcomes and potential outcomes, we borrow the exposure mapping framework \citep{aronow2017estimating, manski2013identification} and assume that the exposure mapping is known. 
Let $\Delta_i$ be the set of all possible exposure conditions for unit $i$. 
An exposure mapping for unit $i$ is a dimension reduction mapping $f_i: \{0,1\}^n \to \Delta_i$ that maps an assignment vector $\bm{w}$ to an exposure condition $f_i(\bm{w}) = \delta_i \in \Delta_i$. 
We assume that the potential outcome of unit $i$ depends on the assignment vector $\bm{w}$ only through the exposure condition $\delta_i$. 

\begin{assumption}[Known Exposure Mapping]
\label{asp:Exposure}
For any two vectors of assignments $\bm{w}, \bm{w}' \in \{0,1\}^n$ and for any unit $i \in [n]$, if $f_i(\bm{w}) = f_i(\bm{w}')$, then
\begin{align*}
Y_i(\bm{w}) = Y_i(\bm{w}').
\end{align*}
\end{assumption}

Suppose we only focus on two exposure conditions $\eT$ and $\eC$ and suppose they are both attainable, that is, $\eT, \eC \in \Delta_i$ for any $i \in [n]$.
Under Assumption~\ref{asp:Exposure}, each observed outcome is related to its respective potential outcomes through 
\begin{align*}
Y_i = \left\{
\begin{aligned}
Y_i(\eT), & \quad \text{if } f_i(\bm{w})=\eT, \\ 
Y_i(\eC), & \quad \text{if } f_i(\bm{w})=\eC, \\
\end{aligned}\right.
\end{align*}
where we reload notation and define $Y_i(\eT)$ as the common value of $Y_i(\bm{w})$ for any $\bm{w}$ such that $f_i(\bm{w}) = \eT$, and define $Y_i(\eC)$ as the common value of $Y_i(\bm{w})$ for any $\bm{w}$ such that $f_i(\bm{w}) = \eC$. 
We collect the potential outcomes into vector form $\bm{Y}(\eT) = (Y_1(\eT), ..., Y_n(\eT))^\top$ and $\bm{Y}(\eC) = (Y_1(\eC), ..., Y_n(\eC))^\top$.
We reload notation and denote the $2n$-dimensional potential outcomes vector $\bm{Y} = (\bm{Y}(\eT)^\top, -\bm{Y}(\eC)^\top)^\top$. 

We are interested in estimating the average treatment effect of exposure condition $\eT$ relative to exposure condition $\eC$ in the finite population,
\begin{align*}
\tau(\eT, \eC) = \frac{1}{n} \sum_{i=1}^n \Big( Y_i(\eT) - Y_i(\eC) \Big).
\end{align*}

In this paper, we take the design of experiment $\cW$ as given. 
We define the following two events
\begin{align*}
\cE_i(1) = \{f_i(\bm{W}) = \eT\}, \qquad \text{and} \qquad \cE_i(0) = \{f_i(\bm{W}) = \eC\}.
\end{align*}
Using the definitions of these two events, we reload notation and denote the exposure indicators
\begin{align*}
\bm{D} = (\bI\{\cE_1(1)\}, ..., \bI\{\cE_n(1)\}, \bI\{\cE_1(0)\}, ..., \bI\{\cE_n(0)\})^\top.
\end{align*}
We denote the following two $n \times n$ diagonal matrices $\bm{\Pi}(1) = \bm{\mathrm{diag}}(\Pr_{\cW}(\cE_1(1)), ..., \Pr_{\cW}(\cE_n(1)))$ and $\bm{\Pi}(0) = \bm{\mathrm{diag}}(\Pr_{\cW}(\cE_1(0)), ..., \Pr_{\cW}(\cE_n(0)))$.
We reload notation and define the following $2n \times 2n$ diagonal matrix
\begin{align*}
\bm{\Pi} = 
\begin{bmatrix}
\bm{\Pi}(1)         & \bm{0}_{n \times n} \\
\bm{0}_{n \times n} & \bm{\Pi}(0)         
\end{bmatrix},
\end{align*}
where $\bm{0}_{n \times n}$ stands for a $n \times n$ matrix with each element equal to $0$.
Matrix $\bm{\Pi}$ reflects the marginal probabilities of $\bm{D}$. 
Denote the covariance matrix as $\bm{\Omega} = \Var_{\cW}(\bm{D})$.

We consider designs of experiments $\cW$ that satisfy Assumption~\ref{asp:ExpDesign:general} below.

\begin{assumption+}{\ref{asp:DesignProbabilities}$^\dagger$}
\label{asp:ExpDesign:general}
The design of experiment $\cW$ satisfies:
\begin{enumerate}[label=(\roman*)]
\item \textbf{(Positivity)} The marginal probabilities under both exposure conditions converge to zero at a slow rate, that is, there exist constants $\underline{\pi} > 0$ and $0 \leq \beta < 1$ such that for $i \in [n]$,
\begin{align*}
\Pr\nolimits_{\cW}(\cE_i(1)) \geq \underline{\pi} n^{-\beta} > 0, \qquad \Pr\nolimits_{\cW}(\cE_i(0)) \geq \underline{\pi} n^{-\beta} > 0.
\end{align*}
\item \textbf{(Dependency Neighborhood)} The correlation between the $2n$ exposure indicators $\bm{D}$ are restricted to some neighborhoods, that is, there exist constants $\overline{d} > 0$ and $0 \leq \alpha < 1$ such that for $l \in [2n]$, $D_l$ is independent of $\{D_j\}_{j\notin\cN_l}$, and 
\begin{align*}
|\cN_l| \leq \overline{d} n^\alpha.
\end{align*}
\end{enumerate}
\end{assumption+}

Assumption~\ref{asp:ExpDesign:general}-(i) is weaker than the standard assumptions that usually appear in the design based network interference literature \citep{aronow2017estimating}.
Assumption~\ref{asp:ExpDesign:general}-(ii) is a mild assumption to describe weakly dependent random variables \citep{ross2011fundamentals}, and has been adopted in the network interference literature \citep{chin2018central}.
In the network interference setup, the maximum degree of the network affects both the magnitudes of the marginal probabilities under both exposure conditions, as reflected in Assumption~\ref{asp:ExpDesign:general}-(i), and the size of the dependency neighborhood, as reflected in Assumption~\ref{asp:ExpDesign:general}-(ii). 
We will explicitly specify the required upper bounds on $\alpha$ and $\beta$ when we present our main results. 
We provide Examples~\ref{exa:SUTVA} and~\ref{exa:GATE} below to illustrate the above notations and assumptions.

\begin{example}[SUTVA]
\label{exa:SUTVA}
Let the exposure mapping be $f_i(\bm{w}) = w_i$. 
Let the potential outcomes under the two exposure conditions be $Y_i(\eT) = Y_i(1)$ and $Y_i(\eC) = Y_i(0)$.
In this case, $\cE_i(1) = \{W_i=1\}$ and $\cE_i(0) = \{W_i=0\}$. 
The average treatment effect of exposure condition $\eT$ relative to $\eC$, given as $\tau(\eT,\eC)$, reduces to the average treatment effect $\tau = \frac{1}{n} \sum_{i=1}^n \big(Y_i(1) - Y_i(0)\big)$.
The covariance matrix $\bm{\Omega}$ is given by
\begin{align*}
\bm{\Omega} = 
\begin{bmatrix}
 \bm{\Sigma} & -\bm{\Sigma} \\
-\bm{\Sigma} &  \bm{\Sigma}
\end{bmatrix}.
\end{align*}
This covariance matrix $\bm{\Omega}$ has a special block structure, where each block can be expressed using the covariance matrix $\bm{\Sigma} = \Var_{\cW}(\bm{W})$, where $\bm{W}$ is a $n$-dimensional vector of treatment indicators.

When the experiment is a Bernoulli randomized experiment where, for each unit $i \in [n]$, we independently assign unit $i$ into the treatment or the control group with probability $\frac{1}{2}$, this experimental design under SUTVA satisfies Assumption~\ref{asp:ExpDesign:general}-(i) because $\Pr\nolimits_{\cW}(\cE_i(1)) = \frac{1}{2}$ and $\Pr\nolimits_{\cW}(\cE_i(0)) = \frac{1}{2}$ for any $i \in [n]$, and satisfies Assumption~\ref{asp:ExpDesign:general}-(ii) because $|\cN_i| = 2$ for any $i \in [n]$.
See Example~\ref{exa:BernoulliExp} in Section~\ref{sec:CausalInference} for more details.
\hfill \halmos
\end{example}

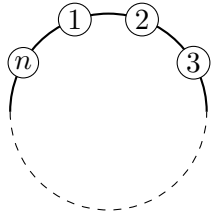
\begin{figure}[tbh]
\centering
\begin{tikzpicture}[scale=1.3]

\def\r{1}

\draw[dashed] (180:\r) arc[start angle=180,end angle=360,radius=\r];

\draw[thick] (0:\r) arc[start angle=0,end angle=180,radius=\r];

\node[circle, draw, fill=white, inner sep=1.2pt] at (150:\r) {$n$};
\node[circle, draw, fill=white, inner sep=1.2pt] at (110:\r) {$1$};
\node[circle, draw, fill=white, inner sep=1.2pt] at (70:\r)  {$2$};
\node[circle, draw, fill=white, inner sep=1.2pt] at (30:\r)  {$3$};

\end{tikzpicture}
\caption{An illustration of a circle network.}
\label{fig:Circle}
\end{figure}

\begin{example}[Neighborhood Interference]
\label{exa:GATE}
Let there be a circle network as shown in Figure~\ref{fig:Circle}.
For any $i \in [n]$, let $\cA_i$ be the set of units that are adjacent to unit $i$, including unit $i$ itself. 
For each unit $i \in [n]$, we independently assign unit $i$ into the treatment or the control group with probability $\frac{1}{2}$.
Let the exposure mapping be $f_i(\bm{w}) = \bm{w}_{\cA_i}$, where $\bm{w}_{\cA_i}$ stands for the sub-vector of treatment assignments for unit $i$ and its adjacent units. 
Let the potential outcomes under the two exposure conditions be $Y_i(\eT) = Y_i(\bm{1}_{\cA_i})$ and $Y_i(\eC) = Y_i(\bm{0}_{\cA_i})$. 
In this case, $\cE_i(1) = \{\bm{W}_{\cA_i} = \bm{1}_{\cA_i}\}$ and $\cE_i(0) = \{\bm{W}_{\cA_i} = \bm{0}_{\cA_i}\}$.
The average treatment effect of exposure condition $\eT$ relative to $\eC$, given as $\tau(\eT,\eC)$, reduces to the global average treatment effect $\tau^{\GATE} = \frac{1}{n} \sum_{i=1}^n \big( Y_i(\bm{1}_{\cA_i}) - Y_i(\bm{0}_{\cA_i}) \big)$.
The covariance matrix $\bm{\Omega}$ is given by
\begin{align*}
\bm{\Omega} = 
\begin{bmatrix}
\bm{\Omega}^{(1)} & \bm{\Omega}^{(2)} \\
\bm{\Omega}^{(2)} & \bm{\Omega}^{(1)}
\end{bmatrix},
\end{align*}
where 
\begin{align*}
\small 
\bm{\Omega}^{(1)} =
\left[
\begin{array}{*{7}{w{c}{1.9em}}}
\frac{7}{64} & \frac{3}{64} & \frac{1}{64} & 0            & \ldots & \frac{1}{64} & \frac{3}{64} \\ [0.4em]
\frac{3}{64} & \frac{7}{64} & \frac{3}{64} & \frac{1}{64} & \ldots & 0            & \frac{1}{64} \\ [0.4em]
\frac{1}{64} & \frac{3}{64} & \frac{7}{64} & \frac{3}{64} & \ldots & 0            & 0            \\ [0.4em]
0            & \frac{1}{64} & \frac{3}{64} & \frac{7}{64} & \ldots & 0            & 0            \\ [0.4em]
\vdots       & \vdots       & \vdots       & \vdots       & \ddots & \vdots       & \vdots       \\ [0.4em]
\frac{1}{64} & 0            & 0            & 0            & \ldots & \frac{7}{64} & \frac{3}{64} \\ [0.4em]
\frac{3}{64} & \frac{1}{64} & 0            & 0            & \ldots & \frac{3}{64} & \frac{7}{64}
\end{array}
\right], \quad 
\bm{\Omega}^{(2)} =
\left[
\begin{array}{*{7}{w{c}{1.9em}}}
-\frac{1}{64} & -\frac{1}{64} & -\frac{1}{64} & 0             & \ldots & -\frac{1}{64} & -\frac{1}{64} \\ [0.4em]
-\frac{1}{64} & -\frac{1}{64} & -\frac{1}{64} & -\frac{1}{64} & \ldots & 0             & -\frac{1}{64} \\ [0.4em]
-\frac{1}{64} & -\frac{1}{64} & -\frac{1}{64} & -\frac{1}{64} & \ldots & 0             & 0             \\ [0.4em]
0             & -\frac{1}{64} & -\frac{1}{64} & -\frac{1}{64} & \ldots & 0             & 0             \\ [0.4em]
\vdots        & \vdots        & \vdots        & \vdots        & \ddots & \vdots        & \vdots        \\ [0.4em]
-\frac{1}{64} & 0             & 0             & 0             & \ldots & -\frac{1}{64} & -\frac{1}{64} \\ [0.4em]
-\frac{1}{64} & -\frac{1}{64} & 0             & 0             & \ldots & -\frac{1}{64} & -\frac{1}{64}
\end{array}
\right].
\end{align*}
Here, $\bm{\Omega}^{(1)}$ captures the covariance between exposure indicators under the same exposure condition, and $\bm{\Omega}^{(2)}$ captures the covariance between exposure indicators under opposite exposure conditions.
This experimental design under neighborhood interference satisfies Assumption~\ref{asp:ExpDesign:general}-(i) because $\Pr\nolimits_{\cW}(\cE_i(1)) = \frac{1}{8}$ and $\Pr\nolimits_{\cW}(\cE_i(0)) = \frac{1}{8}$ for any $i \in [n]$, and satisfies Assumption~\ref{asp:ExpDesign:general}-(ii) because $|\cN_i| = 10$ for any $i \in [n]$.
\hfill \halmos
\end{example}

\subsection{Control Variate Estimators for Causal Inference under Network Interference}
\label{sec:CVEstimator:Interference}

In the network interference setting, one of the most popular ways of estimating $\tau(\eT, \eC)$ is to consider the Horvitz-Thompson estimator defined as
\begin{align*}
\widehat{\tau}^{\HT} = \widehat{\mu}^{\HT}(\eT) - \widehat{\mu}^{\HT}(\eC) = \frac{1}{n} \sum_{i=1}^n \Big( \frac{Y_i \bI\{\cE_i(1)\}}{\Pr(\cE_i(1))} - \frac{Y_i \bI\{\cE_i(0)\}}{\Pr(\cE_i(0))} \Big).
\end{align*}
This Horvitz-Thompson estimator is unbiased \citep{aronow2017estimating}, that is, 
\begin{align*}
\bE_{\cW}\big[\widehat{\tau}^{\HT}\big] = \tau(\eT, \eC).
\end{align*}
The Horvitz-Thompson estimator under network interference is usually known to have a large variance.
This is because the marginal probabilities $\Pr(\cE_i(1))$ and $\Pr(\cE_i(0))$ depend on the degrees of the network, and therefore they are usually heterogeneous and small. 
Again, we draw on the control variate techniques, take the baseline estimator $\widehat{\tau}^{\HT}$, and propose a family of control variate estimators to reduce the variance. 

We consider the following control variate estimator using two control variates $\widehat{X}(1)$ and $\widehat{X}(0)$ and two coefficients $\gamma(1)$ and $\gamma(0)$,
\begin{align*}
\widehat{\tau}^\CV(\gamma(1), \gamma(0)) = \widehat{\tau}^\HT - \gamma(1) \Big(\widehat{X}(1) - \bE[\widehat{X}(1)]\Big) - \gamma(0) \Big(\widehat{X}(0) - \bE[\widehat{X}(0)]\Big).
\end{align*}
When $\gamma(1)$ and $\gamma(0)$ are two constants, the control variate estimator is unbiased, that is, 
\begin{align*}
\bE_{\cW}\big[ \widehat{\tau}^\CV(\gamma(1), \gamma(0)) \big] = \bE_{\cW}\big[ \widehat{\tau}^\HT \big] - \gamma(1) \bE_\cW\big[\widehat{X}(1) - \bE[\widehat{X}(1)]\big] - \gamma(0) \bE_\cW\big[\widehat{X}(0) - \bE[\widehat{X}(0)]\big] = \tau.
\end{align*}
If the coefficients $\gamma(1)$ and $\gamma(0)$ are selected properly, the control variate estimator has a smaller variance than the baseline Horvitz-Thompson estimator.
It is easy to show that the variance minimizing coefficients $\gamma^*(1)$ and $\gamma^*(0)$ for any generic control variates $\widehat{X}(1)$ and $\widehat{X}(0)$ are given by
\begin{align*}
\begin{bmatrix}
\gamma^*(1) \\
\gamma^*(0) 
\end{bmatrix}
= 
\begin{bmatrix}
\Var\big(\widehat{X}(1)\big) & \Cov\big(\widehat{X}(1), \widehat{X}(0)\big) \\
\Cov\big(\widehat{X}(1), \widehat{X}(0)\big) & \Var\big(\widehat{X}(0)\big) 
\end{bmatrix}^{-1}
\cdot
\begin{bmatrix}
\Cov\big( \widehat{\tau}^\HT, \widehat{X}(1) \big) \\
\Cov\big( \widehat{\tau}^\HT, \widehat{X}(0) \big)
\end{bmatrix}.
\end{align*}

In the network interference setting, we propose a family of control variates given in the following form
\begin{align*}
\widehat{X}(1) = \frac{1}{n} \sum_{i=1}^n a_i(1) \bI\{\cE_i(1)\}, && \widehat{X}(0) = \frac{1}{n} \sum_{i=1}^n a_i(0) \bI\{\cE_i(0)\},
\end{align*}
where $a_1(1), a_2(1), ..., a_n(1)$ and $a_1(0), a_2(0), ..., a_n(0)$ are any arbitrary constants that we can choose.
We collect these constants into vector form $\bm{a}(1) = (a_1(1), a_2(1), ..., a_n(1))^\top$ and $\bm{a}(0) = (a_1(0), a_2(0), ..., a_n(0))^\top$ and refer to them as the bases of the control variates.
Note that we use the exposure indicators $\bI\{\cE_i(1)\}$ and $\bI\{\cE_i(0)\}$ as the control variates. 
This family of control variates is in contrast to the regression adjustment works that use auxiliary variables such as covariates as the control variates \citep{basse2018analyzing, gao2025causal, lin2013agnostic, lu2024adjusting, ritzwoller2025regression, wang2025covariate}. 

The mean values of the above control variates are given by $\bE_{\cW}\big[\widehat{X}(1)\big] = \frac{1}{n} \sum_{i=1}^n a_i(1) \Pr_{\cW}(\cE_i(1))$ and $\bE_{\cW}\big[\widehat{X}(0)\big] = \frac{1}{n} \sum_{i=1}^n a_i(0) \Pr_{\cW}(\cE_i(0))$, which are known because the experimental design $\cW$ is known.
Therefore, we can derive the variance minimizing coefficients $\gamma^*(1)$ and $\gamma^*(0)$ in Lemma~\ref{lem:OPTgammas:general}. 
We present Lemma~\ref{lem:OPTgammas:general} using succinct notations. 
Recall the $2n$-dimensional vector $\bm{Y}$, diagonal matrix $\bm{\Pi}$, and the covariance matrix $\bm{\Omega}$.
We denote the $2n \times 2$ block diagonal bases matrix 
\begin{align*}
\bm{B} = 
\begin{bmatrix}
\bm{a}(1) & \bm{0}_n \\
\bm{0}_n  & \bm{a}(0)
\end{bmatrix},
\end{align*}
where $\bm{0}_n$ is a $n$-dimensional vector with each element equal to $0$. 
Let $\cB$ denote the set of all such block diagonal matrices. 
Using these notations, we introduce Lemma~\ref{lem:OPTgammas:general} as follows.

\begin{lemma}[Variance Minimizing Coefficients]
\label{lem:OPTgammas:general}
In the network interference setting, under Assumption~\ref{asp:Exposure}, the variance minimizing coefficients $\gamma^*(1)$ and $\gamma^*(0)$ are given by
\begin{align*}
\big(\gamma^*(1), \gamma^*(0) \big)^\top = \big(\bm{B}^\top \bm{\Omega} \bm{B}\big)^{-1} \bm{B}^\top \bm{\Omega} \bm{\Pi}^{-1} \bm{Y}. 
\end{align*}
and the variance of the control variate estimator under these coefficients is given by
\begin{align*}
\Var\big( \widehat{\tau}^{\CV}(\gamma^*(1), \gamma^*(0)) \big) = \frac{1}{n^2} \bigg( \bm{Y}^\top \bm{\Pi}^{-1} \bm{\Omega} \bm{\Pi}^{-1} \bm{Y} - \bm{Y}^\top \bm{\Pi}^{-1} \bm{\Omega} \bm{B} \big(\bm{B}^\top \bm{\Omega} \bm{B}\big)^{-1} \bm{B}^\top \bm{\Omega} \bm{\Pi}^{-1} \bm{Y} \bigg). 
\end{align*}
\end{lemma}

The proof of Lemma~\ref{lem:OPTgammas:general} is given in Section~\ref{sec:MissingProofs}.
Intuitively, we try to use the randomness in two control variates $\widehat{X}(1)$ and $\widehat{X}(0)$ to explain as much randomness in $\widehat{\tau}^{\HT}$ as possible. 
The variance $\Var( \widehat{\tau}^{\CV}(\gamma^*(1), \gamma^*(0)))$ is given by the residual from projecting the potential outcomes $\bm{\Pi}^{-1} \bm{Y}$ onto the column space of $\bm{B}$.
Because the variance minimizing coefficients $\gamma^*(1)$ and $\gamma^*(0)$ depend on unknown potential outcomes $\bm{Y}$, we cannot directly observe $\gamma^*(1)$ and $\gamma^*(0)$; instead, we have to estimate $\gamma^*(1)$ and $\gamma^*(0)$ from the data.

\subsection{Estimating the Variance Minimizing Coefficients}
\label{sec:EstimatingCoefficients:general}

To estimate $\gamma^*(1)$ and $\gamma^*(0)$ from the data, we can replace the unknown potential outcomes by their Horvitz-Thompson estimates.
Let $\widehat{Y}_i(\eT) = Y_i \frac{\bI\{\cE_i(1)\}}{\Pr_{\cW}(\cE_i(1))}$ and $\widehat{Y}_i(\eC) = Y_i \frac{\bI\{\cE_i(0)\}}{\Pr_{\cW}(\cE_i(0))}$ for $i \in [n]$. 
We collect $\widehat{\bm{Y}} = \big(\widehat{Y}_1(\eT), ..., \widehat{Y}_n(\eT), - \widehat{Y}_1(\eC), ..., - \widehat{Y}_n(\eC)\big)$. 
Then the coefficients $\gamma^*(1)$ and $\gamma^*(0)$ can be estimated using
\begin{align*}
\big(\widehat{\gamma}^\HT(1), \widehat{\gamma}^\HT(0)\big)^\top = \big(\bm{B}^\top \bm{\Omega} \bm{B}\big)^{-1} \bm{B}^\top \bm{\Omega} \bm{\Pi}^{-1} \widehat{\bm{Y}}.
\end{align*}
Since $\widehat{Y}_i(\eT)$ and $\widehat{Y}_i(\eC)$ are unbiased estimators of $Y_i(\eT)$ and $Y_i(\eC)$, respectively, the estimated coefficients $\widehat{\gamma}^{\HT}(1)$ and $\widehat{\gamma}^{\HT}(0)$ as defined above are also unbiased estimators of the optimal coefficients $\gamma^*(1)$ and $\gamma^*(0)$. 

Note that, the estimated coefficients $\widehat{\gamma}^{\HT}(1)$ and $\widehat{\gamma}^{\HT}(0)$ obtained in this way are not constants, but random variables correlated with the control variates $\widehat{X}(1)$ and $\widehat{X}(0)$.
Using the estimated coefficients $\widehat{\gamma}^{\HT}(1)$ and $\widehat{\gamma}^{\HT}(0)$ introduces a bias to the control variate estimator.
Nonetheless, we can show that the bias is small and that $\widehat{\tau}^\CV(\widehat{\gamma}^{\HT}(1), \widehat{\gamma}^{\HT}(0))$ is asymptotically equivalent to $\widehat{\tau}^\CV(\gamma^*(1), \gamma^*(0))$, as long as the bases matrix $\bm{B}$ satisfies certain regularity conditions. 
See Assumption~\ref{asp:RegularBasesInterference} below.

\begin{assumption+}{\ref{asp:RegularBasis}$^\dagger$}
\label{asp:RegularBasesInterference}
The bases matrix $\bm{B}$ satisfies two conditions:
\begin{enumerate}[label=(\roman*)]
\item \textbf{(Boundedness)} There exists a constant $\overline{a}$ such that for each $i \in [n]$,
\begin{align*}
\vert a_i(1) \vert \leq \overline{a}, \quad \vert a_i(0) \vert \leq \overline{a}.
\end{align*}
\item \textbf{(Non-degeneracy)} Let $\lambda_{\min}(\cdot)$ be the smallest eigenvalue of a matrix. 
There exists a constant $\underline{\lambda}_{\Omega}$ such that 
\begin{align*}
\lambda_{\min}(\bm{B}^\top \bm{\Omega} \bm{B}) \geq \underline{\lambda}_{\Omega} n.
\end{align*}
\end{enumerate}
\end{assumption+}

Assumption~\ref{asp:RegularBasesInterference}-(i) is a standard boundedness assumption. 
Assumption~\ref{asp:RegularBasesInterference}-(ii) is a standard non-degeneracy assumption on the control variates $\widehat{X}(1)$ and $\widehat{X}(0)$. 
Denote $\widehat{\bm{X}} = (\widehat{X}(1), \widehat{X}(0))^\top$. 
Note that we have $\Var_{\cW}(\widehat{\bm{X}}) = n^{-2} \bm{B}^\top \bm{\Omega} \bm{B}$, and therefore Assumption~\ref{asp:RegularBasesInterference}-(ii) is equivalent to requiring that $\Var_{\cW}(\widehat{\bm{X}}) \succeq \underline{\lambda}_{\Omega} n^{-1} \bm{I}_2$, where $\bm{I}_2$ stands for a $2 \times 2$ identity matrix. 
Under Assumption~\ref{asp:RegularBasesInterference}, we establish Theorem~\ref{thm:AsymptoticInterference} below.

\begin{theorem}
\label{thm:AsymptoticInterference}
Under Assumptions~\ref{asp:ExpDesign:general} and~\ref{asp:RegularBasesInterference} where in Assumption~\ref{asp:ExpDesign:general} we assume $3 \alpha + 4 \beta < 1$, and assuming the potential outcomes $\vert Y_i(\eT) \vert \leq \overline{y}$ and $\vert Y_i(\eC) \vert \leq \overline{y}$ are all bounded, the control variate estimator using the estimated coefficients $\widehat{\tau}^\CV(\widehat{\gamma}^\HT(1), \widehat{\gamma}^\HT(0))$ is a consistent estimator of $\tau(\eT,\eC)$, that is, 
\begin{align*}
\lim_{n \to +\infty} \widehat{\tau}^\CV(\widehat{\gamma}^\HT(1), \widehat{\gamma}^\HT(0)) \xrightarrow{p} \tau(\eT,\eC).
\end{align*}
Additionally, if in Assumption~\ref{asp:ExpDesign:general} we assume $4 \alpha + 4 \beta < 1$, then the asymptotic variance of the control variate estimator using the estimated coefficients $\Var(\widehat{\tau}^\CV(\widehat{\gamma}^\HT(1), \widehat{\gamma}^\HT(0)))$ is the same as the asymptotic variance of the control variate estimator using the optimal coefficients $\Var(\widehat{\tau}^\CV(\gamma^*(1), \gamma^*(0)))$, that is, 
\begin{align*}
\lim_{n \to +\infty} n \Big( \Var\big( \widehat{\tau}^\CV(\widehat{\gamma}^\HT(1), \widehat{\gamma}^\HT(0)) \big) - \Var\big( \widehat{\tau}^\CV(\gamma^*(1), \gamma^*(0)) \big) \Big) \to 0.
\end{align*}
Additionally, if in Assumption~\ref{asp:ExpDesign:general} we assume $7 \alpha + 6 \beta < 1$ and further assuming there exists a constant $\underline{c} > 0$ such that for sufficiently large $n$, $n \Var\big( \widehat{\tau}^{\CV}(\gamma^*(1), \gamma^*(0)) \big) \geq \underline{c}$, then the control variate estimator using the estimated coefficients $\widehat{\tau}^\CV(\widehat{\gamma}^\HT(1), \widehat{\gamma}^\HT(0))$ is asymptotically normal, that is, 
\begin{align*}
\lim_{n \to +\infty} \frac{\widehat{\tau}^\CV(\widehat{\gamma}^\HT(1), \widehat{\gamma}^\HT(0)) - \tau(\eT,\eC)}{\sqrt{\Var\big(\widehat{\tau}^\CV(\widehat{\gamma}^\HT(1), \widehat{\gamma}^\HT(0))\big)}} \xrightarrow{d} \cN(0,1).
\end{align*}
\end{theorem}

The proof of Theorem~\ref{thm:AsymptoticInterference} is given in Section~\ref{sec:MissingProofs}. 
As we have stronger results, from consistency to asymptotic variance equivalence and to asymptotic normality, the requirement on the marginal probabilities under both exposure conditions and on the dependency neighborhoods becomes correspondingly stronger, strengthening from $3 \alpha + 4 \beta < 1$ to $4 \alpha + 4 \beta < 1$ and to $7 \alpha + 6 \beta < 1$. 
Taken together, Theorem~\ref{thm:AsymptoticInterference} shows that estimating $\gamma^*(1)$ and $\gamma^*(0)$ by $\widehat{\gamma}^{\HT}(1)$ and $\widehat{\gamma}^{\HT}(0)$ does not affect the first order asymptotic properties of the control variate estimator. 
The magnitude of variance reduction is determined by the choice of the bases $\bm{a}(1)$ and $\bm{a}(0)$.
In Section~\ref{sec:OptimalControlVariates:general} below, we study how to choose the bases $\bm{a}(1)$ and $\bm{a}(0)$ for the optimal variance reduction.

\subsection{The Optimal Control Variates}
\label{sec:OptimalControlVariates:general}

Now we focus on finding the optimal bases $\bm{a}(1)$ and $\bm{a}(0)$ that minimize $\Var\big( \widehat{\tau}^{\CV}(\gamma^*(1), \gamma^*(0)) \big)$ the variance of the control variate estimator for fixed $n$. 
If the vector of potential outcomes $\bm{Y}(\eT)$ and $\bm{Y}(\eC)$ were known to take values $Y_i(\eT) = y_i(\eT)$ and $Y_i(\eC) = y_i(\eC)$ for any $i \in [n]$, then we denote $\bm{y} = (\bm{y}(\eT)^\top, -\bm{y}(\eC)^\top)^\top$. 
We see that the problem 
\begin{align*}
\min_{\bm{B} \in \cB} \ \Var\big( \widehat{\tau}^{\CV}(\gamma^*(1), \gamma^*(0)) \big) \ = \ \min_{\bm{B} \in \cB} \frac{1}{n^2} \bigg( \bm{y}^\top \bm{\Pi}^{-1} \bm{\Omega} \bm{\Pi}^{-1} \bm{y} - \bm{y}^\top \bm{\Pi}^{-1} \bm{\Omega} \bm{B} \big(\bm{B}^\top \bm{\Omega} \bm{B}\big)^{-1} \bm{B}^\top \bm{\Omega} \bm{\Pi}^{-1} \bm{y} \bigg).
\end{align*}
can be easily solved by choosing $\bm{B}$ such that $\bm{\Pi}^{-1} \bm{y}$ is in the column space of $\bm{B}$, such as $\bm{a}(1) = \bm{\Pi}(1)^{-1} \bm{y}(\eT)$ and $\bm{a}(0) = -\bm{\Pi}(0)^{-1}\bm{y}(\eC)$.
In this case, the variance $\Var\big( \widehat{\tau}^{\CV}(\gamma^*(1), \gamma^*(0)) \big)$ can be reduced to exactly zero. 

However, the above choice of $\bm{B}$ is infeasible because the potential outcomes $\bm{Y}(\eT)$ and $\bm{Y}(\eC)$ are unknown. 
Therefore, we find the optimal bases $\bm{a}(1)$ and $\bm{a}(0)$ under uncertainty of the potential outcomes, which leads to a decision making problem under uncertainty. 
We adopt a stochastic optimization perspective \citep{zhao2024experimental} and model the unknown outcomes $\big(Y_i(\eT), Y_i(\eC)\big)$ to be i.i.d. sampled from an unknown joint distribution $\cY_{\delta}$.
See Assumption~\ref{asp:iidInterference} below. 

\begin{assumption+}{\ref{asp:iid}$^\dagger$}
\label{asp:iidInterference}
We assume that for each $i \in [n]$, the pair of potential outcomes $\big(Y_i(\eT), Y_i(\eC)\big)$ are i.i.d. sampled from an unknown joint distribution $\cY_{\delta}$, whose marginal distributions are denoted as
$\cY_{\delta}(1)$ and $\cY_{\delta}(0)$.
\end{assumption+}

Denote the population means as $\mu(1) = \bE_{Y_i(\eT) \sim \cY_{\delta}(1)}[Y_i(\eT)]$ and $\mu(0) = \bE_{Y_i(\eC) \sim \cY_{\delta}(0)}[Y_i(\eC)]$, the population variances as $\sigma^2(1) = \Var_{Y_i(\eT) \sim \cY_{\delta}(1)}(Y_i(\eT))$ and $\sigma^2(0) = \Var_{Y_i(\eC) \sim \cY_{\delta}(0)}(Y_i(\eC))$, and the population covariance as $\sigma(1,0) = \Cov_{Y_i(\eT), Y_i(\eC) \sim \cY_{\delta}}(Y_i(\eT), Y_i(\eC))$. 
Denote $\cY_{\delta}^n$ to be the joint probability distribution of the $n$ independent pairs of potential outcomes $\bm{Y}(\eT), \bm{Y}(\eC)$.
Under Assumption~\ref{asp:iidInterference}, we formulate the following constrained stochastic optimization problem
\begin{align}
\bm{B}^* \in \argmin_{\bm{B} \in \cB} \bE_{\bm{Y}(\eT), \bm{Y}(\eC) \sim \cY_{\delta}^n} \Big[ \Var\big( \widehat{\tau}^{\CV}(\gamma^*(1), \gamma^*(0)) \big) \Big]. \label{eqn:FormulationCausal:general}
\end{align}
The optimal control variates use the bases that solve \eqref{eqn:FormulationCausal:general}, which is given by Theorem~\ref{thm:OPTInterference} below.
Recall that $\bm{\Omega}$ is symmetric and positive semidefinite, we can define the square root matrix $\bm{\Omega}^{\frac{1}{2}}$ and the pseudo-inverse square root matrix $\bm{\Omega}^{-\frac{1}{2}}$. 

\begin{theorem}[Optimal Control Variates]
\label{thm:OPTInterference}
For notational simplicity, we reload the notation $\bm{M}$ and define $\bm{M}$ to be a $2n \times 2n$ symmetric matrix 
\begin{align*}
\bm{M} = \bm{\Omega}^{\frac{1}{2}} \bm{\Pi}^{-1} 
\begin{bmatrix}
\mu^2(1) \bm{1}_n \bm{1}_n^\top + \sigma^2(1) \bm{I}_n & - \mu(1)\mu(0) \bm{1}_n \bm{1}_n^\top - \sigma(1,0) \bm{I}_n \\
- \mu(1)\mu(0) \bm{1}_n \bm{1}_n^\top - \sigma(1,0) \bm{I}_n & \mu^2(0) \bm{1}_n \bm{1}_n^\top + \sigma^2(0) \bm{I}_n
\end{bmatrix}
\bm{\Pi}^{-1} \bm{\Omega}^{\frac{1}{2}},
\end{align*}
where $\bm{I}_n$ is a $n \times n$ identity matrix and $\bm{1}_n$ is a $n$-dimensional vector with each element equal to $1$. 
Let $\lambda_l(\bm{M})$ be the $l$-th largest eigenvalue of $\bm{M}$, and let $\bm{u}_l(\bm{M})$ be the eigenvector corresponding to $\lambda_l(\bm{M})$.
Under Assumptions~\ref{asp:iidInterference}, and~\ref{asp:Exposure}, the optimal bases are given by
\begin{align}
\bm{B}^* \in \argmax_{\bm{B} \in \cB} \Tr\Big( \big(\bm{B}^\top \bm{\Omega} \bm{B}\big)^{-1} \bm{B}^\top \bm{\Omega}^{\frac{1}{2}} \bm{M} \bm{\Omega}^{\frac{1}{2}} \bm{B} \Big). \label{eqn:ConstrainedKyFan}
\end{align}
And in expectation, the optimal variance reduction can be upper bounded by
\begin{align*}
\bE_{\bm{Y}(\eT), \bm{Y}(\eC) \sim \cY_{\delta}^n} \Big[ \Var\big( \widehat{\tau}^{\HT} \big) - \Var\big( \widehat{\tau}^{\CV}(\gamma^*(1), \gamma^*(0)) \big) \Big] \leq \frac{\lambda_1(\bm{M}) + \lambda_2(\bm{M})}{n^2}.
\end{align*}
\end{theorem}

The proof of Theorem~\ref{thm:OPTInterference} is given in Section~\ref{sec:MissingProofs}.
Theorem~\ref{thm:OPTInterference} involves solving a non-convex maximization problem \eqref{eqn:ConstrainedKyFan}, for which a closed-form solution is generally unavailable. 
In Section~\ref{sec:OPTComputations}, we discuss how to compute the near-optimal control variates.
We present a $\frac{1}{2}$-approximate solution for \eqref{eqn:ConstrainedKyFan}, and then use its solution to initialize an alternating local search heuristic, which monotonically improves the objective and converges to a local maximum.

Similar to Theorem~\ref{thm:OptimalControlVariate}, Theorem~\ref{thm:OPTInterference} also combines a design-based perspective and a model-based perspective, where we assume the unknown potential outcomes $\bm{Y}$ to be i.i.d. samples to guide the choice of bases. 
But we rely on the experimental design as the only source of randomness when we estimate the average treatment effect. 
We point out that Assumption~\ref{asp:iidInterference} is only required to show optimality of the bases $\bm{B}^*$.
The bases $\bm{B}^*$ and the control variate estimator $\Var(\widehat{\tau}^{\CV}(\widehat{\gamma}^{\HT}(1), \widehat{\gamma}^{\HT}(0)))$ are still well defined without Assumption~\ref{asp:iidInterference}.

In Theorem~\ref{thm:OPTInterference}, the optimal bases $\bm{B}^*$ depend on the unknown moments of the joint distribution $\cY_{\delta}$. 
We can perform sample splitting to estimate the moments of the joint distribution $\cY_{\delta}$, and then use them to find the optimal bases $\bm{B}^*$. 
In particular, we can estimate the population means $\mu(1)$ and $\mu(0)$, as well as the population variances $\sigma^2(1)$ and $\sigma^2(0)$.
But the population covariance $\sigma(1,0)$ is not directly estimable from the data. 
In practice, we propose to assume that the potential outcomes under the two exposure conditions are independent so $\sigma(1,0) = 0$, and use the following surrogate matrix
\begin{align*}
\widetilde{\bm{M}} = \bm{\Omega}^{\frac{1}{2}} \bm{\Pi}^{-1} 
\begin{bmatrix}
\mu^2(1) \bm{1}_n \bm{1}_n^\top + \sigma^2(1) \bm{I}_n & - \mu(1)\mu(0) \bm{1}_n \bm{1}_n^\top \\
- \mu(1)\mu(0) \bm{1}_n \bm{1}_n^\top & \mu^2(0) \bm{1}_n \bm{1}_n^\top + \sigma^2(0) \bm{I}_n
\end{bmatrix}
\bm{\Pi}^{-1} \bm{\Omega}^{\frac{1}{2}}.
\end{align*}

Theorem~\ref{thm:OPTInterference} nests the SUTVA setting as a special case. 
Although the constrained maximization problem \eqref{eqn:ConstrainedKyFan} does not generally have a closed form solution, the SUTVA setting makes the constraint nonbinding and the optimal bases matrix that solves the maximization problem \eqref{eqn:ConstrainedKyFan} has a closed form solution. 
And in expectation, the optimal variance reduction attains the upper bound given in Theorem~\ref{thm:OPTInterference}. 
See Corollary~\ref{coro:SUTVASpecialCase} below.

\begin{corollary}[SUTVA]
\label{coro:SUTVASpecialCase}
Consider the special case of causal inference under SUTVA.
Under Assumption~\ref{asp:iidInterference}, the constraint $\bm{B} \in \cB$ in the maximization problem \eqref{eqn:ConstrainedKyFan} is nonbinding. 
And in expectation, the optimal variance reduction is exactly equal to $(\lambda_1(\bm{M}) + \lambda_2(\bm{M})) n^{-2}$.
\end{corollary}

\subsection{Connections to Hajek Estimator}
\label{sec:ConnectionsInterference}

\subsubsection*{Optimal Basis and IPW Basis.}
We define the following basis as the IPW basis,
\begin{align*}
a_i(1) = \frac{1}{\Pr_{\cW}(\cE_i(1))}, \qquad \text{and} \qquad a_i(0) = - \frac{1}{\Pr_{\cW}(\cE_i(0))}.
\end{align*}
We show in Corollary~\ref{coro:NoiselessPOExposures} below that, when the potential outcomes have near-zero variance, the optimal basis reduces to the IPW basis. 

\begin{corollary}[Noiseless Potential Outcomes]
\label{coro:NoiselessPOExposures}
Under Assumption~\ref{asp:Exposure}, and assuming that the potential outcomes $Y_i(\eT) = y(\eT)$ and $Y_i(\eC) = y(\eC)$ take the same two unknown constants, then one set of optimal bases is given, for any $i \in [n]$, by
\begin{align*}
a_i(1) = \frac{1}{\Pr_{\cW}(\cE_i(1))}, \qquad \text{and} \qquad a_i(0) = - \frac{1}{\Pr_{\cW}(\cE_i(0))}.
\end{align*}
\end{corollary}

The proof of Corollary~\ref{coro:NoiselessPOExposures} is given in Section~\ref{sec:MissingProofs}. 
In this special case when all potential outcomes take the same two constants, the only fluctuation of the Horvitz-Thompson estimator comes from the random components $\frac{1}{n} \sum_{i=1}^n \frac{\bI\{\cE_i(1)\}}{\Pr_{\cW}(\cE_i(1))}$ and $\frac{1}{n} \sum_{i=1}^n \frac{\bI\{\cE_i(0)\}}{\Pr_{\cW}(\cE_i(0))}$. 
Corollary~\ref{coro:NoiselessPOExposures} suggests to correct it by using the exact same terms $\widehat{X}(1) = \frac{1}{n} \sum_{i=1}^n \frac{\bI\{\cE_i(1)\}}{\Pr_{\cW}(\cE_i(1))}$ and $\widehat{X}(0) = \frac{1}{n} \sum_{i=1}^n \frac{\bI\{\cE_i(0)\}}{\Pr_{\cW}(\cE_i(0))}$ as the control variates. 
We will show below that the first order approximations of the Hajek estimator can be viewed as a control variate estimator using the IPW basis.

\subsubsection*{Hajek Estimator.}
In the network interference setup, the Hajek estimator is defined as 
\begin{align*}
\widehat{\tau}^\Hajek(\eT, \eC) = \frac{\widehat{\mu}^{\HT}(\eT)}{\widehat{1}(\eT)} - \frac{\widehat{\mu}^\HT(\eC)}{\widehat{1}(\eC)}, 
\end{align*}
where
\begin{align}
\widehat{1}(\eT) = \frac{1}{n} \sum_{i=1}^n \frac{\bI\{\cE_i(1)\}}{\Pr_{\cW}(\cE_i(1))}, \qquad \widehat{1}(\eC) = \frac{1}{n} \sum_{i=1}^n \frac{\bI\{\cE_i(0)\}}{\Pr_{\cW}(\cE_i(0))}. \label{eqn:EstimatedOne:Interference}
\end{align}
The Hajek estimator in the network interference setup simply combines two Hajek estimators. 
Following similar discussions as in Section~\ref{sec:Connections}, we make a first order approximation by using Taylor expansion around $\widehat{1}(\eT) \approx 1$ and $\widehat{1}(\eC) \approx 1$, which gives
\begin{multline*}
\widehat{\tau}^{\Hajek}(\eT, \eC) \approx \widehat{\tau}^{\HT}(\eT, \eC) \\
- \widehat{\mu}^{\HT}(\eT) \Big( \frac{1}{n} \sum_{i=1}^n \frac{\bI\{\cE_i(1)\}}{\Pr_{\cW}(\cE_i(1))} - 1 \Big) + \widehat{\mu}^{\HT}(\eC) \Big( \frac{1}{n} \sum_{i=1}^n \frac{\bI\{\cE_i(0)\}}{\Pr_{\cW}(\cE_i(0))} - 1 \Big).
\end{multline*}
The first order approximation of the Hajek estimator can be viewed as a control variate estimator using the bases $a_i(1) = \frac{1}{\Pr_{\cW}(\cE_i(1))}$ and $a_i(0) = - \frac{1}{\Pr_{\cW}(\cE_i(0))}, \forall i \in [n]$.

The variance difference between the control variate estimator and the Hajek estimator is especially significant under network interference.
This is because the two control variates $\widehat{1}(\eT) = \frac{1}{n} \sum_{i=1}^n \frac{\bI\{\cE_i(1)\}}{\Pr_{\cW}(\cE_i(1))}$ and $\widehat{1}(\eC) = \frac{1}{n} \sum_{i=1}^n \frac{\bI\{\cE_i(0)\}}{\Pr_{\cW}(\cE_i(0))}$ can be correlated because the exposure events $\cE_i(1)$ and $\cE_i(0)$ are induced by both the same treatment assignment vector $\bm{W}$ and the same underlying network structure. 
The Hajek estimator implicitly chooses its coefficients by viewing the two control variates $\widehat{1}(\eT)$ and $\widehat{1}(\eC)$ separately, whereas the optimal control variate coefficients account for the correlation between them.
Consequently, even when restricted to the same control variates $\widehat{1}(\eT)$ and $\widehat{1}(\eC)$, the optimal control variate estimator can achieve additional variance reduction by accounting for the cross exposure correlation.

\subsection{Computation of Optimal Control Variates}
\label{sec:OPTComputations}

We now discuss the computation of the optimal control variates, which solves the constrained maximization problem \eqref{eqn:ConstrainedKyFan}. 
For convenience, we restate the problem below,
\begin{align}
\bm{B}^* \in \argmax_{\bm{B} \in \cB} \Tr\Big( \big(\bm{B}^\top \bm{\Omega} \bm{B}\big)^{-1} \bm{B}^\top \bm{\Omega}^{\frac{1}{2}} \bm{M} \bm{\Omega}^{\frac{1}{2}} \bm{B} \Big). \tag{\ref{eqn:ConstrainedKyFan}}
\end{align}
First, we propose a $\frac{1}{2}$-approximate solution for \eqref{eqn:ConstrainedKyFan}.
Second, we propose an alternating local search heuristic.
We use the $\frac{1}{2}$-approximate solution to initialize the alternating local search heuristic, which monotonically improves the objective and converges to a local maximum.

\subsubsection*{A $\frac{1}{2}$-Approximate Solution.}
We present a $\frac{1}{2}$-approximate solution for problem \eqref{eqn:ConstrainedKyFan}.
We first consider a rank-$1$ problem and use its solution to construct a feasible solution for \eqref{eqn:ConstrainedKyFan}.
The rank-$1$ problem is as follows,
\begin{align*}
\max_{\bm{b} \in \bR^{2n}} \frac{\bm{b}^\top \bm{\Omega}^{\frac{1}{2}} \bm{M} \bm{\Omega}^{\frac{1}{2}} \bm{b}}{\bm{b}^\top \bm{\Omega} \bm{b}}.
\end{align*}
Suppose $\widetilde{\bm{b}} = (\widetilde{\bm{b}}_1^\top, \widetilde{\bm{b}}_2^\top)^\top$ is one optimal solution to the above rank-$1$ problem. 
Then we pick the bases to be $\bm{a}(1) = \widetilde{\bm{b}}_1$ and $\bm{a}(0) = \widetilde{\bm{b}}_2$, which gives a feasible solution 
\begin{align}
\widetilde{\bm{B}} = 
\begin{bmatrix}
\widetilde{\bm{b}}_1 & \bm{0}_n             \\
\bm{0}_n             & \widetilde{\bm{b}}_2 
\end{bmatrix}. \label{eqn:FeasibleSolution}
\end{align}
We show that $\widetilde{\bm{B}}$ as constructed in \eqref{eqn:FeasibleSolution} above is a $\frac{1}{2}$-approximate solution. 
The proof of Theorem~\ref{thm:1/2approx} is given in Section~\ref{sec:MissingProofs}.

\begin{theorem}[$\frac{1}{2}$-Approximation]
\label{thm:1/2approx}
$\widetilde{\bm{B}}$ is a $\frac{1}{2}$-approximate solution to \eqref{eqn:ConstrainedKyFan}, that is, 
\begin{align*}
\Tr\Big( \big(\widetilde{\bm{B}}^\top \bm{\Omega} \widetilde{\bm{B}}\big)^{-1} \widetilde{\bm{B}}^\top \bm{\Omega}^{\frac{1}{2}} \bm{M} \bm{\Omega}^{\frac{1}{2}} \widetilde{\bm{B}} \Big) \geq \frac{1}{2} \max_{\bm{B} \in \cB} \Tr\Big( \big(\bm{B}^\top \bm{\Omega} \bm{B}\big)^{-1} \bm{B}^\top \bm{\Omega}^{\frac{1}{2}} \bm{M} \bm{\Omega}^{\frac{1}{2}} \bm{B} \Big).
\end{align*}
\end{theorem}

\subsubsection*{An Alternating Local Search Heuristic.} 
We present an alternating local search heuristic that takes the $\frac{1}{2}$-approximate solution $\widetilde{\bm{B}}$ to initialize. 
For notational convenience, denote 
\begin{align*}
\psi(\bm{b}_1, \bm{b}_2) = \Tr\Big( \big(\bm{B}^\top \bm{\Omega} \bm{B}\big)^{-1} \bm{B}^\top \bm{\Omega}^{\frac{1}{2}} \bm{M} \bm{\Omega}^{\frac{1}{2}} \bm{B} \Big), \qquad \text{where} \qquad \bm{B} = 
\begin{bmatrix}
\bm{b}_1 & \bm{0}_n \\
\bm{0}_n & \bm{b}_2 
\end{bmatrix},
\end{align*}
to emphasize the dependence on the two directions. 
Using this notation, we present the heuristic in Algorithm~\ref{alg:AlternatingHeuristic}.

\begin{algorithm}[!tb]
\caption{An Alternating Local Search Heuristic}
\label{alg:AlternatingHeuristic}
\textbf{Input}: Numerical tolerance $\epsilon$.
\begin{algorithmic}[1]
\State \textbf{Initialize}: $k \gets 1$; $\bm{a}^0(1) \gets \widetilde{\bm{b}}_1$; $\bm{a}^0(0) \gets \widetilde{\bm{b}}_2$; \Comment{$\widetilde{\bm{b}}_1$ and $\widetilde{\bm{b}}_2$ are obtained in \eqref{eqn:FeasibleSolution}}
\For{$k = 1, 2, ...$} 
\State $\bm{a}^k(1) \gets \argmax_{\bm{a}(1)} \psi(\bm{a}(1), \bm{a}^{k-1}(0))$;
\State $\bm{a}^k(0) \gets \argmax_{\bm{a}(0)} \psi(\bm{a}^{k}(1), \bm{a}(0))$;
\If{$\psi(\bm{a}^k(1), \bm{a}^k(0)) - \psi(\bm{a}^{k-1}(1), \bm{a}^{k-1}(0)) \leq \epsilon$} \State \textbf{break}
\EndIf
\EndFor
\State \Return $\big(\bm{a}^k(1), \bm{a}^k(0)\big)$ \label{marker}
\end{algorithmic}
\end{algorithm}

The alternating local search heuristic alternates between the two directions $\bm{a}(1)$ and $\bm{a}(0)$ to monotonically improve the objective. 
We will show in Theorem~\ref{thm:AlternatingLocalSearch} below that the objective value converges. 
On the other hand, the solutions may not converge. 
Note that any scaling does not change the objective of \eqref{eqn:ConstrainedKyFan}, that is, for any $c_1, c_0 \ne 0$, $\psi(\bm{b}_1, \bm{b}_2) = \psi(c_1 \bm{b}_1, c_0 \bm{b}_2)$.
Therefore, convergence of the optimal solutions requires an additional normalization as well as a spectral gap assumption to avoid multiple local optimal solutions.
The proof of Theorem~\ref{thm:AlternatingLocalSearch} is given in Section~\ref{sec:MissingProofs}.

\begin{theorem}[Monotone Convergence]
\label{thm:AlternatingLocalSearch}
Let $\{(\bm{a}^k(1),\bm{a}^k(0))\}_{k\geq 0}$ be the sequence generated by Algorithm~\ref{alg:AlternatingHeuristic} without the stopping rule.
Then the objective values $\psi^k:=\psi(\bm{a}^k(1),\bm{a}^k(0))$ are non-decreasing and converge, that is, there exists $\psi^\infty \in \bR$ such that
\begin{align*}
\lim_{k \to +\infty} \psi^k = \psi^\infty.
\end{align*}
\end{theorem}

\section{Applications to Real Data}
\label{sec:RealData}

\subsection{Application to Swiss Environmental Panel Survey Data}
\label{sec:SwissData}

\subsubsection*{Empirical setup.}
The Swiss Environmental Panel conducted three waves of survey among Swiss residents to measure public support for costly food waste regulation \citep{quoss2021swiss}.
They surveyed the same group of people three times between 2018 and 2019. 
Wave~3 asked about how much food waste people attribute to households, which is the focus on our application.

\subsubsection*{Estimators.}
We compare five estimators of $\mu^w_n$: the Horvitz-Thompson estimator $\widehat{\mu}^{\HT}$, the Hajek estimator $\widehat{\mu}^{\Hajek}$, and three control variate estimators $\widehat{\mu}^{\CV}(\widehat{\gamma}^{\HT})$ using the estimated coefficient $\widehat{\gamma}^{\HT}$, but each using a different basis $\bm{a}$ as follows.
\begin{enumerate}
\item \textit{CV, constant basis}: $a_i = 1$ for any $i \in [n]$.
\item \textit{CV, IPW basis}: $a_i = 1/\pi_i$ for any $i \in [n]$.
\item \textit{CV, sample splitting}: a feasible version of the optimal basis. We split the $14961$ units into two halves, use each half to estimate a basis, apply the estimated basis to the other half to calculate a control variate estimator, and average the two resulting control variate estimators.
\end{enumerate}

\subsubsection*{Results.}
Table~\ref{tbl:RealDataResults} reports the point estimate, the standard error calculated using the delta method, and the resulting $95\%$ confidence interval for each estimator.
All five estimators provide almost the same point estimate.
In fact, the first four estimators provide exactly the same point estimate, because $\pi_i$ is computed from \eqref{eqn:ResponseRates} using the same realized sample, so $\sum_{i=1}^n W_i = \sum_{i=1}^n \pi_i$ and the control variate equals its expectation $\widehat{X} = \bE[\widehat{X}]$ for any basis that is constant within each region.
The survey suggests that respondents attribute a little over a fifth of Switzerland's food waste to households on average. 
Despite having almost the same point estimate, Hajek and all three control variate estimators cut the standard error by nearly a half relative to the Horvitz-Thompson estimator. 

\begin{table}[!tb]
\centering
\TABLE{Point estimate, standard error, and $95\%$ confidence interval for the share of food waste attributed to households
\label{tbl:RealDataResults}}
{
\begin{tabular}{l>{\centering}p{3cm}>{\centering}p{3cm}>{\centering}p{3cm}>{\centering}p{2.2cm}c}
Estimator             & Estimate & SE     & 95\% CI            & SE vs. $\HT$ & \\ \hline
Horvitz-Thompson      & $21.64$  & $0.41$ & $[20.84, 22.44]$   & --           & \\
Hajek                 & $21.64$  & $0.22$ & $[21.21, 22.08]$   & $-45.6\%$    & \\
CV, constant basis    & $21.64$  & $0.22$ & $[21.20, 22.08]$   & $-45.0\%$    & \\
CV, IPW basis         & $21.64$  & $0.22$ & $[21.21, 22.08]$   & $-45.6\%$    & \\
CV, sample splitting  & $21.63$  & $0.22$ & $[21.20, 22.07]$   & $-45.6\%$    &
\end{tabular}
}
{\textit{Note:} Column ``SE vs. $\HT$'' reports the percentage reductions in standard error relative to the Horvitz-Thompson estimator.}
\end{table}

\subsection{Application to Insurance Network Experimental Data}
\label{sec:InsuranceNetworkData}

\subsubsection*{Empirical setup.}
The People Insurance Company of China conducted a randomized experiment in two counties in rural Jiangxi, China, in spring 2010 to study how financial education and social networks affect the decisions of rice producing households in purchasing agricultural weather insurance \citep{cai2015social}.
A few days before the experiment, they collected up to five friends from each household, and used the list of friends to create a pre-experimental measure of the social network. 

\subsubsection*{Estimators.}
We compare five estimators: the Horvitz-Thompson estimator $\widehat{\tau}^{\HT}$, the Hajek estimator $\widehat{\tau}^{\Hajek}$, and three control variate estimators using different bases vectors as follows.
\begin{enumerate}
\item \textit{CV, IPW bases}: $a_i(1) = \frac{1}{\Pr_{\cW}(\cE_i(1))}$ and $a_i(0) = - \frac{1}{\Pr_{\cW}(\cE_i(0))}$ for any $i \in [n]$. These are the bases of the first order approximation of the Hajek estimator discussed in Section~\ref{sec:ConnectionsInterference}.
\item \textit{CV, sample splitting}: a feasible version of the optimal bases. 
The first version (\textit{CV, sample splitting 1}) is given by the $\frac{1}{2}$-approximate solution; the second version (\textit{CV, sample splitting 2}) is computed by using the $\frac{1}{2}$-approximate solution to initialize the alternating local search heuristic. 
We randomly split the households into two unequal halves such that no cluster is cut into two. 
The random split depends only on the network and target population sizes, not on realized exposures or outcomes.
We use each half to estimate the moments of the two outcome distributions, substitute these estimates to obtain a pair of bases, apply those bases to the other half, and find the weighted average of the two resulting control variate estimators. 
\end{enumerate}

\subsubsection*{Self-normalized coefficients.}
For the three control variate estimators, we do not use the estimated coefficients $\widehat{\gamma}^{\HT}(1)$ and $\widehat{\gamma}^{\HT}(0)$ directly.
Instead we use their self-normalized versions
\begin{align}
\widehat{\gamma}^{\Hajek}(1) = \frac{\widehat{\gamma}^{\HT}(1)}{\widehat{1}(\eT)}, \qquad \text{and} \qquad \widehat{\gamma}^{\Hajek}(0) = \frac{\widehat{\gamma}^{\HT}(0)}{\widehat{1}(\eC)}, \label{eqn:selfnormalized:estimatedgammas}
\end{align}
where $\widehat{1}(\eT)$ and $\widehat{1}(\eC)$ are defined in \eqref{eqn:EstimatedOne:Interference}. 
The reason for this self-normalization is numerical stability, which follows the same idea that \citet{khan2023adaptive} applied to $\widehat{\gamma}^{\HT}$ in the survey sampling setup.
This self-normalization does not change the first order asymptotic behavior described in Theorem~\ref{thm:AsymptoticInterference}.
But it substantially stabilizes the estimators in finite samples.

\begin{table}[!tb]
\centering
\TABLE{Point estimates, standard errors, and $95\%$ confidence intervals for the effect of one additional friend receiving the intensive session in the first round
\label{tbl:InsuranceNetworkResults}}
{
\begin{tabular}{l>{\centering}p{3cm}>{\centering}p{3cm}>{\centering}p{3cm}>{\centering}p{2.2cm}c}
Estimator              & Estimate & SE      & 95\% CI           & SE vs. $\HT$ & \\ \hline
\multicolumn{6}{l}{\textit{Panel A: $n_1 = 881$}} \\
Horvitz-Thompson       & $0.076$  & $0.077$ & $[-0.076, 0.227]$ & --           & \\
Hajek                  & $0.092$  & $0.060$ & $[-0.025, 0.209]$ & $-22.6\%$    & \\
CV, IPW basis          & $0.086$  & $0.066$ & $[-0.043, 0.216]$ & $-14.5\%$    & \\
CV, sample splitting 1 & $0.097$  & $0.068$ & $[-0.035, 0.230]$ & $-12.5\%$    & \\
CV, sample splitting 2 & $0.098$  & $0.068$ & $[-0.035, 0.231]$ & $-12.2\%$    & \\[1.2ex]
\multicolumn{6}{l}{\textit{Panel B: $n_2 = 481$}} \\
Horvitz-Thompson       & $-0.028$ & $0.097$ & $[-0.219, 0.163]$ & --           & \\
Hajek                  & $0.140$  & $0.085$ & $[-0.026, 0.306]$ & $-13.2\%$    & \\
CV, IPW basis          & $0.128$  & $0.073$ & $[-0.015, 0.272]$ & $-24.9\%$    & \\
CV, sample splitting 1 & $0.136$  & $0.071$ & $[-0.004, 0.275]$ & $-27.0\%$    & \\
CV, sample splitting 2 & $0.135$  & $0.071$ & $[-0.004, 0.275]$ & $-27.0\%$    & \\[1.2ex]
\multicolumn{6}{l}{\textit{Panel C: $n_3 = 168$}} \\
Horvitz-Thompson       & $-0.279$ & $0.153$ & $[-0.579, 0.020]$ & --           & \\
Hajek                  & $-0.037$ & $0.297$ & $[-0.619, 0.545]$ & $+94.3\%$    & \\
CV, IPW basis          & $-0.089$ & $0.108$ & $[-0.301, 0.123]$ & $-29.3\%$    & \\
CV, sample splitting 1 & $-0.117$ & $0.104$ & $[-0.321, 0.086]$ & $-32.0\%$    & \\
CV, sample splitting 2 & $-0.118$ & $0.104$ & $[-0.322, 0.086]$ & $-32.0\%$    &
\end{tabular}
}
{\textit{Note:} Column ``SE vs. $\HT$'' reports the percentage change in standard error relative to the Horvitz-Thompson estimator; negative values indicate reductions in standard error. Panel C has very few samples, and its results should therefore be interpreted cautiously.}
\end{table}

\subsubsection*{Results.}
Table~\ref{tbl:InsuranceNetworkResults} reports the point estimate, the standard error calculated using the delta method, and the resulting conservative $95\%$ confidence interval for each estimator.
For the two control variate estimators using sample splitting, the reported standard errors are conditional on the estimated bases. 
They do not incorporate the uncertainty of estimating the bases.

In the first panel, all three control variate estimators reduce the standard error by approximately $14\%$ relative to the Horvitz-Thompson estimator, although Hajek has the smallest standard error.
In the second panel, all three control variate estimators have smaller standard errors than both the Horvitz-Thompson estimator and the the Hajek estimator. 
In the third panel, all three control variate estimators have smaller standard errors than the Horvitz-Thompson estimator, and the Hajek estimator has much worse performance.
These results suggest that the control variate estimators all consistently outperform the Horvitz-Thompson estimator.
While the Hajek estimator sometimes have better performance, it does not have any theoretical guarantee in finite sample, especially when the sample size is small.

\section{Simulations Using Synthetic Data}
\label{sec:Simulations}

\subsection{Simulations for Survey Sampling}
\label{sec:Simulations:ss}

In this section, we conduct simulations in the survey sampling setup using synthetic data. 
Because the data generating process for the outcomes and the sampling design are both known by construction, we can compare the performances of the control variate estimators with a number of benchmarks in the literature.

\subsubsection*{Simulation setup.}
We consider a finite population of $n = 100$ units. 
The outcomes $Y_1,\dots,Y_n$ are i.i.d. sampled from a Log-normal$(0, 2)$ distribution, and the sampling probabilities $\pi_i$ are i.i.d. sampled from a Uniform$(0.1, 0.3)$ distribution.
Both are fixed for the rest of the simulation. 
We are interested in estimating the population mean $\mu_n = \frac{1}{n}\sum_{i=1}^n Y_i$, whose realized value is $\mu_n \approx 5.89$.
To implement different estimators, we repeat the random sampling $10000$ times. 
In each repetition, we draw a random subset of the $100$ units following Bernoulli sampling as in Example~\ref{exa:Bernoulli}.
Each unit $i$ is sampled independently with probability $\pi_i$.
Only the sampling indicators $W_1,\dots,W_n$ are drawn again across each repetition.

\subsubsection*{Estimators compared.}
We compare six estimators. 
The first is the Horvitz-Thompson estimator $\widehat{\mu}^{\HT}$.
The second is the Hajek estimator $\widehat{\mu}^{\Hajek}$. 
The remaining four are control variate estimators $\widehat{\mu}^{\CV}(\widehat{\gamma}^{\HT})$ using the estimated coefficient $\widehat{\gamma}^{\HT}$, but each using a different basis $\bm{a}$ as follows.

\begin{enumerate}
\item \textit{CV, constant basis}: $a_i = 1$ for any $i \in [n]$.
\item \textit{CV, IPW basis}: $a_i = \frac{1}{\pi_i}$ for any $i \in [n]$. This is the basis behind the Hajek, normalized, AIPW, and TMLE estimators discussed in Section~\ref{sec:Connections}.
\item \textit{CV, optimal basis}: the basis $\bm{a}^*$ from Theorem~\ref{thm:OptimalControlVariate}. Because the outcomes are drawn from a known distribution, the mean and standard deviation are known exactly. 
This estimator assumes knowledge of the outcome distribution in advance, which is usually infeasible.
\item \textit{CV, sample splitting}: a feasible version of the optimal basis. We split the $100$ units into two halves, use each half to estimate the mean and the standard deviation of the outcome distribution, substitute these estimates for the unknown moments to obtain a basis, apply that basis to the other half, and average the two resulting control variate estimators.
\end{enumerate}

\begin{table}[!tb]
\centering
\TABLE{Simulated variance and mean square error when $n = 100$
\label{tbl:SimulationResults}}
{
\begin{tabular}{l>{\centering}p{1.75cm}>{\centering}p{1.75cm}>{\centering}p{1.75cm}>{\centering}p{1.95cm}>{\centering}p{1.75cm}>{\centering}p{1.95cm}c}
Estimator            & Bias     & \% Bias   & Var   & Var vs. $\HT$ & MSE   & MSE vs. $\HT$ & \\ \hline
Horvitz-Thompson     & $0.075$  & $1.28\%$  & 26.73 & --            & 26.74 & --            & \\
Hajek                & $-0.010$ & $0.18\%$  & 23.77 & $-11.09\%$    & 23.77 & $-11.11\%$    & \\
CV, constant basis   & $-0.210$ & $3.57\%$  & 23.86 & $-10.75\%$    & 23.90 & $-10.60\%$    & \\
CV, IPW basis        & $-0.431$ & $7.32\%$  & 22.00 & $-17.71\%$    & 22.18 & $-17.04\%$    & \\
CV, optimal basis    & $-0.694$ & $11.80\%$ & 20.68 & $-22.65\%$    & 21.16 & $-20.86\%$    & \\
CV, sample splitting & $-0.878$ & $14.91\%$ & 21.01 & $-21.39\%$    & 21.78 & $-18.53\%$    & \\
\end{tabular}
}
{\textit{Note:} The column ``\% Bias'' reports the absolute percentage bias relative to $\mu_n \approx 5.89$. The columns ``Var vs. $\HT$'' and ``MSE vs. $\HT$'' report the percentage reductions relative to the Horvitz-Thompson estimator.}
\end{table}

\subsubsection*{Results.}
Table~\ref{tbl:SimulationResults} reports the bias, variance and mean squared error (MSE) of the six estimators across the $10000$ replications.
All four control variate estimators have larger bias but lower variance and MSE than the Horvitz-Thompson estimator. 
Even the control variate estimator using the constant basis, whose basis uses only the sampling indicators and no other information about the design, reduces the variance and MSE relative to the Horvitz-Thompson estimator. 
The control variate estimator using the IPW basis reduces the variance and MSE more than the Hajek estimator itself.
The control variate estimator using the optimal basis gives the largest variance reduction; yet it is infeasible.
The feasible control variate estimator using sample splitting has similar performance as the oracle optimal basis.
This suggests that estimating the basis from the data only loses a small amount of the variance reduction in this example.

\subsection{Simulations for Causal Inference under Network Interference}
\label{sec:Simulations:interference}

In this section, we conduct simulations in the causal inference setup under network interference using synthetic data.
As in Section~\ref{sec:Simulations:ss}, both the data generating process for the potential outcomes and the experimental design are known by construction, so we can compare the performances of the control variate estimators with a number of benchmarks in the literature.

\subsubsection*{Simulation setup.}
We consider a finite population of $n = 100$ units connected by a network that consists of $40$ disjoint clusters, which is the partial interference structure in \citet{hudgens2008toward}.
Units in the same cluster are fully connected to one another, and units in different clusters are not connected.
The cluster sizes cycle through $1, 2, 3, 4$, so that there are ten clusters of each size.
For each unit $i \in [n]$, the set of units adjacent to unit $i$ including unit $i$ itself, $\cA_i$, is the cluster that unit $i$ belongs to.
We adopt the neighborhood interference assumption in Example~\ref{exa:GATE}. 
The exposure mapping is given by $f_i(\bm{w}) = \bm{w}_{\cA_i}$ and the two exposure conditions are $Y_i(\eT) = Y_i(\bm{1}_{\cA_i})$ and $Y_i(\eC) = Y_i(\bm{0}_{\cA_i})$.
We wish to estimate the global average treatment effect $\tau^{\GATE}$.

The potential outcomes $Y_1(\eT),\dots,Y_n(\eT)$ are i.i.d. sampled from a Normal$(0.5, 1)$ distribution and the potential outcomes $Y_1(\eC),\dots,Y_n(\eC)$ are i.i.d. sampled from a Normal$(0.25, 1)$ distribution, independently of one another.
Both are fixed for the rest of the simulation.
The realized value of the global average treatment effect is $\tau^{\GATE} \approx 0.310$.
To implement different estimators, we repeat the randomized experiment $10000$ times.
In each repetition, we independently assign each unit $i$ into the treatment group with probability $\frac{1}{2}$.
Under this design, the two exposure probabilities coincide and are given by $\Pr_{\cW}(\cE_i(1)) = \Pr_{\cW}(\cE_i(0)) = \big(\frac{1}{2}\big)^{\vert \cA_i \vert}$.
Only the treatment assignments $W_1,\dots,W_n$ are drawn again across each repetition.

\subsubsection*{Estimators compared.}
We compare seven estimators.
The first is the Horvitz-Thompson estimator $\widehat{\tau}^{\HT}$.
The second is the Hajek estimator $\widehat{\tau}^{\Hajek}(\eT, \eC)$.
The remaining five are control variate estimators $\widehat{\tau}^{\CV}(\widehat{\gamma}(1), \widehat{\gamma}(0))$ using estimated coefficients, but each using a different pair of bases $\bm{a}(1)$ and $\bm{a}(0)$ as follows.

\begin{enumerate}
\item \textit{CV, IPW bases}: $a_i(1) = \frac{1}{\Pr_{\cW}(\cE_i(1))}$ and $a_i(0) = - \frac{1}{\Pr_{\cW}(\cE_i(0))}$ for any $i \in [n]$. These are the bases of the first order approximation of the Hajek estimator discussed in Section~\ref{sec:ConnectionsInterference}.
\item \textit{CV, optimal bases}: the bases $\bm{B}^*$ from Theorem~\ref{thm:OPTInterference}.
The first version (\textit{CV, optimal bases 1}) is given by the $\frac{1}{2}$-approximate solution; the second version (\textit{CV, optimal bases 2}) is computed by using the $\frac{1}{2}$-approximate solution to initialize the alternating local search heuristic. 
Because the potential outcomes are drawn from known distributions, the moments are known exactly. 
Since we have $\sigma(1,0) = 0$, the surrogate matrix $\widetilde{\bm{M}}$ coincides with $\bm{M}$. 
This estimator assumes knowledge of the outcome distributions in advance, which is usually infeasible.
\item \textit{CV, sample splitting}: a feasible version of the optimal bases. The first version (\textit{CV, sample splitting 1}) is given by the $\frac{1}{2}$-approximate solution; the second version (\textit{CV, sample splitting 2}) is computed by using the $\frac{1}{2}$-approximate solution to initialize the alternating local search heuristic. We split the $100$ units into two halves such that each half contains $20$ clusters and no cluster is cut in two. We use each half to estimate the means and the variances of the two outcome distributions, substitute these estimates into $\widetilde{\bm{M}}$ to obtain a pair of bases, apply those bases to the other half, and average the two resulting control variate estimators. 
\end{enumerate}

\subsubsection*{Self-normalized coefficients.}
For the five control variate estimators, we do not use the estimated coefficients $\widehat{\gamma}^{\HT}(1)$ and $\widehat{\gamma}^{\HT}(0)$ directly.
Instead we use their self-normalized versions as defined in \eqref{eqn:selfnormalized:estimatedgammas}.
This self-normalization does not change the first order asymptotic behavior described in Theorem~\ref{thm:AsymptoticInterference}.
But it substantially stabilizes the estimators in finite samples.

\begin{table}[!tb]
\centering
\TABLE{Simulated variance and mean square error under network interference when $n = 100$
\label{tbl:SimulationResultsInterference}}
{
\begin{tabular}{l>{\centering}p{1.75cm}>{\centering}p{1.75cm}>{\centering}p{1.75cm}>{\centering}p{1.95cm}>{\centering}p{1.75cm}>{\centering}p{1.95cm}c}
Estimator              & Bias     & \% Bias   & Var    & Var vs. $\HT$ & MSE    & MSE vs. $\HT$ & \\ \hline
Horvitz-Thompson       & $-0.004$ & $1.29\%$  & 0.3542 & --            & 0.3542 & --            & \\
Hajek                  & $-0.039$ & $12.47\%$ & 0.2144 & $-39.46\%$    & 0.2159 & $-39.04\%$    & \\
CV, IPW bases          & $-0.097$ & $31.35\%$ & 0.1729 & $-51.17\%$    & 0.1824 & $-48.51\%$    & \\
CV, optimal bases 1    & $-0.100$ & $32.34\%$ & 0.1708 & $-51.79\%$    & 0.1808 & $-48.95\%$    & \\
CV, optimal bases 2    & $-0.096$ & $30.82\%$ & 0.1679 & $-52.59\%$    & 0.1771 & $-50.01\%$    & \\
CV, sample splitting 1 & $-0.170$ & $54.93\%$ & 0.1178 & $-66.75\%$    & 0.1468 & $-58.55\%$    & \\
CV, sample splitting 2 & $-0.179$ & $57.57\%$ & 0.1517 & $-57.18\%$    & 0.1836 & $-48.18\%$    & \\
\end{tabular}
}
{\textit{Note:} The column ``\% Bias'' reports the absolute percentage bias relative to $\tau^{\GATE} \approx 0.310$. The columns ``Var vs. $\HT$'' and ``MSE vs. $\HT$'' report the percentage reductions relative to the Horvitz-Thompson estimator.}
\end{table}

\subsubsection*{Results.}
Table~\ref{tbl:SimulationResultsInterference} reports the bias, variance and mean squared error (MSE) of the seven estimators across the $10000$ replications.
All the three control variate estimators have larger bias but lower variance and MSE than the Horvitz-Thompson estimator.
The control variate estimator using the IPW basis reduces the variance and MSE than the Hajek estimator.
Although they use the same basis, the control variate estimator differs from the first-order approximation to the Hajek estimator in two respects: (i) its coefficients are the self normalized versions $\widehat{\gamma}^{\Hajek}(1)$ and $\widehat{\gamma}^{\Hajek}(0)$, and (ii) its coefficients takes into account of the correlation between the two control variates. 
We believe the first difference is the primary driver of the variance and MSE reduction.
The two control variate estimators using the optimal basis gives larger variance and MSE reduction; yet it is infeasible.
The two feasible control variate estimators using sample splitting has similar performance as the oracle optimal basis.
Unlike the survey sampling setup, the bias from using estimated coefficients is not negligible in all the control variate estimators.
This is because the exposure probabilities can go very small, which is exactly the regime that Assumption~\ref{asp:ExpDesign:general}-(i) restricts.
So a larger $n$ is needed before the asymptotics of Theorem~\ref{thm:AsymptoticInterference} take effect.

\section{Conclusions}
\label{sec:Conclusions}

In this paper, we propose a family of control variate estimators for variance reduction in design-based survey sampling and causal inference, with and without network interference. 
We construct these control variates directly from the sampling, treatment, or exposure indicators, which is in contrast to the conventional use of auxiliary variables as covariates. 
We provide a unified control variate framework and show that the Hajek estimator, the normalized estimator, the AIPW estimator, and the TMLE estimator all correspond to control variate estimators that use the same IPW basis. 
We show that IPW basis is asymptotically optimal under certain conditions, providing a justification of why these estimators are asymptotically efficient.
In finite sample, the IPW basis is different from the optimal basis.
And we show that the optimal control variates are given by the leading eigenvectors of matrices that depend on both the randomized design and the moments of the outcomes. 
We also provide computational methods for computing the optimal basis under network interference.

We conclude this paper with three potential limitations.
First, the optimal bases depend on moments of the outcome distribution, which are usually unknown in practice. 
We use sample splitting to estimate these moments, but the quality of the resulting basis can depend on the accuracy of these moment estimates. 
For causal inference with two potential outcomes, the covariance between the two potential outcomes is not directly identifiable from the observed data, and we therefore assume this covariance to be zero and use a surrogate matrix. 
It remains a future research direction to develop robust choices of bases that perform well under misspecified moment assumptions.

Second, throughout the paper we take the sampling, treatment, or exposure probabilities as known. 
This is natural in randomized experiments and survey sampling.
But these probabilities are not available in observational studies and need to be estimated, such as what we have done in Section~\ref{sec:SwissData}. 
It remains a future research direction to develop the theory for estimated probabilities. 

Third, we use the sampling, treatment, or exposure indicators as control variates.
We do not use auxiliary variables such as covariates that may be available in many empirical applications. 
It remains a future research direction to combine the design-induced control variates studied in this paper with auxiliary variables.

\bibliographystyle{informs2014} 
\bibliography{bibliography} 

\begin{thebibliography}{68}
\providecommand{\natexlab}[1]{#1}
\providecommand{\url}[1]{\texttt{#1}}
\providecommand{\urlprefix}{URL }

\bibitem[{Aronow \protect\BIBand{} Samii(2017)}]{aronow2017estimating}
Aronow PM, Samii C (2017) Estimating average causal effects under general
  interference, with application to a social network experiment. \emph{The
  Annals of Applied Statistics} 11(4):1912--1947.

\bibitem[{Asmussen \protect\BIBand{} Glynn(2007)}]{asmussen2007stochastic}
Asmussen S, Glynn PW (2007) \emph{Stochastic simulation: algorithms and
  analysis}, volume~57 (Springer).

\bibitem[{Bai(2022)}]{bai2022optimality}
Bai Y (2022) Optimality of matched-pair designs in randomized controlled
  trials. \emph{American Economic Review} 112(12):3911--3940.

\bibitem[{Basse \protect\BIBand{} Feller(2018)}]{basse2018analyzing}
Basse G, Feller A (2018) Analyzing two-stage experiments in the presence of
  interference. \emph{Journal of the American Statistical Association}
  113(521):41--55.

\bibitem[{Basu(1971)}]{basu1971essay}
Basu D (1971) An essay on the logical foundations of survey sampling, part i.
  foundations of statistical inferences, vp godambe and da sprott. \emph{New
  York} 203--233.

\bibitem[{Bond et~al.(2012)Bond, Fariss, Jones, Kramer, Marlow, Settle,
  \protect\BIBand{} Fowler}]{bond201261}
Bond RM, Fariss CJ, Jones JJ, Kramer AD, Marlow C, Settle JE, Fowler JH (2012)
  A 61-million-person experiment in social influence and political
  mobilization. \emph{Nature} 489(7415):295--298.

\bibitem[{Cai et~al.(2015)Cai, Janvry, \protect\BIBand{}
  Sadoulet}]{cai2015social}
Cai J, Janvry AD, Sadoulet E (2015) Social networks and the decision to insure.
  \emph{American Economic Journal: Applied Economics} 7(2):81--108.

\bibitem[{Cassel et~al.(1976)Cassel, S{\"a}rndal, \protect\BIBand{}
  Wretman}]{cassel1976some}
Cassel CM, S{\"a}rndal CE, Wretman JH (1976) Some results on generalized
  difference estimation and generalized regression estimation for finite
  populations. \emph{Biometrika} 63(3):615--620.

\bibitem[{Chin(2018)}]{chin2018central}
Chin A (2018) Central limit theorems via stein's method for randomized
  experiments under interference. \emph{arXiv preprint arXiv:1804.03105} .

\bibitem[{Cr{\'e}pon et~al.(2013)Cr{\'e}pon, Duflo, Gurgand, Rathelot,
  \protect\BIBand{} Zamora}]{crepon2013labor}
Cr{\'e}pon B, Duflo E, Gurgand M, Rathelot R, Zamora P (2013) Do labor market
  policies have displacement effects? evidence from a clustered randomized
  experiment. \emph{The quarterly journal of economics} 128(2):531--580.

\bibitem[{Deville \protect\BIBand{} S{\"a}rndal(1992)}]{deville1992calibration}
Deville JC, S{\"a}rndal CE (1992) Calibration estimators in survey sampling.
  \emph{Journal of the American statistical Association} 87(418):376--382.

\bibitem[{Ding(2023)}]{ding2023first}
Ding P (2023) A first course in causal inference. \emph{arXiv preprint
  arXiv:2305.18793} .

\bibitem[{Fan(1949)}]{fan1949theorem}
Fan K (1949) On a theorem of weyl concerning eigenvalues of linear
  transformations i. \emph{Proceedings of the National Academy of Sciences}
  35(11):652--655.

\bibitem[{Feit \protect\BIBand{} Berman(2019)}]{feit2019test}
Feit EM, Berman R (2019) Test \& roll: Profit-maximizing a/b tests.
  \emph{Marketing Science} 38(6):1038--1058.

\bibitem[{Fesenfeld et~al.(2022)Fesenfeld, Rudolph, \protect\BIBand{}
  Bernauer}]{fesenfeld2022policy}
Fesenfeld L, Rudolph L, Bernauer T (2022) Policy framing, design and feedback
  can increase public support for costly food waste regulation. \emph{Nature
  Food} 3(3):227--235.

\bibitem[{Fieller \protect\BIBand{} Hartley(1954)}]{fieller1954sampling}
Fieller EC, Hartley H (1954) Sampling with control variables. \emph{Biometrika}
  41(3/4):494--501.

\bibitem[{Gao \protect\BIBand{} Ding(2025)}]{gao2025causal}
Gao M, Ding P (2025) Causal inference in network experiments: regression-based
  analysis and design-based properties. \emph{Journal of Econometrics}
  252:106119.

\bibitem[{Glasserman(2004)}]{glasserman2004monte}
Glasserman P (2004) \emph{Monte Carlo methods in financial engineering},
  volume~53 (Springer).

\bibitem[{Hansen \protect\BIBand{} Hurwitz(1943)}]{hansen1943theory}
Hansen MH, Hurwitz WN (1943) On the theory of sampling from finite populations.
  \emph{The Annals of Mathematical Statistics} 14(4):333--362.

\bibitem[{Hesterberg(1995)}]{hesterberg1995weighted}
Hesterberg T (1995) Weighted average importance sampling and defensive mixture
  distributions. \emph{Technometrics} 37(2):185--194.

\bibitem[{Hesterberg(1988)}]{hesterberg1988advances}
Hesterberg TC (1988) \emph{Advances in importance sampling} (Stanford
  University).

\bibitem[{Hickernell et~al.(2005)Hickernell, Lemieux, \protect\BIBand{}
  Owen}]{hickernell2005control}
Hickernell FJ, Lemieux C, Owen AB (2005) Control variates for quasi-monte carlo
  .

\bibitem[{Holland(1986)}]{holland1986statistics}
Holland PW (1986) Statistics and causal inference. \emph{Journal of the
  American statistical Association} 81(396):945--960.

\bibitem[{Horn \protect\BIBand{} Johnson(2012)}]{horn2012matrix}
Horn RA, Johnson CR (2012) \emph{Matrix analysis} (Cambridge university press).

\bibitem[{Horvitz \protect\BIBand{} Thompson(1952)}]{horvitz1952generalization}
Horvitz DG, Thompson DJ (1952) A generalization of sampling without replacement
  from a finite universe. \emph{Journal of the American statistical
  Association} 47(260):663--685.

\bibitem[{Hudgens \protect\BIBand{} Halloran(2008)}]{hudgens2008toward}
Hudgens MG, Halloran ME (2008) Toward causal inference with interference.
  \emph{Journal of the American Statistical Association} 103(482):832--842.

\bibitem[{Imbens \protect\BIBand{} Rubin(2015)}]{imbens2015causal}
Imbens GW, Rubin DB (2015) \emph{Causal inference in statistics, social, and
  biomedical sciences} (Cambridge university press).

\bibitem[{Karrer et~al.(2021)Karrer, Shi, Bhole, Goldman, Palmer, Gelman,
  Konutgan, \protect\BIBand{} Sun}]{karrer2021network}
Karrer B, Shi L, Bhole M, Goldman M, Palmer T, Gelman C, Konutgan M, Sun F
  (2021) Network experimentation at scale. \emph{Proceedings of the 27th acm
  sigkdd conference on knowledge discovery \& data mining}, 3106--3116.

\bibitem[{Khan \protect\BIBand{} Ugander(2023)}]{khan2023adaptive}
Khan S, Ugander J (2023) Adaptive normalization for ipw estimation.
  \emph{Journal of Causal Inference} 11(1):20220019.

\bibitem[{Kohavi et~al.(2007)Kohavi, Henne, \protect\BIBand{}
  Sommerfield}]{kohavi2007practical}
Kohavi R, Henne RM, Sommerfield D (2007) Practical guide to controlled
  experiments on the web: listen to your customers not to the hippo.
  \emph{Proceedings of the 13th ACM SIGKDD international conference on
  Knowledge discovery and data mining}, 959--967.

\bibitem[{Lewis \protect\BIBand{} Rao(2015)}]{lewis2015unfavorable}
Lewis RA, Rao JM (2015) The unfavorable economics of measuring the returns to
  advertising. \emph{The Quarterly Journal of Economics} 130(4):1941--1973.

\bibitem[{Lin(2013)}]{lin2013agnostic}
Lin W (2013) Agnostic notes on regression adjustments to experimental data:
  Reexamining freedman's critique. \emph{The Annals of Applied Statistics}
  295--318.

\bibitem[{Lindley(1972)}]{lindley1972bayesian}
Lindley DV (1972) \emph{Bayesian statistics: A review} (SIAM).

\bibitem[{Little(1983)}]{little1983estimating}
Little RJ (1983) Estimating a finite population mean from unequal probability
  samples. \emph{Journal of the American Statistical Association}
  78(383):596--604.

\bibitem[{Little(2004)}]{little2004model}
Little RJ (2004) To model or not to model? competing modes of inference for
  finite population sampling. \emph{Journal of the American Statistical
  Association} 99(466):546--556.

\bibitem[{Lu et~al.(2024)Lu, Wang, \protect\BIBand{} Zhang}]{lu2024adjusting}
Lu X, Wang Y, Zhang Z (2024) Adjusting auxiliary variables under approximate
  neighborhood interference. \emph{arXiv preprint arXiv:2411.19789} .

\bibitem[{Manski(2004)}]{manski2004statistical}
Manski CF (2004) Statistical treatment rules for heterogeneous populations.
  \emph{Econometrica} 72(4):1221--1246.

\bibitem[{Manski(2013)}]{manski2013identification}
Manski CF (2013) Identification of treatment response with social interactions.
  \emph{The Econometrics Journal} 16(1):S1--S23.

\bibitem[{Neyman(1923)}]{neyman1923application}
Neyman J (1923) On the application of probability theory to agricultural
  experiments. essay on principles. \emph{Ann. Agricultural Sciences} 1--51.

\bibitem[{Owen(2013)}]{owen2013monte}
Owen AB (2013) Monte carlo theory, methods and examples.

\bibitem[{Quo{\ss} et~al.(2021)Quo{\ss}, Rudolph, Gomm, Wehrli,
  \protect\BIBand{} Bernauer}]{quoss2021swiss}
Quo{\ss} F, Rudolph L, Gomm S, Wehrli S, Bernauer T (2021) Swiss environmental
  panel study 2018-2021, wave 1-6, cumulative data. Dataset,
  \urlprefix\url{http://dx.doi.org/10.23662/FORS-DS-1220-2}, {E}TH Zurich,
  Institute of Science, Technology and Policy. Distributed by FORS, Lausanne.

\bibitem[{Ritzwoller(2025)}]{ritzwoller2025regression}
Ritzwoller DM (2025) Regression adjustments for disentangling spillover effects
  .

\bibitem[{Robins \protect\BIBand{} Rotnitzky(1995)}]{robins1995semiparametric}
Robins JM, Rotnitzky A (1995) Semiparametric efficiency in multivariate
  regression models with missing data. \emph{Journal of the American
  Statistical Association} 90(429):122--129.

\bibitem[{Robins et~al.(1994)Robins, Rotnitzky, \protect\BIBand{}
  Zhao}]{robins1994estimation}
Robins JM, Rotnitzky A, Zhao LP (1994) Estimation of regression coefficients
  when some regressors are not always observed. \emph{Journal of the American
  statistical Association} 89(427):846--866.

\bibitem[{Rosenbaum(1987)}]{rosenbaum1987model}
Rosenbaum PR (1987) Model-based direct adjustment. \emph{Journal of the
  American statistical Association} 82(398):387--394.

\bibitem[{Rosenbaum \protect\BIBand{} Rubin(1983)}]{rosenbaum1983central}
Rosenbaum PR, Rubin DB (1983) The central role of the propensity score in
  observational studies for causal effects. \emph{Biometrika} 70(1):41--55.

\bibitem[{Ross(2011)}]{ross2011fundamentals}
Ross N (2011) Fundamentals of stein’s method .

\bibitem[{Ross(2013)}]{ross2013simulation}
Ross SM (2013) \emph{Simulation} (Academic Press).

\bibitem[{Rubin(1974)}]{rubin1974estimating}
Rubin DB (1974) Estimating causal effects of treatments in randomized and
  nonrandomized studies. \emph{Journal of educational Psychology} 66(5):688.

\bibitem[{Rubin(1978)}]{rubin1978bayesian}
Rubin DB (1978) Bayesian inference for causal effects: The role of
  randomization. \emph{The Annals of statistics} 34--58.

\bibitem[{Rubin(1980)}]{rubin1980discussion}
Rubin DB (1980) Discussion of ``randomization analysis of experimental data in
  the fisher randomization test''' by d. basu. \emph{Journal of the American
  statistical association} 75:591--593.

\bibitem[{S{\"a}rndal et~al.(1978)S{\"a}rndal, Thomsen, Hoem, Lindley,
  Barndorff-Nielsen, \protect\BIBand{} Dalenius}]{sarndal1978design}
S{\"a}rndal CE, Thomsen I, Hoem JM, Lindley D, Barndorff-Nielsen O, Dalenius T
  (1978) Design-based and model-based inference in survey sampling [with
  discussion and reply]. \emph{Scandinavian Journal of Statistics} 27--52.

\bibitem[{Savage(1951)}]{savage1951theory}
Savage LJ (1951) The theory of statistical decision. \emph{Journal of the
  American Statistical association} 46(253):55--67.

\bibitem[{Scharfstein et~al.(1999)Scharfstein, Rotnitzky, \protect\BIBand{}
  Robins}]{scharfstein1999adjusting}
Scharfstein DO, Rotnitzky A, Robins JM (1999) Adjusting for nonignorable
  drop-out using semiparametric nonresponse models. \emph{Journal of the
  American Statistical Association} 94(448):1096--1120.

\bibitem[{Stoye(2009)}]{stoye2009minimax}
Stoye J (2009) Minimax regret treatment choice with finite samples.
  \emph{Journal of Econometrics} 151(1):70--81.

\bibitem[{Swaminathan \protect\BIBand{} Joachims(2015)}]{swaminathan2015self}
Swaminathan A, Joachims T (2015) The self-normalized estimator for
  counterfactual learning. \emph{advances in neural information processing
  systems} 28.

\bibitem[{Tan(2010)}]{tan2010bounded}
Tan Z (2010) Bounded, efficient and doubly robust estimation with inverse
  weighting. \emph{Biometrika} 97(3):661--682.

\bibitem[{Tchetgen~Tchetgen \protect\BIBand{}
  VanderWeele(2012)}]{tchetgen2012causal}
Tchetgen~Tchetgen EJ, VanderWeele TJ (2012) On causal inference in the presence
  of interference. \emph{Statistical methods in medical research} 21(1):55--75.

\bibitem[{Trotter \protect\BIBand{} Tukey(1956)}]{tukey1956conditional}
Trotter HF, Tukey JW (1956) Conditional monte carlo for normal samples.
  \emph{Proc. Symp. on Monte Carlo Methods}, 64--79 (John Wiley and Sons).

\bibitem[{Van Der~Laan \protect\BIBand{} Rubin(2006)}]{van2006targeted}
Van Der~Laan MJ, Rubin D (2006) Targeted maximum likelihood learning .

\bibitem[{Van Der~Vaart \protect\BIBand{} Wellner(1996)}]{van1996weak}
Van Der~Vaart AW, Wellner JA (1996) Weak convergence. \emph{Weak convergence
  and empirical processes: with applications to statistics}, 16--28 (Springer).

\bibitem[{Wager(2020)}]{wager2020stats}
Wager S (2020) Stats 361: Causal inference. URL:
  https://web.stanford.edu/~swager/stats361.pdf.

\bibitem[{Wald(1949)}]{wald1949statistical}
Wald A (1949) Statistical decision functions. \emph{The Annals of Mathematical
  Statistics} 165--205.

\bibitem[{Wang \protect\BIBand{} Li(2025)}]{wang2025covariate}
Wang X, Li S (2025) Covariate adjustment cannot hurt: Treatment effect
  estimation under interference with low-order outcome interactions.
  \emph{arXiv preprint arXiv:2509.03050} .

\bibitem[{Wooldridge(2010)}]{wooldridge2010econometric}
Wooldridge JM (2010) \emph{Econometric analysis of cross section and panel
  data} (MIT press).

\bibitem[{Wu(1981)}]{wu1981robustness}
Wu CF (1981) On the robustness and efficiency of some randomized designs.
  \emph{The Annals of Statistics} 1168--1177.

\bibitem[{Zhao(2023)}]{zhao2023adaptive}
Zhao J (2023) Adaptive neyman allocation. \emph{arXiv preprint
  arXiv:2309.08808} .

\bibitem[{Zhao(2024)}]{zhao2024experimental}
Zhao J (2024) Experimental design for causal inference through an optimization
  lens. \emph{Tutorials in Operations Research: Smarter Decisions for a Better
  World}, 146--188 (INFORMS).

\end{thebibliography}

\clearpage



\setcounter{page}{1}                         
\pagenumbering{arabic}
\renewcommand*{\thepage}{ec\arabic{page}}
\setcounter{section}{0}                      
\renewcommand{\thesection}{EC.\arabic{section}}   
\setcounter{equation}{0}
\renewcommand{\theequation}{EC.\arabic{equation}}
\setcounter{theorem}{0}
\renewcommand{\thetheorem}{EC.\arabic{theorem}}
\setcounter{lemma}{0}
\renewcommand{\thelemma}{EC.\arabic{lemma}}
\setcounter{example}{0}
\renewcommand{\theexample}{EC.\arabic{example}}
\setcounter{assumption}{0}
\renewcommand{\theassumption}{EC.\arabic{assumption}}
\setcounter{figure}{0}
\renewcommand{\thefigure}{EC.\arabic{figure}}
\setcounter{table}{0}
\renewcommand{\thetable}{EC.\arabic{table}}

\ECHead{Online Appendix}

\section{Causal Inference under SUTVA}
\label{sec:CausalInference}

\subsection{Experimental Design}
\label{sec:ExpDesign}

As Example~\ref{exa:SUTVA} illustrates, the causal inference problem under the Stable Unit Treatment Value Assumption (SUTVA, \citealt{rubin1980discussion, holland1986statistics}) is a special case of network interference. 
In this section, we provide a self-contained and specialized discussion to study the causal inference problem under SUTVA.

We consider the following causal inference problem.
Let there be $n$ units in a finite population. 
Let there be two versions of treatments.
We use ``treatment'' and ``control,'' or $1$ and $0$, respectively, to stand for these two versions of treatments. 
For each $i \in [n]$, let $W_i \in \{0,1\}$ stand for the treatment assignment that unit $i$ receives. 
We consider an experimental design setting where we determine the random treatment assignments $\bm{W} = (W_1, ..., W_n)^\top$ through a known joint probability distribution $\cW$. 
Following convention, we use $W_i$ for a random treatment assignment, and $w_i$ for one realization. 

Following the potential outcomes framework \citep{neyman1923application, rubin1974estimating} and under SUTVA, each unit $i$ has a pair of potential outcomes $(Y_i(1), Y_i(0))$.
We collect them into vector form $\bm{Y}(1) = (Y_1(1),Y_2(1),...,Y_n(1))^\top$ and $\bm{Y}(0) = (Y_1(0),Y_2(0),...,Y_n(0))^\top$.
Each observed outcome is related to its respective potential outcomes through 
\begin{align*}
Y_i = \left\{
\begin{aligned}
Y_i(1), & \quad \text{if } W_i=1, \\ 
Y_i(0), & \quad \text{if } W_i=0. \\
\end{aligned}\right.
\end{align*}
We are interested in estimating the average treatment effect of the finite population,
\begin{align}
\tau = \mu_n(1) - \mu_n(0) = \frac{1}{n} \sum_{i=1}^n \Big( Y_i(1) - Y_i(0) \Big). \label{eqn:ATE}
\end{align}

In this paper, we take the design of experiment $\cW$ as given. 
Denote the vector of treatment indicators as $\bm{D} = (\bI\{W_1=1\}, ..., \bI\{W_n=1\})^\top$.
Denote the marginal treatment probabilities as $\pi_i = \Pr_{\cW}(W_i=1)$ for any $i \in [n]$. 
Denote the following two diagonal matrices as $\bm{\Pi}(1) = \bm{\mathrm{diag}}(\pi_1, \pi_2, ..., \pi_n)$ and $\bm{\Pi}(0) = \bm{\mathrm{diag}}(1-\pi_1, 1-\pi_2, ..., 1-\pi_n)$, which reflect the marginal probabilities of $\bm{D}$. 
Denote the covariance matrix as $\bm{\Sigma} = \Var_{\cW}(\bm{D})$.
We consider designs of experiments $\cW$ that satisfy Assumption~\ref{asp:ExpDesign} below.

\begin{assumption}
\label{asp:ExpDesign}
The design of experiment $\cW$ satisfies:
\begin{enumerate}[label=(\roman*)]
\item \textbf{(Positivity)} The marginal treatment probabilities converge to zero at a slow rate, that is, there exist constants $\underline{\pi} > 0$ and $0 \leq \beta < 1$ such that for $i \in [n]$,
\begin{align*}
\pi_i \geq \underline{\pi} n^{-\beta} > 0, \qquad 1-\pi_i \geq \underline{\pi} n^{-\beta} > 0.
\end{align*}
\item \textbf{(Dependency Neighborhood)} The correlation between treatment indicators are restricted to some neighborhoods, that is, there exist constants $\overline{d} > 0$ and $0 \leq \alpha < 1$ such that for $i \in [n]$,
\begin{align*}
|\cN_i| \leq \overline{d} n^\alpha.
\end{align*}
\end{enumerate}
\end{assumption}

Assumption~\ref{asp:ExpDesign}-(i) is a standard assumption in the design based causal inference literature \citep{ding2023first, imbens2015causal, wager2020stats}.
Assumption~\ref{asp:ExpDesign}-(ii) is a mild assumption to describe weakly dependent random variables \citep{ross2011fundamentals}.
We will explicitly specify the required upper bounds on $\alpha$ and $\beta$ when we present our main results. 
We provide Example~\ref{exa:BernoulliExp} below to illustrate the above notations and assumptions.

\begin{example}[Bernoulli Randomized Experiments]
\label{exa:BernoulliExp}
For each unit $i \in [n]$, we independently assign unit $i$ into treatment or control with probability $\pi_i$ and $1-\pi_i$, respectively.
The covariance matrix is given by
\begin{align*}
\bm{\Sigma} = 
\begin{bmatrix}
\pi_1 (1 - \pi_1) & 0                 & \dots  & 0                 \\
0                 & \pi_2 (1 - \pi_2) & \dots  & 0                 \\
\vdots            & \vdots            & \ddots & \vdots            \\
0                 & 0                 & \dots  & \pi_n (1 - \pi_n)
\end{bmatrix}.
\end{align*}
This is exactly the same covariance matrix that we have seen in Example~\ref{exa:Bernoulli}.
A Bernoulli randomized experiment satisfies Assumption~\ref{asp:ExpDesign}-(i) if $\pi_i \geq \underline{\pi} > 0$. 
It also satisfies Assumption~\ref{asp:ExpDesign}-(ii) because $|\cN_i| = 1$ for $i \in [n]$.
\hfill \halmos
\end{example}

\subsection{Control Variate Estimators for Causal Inference}
\label{sec:CVEstimator:Causal}

Under SUTVA, one of the most popular ways of estimating $\tau$ is to consider the Horvitz-Thompson estimator defined as
\begin{align*}
\widehat{\tau}^{\HT} = \widehat{\mu}^{\HT}(1) - \widehat{\mu}^{\HT}(0) = \frac{1}{n} \sum_{i=1}^n \Big( \frac{Y_i \bI\{W_i=1\}}{\pi_i} - \frac{Y_i \bI\{W_i=0\}}{1-\pi_i} \Big).
\end{align*}
It is easy to see that the Horvitz-Thompson estimator is unbiased, that is, 
\begin{align*}
\bE_{\cW}\big[\widehat{\tau}^{\HT}\big] = \tau.
\end{align*}
Similar to survey sampling, the Horvitz-Thompson estimator for causal inference is sometimes known to have a large variance, especially when some of the marginal probabilities $\pi_i$ are close to $0$ or $1$. 
Again, we draw on the control variate techniques, take the baseline estimator $\widehat{\tau}^{\HT}$, and propose a family of control variate estimators to reduce the variance.

We consider the following control variate estimator using two control variates $\widehat{X}(1)$ and $\widehat{X}(0)$ and two coefficients $\gamma(1)$ and $\gamma(0)$,
\begin{align*}
\widehat{\tau}^\CV(\gamma(1), \gamma(0)) = \widehat{\tau}^\HT - \gamma(1) \Big(\widehat{X}(1) - \bE[\widehat{X}(1)]\Big) - \gamma(0) \Big(\widehat{X}(0) - \bE[\widehat{X}(0)]\Big).
\end{align*}
When $\gamma(1)$ and $\gamma(0)$ are two constant coefficients, the control variate estimator is unbiased, that is, 
\begin{align*}
\bE_{\cW}\big[ \widehat{\tau}^\CV(\gamma(1), \gamma(0)) \big] = \bE_{\cW}\big[ \widehat{\tau}^\HT \big] - \gamma(1) \bE_\cW\big[\widehat{X}(1) - \bE[\widehat{X}(1)]\big] - \gamma(0) \bE_\cW\big[\widehat{X}(0) - \bE[\widehat{X}(0)]\big] = \tau.
\end{align*}
If the coefficients $\gamma(1)$ and $\gamma(0)$ are selected properly, the control variate estimator has a smaller variance than the baseline Horvitz-Thompson estimator.
It is easy to show that the variance minimizing coefficients $\gamma^*(1)$ and $\gamma^*(0)$ for any generic control variates $\widehat{X}(1)$ and $\widehat{X}(0)$ are given by
\begin{align}
\begin{bmatrix}
\gamma^*(1) \\
\gamma^*(0) 
\end{bmatrix}
= 
\begin{bmatrix}
\Var\big(\widehat{X}(1)\big) & \Cov\big(\widehat{X}(1), \widehat{X}(0)\big) \\
\Cov\big(\widehat{X}(1), \widehat{X}(0)\big) & \Var\big(\widehat{X}(0)\big) 
\end{bmatrix}^{-1}
\cdot
\begin{bmatrix}
\Cov\big( \widehat{\tau}^\HT, \widehat{X}(1) \big) \\
\Cov\big( \widehat{\tau}^\HT, \widehat{X}(0) \big)
\end{bmatrix}.
\label{eqn:OPTgammas:CI}
\end{align}

Similar to survey sampling, we propose a family of control variates given in the following form
\begin{align}
\widehat{X}(1) = \frac{1}{n} \sum_{i=1}^n a_i(1) \bI\{W_i=1\}, && \widehat{X}(0) = \frac{1}{n} \sum_{i=1}^n a_i(0) \bI\{W_i=0\}, \label{eqn:CVCausalInference}
\end{align}
where $a_1(1), a_2(1), ..., a_n(1)$ and $a_1(0), a_2(0), ..., a_n(0)$ are any arbitrary constants that we can choose.
We collect these constants into vector form $\bm{a}(1) = (a_1(1), a_2(1), ..., a_n(1))^\top$ and $\bm{a}(0) = (a_1(0), a_2(0), ..., a_n(0))^\top$ and refer to them as the bases of the control variates.
The mean values of this family of control variates are given by $\bE_{\cW}\big[\widehat{X}(1)\big] = \frac{1}{n} \sum_{i=1}^n a_i(1) \pi_i$ and $\bE_{\cW}\big[\widehat{X}(0)\big] = \frac{1}{n} \sum_{i=1}^n a_i(0) (1-\pi_i)$ that are usually known in advance.

Similar to Lemma~\ref{lem:Optimalgamma}, for this family of control variates in \eqref{eqn:CVCausalInference}, we can derive the variance minimizing coefficients $\gamma^*(1)$ and $\gamma^*(0)$ in Lemma~\ref{lem:OPTgammas}. 
We present Lemma~\ref{lem:OPTgammas} using succinct notations. 
Recall the diagonal matrices $\bm{\Pi}(1), \bm{\Pi}(0)$ and the covariance matrix $\bm{\Sigma}$.
Denote the baseline outcomes $\bm{G} = \bm{\Pi}(1)^{-1} \bm{Y}(1) + \bm{\Pi}(0)^{-1} \bm{Y}(0)$ as similarly defined in \citet{bai2022optimality}.
Denote $\bm{A} = [\bm{a}(1), -\bm{a}(0)]$ to be a $(n \times 2)$ bases matrix.
Because $\bm{\Sigma}$ is a symmetric and positive semidefinite matrix, $\bm{A}^\top \bm{\Sigma} \bm{A}$ is also a symmetric and positive semidefinite matrix so its pseudo inverse matrix $(\bm{A}^\top \bm{\Sigma} \bm{A})^{-1}$ exists. 
Using these notations, we introduce Lemma~\ref{lem:OPTgammas} as follows.

\begin{lemma}[Variance Minimizing Coefficients]
\label{lem:OPTgammas}
In the causal inference setting under SUTVA, the variance minimizing coefficients $\gamma^*(1)$ and $\gamma^*(0)$ are given by
\begin{align*}
\big(\gamma^*(1), \gamma^*(0) \big)^\top = \big(\bm{A}^\top \bm{\Sigma} \bm{A}\big)^{-1} \bm{A}^\top \bm{\Sigma} \bm{G}. 
\end{align*}
and the variance of the control variate estimator under these coefficients is given by
\begin{align*}
\Var\big( \widehat{\tau}^{\CV}(\gamma^*(1), \gamma^*(0)) \big) = \frac{1}{n^2} \bigg( \bm{G}^\top \bm{\Sigma} \bm{G} - \bm{G}^\top \bm{\Sigma} \bm{A} \big(\bm{A}^\top \bm{\Sigma} \bm{A}\big)^{-1} \bm{A}^\top \bm{\Sigma} \bm{G} \bigg). 
\end{align*}
\end{lemma}

The proof of Lemma~\ref{lem:OPTgammas} is given in Section~\ref{sec:MissingProofs}.
Intuitively, we try to use the randomness in two control variates $\widehat{X}(1)$ and $\widehat{X}(0)$ to explain as much randomness in $\widehat{\tau}^{\HT}$ as possible. 
The variance $\Var( \widehat{\tau}^{\CV}(\gamma^*(1), \gamma^*(0)))$ is given by the residual from projecting the baseline outcomes $\bm{G}$ onto the column space of $\bm{A} = [\bm{a}(1), -\bm{a}(0)]$.
Because the variance minimizing coefficients $\gamma^*(1)$ and $\gamma^*(0)$ depend on unknown potential outcomes $\bm{Y}(1)$ and $\bm{Y}(0)$, we cannot directly observe $\gamma^*(1)$ and $\gamma^*(0)$; instead, we have to estimate $\gamma^*(1)$ and $\gamma^*(0)$ from the data.

\subsection{Estimating the Variance Minimizing Coefficients under SUTVA}
\label{sec:EstimatingCoefficients}

To estimate $(\gamma^*(1), \gamma^*(0))$ from the data, we can replace the unknown potential outcomes by their Horvitz-Thompson estimates.
Let $\widehat{G}_i = Y_i \big(\frac{\bI\{W_i=1\}}{\pi_i^2} + \frac{\bI\{W_i=0\}}{(1-\pi_i)^2}\big)$ and collect $\widehat{\bm{G}} = (\widehat{G}_1, ..., \widehat{G}_n)$.
Then the coefficients $(\gamma^*(1), \gamma^*(0))$ can be estimated using
\begin{align*}
\big(\widehat{\gamma}^\HT(1), \widehat{\gamma}^\HT(0)\big)^\top = \big(\bm{A}^\top \bm{\Sigma} \bm{A}\big)^{-1} \bm{A}^\top \bm{\Sigma} \widehat{\bm{G}}.
\end{align*}
Since $\widehat{G}_i$ is an unbiased estimator of $G_i$, $(\widehat{\gamma}^{\HT}(1), \widehat{\gamma}^{\HT}(0))$ as defined above are also unbiased estimators of $(\gamma^*(1), \gamma^*(0))$. 

Note that, $(\widehat{\gamma}^{\HT}(1), \widehat{\gamma}^{\HT}(0))$ obtained in this way are not constants, but random variables correlated with the control variates $\widehat{X}(1)$ and $\widehat{X}(0)$.
Using the estimated coefficients $(\widehat{\gamma}^{\HT}(1), \widehat{\gamma}^{\HT}(0))$ introduces a bias to the control variate estimator.
Nonetheless, we can show that the bias is small and that $\widehat{\tau}^\CV(\widehat{\gamma}^{\HT}(1), \widehat{\gamma}^{\HT}(0))$ is asymptotically equivalent to $\widehat{\tau}^\CV(\gamma^*(1), \gamma^*(0))$, as long as the bases matrix $\bm{A} = [\bm{a}(1), -\bm{a}(0)]$ satisfies certain regularity conditions. 
See Assumption~\ref{asp:RegularBasesCI} below.

\begin{assumption}
\label{asp:RegularBasesCI}
The bases matrix $\bm{A} = [\bm{a}(1), -\bm{a}(0)]$ satisfies two conditions:
\begin{enumerate}[label=(\roman*)]
\item \textbf{(Boundedness)} There exists a constant $\overline{a}$ such that for each $i \in [n]$,
\begin{align*}
\vert a_i(1) \vert \leq \overline{a}, \quad \vert a_i(0) \vert \leq \overline{a}.
\end{align*}
\item \textbf{(Non-degeneracy)} Let $\lambda_{\min}(\cdot)$ be the smallest eigenvalue of a matrix. 
There exists a constant $\underline{\lambda}_{\Sigma}$ such that 
\begin{align*}
\lambda_{\min}\big(\bm{A}^\top \bm{\Sigma} \bm{A}\big) \geq \underline{\lambda}_{\Sigma} n.
\end{align*}
\end{enumerate}
\end{assumption}

Assumption~\ref{asp:RegularBasesCI}-(i) is a standard boundedness assumption. 
Assumption~\ref{asp:RegularBasesCI}-(ii) is a standard non-degeneracy assumption on the control variates $\widehat{X}(1)$ and $\widehat{X}(0)$. 
Denote $\widehat{\bm{X}} = (\widehat{X}(1), \widehat{X}(0))^\top$.
Note that we have $\Var_{\cW}(\widehat{\bm{X}}) = n^{-2} \bm{A}^\top \bm{\Sigma} \bm{A}$, and therefore Assumption~\ref{asp:RegularBasesCI}-(ii) is equivalent to requiring that $\Var_{\cW}(\widehat{\bm{X}}) \succeq \underline{\lambda}_{\Sigma} n^{-1} \bm{I}_2.$
Under Assumption~\ref{asp:RegularBasesCI}, we establish Theorem~\ref{thm:AsymptoticCI} below.

\begin{theorem}
\label{thm:AsymptoticCI}
Under Assumptions~\ref{asp:ExpDesign} and~\ref{asp:RegularBasesCI} where in Assumption~\ref{asp:ExpDesign} we assume $3 \alpha + 4 \beta < 1$, and assuming the potential outcomes $\vert Y_i(1) \vert \leq \overline{y}$ and $\vert Y_i(0) \vert \leq \overline{y}$ are all bounded, the control variate estimator using the estimated coefficients $\widehat{\tau}^\CV(\widehat{\gamma}^\HT(1), \widehat{\gamma}^\HT(0))$ is a consistent estimator of $\tau$, that is, 
\begin{align*}
\lim_{n \to +\infty} \widehat{\tau}^\CV(\widehat{\gamma}^\HT(1), \widehat{\gamma}^\HT(0)) \xrightarrow{p} \tau.
\end{align*}
Additionally, if in Assumption~\ref{asp:ExpDesign} we assume $4 \alpha + 4 \beta < 1$, then the asymptotic variance of the control variate estimator using the estimated coefficients $\Var(\widehat{\tau}^\CV(\widehat{\gamma}^\HT(1), \widehat{\gamma}^\HT(0)))$ is the same as the asymptotic variance of the control variate estimator using the optimal coefficients $\Var(\widehat{\tau}^\CV(\gamma^*(1), \gamma^*(0)))$, that is, 
\begin{align*}
\lim_{n \to +\infty} n \Big( \Var\big( \widehat{\tau}^\CV(\widehat{\gamma}^\HT(1), \widehat{\gamma}^\HT(0)) \big) - \Var\big( \widehat{\tau}^\CV(\gamma^*(1), \gamma^*(0)) \big) \Big) \to 0.
\end{align*}
Additionally, if in Assumption~\ref{asp:ExpDesign}-(ii) we assume $7 \alpha + 6 \beta < 1$ and further assuming there exists a constant $\underline{c} > 0$ such that for sufficiently large $n$, $n \Var\big( \widehat{\tau}^{\CV}(\gamma^*(1), \gamma^*(0)) \big) \geq \underline{c}$, then the control variate estimator using the estimated coefficients $\widehat{\tau}^\CV(\widehat{\gamma}^\HT(1), \widehat{\gamma}^\HT(0))$ is asymptotically normal, that is, 
\begin{align*}
\lim_{n \to +\infty} \frac{\widehat{\tau}^\CV(\widehat{\gamma}^\HT(1), \widehat{\gamma}^\HT(0)) - \tau}
{\sqrt{\Var\big( \widehat{\tau}^\CV(\widehat{\gamma}^\HT(1), \widehat{\gamma}^\HT(0)) \big)}}
\xrightarrow{d} \cN(0,1).
\end{align*}
\end{theorem}

The proof of Theorem~\ref{thm:AsymptoticCI} is given in Section~\ref{sec:MissingProofs}. 
As we have stronger results, from consistency to asymptotic variance equivalence and to asymptotic normality, the requirement on the marginal probabilities of treatment and control and on the dependency neighborhoods becomes correspondingly stronger, with $\alpha$ strengthening from $3 \alpha + 4 \beta < 1$ to $4 \alpha + 4 \beta < 1$ and to $7 \alpha + 6 \beta < 1$. 
Taken together, Theorem~\ref{thm:AsymptoticCI} shows that estimating the coefficients $\gamma^*(1)$ and $\gamma^*(0)$ by $\widehat{\gamma}^{\HT}(1)$ and $\widehat{\gamma}^{\HT}(0)$ does not affect the first order asymptotic properties of the control variate estimator. 
In particular, estimating $\gamma^*(1)$ and $\gamma^*(0)$ by $\widehat{\gamma}^{\HT}(1)$ and $\widehat{\gamma}^{\HT}(0)$ leads to the same variance reduction asymptotically, as given by Lemma~\ref{lem:OPTgammas}. 
The magnitude of variance reduction is determined by the alignment between the baseline outcomes $\bm{G}$ and the basis $\bm{A}$. 
In Section~\ref{sec:OptimalControlVariates} below, we study how to choose a basis matrix $\bm{A}$ for the optimal variance reduction.

\subsection{The Optimal Control Variates}
\label{sec:OptimalControlVariates}

Now we focus on finding the optimal bases $\bm{a}(1)$ and $\bm{a}(0)$ that minimize $\Var\big( \widehat{\tau}^{\CV}(\gamma^*(1), \gamma^*(0)) \big)$ the variance of the control variate estimator for fixed $n$. 
If the potential outcomes $\bm{Y}(1)$ and $\bm{Y}(0)$ were known to take values $Y_i(1) = y_i(1)$ and $Y_i(0) = y_i(0)$ for any $i \in [n]$, then denote $\bm{g} = \bm{\Pi}(1)^{-1} \bm{y}(1) + \bm{\Pi}(0)^{-1} \bm{y}(0)$, and we see that the problem 
\begin{align*}
\min_{\bm{A}} \ \Var\big( \widehat{\tau}^{\CV}(\gamma^*(1), \gamma^*(0)) \big) \ = \ \min_{\bm{A}} \frac{1}{n^2} \bigg( \bm{g}^\top \bm{\Sigma} \bm{g} - \bm{g}^\top \bm{\Sigma} \bm{A} \big(\bm{A}^\top \bm{\Sigma} \bm{A}\big)^{-1} \bm{A}^\top \bm{\Sigma} \bm{g} \bigg)
\end{align*}
can be easily solved by choosing $\bm{A}$ such that $\bm{g}$ is in the column space of $\bm{A}$, such as $\bm{a}(1) = \bm{\Pi}(1)^{-1} \bm{y}(1)$ and $\bm{a}(0) = -\bm{\Pi}(0)^{-1}\bm{y}(0)$.
In this case, the variance $\Var\big( \widehat{\tau}^{\CV}(\gamma^*(1), \gamma^*(0)) \big)$ can be reduced to exactly zero. 

However, the above choice of $\bm{A}$ is infeasible because the potential outcomes $\bm{Y}(1)$ and $\bm{Y}(0)$ are unknown. 
Therefore, we find the optimal bases $\bm{a}(1)$ and $\bm{a}(0)$ under uncertainty of the potential outcomes, which leads to a decision making problem under uncertainty. 
We adopt a stochastic optimization perspective \citep{zhao2024experimental} and model each pair of unknown outcomes $Y_i(1), Y_i(0)$ to be i.i.d. sampled from an unknown joint distribution $\cY_{\PO}$.
See Assumption~\ref{asp:iidPO} below. 

\begin{assumption}
\label{asp:iidPO}
We assume that for each $i \in [n]$, the potential outcomes $(Y_i(1), Y_i(0))$ are i.i.d. sampled from an unknown joint distribution $\cY_{\PO}$, whose marginal distributions are denoted as
$\cY(1)$ and $\cY(0)$.
\end{assumption}

Denote the population means as $\mu(1) = \bE_{Y_i(1) \sim \cY(1)}[Y_i(1)]$ and $\mu(0) = \bE_{Y_i(0) \sim \cY(0)}[Y_i(0)]$, the population variances as $\sigma^2(1) = \Var_{Y_i(1) \sim \cY(1)}(Y_i(1))$ and $\sigma^2(0) = \Var_{Y_i(0) \sim \cY(0)}(Y_i(0))$, and the population covariance as $\sigma(1,0) = \Cov_{Y_i(1), Y_i(0) \sim \cY_{\PO}}(Y_i(1), Y_i(0))$. 
Denote $\cY_{\PO}^n$ to be the joint probability distribution of the outcomes $\bm{Y}(1), \bm{Y}(0)$.
Under Assumption~\ref{asp:iidPO}, we formulate the following optimization problem
\begin{align}
\bm{A}^* \in \argmin_{\bm{A}} \bE_{\bm{Y}(1), \bm{Y}(0) \sim \cY_{\PO}^n} \Big[ \Var\big( \widehat{\tau}^{\CV}(\gamma^*(1), \gamma^*(0)) \big) \Big]. \label{eqn:FormulationCausal}
\end{align}
The optimal control variates use the bases that solve \eqref{eqn:FormulationCausal}, which is given by Theorem~\ref{thm:OptimalControlVariates} below.
Recall the square root matrix $\bm{\Sigma}^{\frac{1}{2}}$ and the pseudo-inverse square root matrix $\bm{\Sigma}^{-\frac{1}{2}}$. 

\begin{theorem}[Optimal Control Variates]
\label{thm:OptimalControlVariates}
For notational simplicity, we reload the notation $\bm{M}$ and define $\bm{M}$ to be a symmetric matrix 
\begin{align*}
\bm{M} = & \bm{\Sigma}^{\frac{1}{2}} \bm{\Pi}(1)^{-1} \big( \mu^2(1) \bm{1}_n \bm{1}_n^\top + \sigma^2(1) \bm{I}_n \big) \bm{\Pi}(1)^{-1} \bm{\Sigma}^{\frac{1}{2}} \\
& + \bm{\Sigma}^{\frac{1}{2}} \bm{\Pi}(1)^{-1} \big( \mu(1) \mu(0) \bm{1}_n \bm{1}_n^\top + \sigma(1,0) \bm{I}_n \big) \bm{\Pi}(0)^{-1} \bm{\Sigma}^{\frac{1}{2}} \\
& + \bm{\Sigma}^{\frac{1}{2}} \bm{\Pi}(0)^{-1} \big( \mu(1) \mu(0) \bm{1}_n \bm{1}_n^\top + \sigma(1,0) \bm{I}_n \big) \bm{\Pi}(1)^{-1} \bm{\Sigma}^{\frac{1}{2}} \\
& + \bm{\Sigma}^{\frac{1}{2}} \bm{\Pi}(0)^{-1} \big( \mu^2(0) \bm{1}_n \bm{1}_n^\top + \sigma^2(0) \bm{I}_n \big) \bm{\Pi}(0)^{-1} \bm{\Sigma}^{\frac{1}{2}},
\end{align*}
where $\bm{I}_n$ is a $n \times n$ identity matrix and $\bm{1}_n$ is a $n$-dimensional vector with each element equal to $1$. 
Let $\lambda_l(\bm{M})$ be the $l$-th largest eigenvalue of $\bm{M}$, and let $\bm{u}_l(\bm{M})$ be the eigenvector that corresponds to $\lambda_l(\bm{M})$.
Under Assumptions~\ref{asp:ExpDesign}-(i) and~\ref{asp:iidPO}, the optimal bases are given by 
\begin{align*}
\bm{A}^* = \bm{\Sigma}^{-\frac{1}{2}} \big[\bm{u}_1(\bm{M}), \bm{u}_2(\bm{M})\big] \bm{C}
\end{align*}
where $\bm{C}$ is any invertible, symmetric and semidefinite matrix. 
And in expectation, the optimal variance reduction is given by
\begin{align*}
\bE_{\bm{Y}(1), \bm{Y}(0) \sim \cY_{\PO}^n} \Big[ \Var\big( \widehat{\tau}^{\HT} \big) - \Var\big( \widehat{\tau}^{\CV}(\gamma^*(1), \gamma^*(0)) \big) \Big] = \frac{\lambda_1(\bm{M}) + \lambda_2(\bm{M})}{n^2}.
\end{align*}
\end{theorem}

The proof of Theorem~\ref{thm:OptimalControlVariates} is given in Section~\ref{sec:MissingProofs}.
Similar to Theorem~\ref{thm:OptimalControlVariate}, Theorem~\ref{thm:OptimalControlVariates} also combines a design-based perspective and a model-based perspective, where we assume a model for the unknown potential outcomes $\bm{Y}(1), \bm{Y}(0)$ to guide the choice of bases. 
But we rely on the experimental design as the only source of randomness when we estimate the average treatment effect. 
We point out that Assumption~\ref{asp:iidPO} is only required to show optimality of the bases $\bm{A}^*$.
The basis $\bm{A}^*$ is still well defined without Assumption~\ref{asp:iidPO}.

In Theorem~\ref{thm:OptimalControlVariates}, the optimal bases $\bm{A}^*$ depends on the unknown moments of the joint distribution $\cY_{\PO}$. 
We can perform sample splitting to estimate the moments of the joint distribution $\cY_{\PO}$, and then use them to find the optimal bases $\bm{A}^*$. 
In particular, we can estimate the population means $\mu(1)$ and $\mu(0)$, as well as the population variances $\sigma^2(1)$ and $\sigma^2(0)$.
But the population covariance $\sigma(1,0)$ is not directly estimable from the data. 
In practice, we propose to assume that the potential outcomes under treatment and control are independent so $\sigma(1,0) = 0$, and use the following surrogate matrix
\begin{align*}
\bm{M} = & \bm{\Sigma}^{\frac{1}{2}} \bm{\Pi}(1)^{-1} \big( \mu^2(1) \bm{1}_n \bm{1}_n^\top + \sigma^2(1) \bm{I}_n \big) \bm{\Pi}(1)^{-1} \bm{\Sigma}^{\frac{1}{2}} \\
& + \bm{\Sigma}^{\frac{1}{2}} \bm{\Pi}(1)^{-1} \big( \mu(1) \mu(0) \bm{1}_n \bm{1}_n^\top \big) \bm{\Pi}(0)^{-1} \bm{\Sigma}^{\frac{1}{2}} \\
& + \bm{\Sigma}^{\frac{1}{2}} \bm{\Pi}(0)^{-1} \big( \mu(1) \mu(0) \bm{1}_n \bm{1}_n^\top \big) \bm{\Pi}(1)^{-1} \bm{\Sigma}^{\frac{1}{2}} \\
& + \bm{\Sigma}^{\frac{1}{2}} \bm{\Pi}(0)^{-1} \big( \mu^2(0) \bm{1}_n \bm{1}_n^\top + \sigma^2(0) \bm{I}_n \big) \bm{\Pi}(0)^{-1} \bm{\Sigma}^{\frac{1}{2}},
\end{align*}

Although Corollary~\ref{coro:SUTVASpecialCase} in the network interference setting and Theorem~\ref{thm:OptimalControlVariates} in the SUTVA setting use different matrix notations, they are actually equivalent. 
See Example~\ref{exa:SUTVAnInterference} below.

\begin{example}[SUTVA, Example~\ref{exa:SUTVA} Continued]
\label{exa:SUTVAnInterference}
In Example~\ref{exa:SUTVA} we show that the network interference setting nests the SUTVA setting, and $\bm{\Omega}$ has a special block structure under SUTVA,
\begin{align}
\bm{\Omega} =
\begin{bmatrix}
\bm{\Sigma} & -\bm{\Sigma}\\
-\bm{\Sigma} & \bm{\Sigma}
\end{bmatrix}. \label{eqn:SpecialBlockStructure}
\end{align}

First, we compare Lemmas~\ref{lem:OPTgammas:general} and~\ref{lem:OPTgammas}. \vspace{3pt}
Note that there is a one-to-one correspondence between a set of bases in the network interference setting 
$\bm{B} =
\begin{bmatrix}
\bm{a}(1) & \bm{0}_n  \\
\bm{0}_n  & \bm{a}(0) 
\end{bmatrix} \vspace{5pt}$ 
and a set of bases in the SUTVA setting $\bm{A} = \big[\bm{a}(1),-\bm{a}(0)\big]$.
So we have $\bm{B}^\top \bm{\Omega} \bm{B} = \bm{A}^\top \bm{\Sigma} \bm{A}$.
Next, recall that $\bm{G} = \bm{\Pi}(1)^{-1} \bm{Y}(1) + \bm{\Pi}(0)^{-1} \bm{Y}(0)$. 
So we have $\bm{B}^\top \bm{\Omega} \bm{\Pi}^{-1} \bm{Y} = \bm{A}^\top \bm{\Sigma} \bm{G}$.
Therefore, Lemma~\ref{lem:OPTgammas:general} in the network interference setting nests Lemma~\ref{lem:OPTgammas} in the SUTVA setting as a special case. 

Second, we compare Corollary~\ref{coro:SUTVASpecialCase} and Theorem~\ref{thm:OptimalControlVariates}.
In the general network interference setting, the block diagonal constraint $\bm{B} \in \cB$ can make the optimization problem harder than the unconstrained problem. 
Under SUTVA, however, the special block structure in \eqref{eqn:SpecialBlockStructure} makes this constraint nonbinding: for any $2n \times 2$ bases matrix, we can explicitly construct a bases matrix satisfying the block diagonal constraint $\bm{B} \in \cB$ that attains exactly the same objective value. 
See Section~\ref{sec:MissingProofs} for details. 
This shows that Theorem~\ref{thm:OPTInterference} in the network interference setting nests Theorem~\ref{thm:OptimalControlVariates} in the SUTVA setting as a special case. 
\hfill \halmos
\end{example}

Theorem~\ref{thm:OptimalControlVariates} involves a scaling matrix $\bm{C}$ which makes the optimal bases $\bm{A}^*$ not unique.
We could always choose an invertible, symmetric and semidefinite matrix $\bm{C}$ that satisfies some nice regularity conditions. 
See Proposition~\ref{prop:OrthonormalOptA} below.

\begin{proposition}[Orthonormal Bases]
\label{prop:OrthonormalOptA}
Let $\bm{M}$ be as defined in Theorem~\ref{thm:OptimalControlVariates}. 
When $\bm{C} = n^{\frac{1}{2}} \bm{I}_2$ where $\bm{I}_2$ stands for a $2 \times 2$ identity matrix, the optimal bases $\bm{a}(1)$ and $\bm{a}(0)$ in Theorem~\ref{thm:OptimalControlVariates} are given by
\begin{align*}
\bm{a}(1) = n^{\frac{1}{2}} \bm{\Sigma}^{-\frac{1}{2}} \bm{u}_1(\bm{M}), \qquad \text{and} \qquad \bm{a}(0) = - n^{\frac{1}{2}} \bm{\Sigma}^{-\frac{1}{2}} \bm{u}_2(\bm{M}),
\end{align*}
and satisfies the following regularity conditions
\begin{align*}
\bm{a}(1)^\top \bm{\Sigma} \bm{a}(0) = 0, \qquad \text{and} \qquad \bm{a}(1)^\top \bm{\Sigma} \bm{a}(1) = \bm{a}(0)^\top \bm{\Sigma} \bm{a}(0) = n.
\end{align*}
\end{proposition}

The proof of Proposition~\ref{prop:OrthonormalOptA} is given in Section~\ref{sec:MissingProofs}.
Proposition~\ref{prop:OrthonormalOptA} recommends one specific set of bases to choose.
This set of bases satisfies two properties: (i) the two control variates $\widehat{X}(1)$ and $\widehat{X}(0)$ are uncorrelated, that is, $\Cov(\widehat{X}(1), \widehat{X}(0)) = 0$; and (ii) the two control variates $\widehat{X}(1)$ and $\widehat{X}(0)$ have appropriate variances, that is, $\Var(\widehat{X}(1)) = \Var(\widehat{X}(0)) = n^{-1}$.

To conclude this section, we consider a remarkable special case of Theorem~\ref{thm:OptimalControlVariates} when $\sigma^2(1) = \sigma^2(0) = \sigma(1,0) = 0$.

\begin{corollary}[Noiseless Potential Outcomes]
\label{coro:NoiselessPotentialOutcomes}
Assume that the potential outcomes $Y_i(1) = y(1)$ and $Y_i(0) = y(0)$ take the same unknown constants, then one set of optimal bases is given, for any $i \in [n]$, by
\begin{align*}
a_i(1) = \frac{1}{\pi_i}, \qquad \text{and} \qquad a_i(0) = - \frac{1}{1-\pi_i}.
\end{align*}
\end{corollary}

The proof of Corollary~\ref{coro:NoiselessPotentialOutcomes} is given in Section~\ref{sec:MissingProofs}. 
In this special case when all potential outcomes take the same two constants, the only fluctuation of the Horvitz-Thompson estimator comes from the random components $\frac{1}{n} \sum_{i=1}^n \frac{\bI\{W_i=1\}}{\pi_i}$ and $\frac{1}{n} \sum_{i=1}^n \frac{\bI\{W_i=0\}}{1-\pi_i}$. 
Corollary~\ref{coro:NoiselessPotentialOutcomes} suggests to correct it by using the exact same terms $\widehat{X}(1) = \frac{1}{n} \sum_{i=1}^n \frac{\bI\{W_i=1\}}{\pi_i}$ and $\widehat{X}(0) = \frac{1}{n} \sum_{i=1}^n \frac{\bI\{W_i=0\}}{1-\pi_i}$ as the control variates. 
This draws a connection between the control variate estimator and the Hajek estimator, which we will discuss in Section~\ref{sec:ConnectionsCI} below.

\subsection{Connections to Hajek Estimator}
\label{sec:ConnectionsCI}

\subsubsection*{A Naive Control Variate Estimator.}
As an alternative to the variance minimizing coefficients provided in Lemma~\ref{lem:OPTgammas}, we can consider a naive control variate estimator using coefficients $\gamma^{\Naive}(1)$ and $\gamma^{\Naive}(0)$ defined as follows.
\begin{align*}
\gamma^{\Naive}(1) = \frac{\Cov\big(\widehat{\mu}^{\HT}(1), \widehat{X}(1)\big)}{\Var\big(\widehat{X}(1)\big)}, \qquad \text{and} \qquad \gamma^{\Naive}(0) = - \frac{\Cov\big(\widehat{\mu}^{\HT}(0), \widehat{X}(0)\big)}{\Var\big(\widehat{X}(0)\big)}.
\end{align*}
The naive coefficients consider the treatment and control components of the Horvitz-Thompson estimator separately. 
Specifically, $\gamma^{\Naive}(1)$ is chosen as if $\widehat{X}(1)$ were only used to reduce the variance of $\widehat{\mu}^{\HT}(1)$, and $\gamma^{\Naive}(0)$ is chosen as if $\widehat{X}(0)$ were only used to reduce the variance of $\widehat{\mu}^{\HT}(0)$. 
By doing so, it ignores the correlation between $\widehat{\mu}^{\HT}(0)$ and $\widehat{X}(1)$, as well as between $\widehat{\mu}^{\HT}(1)$ and $\widehat{X}(0)$.
Therefore, the variance of the naive control variate estimator is greater or equal to the optimal control variate estimator. 
The magnitude of this variance gap depends on the choice of bases $\bm{a}(1)$ and $\bm{a}(0)$. 
We illustrate this variance gap in Examples~\ref{exa:NaiveBenchmarkOrthonormalBases} and~\ref{exa:NaiveBenchmarkIPWBases} below. 

\begin{example}
\label{exa:NaiveBenchmarkOrthonormalBases}
Let $\bm{M}$ be as defined in Theorem~\ref{thm:OptimalControlVariates}. 
Using the bases
\begin{align*}
\bm{a}(1) = n^{\frac{1}{2}} \bm{\Sigma}^{-\frac{1}{2}} \bm{u}_1(\bm{M}), \qquad \text{and} \qquad \bm{a}(0) = - n^{\frac{1}{2}} \bm{\Sigma}^{-\frac{1}{2}} \bm{u}_2(\bm{M}),
\end{align*}
we have
\begin{multline*}
\Var\big( \widehat{\tau}^{\CV}(\gamma^{\Naive}(1), \gamma^{\Naive}(0)) \big) - \Var\big( \widehat{\tau}^{\CV}(\gamma^*(1), \gamma^*(0)) \big) \\
= \ \frac{1}{n^2} \bigg( \big(\bm{Y}(0)^\top \bm{\Pi}(0)^{-1} \bm{\Sigma}^{\frac{1}{2}} \bm{u}_1(\bm{M})\big)^2 + \big(\bm{Y}(1)^\top \bm{\Pi}(1)^{-1} \bm{\Sigma}^{\frac{1}{2}} \bm{u}_2(\bm{M})\big)^2 \bigg) \geq 0.
\end{multline*}
\end{example}

\begin{example}
\label{exa:NaiveBenchmarkIPWBases}
Denote $\Theta(1) = \frac{n \sum_{i=1}^n \frac{1-\pi_i}{\pi_i} Y_i(1)}{\sum_{i=1}^n \frac{1-\pi_i}{\pi_i}} - \sum_{i=1}^n Y_i(1)$ and $\Theta(0) = \frac{n \sum_{i=1}^n \frac{\pi_i}{1-\pi_i} Y_i(0)}{\sum_{i=1}^n \frac{\pi_i}{1-\pi_i}} - \sum_{i=1}^n Y_i(0)$.
Using the bases
\begin{align*}
a_i(1) = \frac{1}{\pi_i}, \qquad \text{and} \qquad a_i(0) = - \frac{1}{1-\pi_i},
\end{align*}
for any $i \in [n]$, we have
\begin{multline*}
\Var\big( \widehat{\tau}^{\CV}(\gamma^{\Naive}(1), \gamma^{\Naive}(0)) \big) - \Var\big( \widehat{\tau}^{\CV}(\gamma^*(1), \gamma^*(0)) \big) \\
= \frac{1}{n^2 \big(\sum_{i=1}^n \frac{1-\pi_i}{\pi_i} \sum_{i=1}^n \frac{\pi_i}{1-\pi_i} - n^2\big)} \sum_{i=1}^n \bigg( \sqrt{\frac{1-\pi_i}{\pi_i}} \Theta(1) - \sqrt{\frac{\pi_i}{1-\pi_i}} \Theta(0) \bigg)^2 \geq 0,
\end{multline*}
where the inequality takes equality if and only if one of the following two conditions hold:
\begin{enumerate}
\item $\pi_i = \pi$ for any $i \in [n]$;
\item $Y_i(1) = Y(1)$ and $Y_i(0) = Y(0)$ for any $i \in [n]$.
\end{enumerate}
\end{example}

The proofs of Examples~\ref{exa:NaiveBenchmarkOrthonormalBases} and~\ref{exa:NaiveBenchmarkIPWBases} are given in Section~\ref{sec:MissingProofs}. 
Examples~\ref{exa:NaiveBenchmarkOrthonormalBases} and~\ref{exa:NaiveBenchmarkIPWBases} quantify the loss in variance from using the naive coefficients $\gamma^{\Naive}(1)$ and $\gamma^{\Naive}(0)$ instead of the jointly optimal coefficients $\gamma^*(1)$ and $\gamma^*(0)$. 
In other words, they quantify the loss in variance from ignoring the correlation between $\widehat{\mu}^{\HT}(0)$ and $\widehat{X}(1)$, as well as between $\widehat{\mu}^{\HT}(1)$ and $\widehat{X}(0)$.
But as Example~\ref{exa:NaiveBenchmarkIPWBases} suggests, the loss in variance is equal to zero under the special bases $a_i(1) = \frac{1}{\pi_i}$ and $a_i(0) = - \frac{1}{1-\pi_i}$ for any $i \in [n]$ and when $\pi_i = \pi$ for any $i \in [n]$.

\subsubsection*{Hajek Estimator.}
In the causal inference setup and under SUTVA, the Hajek estimator is defined as 
\begin{align*}
\widehat{\tau}^\Hajek = \frac{\widehat{\mu}^\HT(1)}{\widehat{1}(1)} - \frac{\widehat{\mu}^\HT(0)}{\widehat{1}(0)}, \quad \text{where} \quad \widehat{1}(1) = \frac{1}{n} \sum_{i=1}^n \frac{\bI\{W_i=1\}}{\pi_i}, \quad \widehat{1}(0) = \frac{1}{n} \sum_{i=1}^n \frac{\bI\{W_i=0\}}{1-\pi_i}. 
\end{align*}
The Hajek estimator in the causal inference setup simply combines two Hajek estimators in survey sampling to estimate the average treatment effect. 
Following similar discussions as in Section~\ref{sec:Connections}, we make a first order approximation by using Taylor expansion around $\widehat{1}(1) \approx 1$ and $\widehat{1}(0) \approx 1$, which gives
\begin{align*}
\widehat{\tau}^{\Hajek} \approx \widehat{\tau}^{\HT} - \widehat{\mu}^{\HT}(1) \Big( \frac{1}{n} \sum_{i=1}^n \frac{\bI\{W_i=1\}}{\pi_i} - 1 \Big) + \widehat{\mu}^{\HT}(0) \Big( \frac{1}{n} \sum_{i=1}^n \frac{\bI\{W_i=0\}}{1-\pi_i} - 1 \Big).
\end{align*}
On the other hand, using the bases $a_i(1) = \frac{1}{\pi_i}$ and $a_i(0) = - \frac{1}{1-\pi_i}, \forall i \in [n]$, the control variate estimator using coefficients $\widehat{\gamma}^{\Naive}(1)$ and $\widehat{\gamma}^{\Naive}(0)$ is given as
\begin{align*}
\widehat{\tau}^\CV(\widehat{\gamma}^{\Naive}(1), \widehat{\gamma}^{\Naive}(0)) = \widehat{\tau}^\HT - \widehat{\gamma}^{\Naive}(1) \Big( \frac{1}{n} \sum_{i=1}^n \frac{\bI\{W_i=1\}}{\pi_i} - 1 \Big) + \widehat{\gamma}^{\Naive}(0) \Big( \frac{1}{n} \sum_{i=1}^n \frac{\bI\{W_i=0\}}{1-\pi_i} - 1 \Big),
\end{align*}
where we use the Horvitz-Thompson style estimators $\widehat{\gamma}^{\Naive}(1)$ and $\widehat{\gamma}^{\Naive}(0)$ to estimate the coefficients $\gamma^{\Naive}(1)$ and $\gamma^{\Naive}(0)$ in the naive estimator.
Under Bernoulli randomized experiments and when all the treatment probabilities are equal $\pi_i = \pi$, we have
\begin{align*}
\widehat{\gamma}^{\Naive}(1) = \frac{1}{n} \sum_{i=1}^n \frac{Y_i \bI\{W_i=1\}}{\pi} = \widehat{\mu}^\HT(1), \qquad \text{and} \qquad \widehat{\gamma}^{\Naive}(0) = \frac{1}{n} \sum_{i=1}^n \frac{Y_i \bI\{W_i=0\}}{1-\pi} = \widehat{\mu}^\HT(0).
\end{align*}
Therefore, the first order approximation to the Hajek estimator is a special case of the control variate estimator using the bases $a_i(1) = \frac{1}{\pi_i}$ and $a_i(0) = - \frac{1}{1-\pi_i}, \forall i \in [n]$. 
Combining with Corollary~\ref{coro:NoiselessPotentialOutcomes} and Example~\ref{exa:NaiveBenchmarkIPWBases}, it shows that the first order approximation to the Hajek estimator is asymptotically optimal when the outcomes have near-zero variance, and when data is collected under Bernoulli randomized experiments with all the treatment probabilities being equal.

\section{Additional Matrix Definitions}
\label{sec:AdditionalDefn}

For any symmetric and positive semidefinite matrix $\bm{E}$, it has an eigen-decomposition in the form of 
\begin{align*}
\bm{E} = \bm{U} \bm{\Lambda} \bm{U}^{-1},
\end{align*}
where $\bm{\Lambda} = \bm{\mathrm{diag}}(\lambda_1, \lambda_2, ..., \lambda_n)$ is a diagonal matrix consisting of all the $n$ eigenvalues and each eigenvalue $\lambda_i \geq 0$ is non-negative, and $\bm{U}$ is an orthonormal matrix.
Using the above eigen-decomposition, define the square root matrix as
\begin{align*}
\bm{E}^{\frac{1}{2}} = \bm{U} \bm{\Lambda}^{\frac{1}{2}} \bm{U}^{-1}
\end{align*}
where $\bm{\Lambda}^{\frac{1}{2}} = \bm{\mathrm{diag}}(\lambda_1^{\frac{1}{2}}, \lambda_2^{\frac{1}{2}}, ..., \lambda_n^{\frac{1}{2}})$.
We have $\bm{E}^{\frac{1}{2}} \bm{E}^{\frac{1}{2}} = \bm{E}$.

For any symmetric and positive semidefinite matrix $\bm{E}$, define the pseudo-inverse matrix as
\begin{align*}
\bm{E}^{-1} = \bm{U} \bm{\Lambda}^{-1} \bm{U}^{-1}
\end{align*}
where $\bm{\Lambda}^{-1} = \bm{\mathrm{diag}}(\eta_1, \eta_2, ..., \eta_n)$ is a diagonal matrix where each element $\eta_i$ is defined as
\begin{align*}
\eta_i = \left\{
\begin{aligned}
& \lambda_i^{-1}, & \text{ if } \lambda_i > 0, \\
& 0, & \text{ if } \lambda_i = 0.
\end{aligned}
\right.
\end{align*}
We have $\bm{E}^{-1} \bm{E} = \bm{U} \bI^+ \bm{U}^{-1}$, where $\bI^+ = \bm{\mathrm{diag}}(\bI\{\lambda_1 > 0\}, \bI\{\lambda_2 > 0\}, ..., \bI\{\lambda_n > 0\})$ is a diagonal matrix and $\lambda_i$ stands for the $i$-th largest eigenvalue of matrix $\bm{E}$.

For any symmetric and positive semidefinite matrix $\bm{E}$, define the pseudo-inverse square root matrix as
\begin{align*}
\bm{E}^{-\frac{1}{2}} = \bm{U} \bm{\Lambda}^{-\frac{1}{2}} \bm{U}^{-1}
\end{align*}
where $\bm{\Lambda}^{-\frac{1}{2}} = \bm{\mathrm{diag}}(\eta_1, \eta_2, ..., \eta_n)$ is a diagonal matrix where each element $\eta_i$ is defined as
\begin{align*}
\eta_i = \left\{
\begin{aligned}
& \lambda_i^{-\frac{1}{2}}, & \text{ if } \lambda_i > 0, \\
& 0, & \text{ if } \lambda_i = 0.
\end{aligned}
\right.
\end{align*}
We have $\bm{E}^{-\frac{1}{2}} \bm{E}^{\frac{1}{2}} = \bm{U} \bI^+ \bm{U}^{-1}$, where $\bI^+ = \bm{\mathrm{diag}}(\bI\{\lambda_1 > 0\}, \bI\{\lambda_2 > 0\}, ..., \bI\{\lambda_n > 0\})$ is a diagonal matrix and $\lambda_i$ stands for the $i$-th largest eigenvalue of matrix $\bm{E}$.

\section{Additional Examples}
\label{sec:AdditionalExamples}

We consider Example~\ref{exa:Clustered} below in the survey sampling setup.
\begin{example}[Clustered Sampling]
\label{exa:Clustered}
We partition $n$ units into equal size clusters where the size of each cluster is $s$.
We independently select each cluster with probability $\frac{1}{2}$, and, if a cluster is selected, we sample all the units within the cluster.
Define the $s \times s$ matrix $\bm{\Sigma}_{s} = \frac{1}{4} \bm{1}_s \bm{1}_s^\top$ where $\bm{1}_s$ is an $s$-dimensional vector with each element equal to $1$.
The covariance matrix is given by
\begin{align*}
\bm{\Sigma} = 
\begin{bmatrix}
\bm{\Sigma}_{s} & 0               & \dots  & 0               \\
0               & \bm{\Sigma}_{s} & \dots  & 0               \\
\vdots          & \vdots          & \ddots & \vdots          \\
0               & 0               & \dots  & \bm{\Sigma}_{s}
\end{bmatrix}.
\end{align*}
Clustered sampling satisfies Assumption~\ref{asp:DesignProbabilities}-(i) because $\pi_i \geq \frac{1}{2}$.
Note that $|\cN_i| = s$ for $i \in [n]$.
It satisfies Assumption~\ref{asp:DesignProbabilities}-(ii) if cluster size $s$ is small (e.g., $s < n^\alpha$);
it does not satisfy Assumption~\ref{asp:DesignProbabilities}-(ii) if cluster size $s$ is large (e.g., $s \geq n^\alpha$).
\hfill \halmos
\end{example}

\section{Useful Lemmas}
\label{sec:UsefulLemmas}

\begin{lemma}[Dependency Neighborhood CLT, \citet{ross2011fundamentals} Theorem 3.5]
\label{lem:DependencyNeighborhoodCLT}
Let $\xi_1, \xi_2, ..., \xi_n$ be random variables with $\bE[\xi_i] = 0$ and $\bE[\xi_i^4] < +\infty$, and define $\sigma_n = \sqrt{\Var\Big(\sum_{i=1}^n \xi_i\Big)}$ and $S_n = \frac{\sum_{i=1}^n \xi_i}{\sigma_n}$.
Let $\xi_1, \xi_2, ..., \xi_n$ satisfy the dependency neighborhood assumption (i.e., Assumptions~\ref{asp:DesignProbabilities}-(ii), ~\ref{asp:ExpDesign:general}-(ii), and ~\ref{asp:ExpDesign}-(ii) in the paper) with $|\cN_i| \leq d_n$ for all $i \in [n]$.
Then for $Z$ a standard normal random variable,
\begin{align*}
d_W(S_n, Z) \leq \frac{d_n^2}{\sigma_n^3} \sum_{i=1}^n \bE[\vert \xi_i \vert^3] + \frac{26^\frac{1}{2} d_n^\frac{3}{2}}{\pi^\frac{1}{2} \sigma_n^2} \Big(\sum_{i=1}^n \bE[\xi_i^4]\Big)^\frac{1}{2},
\end{align*}
where $D_W(S_n, Z)$ stands for the Wasserstein distance between $S_n$ and $Z$.
\end{lemma}
Note that, in Lemma~\ref{lem:DependencyNeighborhoodCLT} $\pi \approx 3.14$. This should not be confused with the marginal probabilities $\pi_i$ or their lower bound $\underline{\pi}$.

\begin{lemma}[Rayleigh Quotient, \citet{horn2012matrix} Theorem 4.2.2]
\label{lem:RayleighQuotient}
Let $\cV$ be a unit sphere defined as $\cV = \big\{ \bm{v} \in \bR^n \big\vert \| \bm{v} \|_2 = 1 \big\}$.
For any symmetric real matrix $\bm{E} \in \bR^{n \times n}$, its largest eigenvalue is given by 
\begin{align*}
\lambda_1(\bm{E}) = \max_{\bm{v} \in \cV} \ \bm{v}^\top \bm{E} \bm{v}.
\end{align*}
Its principal eigenvector, that is, the eigenvector corresponding to the largest eigenvalue, is given by
\begin{align*}
\bm{v}_1(\bm{E}) = \argmax_{\bm{v} \in \cV} \ \bm{v}^\top \bm{E} \bm{v}.
\end{align*}
\end{lemma}

\begin{lemma}[Fan's Principle, \citet{fan1949theorem} Theorem 1]
\label{lem:FanPrinciple}
Let $p < n$ be a constant. 
Let $\bm{V} \in \bR^{n \times p}$ be a matrix whose column vectors $\bm{v}_1, ..., \bm{v}_p$ are orthonormal, that is, $\bm{V} \in \cV = \{\bm{V} \vert \bm{V}^\top \bm{V} = \bm{I}_p\}$, where $\bm{I}_p$ is a $(p \times p)$ identity matrix. 
For any symmetric real matrix $\bm{E} \in \bR^{n \times n}$, we have
\begin{align*}
\max_{\bm{V} \in \cV} \ \Tr(\bm{V}^\top \bm{E} \bm{V}) = \sum_{i=1}^p \lambda_i(\bm{E}),
\end{align*}
where $\Tr(\bm{V}^\top \bm{E} \bm{V}) = \sum_{i=1}^p \bm{v}_i^\top \bm{E} \bm{v}_i$ stands for the trace of matrix $\bm{V}^\top \bm{E} \bm{V}$, and $\lambda_1(\bm{E})$, $\lambda_2(\bm{E})$, $...$, $\lambda_p(\bm{E})$ stand for the largest $p$ eigenvalues of $\bm{E}$. 
The optimal solution to the above problem is given by $\bm{v}_1 = \bm{v}_1(\bm{E}), \bm{v}_2 = \bm{v}_2(\bm{E}), ..., \bm{v}_p = \bm{v}_p(\bm{E})$, the eigenvectors corresponding to the largest $p$ eigenvalues. 
\end{lemma}

%

\section{Missing Proofs}
\label{sec:MissingProofs}

\subsection{Proof of Lemma~\ref{lem:Optimalgamma}}

\proof{Proof of Lemma~\ref{lem:Optimalgamma}.}
We follow \eqref{eqn:gammaexpression} to find the expression of $\gamma^*$.
We start with $\Cov\big(\widehat{\mu}^{\HT}, \widehat{X}\big)$.
\begin{align*}
\Cov\big(\widehat{\mu}^{\HT}, \widehat{X}\big) = & \ \bE\big[ \widehat{\mu}^{\HT} \cdot \widehat{X} \big] - \bE\big[ \widehat{\mu}^{\HT} \big] \cdot \bE\big[ \widehat{X} \big] \\
= & \ \frac{1}{n^2} \bE\bigg[ \sum_{i=1}^n \frac{Y_i \bI\{W_i=1\}}{\pi_i} \cdot \sum_{i=1}^n a_i \bI\{W_i=1\} \bigg] - \frac{1}{n^2} \sum_{i=1}^n Y_i \cdot \sum_{i=1}^n a_i \pi_i \\
= & \ \frac{1}{n^2} \sum_{i=1}^n Y_i a_i \bE\bigg[ \frac{\bI\{W_i=1\}}{\pi_i} \bigg] + \frac{1}{n^2} \sum_{i=1}^n \sum_{j \ne i} Y_i a_j \bE\bigg[ \frac{\bI\{W_i=1\} \bI\{W_j=1\}}{\pi_i} \bigg] \\
& \quad - \frac{1}{n^2} \sum_{i=1}^n Y_i a_i \pi_i - \frac{1}{n^2} \sum_{i=1}^n \sum_{j \ne i} Y_i a_j \pi_j \\
= & \ \frac{1}{n^2} \sum_{i=1}^n Y_i a_i \big( 1 - \pi_i \big) + \frac{1}{n^2} \sum_{i=1}^n \sum_{j \ne i} Y_i a_j \Big( \frac{\pi_{ij}}{\pi_i} - \pi_j \Big) \\
= & \ \frac{1}{n^2} \sum_{i=1}^n \frac{Y_i}{\pi_i} (\pi_i - \pi_i^2) a_i + \frac{1}{n^2} \sum_{i=1}^n \sum_{j \ne i} \frac{Y_i}{\pi_i} (\pi_{ij} - \pi_i \pi_j) a_j \\
= & \ \frac{1}{n^2} \bm{Y}^\top \bm{\Pi}^{-1} \bm{\Sigma} \bm{a}.
\end{align*}
Similarly, we have
\begin{align*}
\Var\big( \widehat{X} \big) = & \ \bE\big[ \widehat{X}^2 \big] - \bE\big[ \widehat{X} \big]^2 \\
= & \ \frac{1}{n^2} \bE\bigg[ \Big(\sum_{i=1}^n a_i \bI\{W_i=1\}\Big)^2 \bigg] - \frac{1}{n^2} \bigg(\sum_{i=1}^n a_i \pi_i\bigg)^2 \\
= & \ \frac{1}{n^2} \sum_{i=1}^n a_i^2 \bE\big[ \bI\{W_i=1\} \big] + \frac{1}{n^2} \sum_{i=1}^n \sum_{j \ne i} a_i a_j \bE\big[ \bI\{W_i=1\} \bI\{W_j=1\} \big] \\
& \quad - \frac{1}{n^2} \sum_{i=1}^n a_i^2 \pi_i^2 - \frac{1}{n^2} \sum_{i=1}^n \sum_{j \ne i} a_i \pi_i a_j \pi_j \\
= & \ \frac{1}{n^2} \sum_{i=1}^n a_i (\pi_i - \pi_i^2) a_i + \frac{1}{n^2} \sum_{i=1}^n \sum_{j \ne i} a_i (\pi_{ij} - \pi_i \pi_j) a_j \\
= & \ \frac{1}{n^2} \bm{a}^\top \bm{\Sigma} \bm{a}.
\end{align*}
Combining both, the variance minimizing coefficient can be expressed by
\begin{align*}
\gamma^* = \frac{\bm{Y}^\top \bm{\Pi}^{-1} \bm{\Sigma} \bm{a}}{\bm{a}^\top \bm{\Sigma} \bm{a}}.
\end{align*}

Next, we study the variance of the control variate estimator $\Var\big( \widehat{\mu}^{\CV}(\gamma^*) \big)$.
Note that 
\begin{align*}
\Var\big( \widehat{\mu}^{\CV}(\gamma^*) \big) = & \ \Var\big( \widehat{\mu}^{\HT} \big) - 2 \gamma^* \Cov\big(\widehat{\mu}^{\HT}, \widehat{X}\big) + (\gamma^*)^2\Var\big(\widehat{X}\big) \\
= & \ \Var\big( \widehat{\mu}^{\HT} \big) - \frac{\Cov\big(\widehat{\mu}^{\HT}, \widehat{X}\big)^2}{\Var\big(\widehat{X}\big)},
\end{align*}
where the second equality follows from \eqref{eqn:gammaexpression}.
Next, we have
\begin{align*}
\Var\big( \widehat{\mu}^{\HT} \big) = & \ \bE\big[ (\widehat{\mu}^{\HT})^2 \big] - \bE\big[ \widehat{\mu}^{\HT} \big]^2 \\
= & \ \frac{1}{n^2} \bE\bigg[ \Big(\sum_{i=1}^n \frac{Y_i \bI\{W_i=1\}}{\pi_i} \Big)^2 \bigg] - \frac{1}{n^2} \bigg( \sum_{i=1}^n Y_i \bigg)^2 \\
= & \ \frac{1}{n^2} \sum_{i=1}^n Y_i^2 \bE\bigg[ \frac{\bI\{W_i=1\}}{\pi_i^2} \bigg] + \frac{1}{n^2} \sum_{i=1}^n \sum_{j \ne i} Y_i Y_j \bE\bigg[ \frac{\bI\{W_i=1\} \bI\{W_j=1\}}{\pi_i \pi_j} \bigg] \\
& \quad - \frac{1}{n^2} \sum_{i=1}^n Y_i^2 - \frac{1}{n^2} \sum_{i=1}^n \sum_{j \ne i} Y_i Y_j \\
= & \ \frac{1}{n^2} \sum_{i=1}^n Y_i^2 \Big( \frac{1}{\pi_i} - 1 \Big) + \frac{1}{n^2} \sum_{i=1}^n \sum_{j \ne i} Y_i Y_j \Big( \frac{\pi_{ij}}{\pi_i \pi_j} - 1 \Big) \\
= & \ \frac{1}{n^2} \sum_{i=1}^n \frac{Y_i}{\pi_i} (\pi_i - \pi_i^2) \frac{Y_i}{\pi_i} + \frac{1}{n^2} \sum_{i=1}^n \sum_{j \ne i} \frac{Y_i}{\pi_i} (\pi_{ij} - \pi_i \pi_j) \frac{Y_j}{\pi_j} \\
= & \ \frac{1}{n^2} \bm{Y}^\top \bm{\Pi}^{-1} \bm{\Sigma} \bm{\Pi}^{-1} \bm{Y}.
\end{align*}
Plugging in all the expressions above, we have
\begin{align*}
\Var\big( \widehat{\mu}^{\CV}(\gamma^*) \big) = \frac{1}{n^2} \bigg( \bm{Y}^\top \bm{\Pi}^{-1} \bm{\Sigma} \bm{\Pi}^{-1} \bm{Y} - \frac{\big(\bm{Y}^\top \bm{\Pi}^{-1} \bm{\Sigma} \bm{a} \big)^2}{\bm{a}^\top \bm{\Sigma} \bm{a}} \bigg).
\end{align*}
This finishes the proof.
\hfill \halmos
\endproof

\subsection{Proof of Theorem~\ref{thm:AsymptoticBehavior}}

We present the following Lemma~\ref{lem:UBOptimalgamma} which will be useful in the proof of Theorem~\ref{thm:AsymptoticBehavior}.

\begin{lemma}
\label{lem:UBOptimalgamma}
Under Assumptions~\ref{asp:DesignProbabilities} and~\ref{asp:RegularBasis}, and assuming the outcomes $\vert Y_i \vert \leq \overline{y}$ are all bounded, the optimal coefficient $\gamma^*$ is bounded, that is, 
\begin{align*}
\vert \gamma^* \vert \leq \frac{\overline{y} \ \overline{d}^\frac{1}{2}}{\underline{\pi} \ \underline{c}_{\Sigma}^\frac{1}{2}} n^{\frac{\alpha}{2}+\beta}.
\end{align*}
\end{lemma}

\proof{Proof of Lemma~\ref{lem:UBOptimalgamma}.}
Note that,
\begin{multline*}
\vert \gamma^* \vert 
= \Big\vert \frac{\Cov(\widehat{\mu}^\HT, \widehat{X})}{\Var(\widehat{X})} \Big\vert
\leq \sqrt{\frac{\Var(\widehat{\mu}^\HT)}{\Var(\widehat{X})}}
= \sqrt{\frac{\bm{Y}^\top \bm{\Pi}^{-1} \bm{\Sigma} \bm{\Pi}^{-1} \bm{Y}}{\bm{a}^\top \bm{\Sigma} \bm{a}}} \\
\leq \sqrt{\frac{1}{\bm{a}^\top \bm{\Sigma} \bm{a}} \frac{\overline{y}^2 \overline{d}}{\underline{\pi}^2} n^{1+\alpha+2\beta}}
\leq \sqrt{\frac{\overline{y}^2}{\underline{\pi}^2 \underline{c}_{\Sigma}} n^{\alpha+2\beta}}
= \frac{\overline{y} \ \overline{d}^\frac{1}{2}}{\underline{\pi} \ \underline{c}_{\Sigma}^\frac{1}{2}} n^{\frac{\alpha}{2}+\beta},
\end{multline*}
where the first inequality is due to Cauchy-Schwarz inequality;
the second inequality is due to Assumption~\ref{asp:DesignProbabilities}-(i), the assumption that the outcomes $\vert Y_i \vert \leq \overline{y}$ are all bounded, and Assumption~\ref{asp:DesignProbabilities}-(ii);
the last inequality is due to Assumption~\ref{asp:RegularBasis}-(ii).
\hfill \halmos
\endproof

\

\noindent Now we prove Theorem~\ref{thm:AsymptoticBehavior} as follows.

\proof{Proof of Theorem~\ref{thm:AsymptoticBehavior}.}
We prove the three claims (consistency, asymptotic variance, and asymptotic normality) in Theorem~\ref{thm:AsymptoticBehavior} one by one.

\noindent \textbf{Claim 1: Consistency}. 
In the proof of this claim, we assume $3\alpha + 4\beta < 1$.
We prove consistency through two steps.
First, we show that $\lim_{n \to +\infty} \big( \widehat{\mu}^\CV(\widehat{\gamma}^\HT) - \widehat{\mu}^\CV(\gamma^*) \big) \xrightarrow{p} 0$.
To show this, recall that by definition,
\begin{align*}
\widehat{\gamma}^\HT = \frac{\widehat{\bm{Y}}^\top \bm{\Pi}^{-1} \bm{\Sigma} \bm{a}}{\bm{a}^\top \bm{\Sigma} \bm{a}}.
\end{align*}
For any $i \in [n]$, define
\begin{align*}
Z_i = \frac{Y_i}{\pi_i^2}(\bm{\Sigma} \bm{a})_i
\end{align*}
where $(\bm{\Sigma} \bm{a})_i$ stands for the $i$-th element of vector $\bm{\Sigma} \bm{a}$.
Note that we assume for any $i \in [n]$, $\vert Y_i \vert \leq \overline{y}$ (the additional assumption made in Theorem~\ref{thm:AsymptoticBehavior}), $\vert a_i \vert \leq \overline{a}$ (Assumption~\ref{asp:RegularBasis}-(i)), and $\pi_i \geq \underline{\pi} n^{-\beta}$ (Assumption~\ref{asp:DesignProbabilities}-(i)).
So define $\overline{z} = \dfrac{\overline{y} \ \overline{a} \ \overline{d}}{\underline{\pi}^2}$ and we have
\begin{align*}
\vert Z_i \vert \leq \overline{z} n^{\alpha+2\beta}.
\end{align*}

Using the definition of $Z_i$, the Horvitz-Thompson estimator $\widehat{\gamma}^\HT$ can be written as
\begin{align*}
\widehat{\gamma}^\HT = \frac{1}{\bm{a}^\top \bm{\Sigma} \bm{a}} \sum_{i=1}^n Z_i \bI\{W_i=1\}. 
\end{align*}
It is easy to see that $\widehat{\gamma}^\HT$ is unbiased,
\begin{multline*}
\bE_{\cW}\big[\widehat{\gamma}^\HT\big] = \frac{1}{\bm{a}^\top \bm{\Sigma} \bm{a}} \bE_{\cW}\big[ \widehat{\bm{Y}}^\top \bm{\Pi}^{-1} \bm{\Sigma} \bm{a}\big] \\
= \frac{1}{\bm{a}^\top \bm{\Sigma} \bm{a}} \bE_{\cW}\Big[ \sum_{i=1}^n Z_i \bI\{W_i=1\} \Big] 
= \frac{1}{\bm{a}^\top \bm{\Sigma} \bm{a}} \sum_{i=1}^n Z_i \pi_i 
= \frac{1}{\bm{a}^\top \bm{\Sigma} \bm{a}} \bm{Y}^\top \bm{\Pi}^{-1} \bm{\Sigma} \bm{a}
= \gamma^*.
\end{multline*}
Next, we calculate the variance of $\widehat{\gamma}^\HT$ as follows,
\begin{multline*}
\Var(\widehat{\gamma}^\HT) = \Var\Big( \frac{1}{\bm{a}^\top \bm{\Sigma} \bm{a}} \sum_{i=1}^n Z_i \bI\{W_i=1\} \Big) 
= \frac{1}{(\bm{a}^\top \bm{\Sigma} \bm{a})^2} \sum_{i=1}^n \sum_{j=1}^n Z_i Z_j (\pi_{ij} - \pi_i \pi_j) \\
\leq \frac{\overline{z}^2 n^{2\alpha+4\beta}}{(\bm{a}^\top \bm{\Sigma} \bm{a})^2} \sum_{i=1}^n \sum_{j=1}^n \vert \pi_{ij} - \pi_i \pi_j \vert 
\ \leq \frac{\overline{z}^2 \overline{d} n^{1+3\alpha+4\beta}}{(\bm{a}^\top \bm{\Sigma} \bm{a})^2}
\ \leq \frac{\overline{z}^2 \overline{d}}{\underline{c}_{\Sigma}^2} n^{3\alpha+4\beta-1},
\end{multline*}
where the first inequality is upper bounding $\vert Z_i \vert \leq \overline{z} n^{\alpha+2\beta}$;
the second inequality is due to Assumption~\ref{asp:DesignProbabilities}-(ii);
and the third inequality is due to Assumption~\ref{asp:RegularBasis}-(ii).
Using Chebyshev inequality, for any $\epsilon > 0$,
\begin{align*}
\Pr\big( \big\vert \widehat{\gamma}^\HT - \gamma^* \big\vert \geq \epsilon \big) \leq \frac{\Var(\widehat{\gamma}^\HT)}{\epsilon^2} \leq \frac{\overline{z}^2 \overline{d}}{\underline{c}_{\Sigma}^2} n^{3\alpha+4\beta-1} \cdot \frac{1}{\epsilon^2}.
\end{align*}

Similarly, we calculate the variance of $\widehat{X}$ as follows,
\begin{multline*}
\Var(\widehat{X}) = \Var\Big( \frac{1}{n} \sum_{i=1}^n a_i \bI\{W_i=1\} \Big) 
= \frac{1}{n^2} \sum_{i=1}^n \sum_{j=1}^n a_i a_j (\pi_{ij} - \pi_i \pi_j) \\
\leq \frac{\overline{a}^2}{n^2} \sum_{i=1}^n \sum_{j=1}^n \vert \pi_{ij} - \pi_i \pi_j \vert 
\ \leq \overline{a}^2 \overline{d} n^{\alpha-1},
\end{multline*}
where the first inequality is due to Assumption~\ref{asp:DesignProbabilities}-(i);
and the second inequality is due to Assumption~\ref{asp:DesignProbabilities}-(ii).
Using Chebyshev inequality, for any $\epsilon > 0$, 
\begin{align*}
\Pr\big( \big\vert \widehat{X} - \bE[\widehat{X}] \big\vert \geq \epsilon \big) \leq \frac{\Var(\widehat{X})}{\epsilon^2} \leq \overline{a}^2 \overline{d} n^{\alpha-1} \frac{1}{\epsilon^2}.
\end{align*}
Next, we focus on
\begin{multline*}
\Delta_n = \widehat{\mu}^\CV(\widehat{\gamma}^\HT) - \widehat{\mu}^\CV(\gamma^*) \\
= \big( \widehat{\mu}^\HT - \widehat{\gamma}^\HT (\widehat{X} - \bE[\widehat{X}]) \big) - \big( \widehat{\mu}^\HT - \gamma^* (\widehat{X} - \bE[\widehat{X}]) \big)
= - \big(\widehat{\gamma}^\HT - \gamma^*\big) (\widehat{X} - \bE[\widehat{X}]).
\end{multline*}
We have that, for any $\epsilon > 0$,
\begin{align*}
\lim_{n \to +\infty} \Pr\big( \vert \Delta_n \vert \geq \epsilon \big) \leq & \ \lim_{n \to +\infty} \bigg( \Pr\big( \vert \widehat{\gamma}^\HT - \gamma^* \vert \geq \epsilon^\frac{1}{2} \big) + \Pr\big( \vert \widehat{X} - \bE[\widehat{X}] \vert \geq \epsilon^\frac{1}{2} \big) \bigg)\\
\leq & \ \lim_{n \to +\infty} \bigg( \frac{\overline{z}^2 \overline{d}}{\underline{c}_{\Sigma}^2} n^{3\alpha+4\beta-1} \frac{1}{\epsilon} + \overline{a}^2 \overline{d} n^{\alpha-1} \frac{1}{\epsilon} \bigg) \\
= & \ 0,
\end{align*}
where the first inequality holds because for $\vert \Delta_n \vert \geq \epsilon$ to hold, either $\vert \widehat{\gamma}^\HT - \gamma^* \vert \geq \epsilon^\frac{1}{2}$ or $\vert \widehat{X} - \bE[\widehat{X}] \vert \geq \epsilon^\frac{1}{2}$ needs to hold;
the last equality holds because $3\alpha+4\beta<1$. 
So we have shown that
\begin{align}
\lim_{n \to +\infty} \Big( \widehat{\mu}^\CV(\widehat{\gamma}^\HT) - \widehat{\mu}^\CV(\gamma^*) \Big) \xrightarrow{p} 0. \label{eqn:ConsistencyComponent1}
\end{align}

Next, we show that $\widehat{\mu}^\CV(\gamma^*)$ is a consistent estimator of $\mu_n$.
To show this, recall that
\begin{align*}
\Pr\big( \big\vert \widehat{X} - \bE[\widehat{X}] \big\vert \geq \epsilon \big) \leq \frac{\Var(\widehat{X})}{\epsilon^2} \leq \overline{a}^2 \overline{d} n^{\alpha-1} \frac{1}{\epsilon^2}.
\end{align*}
Note also that
\begin{multline*}
\Var(\widehat{\mu}^\HT) 
= \Var\Big( \frac{1}{n} \sum_{i=1}^n \frac{Y_i}{\pi_i} \bI\{W_i = 1\} \Big) 
= \frac{1}{n^2} \sum_{i=1}^n \sum_{j=1}^n \frac{Y_i Y_j}{\pi_i \pi_j} (\pi_{ij} - \pi_i \pi_j) \\
\leq \frac{\overline{y}^2}{\underline{\pi}^2} n^{2\beta-2} \sum_{i=1}^n \sum_{j=1}^n \vert \pi_{ij} - \pi_i \pi_j \vert 
\ \leq \frac{\overline{y}^2 \overline{d}}{\underline{\pi}^2} n^{\alpha+2\beta-1},
\end{multline*}
where the first inequality is due to Assumption~\ref{asp:RegularBasis}-(i);
and the second inequality is due to Assumption~\ref{asp:DesignProbabilities}-(ii).
Using Chebyshev inequality, for any $\epsilon > 0$, 
\begin{align*}
\Pr\big( \big\vert \widehat{\mu}^\HT - \mu_n \big\vert \geq \epsilon \big) \leq \frac{\Var(\widehat{\mu}^\HT)}{\epsilon^2} \leq \frac{\overline{y}^2 \overline{d}}{\underline{\pi}^2} n^{\alpha+2\beta-1} \frac{1}{\epsilon^2}.
\end{align*}
Next, we focus on
\begin{align*}
\lim_{n \to +\infty} \Pr\big( \vert \widehat{\mu}^\CV(\gamma^*) - \mu_n \vert \geq \epsilon \big) = & \lim_{n \to +\infty} \Pr\Big( \big\vert (\widehat{\mu}^\HT - \mu_n) - \gamma^* (\widehat{X} - \bE[\widehat{X}]) \big\vert \geq \epsilon \Big) \\
\leq & \lim_{n \to +\infty} \bigg( \Pr\Big( \vert \widehat{\mu}^\HT - \mu_n \vert \geq \frac{\epsilon}{2} \Big) + \Pr\Big( \vert \gamma^* \vert \cdot \big\vert \widehat{X} - \bE[\widehat{X}] \big\vert \geq \frac{\epsilon}{2} \Big) \bigg) \\
\leq & \lim_{n \to +\infty} \bigg( \frac{\overline{y}^2 \overline{d}}{\underline{\pi}^2} n^{\alpha+2\beta-1} \frac{4}{\epsilon^2} + \overline{a}^2 \overline{d} n^{\alpha-1} \frac{4 \vert \gamma^* \vert^2}{\epsilon^2} \bigg) \\
\leq & \lim_{n \to +\infty} \bigg( \frac{\overline{y}^2 \overline{d}}{\underline{\pi}^2} n^{\alpha+2\beta-1} \frac{4}{\epsilon^2} + \frac{\overline{y}^2 \overline{a}^2 \overline{d}^3}{\underline{\pi}^2 \underline{c}_{\Sigma}} n^{2\alpha+2\beta-1} \frac{4}{\epsilon^2} \bigg) \\
= & \ 0,
\end{align*}
where the first inequality holds because for $\big\vert (\widehat{\mu}^\HT - \mu_n) - \gamma^* (\widehat{X} - \bE[\widehat{X}]) \big\vert \geq \epsilon$ to hold, either $\vert \widehat{\mu}^\HT - \mu_n \vert \geq \frac{\epsilon}{2}$ or $\vert \gamma^* \vert \cdot \big\vert \widehat{X} - \bE[\widehat{X}] \big\vert \geq \frac{\epsilon}{2}$ needs to hold;
the third inequality holds because of Lemma~\ref{lem:UBOptimalgamma};
the last equality holds because $\alpha+2\beta < 2\alpha+2\beta < 3\alpha+4\beta < 1$.
So we conclude that
\begin{align}
\lim_{n \to +\infty} \widehat{\mu}^\CV(\gamma^*) \xrightarrow{p} \mu_n. \label{eqn:ConsistencyComponent2}
\end{align}
Combining \eqref{eqn:ConsistencyComponent1} and \eqref{eqn:ConsistencyComponent2} we have
\begin{align*}
\lim_{n \to +\infty} \widehat{\mu}^\CV(\widehat{\gamma}^\HT) \xrightarrow{p} \mu_n.
\end{align*}

\noindent \textbf{Claim 2: Asymptotic Variance}. 
In the proof of this claim, we assume $4\alpha+4\beta<1$.

Recall that we define $\Delta_n = \widehat{\mu}^\CV(\widehat{\gamma}^\HT) - \widehat{\mu}^\CV(\gamma^*)$, so we have
\begin{align}
\Var\big(\widehat{\mu}^\CV(\widehat{\gamma}^\HT)\big) = \Var\big(\widehat{\mu}^\CV(\gamma^*)\big) + \Var\big(\Delta_n\big) + 2 \Cov\big(\widehat{\mu}^\CV(\gamma^*), \Delta_n\big). \label{eqn:VarianceDecomposition}
\end{align}
It suffices to prove that both $\lim_{n \to +\infty} n \Var\big(\Delta_n\big) = 0$ and $\lim_{n \to +\infty} n \Cov\big(\widehat{\mu}^\CV(\gamma^*), \Delta_n\big) = 0$.

We first show that $\lim_{n \to +\infty} n \Var\big(\Delta_n\big) = 0$.
It suffices to show that $\lim_{n \to +\infty} n \bE\big[\Delta_n^2\big] = 0$ because $0 \leq \Var\big(\Delta_n\big) \leq \bE\big[\Delta_n^2\big]$.
To show this, note that 
\begin{align*}
\bE\big[\Delta_n^2\big] = \bE\big[(\widehat{\gamma}^\HT - \gamma^*)^2 \cdot (\widehat{X} - \bE[\widehat{X}])^2\big] \leq \bE\big[(\widehat{\gamma}^\HT - \gamma^*)^4\big]^\frac{1}{2} \cdot \bE\big[(\widehat{X} - \bE[\widehat{X}])^4\big]^\frac{1}{2},
\end{align*}
where the inequality holds due to Cauchy-Schwarz inequality.
We next upper bound the two terms separately.
Note that
\begin{align}
\bE\big[(\widehat{\gamma}^\HT - \gamma^*)^4\big] = & \bE\bigg[\Big( \frac{1}{\bm{a}^\top \bm{\Sigma} \bm{a}} \sum_{i=1}^n Z_i (\bI\{W_i=1\} - \pi_i) \Big)^4\bigg] \nonumber \\
\leq & \frac{\overline{z}^4 n^{4\alpha+8\beta}}{\underline{c}_{\Sigma}^4 n^4} \bE\bigg[\Big( \sum_{i=1}^n \big( \bI\{W_i=1\} - \pi_i \big) \Big)^4\bigg] \nonumber \\
\leq & \frac{\overline{z}^4 n^{4\alpha+8\beta}}{\underline{c}_{\Sigma}^4 n^4} \ 125 \ \overline{d}^2 n^{2+2 \alpha} \nonumber \\
= & \frac{125 \overline{y}^4 \overline{a}^4 \overline{d}^6}{\underline{\pi}^4 \underline{c}_{\Sigma}^4} \ n^{6\alpha+8\beta-2}, \label{eqn:FourthMoment_gamma}
\end{align}
where the first inequality is using the upper bound that $|Z_i| \leq \overline{z} n^{\alpha+2\beta}$ and using Assumption~\ref{asp:RegularBasis}-(ii);
the second inequality is doing the combinatorial counting, upper bounding $|\bI\{W_i=1\} - \pi_i| \leq 1$, and using Assumption~\ref{asp:DesignProbabilities}-(ii).
Similarly,
\begin{align}
\bE\big[(\widehat{X} - \bE[\widehat{X}])^4\big] = & \bE\bigg[\Big( \frac{1}{n} \sum_{i=1}^n a_i (\bI\{W_i=1\} - \pi_i) \Big)^4\bigg] \nonumber \\
\leq & \frac{\overline{a}^4}{n^4} \bE\bigg[\Big( \sum_{i=1}^n \big( \bI\{W_i=1\} - \pi_i \big) \Big)^4\bigg] \nonumber \\
\leq & 125 \overline{a}^4 \overline{d}^2 \ n^{2\alpha - 2}, \label{eqn:FourthMoment_X}
\end{align}
where the first inequality is using the upper bound that $|a_i| \leq \overline{a}$;
the second inequality is doing the combinatorial counting, upper bounding $|\bI\{W_i=1\} - \pi_i| \leq 1$, and using Assumption~\ref{asp:DesignProbabilities}-(ii).

Combining \eqref{eqn:FourthMoment_gamma} and \eqref{eqn:FourthMoment_X} we have
\begin{align}
\lim_{n \to +\infty} n \bE\big[\Delta_n^2\big] \leq \lim_{n \to +\infty} \frac{125 \overline{y}^2 \overline{a}^4 \overline{d}^4}{\underline{\pi}^2 \underline{c}_{\Sigma}^2} n^{4\alpha+4\beta-1} = 0, \label{eqn:UBDeltan}
\end{align}
where the last equality holds because $4\alpha+4\beta < 1$.

Next, we show that $\lim_{n \to +\infty} n \Cov\big(\widehat{\mu}^\CV(\gamma^*), \Delta_n\big) = 0$. 
To show this, note that
\begin{multline*}
\Cov\big(\widehat{\mu}^\CV(\gamma^*), \Delta_n\big) 
= \bE\big[ \widehat{\mu}^\CV(\gamma^*) \Delta_n \big] - \bE\big[ \widehat{\mu}^\CV(\gamma^*)\big] \bE\big[ \Delta_n \big] 
= \bE\big[ (\widehat{\mu}^\CV(\gamma^*) - \mu_n) \Delta_n \big] \\
= - \bE\Big[ \big(\widehat{\mu}^\HT - \mu_n - \gamma^* (\widehat{X} - \bE[\widehat{X}]) \big) \cdot (\widehat{\gamma}^\HT - \gamma^*) \cdot (\widehat{X} - \bE[\widehat{X}]) \Big].
\end{multline*}
So we have
\begin{align*}
\Big\vert \Cov\big(\widehat{\mu}^\CV(\gamma^*), \Delta_n\big) \Big\vert \leq \bE\Big[ \big\vert\widehat{\mu}^\HT - \mu_n - \gamma^* (\widehat{X} - \bE[\widehat{X}]) \big\vert^3 \Big]^\frac{1}{3} \cdot \bE\Big[ \vert\widehat{\gamma}^\HT - \gamma^*\vert^3 \Big]^\frac{1}{3} \cdot \bE\Big[ \vert\widehat{X} - \bE[\widehat{X}]\vert^3 \Big]^\frac{1}{3},
\end{align*}
where the inequality holds due to Holder inequality.
We next upper bound the three terms separately.
Note that,
\begin{align}
\bE\Big[ \big\vert\widehat{\mu}^\HT - \mu_n - \gamma^* (\widehat{X} - \bE[\widehat{X}]) \big\vert^3 \Big] = & \bE\bigg[ \Big(\frac{1}{n} \sum_{i=1}^n \Big(\frac{Y_i}{\pi_i} - \gamma^* a_i\Big) (\bI\{W_i=1\} - \pi_i) \Big)^3 \bigg] \nonumber \\
\leq & \frac{1}{n^3} \Big( \frac{\overline{y}}{\underline{\pi}} n^{\beta} + \frac{\overline{y} \ \overline{a} \overline{d}^\frac{1}{2}}{\underline{\pi} \ \underline{c}_{\Sigma}^\frac{1}{2}} n^{\frac{\alpha}{2}+\beta} \Big)^3 \bE\bigg[ \Big( \sum_{i=1}^n (\bI\{W_i=1\} - \pi_i) \Big)^3 \bigg] \nonumber \\
\leq & \frac{1}{n^3} \cdot \frac{8 \overline{y}^3 \overline{a}^3 \overline{d}^\frac{3}{2}}{\underline{\pi}^3 \underline{c}_{\Sigma}^\frac{3}{2}} n^{\frac{3\alpha}{2}+3\beta} \bE\bigg[ \Big( \sum_{i=1}^n (\bI\{W_i=1\} - \pi_i) \Big)^3 \bigg] \nonumber \\
\leq & \frac{1}{n^3} \cdot \frac{8 \overline{y}^3 \overline{a}^3 \overline{d}^\frac{3}{2}}{\underline{\pi}^3 \underline{c}_{\Sigma}^\frac{3}{2}} n^{\frac{3\alpha}{2}+3\beta} 5 \overline{d}^2 n^{1+2\alpha} \nonumber \\
= & \frac{40 \overline{y}^3 \overline{a}^3 \overline{d}^\frac{7}{2}}{\underline{\pi}^3 \underline{c}_{\Sigma}^\frac{3}{2}} n^{\frac{7}{2}\alpha+3\beta-2}, \label{eqn:ThirdMoment_mu}
\end{align}
where the first inequality is upper bounding $|\frac{Y_i}{\pi_i} - \gamma^* a_i|$ using Lemma~\ref{lem:UBOptimalgamma};
the second inequality holds when $n$ is sufficiently large;
the third inequality is doing the combinatorial counting, upper bounding $|\bI\{W_i=1\} - \pi_i| \leq 1$, and using Assumption~\ref{asp:DesignProbabilities}-(ii).
Similarly,
\begin{align}
\bE\Big[ \big\vert\widehat{\gamma}^\HT - \gamma^* \big\vert^3 \Big] = & \bE\bigg[\Big( \frac{1}{\bm{a}^\top \bm{\Sigma} \bm{a}} \sum_{i=1}^n Z_i (\bI\{W_i=1\} - \pi_i) \Big)^3\bigg] \nonumber \\
\leq & \frac{\overline{z}^3 n^{3\alpha+6\beta}}{\underline{c}_{\Sigma}^3 n^3} \bE\bigg[\Big( \sum_{i=1}^n \big( \bI\{W_i=1\} - \pi_i \big) \Big)^3\bigg] \nonumber \\
\leq & \frac{\overline{z}^3 n^{3\alpha+6\beta}}{\underline{c}_{\Sigma}^3 n^3} \ 5 \overline{d}^2 n^{1+2\alpha} \nonumber \\
= & \frac{5 \overline{y}^3 \overline{a}^3 \overline{d}^5}{\underline{\pi}^6 \underline{c}_{\Sigma}^3} n^{5\alpha+6\beta-2}, \label{eqn:ThirdMoment_gamma}
\end{align}
where the first inequality is using the upper bound that $|Z_i| \leq \overline{z} n^\alpha$ and using Assumption~\ref{asp:RegularBasis}-(ii);
the second inequality is doing the combinatorial counting, upper bounding $|\bI\{W_i=1\} - \pi_i| \leq 1$, and using Assumption~\ref{asp:DesignProbabilities}-(ii).
Similarly,
\begin{align}
\bE\Big[ \big\vert\widehat{X} - \bE[\widehat{X}] \big\vert^3 \Big] = & \bE\bigg[\Big( \frac{1}{n} \sum_{i=1}^n a_i (\bI\{W_i=1\} - \pi_i) \Big)^3\bigg] \nonumber \\
\leq & \frac{\overline{a}^3}{n^3} \bE\bigg[\Big( \sum_{i=1}^n \big( \bI\{W_i=1\} - \pi_i \big) \Big)^3\bigg] \nonumber \\
= & 5 \overline{a}^3 \overline{d}^2 n^{2\alpha-2}, \label{eqn:ThirdMoment_X}
\end{align}
where the first inequality is using the upper bound that $|a_i| \leq \overline{a}$;
the second inequality is doing the combinatorial counting, upper bounding $|\bI\{W_i=1\} - \pi_i| \leq 1$, and using Assumption~\ref{asp:DesignProbabilities}-(ii).

Combining \eqref{eqn:ThirdMoment_mu}, \eqref{eqn:ThirdMoment_gamma}, and \eqref{eqn:ThirdMoment_X}, we have
\begin{align*}
\lim_{n \to +\infty} n \Cov\big(\widehat{\mu}^\CV(\gamma^*), \Delta_n\big) \leq \lim_{n \to +\infty} \frac{10 \overline{y}^2 \overline{a}^3 \overline{d}^\frac{7}{2}}{\underline{\pi}^3 \underline{c}_{\Sigma}^\frac{3}{2}} n^{\frac{7}{2}\alpha+3\beta-1} = 0,
\end{align*}
where the last equality holds because $\frac{7}{2}\alpha+3\beta < 4\alpha+4\beta < 1$.

Now that we have shown both $\lim_{n \to +\infty} n \Var\big(\Delta_n\big) = 0$ and $\lim_{n \to +\infty} n \Cov\big(\widehat{\mu}^\CV(\gamma^*), \Delta_n\big) = 0$, we conclude that
\begin{align*}
\lim_{n \to +\infty} n \Var\big(\widehat{\mu}^\CV(\widehat{\gamma}^\HT)\big) = \lim_{n \to +\infty} n \Var\big(\widehat{\mu}^\CV(\gamma^*)\big) 
\end{align*}

\noindent \textbf{Claim 3: Asymptotic Normality}. 
In the proof of this claim, we assume $7\alpha+6\beta < 1$.

First, we show that
\begin{align*}
\lim_{n \to +\infty} \frac{\widehat{\mu}^\CV(\gamma^*) - \mu_n}{\sqrt{\Var\big(\widehat{\mu}^\CV(\gamma^*)\big)}} \xrightarrow{d} \cN(0,1).
\end{align*}
To show this, we introduce some notations. Note that
\begin{align*}
\widehat{\mu}^\CV(\gamma^*) - \mu_n = \frac{1}{n} \sum_{i=1}^n \Big(\frac{Y_i}{\pi_i} - \gamma^* a_i \Big) (\bI\{W_i=1\} - \pi_i).
\end{align*}
Now denote
\begin{align*}
\xi_i = \Big(\frac{Y_i}{\pi_i} - \gamma^* a_i \Big) (\bI\{W_i=1\} - \pi_i),
\end{align*}
and denote $\overline{\xi} = \frac{2 \overline{y} \ \overline{a} \overline{d}^\frac{1}{2}}{\underline{\pi} \ \underline{c}_{\Sigma}^\frac{1}{2}}$,
so we have
\begin{align*}
\vert\xi_i\vert \leq \frac{\overline{y}}{\underline{\pi}} n^{\beta} + \frac{\overline{y} \ \overline{a} \overline{d}^\frac{1}{2}}{\underline{\pi} \ \underline{c}_{\Sigma}^\frac{1}{2}} n^{\frac{\alpha}{2}+\beta} \leq \overline{\xi} n^{\frac{\alpha}{2}+\beta}
\end{align*}
for sufficiently large $n$.
Next, denote
\begin{align*}
\sigma_n = \sqrt{\Var\Big(\sum_{i=1}^n \xi_i\Big)}, \qquad \text{and} \qquad S_n = \frac{\sum_{i=1}^n \xi_i}{\sigma_n}.
\end{align*}
We will use Lemma~\ref{lem:DependencyNeighborhoodCLT} to show the central limit theorem. 
To do so, we bound each term in Lemma~\ref{lem:DependencyNeighborhoodCLT}.
First, recall that we additionally assume that for sufficiently large $n$, $n \Var\big( \widehat{\mu}^{\CV}(\gamma^*) \big) \geq \underline{c}$.
So we have
\begin{align*}
\sigma_n^2 = \Var\Big(\sum_{i=1}^n \xi_i\Big) = n^2 \Var\Big(\frac{1}{n} \sum_{i=1}^n \xi_i\Big) = n^2 \Var\big( \widehat{\mu}^{\CV}(\gamma^*) \big) \geq \underline{c} n.
\end{align*}
We also have
\begin{align*}
\sigma_n^3 \geq \underline{c}^\frac{3}{2} n^\frac{3}{2}.
\end{align*}
Next, we have
\begin{align*}
\sum_{i=1}^n \bE[\vert\xi_i\vert^3] & \leq \sum_{i=1}^n \overline{\xi}^3 n^{\frac{3}{2}\alpha+3\beta} = \overline{\xi}^3 n^{1+\frac{3}{2}\alpha+3\beta}, \\
\sum_{i=1}^n \bE[\xi_i^4] & \leq \sum_{i=1}^n \overline{\xi}^4 n^{2\alpha+4\beta} = \overline{\xi}^4 n^{1+2\alpha+4\beta}.
\end{align*}
Putting all together into Lemma~\ref{lem:DependencyNeighborhoodCLT}, for $Z$ a standard normal random variable, 
\begin{align*}
\lim_{n \to +\infty} d_W(S_n, Z) \leq & \lim_{n \to +\infty} \ \frac{\overline{d}^2 n^{2\alpha}}{\underline{c}^\frac{3}{2} n^\frac{3}{2}} \overline{\xi}^3 n^{1+\frac{3}{2}\alpha+3\beta} + \frac{26^\frac{1}{2} \overline{d}^\frac{3}{2} n^{\frac{3}{2} \alpha}}{\pi^\frac{1}{2} \underline{c} n} \Big(\overline{\xi}^4 n^{1+2\alpha+4\beta}\Big)^\frac{1}{2} \\
= & \lim_{n \to +\infty} \ \frac{\overline{d}^2 \overline{\xi}^3}{\underline{c}^\frac{3}{2}} n^{\frac{7}{2}\alpha+3\beta-\frac{1}{2}} + \frac{26^\frac{1}{2} \overline{d}^\frac{3}{2} \overline{\xi}^2}{\pi^\frac{1}{2} \underline{c}} n^{\frac{5}{2}\alpha+2\beta-\frac{1}{2}} \\
= & \ 0,
\end{align*}
where $D_W(S_n, Z)$ stands for the Wasserstein distance between $S_n$ and $Z$;
the first inequality holds using Lemma~\ref{lem:DependencyNeighborhoodCLT};
the last equality holds because $\frac{5}{2}\alpha+2\beta < \frac{7}{2}\alpha+3\beta < \frac{1}{2}$.
So we have
\begin{align}
\lim_{n \to +\infty} S_n = \lim_{n \to +\infty} \frac{\widehat{\mu}^\CV(\gamma^*) - \mu_n}{\sqrt{\Var\big(\widehat{\mu}^\CV(\gamma^*)\big)}} \xrightarrow{d} \cN(0,1). \label{eqn:auxiliaryCLT}
\end{align}

Next, we establish the connection between $\widehat{\mu}^\CV(\widehat{\gamma}^\HT)$ and $\widehat{\mu}^\CV(\gamma^*)$.
Note that,
\begin{multline}
\frac{\widehat{\mu}^\CV(\widehat{\gamma}^\HT) - \mu_n}{\sqrt{\Var\big(\widehat{\mu}^\CV(\widehat{\gamma}^\HT)\big)}} = \frac{\widehat{\mu}^\CV(\gamma^*) - \mu_n}{\sqrt{\Var\big(\widehat{\mu}^\CV(\widehat{\gamma}^\HT)\big)}} + \frac{\Delta_n}{\sqrt{\Var\big(\widehat{\mu}^\CV(\widehat{\gamma}^\HT)\big)}} \\
= \frac{\widehat{\mu}^\CV(\gamma^*) - \mu_n}{\sqrt{\Var\big(\widehat{\mu}^\CV(\gamma^*)\big)}} \frac{\sqrt{\Var\big(\widehat{\mu}^\CV(\gamma^*)\big)}}{\sqrt{\Var\big(\widehat{\mu}^\CV(\widehat{\gamma}^\HT)\big)}} + \frac{\Delta_n}{\sqrt{\Var\big(\widehat{\mu}^\CV(\gamma^*)\big)}} \frac{\sqrt{\Var\big(\widehat{\mu}^\CV(\gamma^*)\big)}}{\sqrt{\Var\big(\widehat{\mu}^\CV(\widehat{\gamma}^\HT)\big)}}. \label{eqn:Connection}
\end{multline}
From here, note that
\begin{align*}
\bE\Bigg[\frac{\Delta_n}{\sqrt{\Var\big(\widehat{\mu}^\CV(\gamma^*)\big)}} \Bigg] \leq \bE\Bigg[\frac{\Delta_n^2}{\Var\big(\widehat{\mu}^\CV(\gamma^*)\big)} \Bigg]^\frac{1}{2} \leq \sqrt{\frac{125 \overline{y}^2 \overline{a}^4 \overline{d}^4}{\underline{\pi}^2 \underline{c}_{\Sigma}^2 \underline{c}} n^{4\alpha+4\beta-1}},
\end{align*}
where the first inequality is Cauchy-Schwarz inequality;
and the second inequality is using \eqref{eqn:UBDeltan} and assuming $\Var\big(\widehat{\mu}^\CV(\gamma^*)\big) \geq \underline{c} n^{-1}$.
So for any $\epsilon > 0$,
\begin{align*}
\lim_{n \to +\infty} \Pr\Bigg( \Bigg\vert \frac{\Delta_n}{\sqrt{\Var\big(\widehat{\mu}^\CV(\gamma^*)\big)}} \Bigg\vert \geq \epsilon\Bigg) 
\leq \lim_{n \to +\infty} \sqrt{\frac{125 \overline{y}^2 \overline{a}^4 \overline{d}^4}{\underline{\pi}^4 \underline{c}_{\Sigma}^2 \underline{c}} \cdot n^{4\alpha+4\beta-1}} \cdot \frac{1}{\epsilon} 
= 0,
\end{align*}
where the first inequality is due to Markov inequality;
the last and only equality is because $4\alpha+4\beta < 7\alpha+6\beta < 1$
So we have
\begin{align*}
\lim_{n \to +\infty} \frac{\Delta_n}{\sqrt{\Var\big(\widehat{\mu}^\CV(\gamma^*)\big)}} \xrightarrow{p} 0.
\end{align*}
Note also that $\lim_{n \to +\infty} n \Var\big(\widehat{\mu}^\CV(\widehat{\gamma}^\HT)\big) = \lim_{n \to +\infty} n \Var\big(\widehat{\mu}^\CV(\gamma^*)\big)$ and for sufficiently large $n$ we assume $n \Var\big(\widehat{\mu}^\CV(\gamma^*)\big) \geq \underline{c}$.
This ensures that $n \Var\big(\widehat{\mu}^\CV(\widehat{\gamma}^\HT)\big) \geq \underline{c} > 0$ for sufficiently large $n$.
So we have
\begin{multline*}
\lim_{n \to +\infty} \Bigg\vert \frac{\Var\big(\widehat{\mu}^\CV(\gamma^*)\big)}{\Var\big(\widehat{\mu}^\CV(\widehat{\gamma}^\HT)\big)} - 1 \Bigg\vert 
= \lim_{n \to +\infty} \frac{\big\vert n \Var\big(\widehat{\mu}^\CV(\widehat{\gamma}^\HT)\big) - n \Var\big(\widehat{\mu}^\CV(\gamma^*)\big) \big\vert}{n \Var\big(\widehat{\mu}^\CV(\widehat{\gamma}^\HT)\big)} \\
\leq \lim_{n \to +\infty} \frac{\big\vert n \Var\big(\widehat{\mu}^\CV(\widehat{\gamma}^\HT)\big) - n \Var\big(\widehat{\mu}^\CV(\gamma^*)\big) \big\vert}{\underline{c}} = 0,
\end{multline*}
where the last equality is Claim 2 in Theorem~\ref{thm:AsymptoticBehavior}.

Putting the above and \eqref{eqn:auxiliaryCLT} into \eqref{eqn:Connection}, and using the Slutsky theorem we have
\begin{align*}
\lim_{n \to +\infty} \frac{\widehat{\mu}^\CV(\widehat{\gamma}^\HT) - \mu_n}{\sqrt{\Var\big(\widehat{\mu}^\CV(\widehat{\gamma}^\HT)\big)}} \xrightarrow{d} \cN(0,1).
\end{align*}
This finishes the proof.
\hfill \halmos
\endproof

\subsection{Proof of Theorem~\ref{thm:OptimalControlVariate}}

\proof{Proof of Theorem~\ref{thm:OptimalControlVariate}.}
Recall that the objective function of \eqref{eqn:OptFormulation} can be written as
\begin{align*}
\min_{\bm{a}} \ \bE_{\bm{Y} \sim \cY^n} \bigg[\bm{Y}^\top \bm{\Pi}^{-1} \bm{\Sigma} \bm{\Pi}^{-1} \bm{Y} - \frac{\big(\bm{Y}^\top \bm{\Pi}^{-1} \bm{\Sigma} \bm{a} \big)^2}{\bm{a}^\top \bm{\Sigma} \bm{a}}\bigg]. 
\end{align*}
Note that $\bE_{\bm{Y} \sim \cY^n} \big[\bm{Y}^\top \bm{\Pi}^{-1} \bm{\Sigma} \bm{\Pi}^{-1} \bm{Y}\big]$ does not depend on $\bm{a}$.
So this minimization problem can be written as a maximization problem
\begin{align}
\max_{\bm{a}} \ \bE_{\bm{Y} \sim \cY^n} \bigg[\frac{\big(\bm{Y}^\top \bm{\Pi}^{-1} \bm{\Sigma} \bm{a} \big)^2}{\bm{a}^\top \bm{\Sigma} \bm{a}}\bigg] = \max_{\bm{a}} \ \frac{\bm{a}^\top \bm{\Sigma} \bm{\Pi}^{-1} \bE_{\bm{Y} \sim \cY^n}\big[\bm{Y} \bm{Y}^\top\big] \bm{\Pi}^{-1} \bm{\Sigma} \bm{a}}{\bm{a}^\top \bm{\Sigma} \bm{a}}. \label{eqn:VarIntermediate2}
\end{align}
Note that $\bE_{\bm{Y} \sim \cY^n}\big[\bm{Y} \bm{Y}^\top\big] = \sigma^2 \bm{I}_n + \mu^2 \bm{1}_n \bm{1}_n^\top$, where $\bm{I}_n$ is a $n \times n$ identity matrix and $\bm{1}_n \bm{1}_n^\top$ is a $n \times n$ matrix with all elements equal to $1$.
Using the above notation, \eqref{eqn:VarIntermediate2} can be written as
\begin{align}
\max_{\bm{a}} \ \frac{\bm{a}^\top \bm{\Sigma} \bm{\Pi}^{-1} \big(\sigma^2 \bm{I}_n + \mu^2 \bm{1}_n \bm{1}_n^\top\big) \bm{\Pi}^{-1} \bm{\Sigma} \bm{a}}{\bm{a}^\top \bm{\Sigma} \bm{a}}. \label{eqn:VarIntermediate3}
\end{align}

Now we introduce a variable transformation.
For any vector $\bm{a} \in \bR^n$, let $\|\bm{a}\|_2$ be its $L_2$ norm, that is, $\|\bm{a}\|_2 = (\sum_{i=1}^n a_i^2)^{\frac{1}{2}}$.
Next, because $\bm{\Sigma}$ is symmetric and positive semidefinite, we can define its square root matrix as $\bm{\Sigma}^{\frac{1}{2}}$ and its pseudo-inverse square root matrix as $\bm{\Sigma}^{-\frac{1}{2}}$.
See Section~\ref{sec:AdditionalDefn} for their detailed definitions and properties.
Using the above definitions, we define 
\begin{align*}
\bm{v} = \frac{\bm{\Sigma}^{\frac{1}{2}} \bm{a}}{\|\bm{\Sigma}^{\frac{1}{2}} \bm{a}\|_2}
\end{align*}
to be a unit length vector.
Under this variable transformation, \eqref{eqn:VarIntermediate3} can be written as
\begin{align}
\max_{\bm{v} \in \cV} \ \bm{v}^\top \bm{\Sigma}^{\frac{1}{2}} \bm{\Pi}^{-1} \big(\sigma^2 \bm{I}_n + \mu^2 \bm{1}_n \bm{1}_n^\top\big) \bm{\Pi}^{-1} \bm{\Sigma}^{\frac{1}{2}} \bm{v}, \label{eqn:VarIntermediate4}
\end{align}
where $\cV$ is a unit sphere defined as $\cV = \big\{ \bm{v} \in \bR^n \big\vert \| \bm{v} \|_2 = 1 \big\}$.

Now that we know $\bm{\Sigma}^{\frac{1}{2}}$ is symmetric because $\bm{\Sigma}$ is symmetric and positive semidefinite, $\bm{\Pi}^{-1}$ is diagonal thus symmetric, and $\big(\sigma^2 \bm{I}_n + \mu^2 \bm{1}_n \bm{1}_n^\top\big)$ is also symmetric, we have 
\begin{align*}
\Big(\bm{\Sigma}^{\frac{1}{2}} \bm{\Pi}^{-1} \big(\sigma^2 \bm{I}_n + \mu^2 \bm{1}_n \bm{1}_n^\top\big) \bm{\Pi}^{-1} \bm{\Sigma}^{\frac{1}{2}}\Big)^\top 
= \bm{\Sigma}^{\frac{1}{2}} \bm{\Pi}^{-1} \big(\sigma^2 \bm{I}_n + \mu^2 \bm{1}_n \bm{1}_n^\top\big) \bm{\Pi}^{-1} \bm{\Sigma}^{\frac{1}{2}}
\end{align*}
so it is also symmetric.
Now denote $\bm{M} = \bm{\Sigma}^{\frac{1}{2}} \bm{\Pi}^{-1} \big(\sigma^2 \bm{I}_n + \mu^2 \bm{1}_n \bm{1}_n^\top\big) \bm{\Pi}^{-1} \bm{\Sigma}^{\frac{1}{2}}$.
Using Lemma~\ref{lem:RayleighQuotient} the Rayleigh Quotient (\citet{horn2012matrix} Theorem 4.2.2), the optimal objective value of \eqref{eqn:VarIntermediate4} is the largest eigenvalue of $\bm{M}$, denoted as $\lambda_1(\bm{M})$.
So the optimal variance reduction is given by
\begin{align*}
\bE_{\bm{Y} \sim \cY^n} \Big[ \Var\big( \widehat{\mu}^{\HT} \big) - \Var\big( \widehat{\mu}^{\CV}(\gamma^*) \big) \Big] = \frac{\lambda_1(\bm{M})}{n^2}.
\end{align*}

Additionally, using Lemma~\ref{lem:RayleighQuotient} the Rayleigh Quotient (\citet{horn2012matrix} Theorem 4.2.2), the optimal solution to \eqref{eqn:VarIntermediate4} is the principal eigenvector, that is, the eigenvector corresponding to $\lambda_1(\bm{M})$, denoted as $\bm{u}_1\big(\bm{M}\big)$. 
Substituting this back to the definition of $\bm{v}$ we have that the optimal basis is 
\begin{align*}
\bm{a}^* = c \bm{\Sigma}^{-\frac{1}{2}} \bm{u}_1\big(\bm{M}\big)
\end{align*}
where $c \ne 0$ is any constant.
\hfill \halmos
\endproof

\subsection{Proof of Proposition~\ref{prop:RegularOpta}}

We present the following Lemma~\ref{lem:ColumnSpace} which will be useful in the proof of Proposition~\ref{prop:RegularOpta}. 
To introduce Lemma~\ref{lem:ColumnSpace}, we denote the following notation.
For any matrix $\bm{E} \in \bR^{n \times n}$, denote $\mathrm{range}(\bm{E})$ to be the column space of matrix $\bm{E}$, defined as
\begin{align*}
\mathrm{range}(\bm{E}) = \big\{ \bm{E} \bm{x} \ \vert \ \bm{x} \in \bR^{n} \big\}.
\end{align*}

\begin{lemma}
\label{lem:ColumnSpace}
Let $\bm{E} \in \bR^{n \times n}$ be any symmetric matrix. 
Let $\lambda_1(\bm{M})$ be the largest eigenvalue of $\bm{M}$ and $\bm{u}_1(\bm{M})$ be the eigenvector that corresponds to $\lambda_1(\bm{M})$, where $\bm{M}$ is a symmetric matrix defined as $\bm{M} = \bm{\Sigma}^{\frac{1}{2}} \bm{E} \bm{\Sigma}^{\frac{1}{2}}$. 
When $\lambda_1(\bm{M}) > 0$, we have
\begin{align*}
\bm{u}_1(\bm{M}) \in \mathrm{range}(\bm{\Sigma}).
\end{align*}
\end{lemma}

\proof{Proof of Lemma~\ref{lem:ColumnSpace}.}
Note that $\bm{u}_1(\bm{M}) \in \mathrm{range}(\bm{M})$ because $\bm{u}_1(\bm{M})$ is an eigenvector of the symmetric matrix $\bm{M}$.
Next, we show that $\mathrm{range}(\bm{M}) \subseteq \mathrm{range}(\bm{\Sigma})$.
To see this, note that for any vector $\bm{x}$,
\begin{align*}
\bm{M} \bm{x} = \bm{\Sigma}^{\frac{1}{2}} \Big( \bm{E} \bm{\Sigma}^{\frac{1}{2}} \bm{x} \Big).
\end{align*}
Because $\big( \bm{E} \bm{\Sigma}^{\frac{1}{2}} \bm{x} \big)$ may not span the entire space of $\bR^n$, we have that
\begin{align*}
\mathrm{range}(\bm{M}) \subseteq \mathrm{range}(\bm{\Sigma}^\frac{1}{2}).
\end{align*}
Next, because $\bm{\Sigma} = \bm{U} \bm{\Lambda} \bm{U}^{-1}$ and $\bm{\Sigma}^\frac{1}{2} = \bm{U} \bm{\Lambda}^\frac{1}{2} \bm{U}^{-1}$ have the same directions along the eigenvectors corresponding to positive eigenvalues, we have 
\begin{align*}
\mathrm{range}(\bm{\Sigma}^\frac{1}{2}) = \mathrm{range}(\bm{\Sigma}).
\end{align*}

Putting all above together, we have
\begin{align*}
\bm{u}_1(\bm{M}) \in \mathrm{range}(\bm{M}) \subseteq \mathrm{range}(\bm{\Sigma}^\frac{1}{2}) = \mathrm{range}(\bm{\Sigma}).
\end{align*}
This finishes the proof.
\hfill \halmos
\endproof

\

\noindent Now we prove Proposition~\ref{prop:RegularOpta} as follows.

\proof{Proof of Proposition~\ref{prop:RegularOpta}.}
As a succinct notation, we drop $\bm{M}$ and write $\bm{u}_1(\bm{M}) = \bm{u}_1$. 
Denote the following diagonal matrix $\bI^+ = \bm{\mathrm{diag}}(\bI\{\lambda_1 > 0\}, \bI\{\lambda_2 > 0\}, ..., \bI\{\lambda_n > 0\})$, where $\lambda_i$ stands for the $i$-th largest eigenvalue of matrix $\bm{\Sigma}$.
Denote $\bm{P}_\Sigma = \bm{\Sigma}^{-\frac{1}{2}} \bm{\Sigma} \bm{\Sigma}^{-\frac{1}{2}}$.

Recall the eigen-decomposition $\bm{\Sigma} = \bm{U} \bm{\Lambda} \bm{U}^{-1}$.
So we know that
\begin{align*}
\bm{P}_\Sigma = \bm{\Sigma}^{-\frac{1}{2}} \bm{\Sigma} \bm{\Sigma}^{-\frac{1}{2}} = \bm{U} \bI^+ \bm{U}^{-1}
\end{align*}
is a projection matrix, that is, if a vector $\bm{x} \in \mathrm{range}(\bm{\Sigma})$, we have $\bm{P}_\Sigma \bm{x} = \bm{x}$.

Using Lemma~\ref{lem:ColumnSpace}, we have
\begin{align*}
\bm{u}_1 \in \mathrm{range}(\bm{\Sigma}).
\end{align*}
Consequently, $\bm{P}_\Sigma \bm{u}_1 = \bm{u}_1$.
Next, we have
\begin{align*}
(\bm{a}^*)^\top \bm{\Sigma} \bm{a}^* = n \bm{u}_1^\top \bm{\Sigma}^{-\frac{1}{2}} \bm{\Sigma} \bm{\Sigma}^{-\frac{1}{2}} \bm{u}_1 = n \bm{u}_1^\top \bm{u}_1 = n,
\end{align*}
where the last equality is because $\bm{u}_1$ is an eigenvector so it is unit length.
\hfill \halmos
\endproof

\subsection{Proof of Proposition~\ref{prop:OptBasisIPWBasis}}

\proof{Proof of Proposition~\ref{prop:OptBasisIPWBasis}.}
To start, note that
\begin{align*}
\| \bm{a}^* - \bm{a}^\circ \|_2 = & \| c \bm{\Sigma}^{-\frac{1}{2}} \bm{u}_1(\bm{M}) - \bm{\Pi}^{-1} \bm{1}_n \|_2 \\
= & \Big\| c \bm{\Sigma}^{-\frac{1}{2}} \bm{u}_1(\bm{M}) - c \bm{\Sigma}^{-\frac{1}{2}} \Big( \frac{1}{c} \bm{\Sigma}^{\frac{1}{2}} \bm{\Pi}^{-1} \bm{1}_n \Big) \Big\|_2 \\
= & c \cdot \Big\| \bm{\Sigma}^{-\frac{1}{2}} \Big\|_2 \cdot \Big\| \bm{u}_1(\bm{M}) - \frac{1}{c} \bm{\Sigma}^{\frac{1}{2}} \bm{\Pi}^{-1} \bm{1}_n \Big\|_2,
\end{align*}
where $\| \bm{\Sigma}^{-\frac{1}{2}} \|_2$ stands for the spectral norm of matrix $\bm{\Sigma}^{-\frac{1}{2}}$.
Because $\bm{\Sigma}^{-\frac{1}{2}}$ is a diagonal matrix, 
\begin{align*}
\| \bm{\Sigma}^{-\frac{1}{2}} \|_2 = \max_{i \in [n]} \Big( \frac{1}{\pi_i (1-\pi_i)} \Big)^{\frac{1}{2}} \leq \Big( \frac{1}{\underline{\pi} n^{-\beta} (1 - \underline{\pi} n^{-\beta})} \Big)^{\frac{1}{2}}.
\end{align*}
Now we focus on upper bounding $\big\| \bm{u}_1(\bm{M}) - \frac{1}{c} \bm{\Sigma}^{\frac{1}{2}} \bm{\Pi}^{-1} \bm{1}_n \big\|_2$.
Note that
\begin{align*}
\Big\| \bm{u}_1(\bm{M}) - \frac{1}{c} \bm{\Sigma}^{\frac{1}{2}} \bm{\Pi}^{-1} \bm{1}_n \Big\|_2^2 = & \Big\| \bm{u}_1(\bm{M}) \Big\|_2^2 + \Big\| \frac{1}{c} \bm{\Sigma}^{\frac{1}{2}} \bm{\Pi}^{-1} \bm{1}_n \Big\|_2^2 - 2 \big(\bm{u}_1(\bm{M})\big)^\top \frac{1}{c} \bm{\Sigma}^{\frac{1}{2}} \bm{\Pi}^{-1} \bm{1}_n \\
= & 2 \Big( 1 - \frac{1}{c} \big(\bm{u}_1(\bm{M})\big)^\top \bm{\Sigma}^{\frac{1}{2}} \bm{\Pi}^{-1} \bm{1}_n \Big),
\end{align*}
where the second equality holds because both $\bm{u}_1(\bm{M})$ and $\frac{1}{c} \bm{\Sigma}^{\frac{1}{2}} \bm{\Pi}^{-1} \bm{1}_n$ are unit vectors. 

Denote $\zeta = \frac{1}{c} \big(\bm{u}_1(\bm{M})\big)^\top \bm{\Sigma}^{\frac{1}{2}} \bm{\Pi}^{-1} \bm{1}_n$.
Note that $\zeta \geq 0$ by the definition of $c$, and that $\zeta \leq 1$ because both $\bm{u}_1(\bm{M})$ and $\frac{1}{c} \bm{\Sigma}^{\frac{1}{2}} \bm{\Pi}^{-1} \bm{1}_n$ are unit vectors. 
So $0 \leq 1 - \zeta \leq 1 - \zeta^2$.
Next, denote the orthogonal projection matrix
\begin{align*}
\bm{P}_{\perp} = \bm{I}_n - \frac{1}{c^2} \big( \bm{\Sigma}^{\frac{1}{2}} \bm{\Pi}^{-1} \bm{1}_n \big) \big( \bm{\Sigma}^{\frac{1}{2}} \bm{\Pi}^{-1} \bm{1}_n \big)^\top.
\end{align*}
Using the above notations,
\begin{align*}
\Big\| \bm{u}_1(\bm{M}) - \frac{1}{c} \bm{\Sigma}^{\frac{1}{2}} \bm{\Pi}^{-1} \bm{1}_n \Big\|_2^2 
= 2 (1 - \zeta) \leq 2 (1 - \zeta^2) 
= 2 \Big\| \bm{P}_{\perp} \bm{u}_1(\bm{M}) \Big\|_2^2,
\end{align*}
where the last equality can be seen as decomposing the unit vector $\bm{u}_1(\bm{M})$ through two directions that are along the direction of $\bm{\Sigma}^{\frac{1}{2}} \bm{\Pi}^{-1} \bm{1}_n$ and orthogonal to $\bm{\Sigma}^{\frac{1}{2}} \bm{\Pi}^{-1} \bm{1}_n$.

Finally, we focus on upper bounding $\big\| \bm{P}_{\perp} \bm{u}_1(\bm{M}) \big\|_2^2$.
Recall the definition of $\bm{M}$ as follows,
\begin{align*}
\bm{M} = & \bm{\Sigma}^{\frac{1}{2}} \bm{\Pi}^{-1} \big(\sigma^2 \bm{I}_n + \mu^2 \bm{1}_n \bm{1}_n^\top\big) \bm{\Pi}^{-1} \bm{\Sigma}^{\frac{1}{2}} \\
= & \sigma^2 \bm{\Sigma}^{\frac{1}{2}} \bm{\Pi}^{-2} \bm{\Sigma}^{\frac{1}{2}} + \mu^2 \big( \bm{\Sigma}^{\frac{1}{2}} \bm{\Pi}^{-1} \bm{1}_n \big) \big( \bm{\Sigma}^{\frac{1}{2}} \bm{\Pi}^{-1} \bm{1}_n \big)^\top.
\end{align*}
Since $\bm{u}_1(\bm{M})$ is an eigenvector of matrix $\bm{M}$, we have
\begin{align*}
\lambda_1(\bm{M}) \bm{u}_1(\bm{M}) = \bm{M} \bm{u}_1(\bm{M}) = \sigma^2 \bm{\Sigma}^{\frac{1}{2}} \bm{\Pi}^{-2} \bm{\Sigma}^{\frac{1}{2}} \bm{u}_1(\bm{M}) + \mu^2 \big( \bm{\Sigma}^{\frac{1}{2}} \bm{\Pi}^{-1} \bm{1}_n \big) \big( \bm{\Sigma}^{\frac{1}{2}} \bm{\Pi}^{-1} \bm{1}_n \big)^\top \bm{u}_1(\bm{M}).
\end{align*}
Now pre-multiply both sides by the orthogonal projection matrix $\bm{P}_{\perp}$, we have
\begin{align*}
\lambda_1(\bm{M}) \bm{P}_{\perp} \bm{u}_1(\bm{M}) = \sigma^2 \bm{P}_{\perp} \bm{\Sigma}^{\frac{1}{2}} \bm{\Pi}^{-2} \bm{\Sigma}^{\frac{1}{2}} \bm{u}_1(\bm{M}) + \bm{0},
\end{align*}
where the equality holds because $\bm{P}_{\perp} \bm{\Sigma}^{\frac{1}{2}} \bm{\Pi}^{-1} \bm{1}_n = \bm{0}$.
From here, we have
\begin{align*}
\big\| \bm{P}_{\perp} \bm{u}_1(\bm{M}) \big\|_2 = & \frac{\sigma^2}{\lambda_1(\bm{M})} \big\| \bm{P}_{\perp} \bm{\Sigma}^{\frac{1}{2}} \bm{\Pi}^{-2} \bm{\Sigma}^{\frac{1}{2}} \bm{u}_1(\bm{M}) \big\|_2 \\
\leq & \frac{\sigma^2}{\lambda_1(\bm{M})} \big\| \bm{P}_{\perp} \big\|_2 \cdot \big\| \bm{\Sigma}^{\frac{1}{2}} \bm{\Pi}^{-2} \bm{\Sigma}^{\frac{1}{2}} \big\|_2 \cdot \big\| \bm{u}_1(\bm{M}) \big\|_2 \\
\leq & \frac{\sigma^2}{\lambda_1(\bm{M})} \big\| \bm{\Sigma}^{\frac{1}{2}} \bm{\Pi}^{-2} \bm{\Sigma}^{\frac{1}{2}} \big\|_2 \\
\leq & \frac{\sigma^2}{\lambda_1(\bm{M})} \frac{1-\underline{\pi} n^{-\beta}}{\underline{\pi} n^{-\beta}},
\end{align*}
where $\big\| \bm{P}_{\perp} \big\|_2$ and $\big\| \bm{\Sigma}^{\frac{1}{2}} \bm{\Pi}^{-2} \bm{\Sigma}^{\frac{1}{2}} \big\|_2$ stand for the spectral norm of matrices $\bm{P}_{\perp}$ and $\bm{\Sigma}^{\frac{1}{2}} \bm{\Pi}^{-2} \bm{\Sigma}^{\frac{1}{2}}$;
the first inequality holds because $\bm{P}_{\perp}$ is a projection matrix so $\big\|\bm{P}_{\perp}\big\|_2 \leq 1$, and $\bm{u}_1(\bm{M})$ is a unit vector so $\big\| \bm{u}_1(\bm{M}) \big\|_2 = 1$;
the second inequality is because under Bernoulli sampling, $\bm{\Sigma}^{\frac{1}{2}} \bm{\Pi}^{-2} \bm{\Sigma}^{\frac{1}{2}} = \bm{\mathrm{diag}}\big(\frac{1-\pi_1}{\pi_1}, ..., \frac{1-\pi_n}{\pi_n}\big)$ so we have 
\begin{align*}
\big\| \bm{\Sigma}^{\frac{1}{2}} \bm{\Pi}^{-2} \bm{\Sigma}^{\frac{1}{2}} \big\|_2 = \max_{i \in [n]} \frac{1-\pi_i}{\pi_i} \leq \frac{1-\underline{\pi} n^{-\beta}}{\underline{\pi} n^{-\beta}}.
\end{align*}

Finally, we have
\begin{multline*}
\lambda_1(\bm{M}) \geq \frac{1}{c^2} \big( \bm{\Sigma}^{\frac{1}{2}} \bm{\Pi}^{-1} \bm{1}_n \big)^\top \bm{M} \big( \bm{\Sigma}^{\frac{1}{2}} \bm{\Pi}^{-1} \bm{1}_n \big) \\
\geq \frac{1}{c^2} \big( \bm{\Sigma}^{\frac{1}{2}} \bm{\Pi}^{-1} \bm{1}_n \big)^\top \mu^2 \big( \bm{\Sigma}^{\frac{1}{2}} \bm{\Pi}^{-1} \bm{1}_n \big) \big( \bm{\Sigma}^{\frac{1}{2}} \bm{\Pi}^{-1} \bm{1}_n \big)^\top \big( \bm{\Sigma}^{\frac{1}{2}} \bm{\Pi}^{-1} \bm{1}_n \big) = \mu^2 c^2,
\end{multline*}
where the first inequality is because due to Lemma~\ref{lem:RayleighQuotient} Rayleigh quotient and $\frac{1}{c} \bm{\Sigma}^{\frac{1}{2}} \bm{\Pi}^{-1} \bm{1}_n$ is a unit vector;
the second inequality is because $\bm{M} \succeq \mu^2 \big( \bm{\Sigma}^{\frac{1}{2}} \bm{\Pi}^{-1} \bm{1}_n \big) \big( \bm{\Sigma}^{\frac{1}{2}} \bm{\Pi}^{-1} \bm{1}_n \big)^\top$;
the equality is because $\big( \bm{\Sigma}^{\frac{1}{2}} \bm{\Pi}^{-1} \bm{1}_n \big)^\top \big( \bm{\Sigma}^{\frac{1}{2}} \bm{\Pi}^{-1} \bm{1}_n \big) = c^2$.

Putting all together, we have
\begin{align*}
\| \bm{a}^* - \bm{a}^\circ \|_2 \leq & \ c \cdot \Big( \frac{1}{\underline{\pi} n^{-\beta} (1 - \underline{\pi} n^{-\beta})} \Big)^{\frac{1}{2}} \cdot 2^{\frac{1}{2}} \frac{\sigma^2}{\mu^2 c^2} \frac{1-\underline{\pi} n^{-\beta}}{\underline{\pi} n^{-\beta}} \\
= & \ \frac{2^{\frac{1}{2}} \sigma^2}{\mu^2} \cdot \Big( \frac{1}{\underline{\pi} n^{-\beta} (1 - \underline{\pi} n^{-\beta})} \Big)^{\frac{1}{2}} \cdot \Big( \frac{1}{\sum_{i=1}^n \frac{1-\pi_i}{\pi_i}} \Big)^{\frac{1}{2}} \cdot \frac{1-\underline{\pi} n^{-\beta}}{\underline{\pi} n^{-\beta}} \\
\leq & \ \frac{2^{\frac{1}{2}} \sigma^2}{\mu^2} \cdot \Big( \frac{1}{\underline{\pi} n^{-\beta} (1 - \underline{\pi} n^{-\beta})} \Big)^{\frac{1}{2}} \cdot \bigg( \frac{1}{n \frac{\underline{\pi} n^{-\beta}}{1 - \underline{\pi} n^{-\beta}}} \bigg)^{\frac{1}{2}} \cdot \frac{1-\underline{\pi} n^{-\beta}}{\underline{\pi} n^{-\beta}} \\
\leq & \ \frac{2^{\frac{1}{2}} \sigma^2}{\mu^2} \cdot \frac{1-\underline{\pi} n^{-\beta}}{\underline{\pi}^2 n^{\frac{1}{2}-2\beta}},
\end{align*}
where the equality is using the definition of $c$;
the second inequality is lower bounding $\sum_{i=1}^n \frac{1-\pi_i}{\pi_i}$.
So we have
\begin{align*}
\lim_{n \to +\infty} \| \bm{a}^* - \bm{a}^\circ \|_2 = 0,
\end{align*}
because $\beta < \frac{1}{4}$.
\hfill \halmos
\endproof

\subsection{Proof of Corollary~\ref{coro:NoiselessOutcomes}}


\proof{Proof of Corollary~\ref{coro:NoiselessOutcomes}.}
The proof of Corollary~\ref{coro:NoiselessOutcomes} is very straightforward. 
Recall from Section~\ref{sec:OptimalControlVariate} that if the outcomes $\bm{Y}$ were known to take values $Y_i = y_i$ for any $i \in [n]$, the problem 
\begin{align*}
\min_{\bm{a}} \ \Var\big( \widehat{\mu}^{\CV}(\gamma^*) \big) \ = \ \min_{\bm{a}} \ \frac{1}{n^2} \bigg( \bm{y}^\top \bm{\Pi}^{-1} \bm{\Sigma} \bm{\Pi}^{-1} \bm{y} - \frac{\big(\bm{y}^\top \bm{\Pi}^{-1} \bm{\Sigma} \bm{a} \big)^2}{\bm{a}^\top \bm{\Sigma} \bm{a}} \bigg)
\end{align*}
can be easily solved by setting $\bm{a} = \bm{\Pi}^{-1} \bm{y}$.
Now because we assume that the outcomes $Y_i = y$ take the same unknown constant, one optimal solution is given by $\bm{a} = \bm{\Pi}^{-1} \bm{1}_n \cdot y$.
Because the optimal basis is scale invariant, so for any $c \ne 0$, the optimal basis takes the form of $\bm{a}^* = c \bm{\Pi}^{-1} \bm{1}_n$.
\hfill \halmos
\endproof

\subsection{Proof of Lemma~\ref{lem:OPTgammas:general}}
\proof{Proof of Lemma~\ref{lem:OPTgammas:general}.}
Recall that we have defined the $2n$-dimensional exposure condition indicators $\bm{D} = (\bI\{\cE_1(1)\}, ..., \bI\{\cE_n(1)\}, \bI\{\cE_1(0)\}, ..., \bI\{\cE_n(0)\})^\top$ and defined the $2n \times 2n$ covariance matrix $\bm{\Omega} = \Var_{\cW}(\bm{D})$.

For notational convenience, we define the following matrices,
\begin{align*}
\bm{\Sigma}(1,1) =
\begin{bmatrix}
\Var(\bI\{\cE_1(1)\})               & \Cov(\bI\{\cE_1(1)\}, \bI\{\cE_2(1)\}) & \dots  & \Cov(\bI\{\cE_1(1)\}, \bI\{\cE_n(1)\}) \\
\Cov(\bI\{\cE_1(1)\}, \bI\{\cE_2(1)\}) & \Var(\bI\{\cE_2(1)\})               & \dots  & \Cov(\bI\{\cE_2(1)\}, \bI\{\cE_n(1)\}) \\
\vdots                           & \vdots                           & \ddots & \vdots                           \\
\Cov(\bI\{\cE_1(1)\}, \bI\{\cE_n(1)\}) & \Cov(\bI\{\cE_2(1)\}, \bI\{\cE_n(1)\}) & \dots  & \Var(\bI\{\cE_n(1)\})               
\end{bmatrix},
\end{align*}
\begin{align*}
\bm{\Sigma}(1,0) = 
\begin{bmatrix}
\Cov(\bI\{\cE_1(1)\}, \bI\{\cE_1(0)\}) & \Cov(\bI\{\cE_1(1)\}, \bI\{\cE_2(0)\}) & \dots  & \Cov(\bI\{\cE_1(1)\}, \bI\{\cE_n(0)\}) \\
\Cov(\bI\{\cE_1(1)\}, \bI\{\cE_2(0)\}) & \Cov(\bI\{\cE_2(1)\}, \bI\{\cE_2(0)\}) & \dots  & \Cov(\bI\{\cE_2(1)\}, \bI\{\cE_n(0)\}) \\
\vdots                           & \vdots                           & \ddots & \vdots                           \\
\Cov(\bI\{\cE_1(1)\}, \bI\{\cE_n(0)\}) & \Cov(\bI\{\cE_2(1)\}, \bI\{\cE_n(0)\}) & \dots  & \Cov(\bI\{\cE_n(1)\}, \bI\{\cE_n(0)\}) 
\end{bmatrix},
\end{align*}
\begin{align*}
\bm{\Sigma}(0,1) = 
\begin{bmatrix}
\Cov(\bI\{\cE_1(0)\}, \bI\{\cE_1(1)\}) & \Cov(\bI\{\cE_1(0)\}, \bI\{\cE_2(1)\}) & \dots  & \Cov(\bI\{\cE_1(0)\}, \bI\{\cE_n(1)\}) \\
\Cov(\bI\{\cE_1(0)\}, \bI\{\cE_2(1)\}) & \Cov(\bI\{\cE_2(0)\}, \bI\{\cE_2(1)\}) & \dots  & \Cov(\bI\{\cE_2(0)\}, \bI\{\cE_n(1)\}) \\
\vdots                           & \vdots                           & \ddots & \vdots                           \\
\Cov(\bI\{\cE_1(0)\}, \bI\{\cE_n(1)\}) & \Cov(\bI\{\cE_2(0)\}, \bI\{\cE_n(1)\}) & \dots  & \Cov(\bI\{\cE_n(0)\}, \bI\{\cE_n(1)\}) 
\end{bmatrix},
\end{align*}
and
\begin{align*}
\bm{\Sigma}(0,0) = 
\begin{bmatrix}
\Var(\bI\{\cE_1(0)\})               & \Cov(\bI\{\cE_1(0)\}, \bI\{\cE_2(0)\}) & \dots  & \Cov(\bI\{\cE_1(0)\}, \bI\{\cE_n(0)\}) \\
\Cov(\bI\{\cE_1(0)\}, \bI\{\cE_2(0)\}) & \Var(\bI\{\cE_2(0)\})               & \dots  & \Cov(\bI\{\cE_2(0)\}, \bI\{\cE_n(0)\}) \\
\vdots                           & \vdots                           & \ddots & \vdots                           \\
\Cov(\bI\{\cE_1(0)\}, \bI\{\cE_n(0)\}) & \Cov(\bI\{\cE_2(0)\}, \bI\{\cE_n(0)\}) & \dots  & \Var(\bI\{\cE_n(0)\})               
\end{bmatrix}.
\end{align*}
So we have
\begin{align*}
\bm{\Omega} = 
\begin{bmatrix}
\bm{\Sigma}(1,1) & \bm{\Sigma}(1,0) \\
\bm{\Sigma}(0,1) & \bm{\Sigma}(0,0)
\end{bmatrix}.
\end{align*}

Next, we use the above notations to find the expressions of $\gamma^*(1)$ and $\gamma^*(0)$.
We start with $\Cov\big(\widehat{\tau}^{\HT}, \widehat{X}(1)\big)$.
\begin{align*}
& \Cov\big(\widehat{\tau}^{\HT}, \widehat{X}(1)\big) \\
= & \bE\big[ \widehat{\tau}^{\HT} \widehat{X}(1) \big] - \tau \bE\big[ \widehat{X}(1) \big] \\
= & \frac{1}{n^2} \sum_{i=1}^n \sum_{j=1}^n \bE\bigg[ \Big( \frac{Y_i(1) \bI\{\cE_i(1)\}}{\Pr(\cE_i(1))} - \frac{Y_i(0) \bI\{\cE_i(0)\}}{\Pr(\cE_i(0))} \Big) a_j(1) \bI\{\cE_j(1)\} \bigg] - \frac{1}{n^2} \sum_{i=1}^n \sum_{j=1}^n \big(Y_i(1) - Y_i(0)\big) a_j(1) \Pr(\cE_j(1)) \\
= & \frac{1}{n^2} \sum_{i=1}^n \sum_{j=1}^n \frac{Y_i(1)}{\Pr(\cE_i(1))} \Cov(\bI\{\cE_i(1)\}, \bI\{\cE_j(1)\}) a_j(1) - \frac{1}{n^2} \sum_{i=1}^n \sum_{j=1}^n \frac{Y_i(0)}{\Pr(\cE_i(0))} \Cov(\bI\{\cE_i(0)\}, \bI\{\cE_j(1)\}) a_j(1) \\
= & \frac{1}{n^2} \bm{Y}(\eT)^\top \bm{\Pi}(1)^{-1} \bm{\Sigma}(1,1) \bm{a}(1) - \frac{1}{n^2} \bm{Y}(\eC)^\top \bm{\Pi}(0)^{-1} \bm{\Sigma}(0,1) \bm{a}(1).
\end{align*}
Similarly, we derive the expression for $\Cov\big(\widehat{\tau}^{\HT}, \widehat{X}(0)\big)$.
\begin{align*}
& \Cov\big(\widehat{\tau}^{\HT}, \widehat{X}(0)\big) \\ 
= & \bE\big[ \widehat{\tau}^{\HT} \widehat{X}(0) \big] - \tau \bE\big[ \widehat{X}(0) \big] \\ 
= & \frac{1}{n^2} \sum_{i=1}^n \sum_{j=1}^n \bE\bigg[ \Big( \frac{Y_i(1) \bI\{\cE_i(1)\}}{\Pr(\cE_i(1))} - \frac{Y_i(0) \bI\{\cE_i(0)\}}{\Pr(\cE_i(0))} \Big) a_j(0) \bI\{\cE_j(0)\} \bigg] - \frac{1}{n^2} \sum_{i=1}^n \sum_{j=1}^n \big(Y_i(1) - Y_i(0)\big) a_j(0) \Pr(\cE_j(0)) \\ 
= & \frac{1}{n^2} \sum_{i=1}^n \sum_{j=1}^n \frac{Y_i(1)}{\Pr(\cE_i(1))} \Cov(\bI\{\cE_i(1)\}, \bI\{\cE_j(0)\}) a_j(0) - \frac{1}{n^2} \sum_{i=1}^n \sum_{j=1}^n \frac{Y_i(0)}{\Pr(\cE_i(0))} \Cov(\bI\{\cE_i(0)\}, \bI\{\cE_j(0)\}) a_j(0) \\ 
= & \frac{1}{n^2} \bm{Y}(\eT)^\top \bm{\Pi}(1)^{-1} \bm{\Sigma}(1,0) \bm{a}(0) - \frac{1}{n^2} \bm{Y}(\eC)^\top \bm{\Pi}(0)^{-1} \bm{\Sigma}(0,0) \bm{a}(0).
\end{align*}

Next, we derive the expression for $\Cov\big(\widehat{X}(1), \widehat{X}(0)\big)$.
\begin{align*}
\Cov\big(\widehat{X}(1), \widehat{X}(0)\big) = & \bE\big[ \widehat{X}(1) \widehat{X}(0) \big] - \bE\big[ \widehat{X}(1) \big] \bE\big[ \widehat{X}(0) \big] \\ 
= & \frac{1}{n^2} \sum_{i=1}^n \sum_{j=1}^n a_i(1) \bE\big[\bI\{\cE_i(1), \cE_j(0)\}\big] a_j(0) - \frac{1}{n^2} \sum_{i=1}^n \sum_{j=1}^n a_i(1) \Pr(\cE_i(1)) \Pr(\cE_i(0)) a_j(0) \\
= & \frac{1}{n^2} \sum_{i=1}^n \sum_{j=1}^n a_i(1) \Cov\big(\bI\{\cE_i(1)\}, \bI\{\cE_j(0)\}\big) a_j(0) \\
= & \frac{1}{n^2} \bm{a}(1)^\top \bm{\Sigma}(1,0) \bm{a}(0) \\
= & \frac{1}{n^2} \bm{a}(0)^\top \bm{\Sigma}(0,1) \bm{a}(1).
\end{align*}
Similarly, we derive the expression for $\Var\big(\widehat{X}(1)\big)$.
\begin{align*}
\Var\big(\widehat{X}(1)\big) = & \bE\big[ \widehat{X}(1)^2 \big] - \bE\big[ \widehat{X}(1) \big]^2 \\ 
= & \frac{1}{n^2} \sum_{i=1}^n \sum_{j=1}^n a_i(1) \bE\big[\bI\{\cE_i(1), \cE_j(1)\}\big] a_j(1) - \frac{1}{n^2} \sum_{i=1}^n \sum_{j=1}^n a_i(1) \Pr(\cE_i(1)) \Pr(\cE_j(1)) a_j(1) \\
= & \frac{1}{n^2} \sum_{i=1}^n \sum_{j=1}^n a_i(1) \Cov\big(\bI\{\cE_i(1)\}, \bI\{\cE_j(1)\}\big) a_j(1) \\
= & \frac{1}{n^2} \bm{a}(1)^\top \bm{\Sigma}(1,1) \bm{a}(1).
\end{align*}
Finally, we derive the expression for $\Var\big(\widehat{X}(0)\big)$.
\begin{align*}
\Var\big(\widehat{X}(0)\big) = & \bE\big[ \widehat{X}(0)^2 \big] - \bE\big[ \widehat{X}(0) \big]^2 \\ 
= & \frac{1}{n^2} \sum_{i=1}^n \sum_{j=1}^n a_i(0) \bE\big[\bI\{\cE_i(0), \cE_j(0)\}\big] a_j(0) - \frac{1}{n^2} \sum_{i=1}^n \sum_{j=1}^n a_i(0) \Pr(\cE_i(0)) \Pr(\cE_j(0)) a_j(0) \\
= & \frac{1}{n^2} \sum_{i=1}^n \sum_{j=1}^n a_i(0) \Cov\big(\bI\{\cE_i(0)\}, \bI\{\cE_j(0)\}\big) a_j(0) \\
= & \frac{1}{n^2} \bm{a}(0)^\top \bm{\Sigma}(0,0) \bm{a}(0).
\end{align*}

Putting all the above into the expression of $\big(\gamma^*(1), \gamma^*(0)\big)$, we have
\begin{align*}
{\footnotesize
\begin{bmatrix}
\gamma^*(1) \\
\gamma^*(0) 
\end{bmatrix}
= 
\begin{bmatrix}
\bm{a}(1)^\top \bm{\Sigma}(1,1) \bm{a}(1) & \bm{a}(1)^\top \bm{\Sigma}(1,0) \bm{a}(0) \\
\bm{a}(0)^\top \bm{\Sigma}(0,1) \bm{a}(1) & \bm{a}(0)^\top \bm{\Sigma}(0,0) \bm{a}(0) 
\end{bmatrix}^{-1}
\begin{bmatrix}
\bm{a}(1)^\top \bm{\Sigma}(1,1) \bm{\Pi}(1)^{-1} \bm{Y}(\eT) - \bm{a}(1)^\top \bm{\Sigma}(1,0) \bm{\Pi}(0)^{-1} \bm{Y}(\eC) \\
\bm{a}(0)^\top \bm{\Sigma}(0,1) \bm{\Pi}(1)^{-1} \bm{Y}(\eT) - \bm{a}(0)^\top \bm{\Sigma}(0,0) \bm{\Pi}(0)^{-1} \bm{Y}(\eC)
\end{bmatrix}.}
\end{align*}
Note that
\begin{align*}
\begin{bmatrix}
\bm{a}(1)^\top \bm{\Sigma}(1,1) \bm{a}(1) & \bm{a}(1)^\top \bm{\Sigma}(1,0) \bm{a}(0) \\
\bm{a}(0)^\top \bm{\Sigma}(0,1) \bm{a}(1) & \bm{a}(0)^\top \bm{\Sigma}(0,0) \bm{a}(0) 
\end{bmatrix}
= \bm{B}^\top \bm{\Omega} \bm{B},
\end{align*}
and that
\begin{align*}
\begin{bmatrix}
\bm{a}(1)^\top \bm{\Sigma}(1,1) \bm{\Pi}(1)^{-1} \bm{Y}(\eT) - \bm{a}(1)^\top \bm{\Sigma}(1,0) \bm{\Pi}(0)^{-1} \bm{Y}(\eC) \\
\bm{a}(0)^\top \bm{\Sigma}(0,1) \bm{\Pi}(1)^{-1} \bm{Y}(\eT) - \bm{a}(0)^\top \bm{\Sigma}(0,0) \bm{\Pi}(0)^{-1} \bm{Y}(\eC)
\end{bmatrix} 
= \bm{B}^\top \bm{\Omega} \bm{\Pi}^{-1} \bm{Y}.
\end{align*}
These succinct notations leads to 
\begin{align*}
\big(\gamma^*(1), \gamma^*(0) \big)^\top = \big(\bm{B}^\top \bm{\Omega} \bm{B}\big)^{-1} \bm{B}^\top \bm{\Omega} \bm{\Pi}^{-1} \bm{Y},
\end{align*}
which finishes the proof for the coefficients $(\gamma^*(1), \gamma^*(0))^\top$. 

Finally, we focus on the expression for $\Var\big(\widehat{\tau}^\HT\big)$.
\begin{align*}
& \Var\big(\widehat{\tau}^\HT\big) \\
= & \bE\big[ \big( \widehat{\tau}^{\HT} \big)^2 \big] - \tau^2 \\
= & \frac{1}{n^2} \sum_{i=1}^n \sum_{j=1}^n \bE\bigg[ \Big( \frac{Y_i(1) \bI\{\cE_i(1)\}}{\Pr(\cE_i(1))} - \frac{Y_i(0) \bI\{\cE_i(0)\}}{\Pr(\cE_i(0))} \Big) \Big( \frac{Y_j(1) \bI\{\cE_j(1)\}}{\Pr(\cE_j(1))} - \frac{Y_j(0) \bI\{\cE_j(0)\}}{\Pr(\cE_j(0))} \Big) \bigg] \\
& - \frac{1}{n^2} \sum_{i=1}^n \sum_{j=1}^n \big(Y_i(1) - Y_i(0)\big) \big(Y_j(1) - Y_j(0)\big) \\
= & \frac{1}{n^2} \sum_{i=1}^n \sum_{j=1}^n \frac{Y_i(1)}{\Pr(\cE_i(1))} \Cov(\bI\{\cE_i(1)\}, \bI\{\cE_j(1)\}) \frac{Y_j(1)}{\Pr(\cE_j(1))} \\
& - \frac{1}{n^2} \sum_{i=1}^n \sum_{j=1}^n \frac{Y_i(1)}{\Pr(\cE_i(1))} \Cov(\bI\{\cE_i(1)\}, \bI\{\cE_j(0)\}) \frac{Y_j(0)}{\Pr(\cE_j(0))} \\
& - \frac{1}{n^2} \sum_{i=1}^n \sum_{j=1}^n \frac{Y_i(0)}{\Pr(\cE_i(0))} \Cov(\bI\{\cE_i(0)\}, \bI\{\cE_j(1)\}) \frac{Y_j(1)}{\Pr(\cE_j(1))} \\
& + \frac{1}{n^2} \sum_{i=1}^n \sum_{j=1}^n \frac{Y_i(0)}{\Pr(\cE_i(0))} \Cov(\bI\{\cE_i(0)\}, \bI\{\cE_j(0)\}) \frac{Y_j(0)}{\Pr(\cE_j(0))} \\
= & \frac{1}{n^2} \bm{Y}(\eT)^\top \bm{\Pi}(1)^{-1} \bm{\Sigma}(1,1) \bm{\Pi}(1)^{-1} \bm{Y}(\eT) - \frac{1}{n^2} \bm{Y}(\eT)^\top \bm{\Pi}(1)^{-1} \bm{\Sigma}(1,0) \bm{\Pi}(0)^{-1} \bm{Y}(\eC) \\
& - \frac{1}{n^2} \bm{Y}(\eC)^\top \bm{\Pi}(0)^{-1} \bm{\Sigma}(0,1) \bm{\Pi}(1)^{-1} \bm{Y}(\eT) + \frac{1}{n^2} \bm{Y}(\eC)^\top \bm{\Pi}(0)^{-1} \bm{\Sigma}(0,0) \bm{\Pi}(0)^{-1} \bm{Y}(\eC).
\end{align*}

Putting all the above into expression \eqref{eqn:OPTvariance}, we have
\begin{align*}
\Var\big( \widehat{\tau}^{\CV}(\gamma^*(1), \gamma^*(0)) \big) = \frac{1}{n^2} \bigg( \bm{Y}^\top \bm{\Pi}^{-1} \bm{\Omega} \bm{\Pi}^{-1} \bm{Y} - \bm{Y}^\top \bm{\Pi}^{-1} \bm{\Omega} \bm{B} \big(\bm{B}^\top \bm{\Omega} \bm{B}\big)^{-1} \bm{B}^\top \bm{\Omega} \bm{\Pi}^{-1} \bm{Y} \bigg), 
\end{align*} 
which finishes the proof for the variance $\Var\big( \widehat{\tau}^{\CV}(\gamma^*(1), \gamma^*(0)) \big)$. 
\hfill \halmos
\endproof

\subsection{Proof of Theorem~\ref{thm:AsymptoticInterference}}

We present the following Lemma~\ref{lem:UBOptimalgammasInterference} which will be useful in the proof of Theorem~\ref{thm:AsymptoticInterference}. 
We first introduce a vector notation 
\begin{align*}
\bm{\gamma}^* = (\gamma^*(1),\gamma^*(0))^\top.
\end{align*}

\begin{lemma}
\label{lem:UBOptimalgammasInterference}
Under Assumptions~\ref{asp:ExpDesign:general},~\ref{asp:RegularBasesInterference}, and~\ref{asp:Exposure}, and assuming the potential outcomes $|Y_i(\eT)|\leq \overline{y}$ and $|Y_i(\eC)|\leq \overline{y}$ are all bounded, the optimal coefficients \(\gamma^*(1)\) and \(\gamma^*(0)\) are bounded, that is,
\begin{align*}
|\gamma^*(1)| \leq \frac{2^{\frac{1}{2}}\overline{y}\overline{d}^{\frac{1}{2}}}{\underline{\pi}\ \underline{\lambda}_{\Omega}^{\frac{1}{2}}} n^{\frac{\alpha}{2}+\beta},
\qquad
|\gamma^*(0)| \leq \frac{2^{\frac{1}{2}}\overline{y}\overline{d}^{\frac{1}{2}}}{\underline{\pi}\ \underline{\lambda}_{\Omega}^{\frac{1}{2}}} n^{\frac{\alpha}{2}+\beta}.
\end{align*}
\end{lemma}

\proof{Proof of Lemma~\ref{lem:UBOptimalgammasInterference}.}
Recall that the variance minimizing coefficients are given by
\begin{align*}
\bm{\gamma}^* = (\bm{B}^\top\bm{\Omega}\bm{B})^{-1} \bm{B}^\top\bm{\Omega}\bm{\Pi}^{-1}\bm{Y}.
\end{align*}
By Assumption~\ref{asp:ExpDesign:general}-(i) and the boundedness of potential outcomes, every element of $\bm{\Pi}^{-1}\bm{Y}$ is bounded in absolute value by $\frac{\overline{y}}{\underline{\pi}} n^{\beta}$. 
So we have
\begin{align*}
\bm{Y}^\top \bm{\Pi}^{-1}\bm{\Omega}\bm{\Pi}^{-1}\bm{Y} 
\leq \frac{\overline{y}^2}{\underline{\pi}^2} n^{2\beta} \sum_{k=1}^{2n}\sum_{l=1}^{2n} |\Omega_{kl}| 
\leq \frac{2\overline{y}^2\overline{d}}{\underline{\pi}^2} n^{1+\alpha+2\beta},
\end{align*}
where the last inequality is due to Assumption~\ref{asp:ExpDesign:general}-(ii).

Next, note that
\begin{align*}
(\bm{\gamma}^*)^\top \bm{B}^\top\bm{\Omega}\bm{B}\bm{\gamma}^*
= \bm{Y}^\top\bm{\Pi}^{-1}\bm{\Omega}\bm{B} (\bm{B}^\top\bm{\Omega}\bm{B})^{-1} \bm{B}^\top\bm{\Omega}\bm{\Pi}^{-1}\bm{Y}
\leq \bm{Y}^\top\bm{\Pi}^{-1}\bm{\Omega}\bm{\Pi}^{-1}\bm{Y},
\end{align*}
where the inequality follows because the middle term is the variance reduction from projecting $\bm{\Pi}^{-1}\bm{Y}$ onto the column space of $\bm{B}$, and hence cannot exceed the total variance $\bm{Y}^\top\bm{\Pi}^{-1}\bm{\Omega}\bm{\Pi}^{-1}\bm{Y}$.

By Assumption~\ref{asp:RegularBasesInterference}-(ii),
\begin{align*}
(\bm{\gamma}^*)^\top \bm{B}^\top\bm{\Omega}\bm{B}\bm{\gamma}^* \geq \underline{\lambda}_{\Omega} n \Big( (\gamma^*(1))^2+(\gamma^*(0))^2 \Big).
\end{align*}
Putting the above inequalities together, we have
\begin{align*}
(\gamma^*(1))^2+(\gamma^*(0))^2 \leq \frac{2\overline{y}^2\overline{d}}{\underline{\pi}^2\underline{\lambda}_{\Omega}} n^{\alpha+2\beta}.
\end{align*}
Taking square root we finish the proof.
\hfill \halmos
\endproof

\

\noindent Now we prove Theorem~\ref{thm:AsymptoticInterference} as follows.

\proof{Proof of Theorem~\ref{thm:AsymptoticInterference}.}
We prove the three claims in Theorem~\ref{thm:AsymptoticInterference} one by one.

\noindent \textbf{Claim 1: Consistency}. 
In the proof of this claim, we assume $3 \alpha + 4 \beta < 1$.
We prove consistency through two steps.
First, we show that
\begin{align*}
\widehat{\tau}^{\CV}(\widehat{\gamma}^{\HT}(1),\widehat{\gamma}^{\HT}(0)) - \widehat{\tau}^{\CV}(\gamma^*(1),\gamma^*(0)) \xrightarrow{p} 0.
\end{align*}
Note that for each $i$, 
\begin{align*}
\widehat{Y}_i(\eT) - Y_i(\eT) = \big( \bI\{\cE_i(1)\} - \Pr\nolimits_{\cW}(\cE_i(1)) \big) \frac{Y_i(\eT)}{\Pr_{\cW}(\cE_i(1))},
\end{align*}
and
\begin{align*}
\widehat{Y}_i(\eC) - Y_i(\eC) = \big( \bI\{\cE_i(0)\} - \Pr\nolimits_{\cW}(\cE_i(0)) \big) \frac{Y_i(\eC)}{\Pr_{\cW}(\cE_i(0))}.
\end{align*}
Next, define for $l \in \{1,2\}$,
\begin{align*}
Z_{i,l}(1) = \frac{Y_i(\eT)}{\Pr_{\cW}(\cE_i(1))^2} \big( \bm{\Omega}\bm{B}(\bm{B}^\top\bm{\Omega}\bm{B})^{-1} \big)_{il}
\end{align*}
and
\begin{align*}
Z_{i,l}(0) = \frac{Y_i(\eC)}{\Pr_{\cW}(\cE_i(0))^2} \big( \bm{\Omega}\bm{B}(\bm{B}^\top\bm{\Omega}\bm{B})^{-1} \big)_{(n+i)l}.
\end{align*}
where we use $\big(\bm{\Omega}\bm{B}(\bm{B}^\top\bm{\Omega}\bm{B})^{-1}\big)_{il}$ to stand for the element in the $i$-th row and the $l$-st column and $\big(\bm{\Omega}\bm{B}(\bm{B}^\top\bm{\Omega}\bm{B})^{-1}\big)_{(n+i)1}$ to stand for the element in the $n+i$-th row and the $l$-st column.

Using the above definitions, we have
\begin{align*}
\widehat{\gamma}^{\HT}(1) - \gamma^*(1) = \sum_{i=1}^{n} Z_{i,1}(1) \big( \bI\{\cE_i(1)\} - \Pr\nolimits_{\cW}(\cE_i(1)) \big) - \sum_{i=1}^{n} Z_{i,1}(0) \big( \bI\{\cE_i(0)\} - \Pr\nolimits_{\cW}(\cE_i(0)) \big),
\end{align*}
and
\begin{align*}
\widehat{\gamma}^{\HT}(0) - \gamma^*(0) = \sum_{i=1}^{n} Z_{i,2}(1) \big( \bI\{\cE_i(1)\} - \Pr\nolimits_{\cW}(\cE_i(1)) \big) - \sum_{i=1}^{n} Z_{i,2}(0) \big( \bI\{\cE_i(0)\} - \Pr\nolimits_{\cW}(\cE_i(0)) \big).
\end{align*}

We next upper bound $|Z_{i,l}(1)|$ and $|Z_{i,l}(0)|$ for $l\in\{1,2\}$.
By Assumption~\ref{asp:ExpDesign:general}-(i) and boundedness of the potential outcomes,
\begin{align*}
\bigg\vert \frac{Y_i(\eT)}{\Pr_{\cW}(\cE_i(1))^2} \bigg\vert \leq \frac{\overline{y}}{\underline{\pi}^2} n^{2\beta},
\qquad
\bigg\vert \frac{Y_i(\eC)}{\Pr_{\cW}(\cE_i(0))^2} \bigg\vert \leq \frac{\overline{y}}{\underline{\pi}^2} n^{2\beta}.
\end{align*}
By Assumptions~\ref{asp:ExpDesign:general}-(ii) and~\ref{asp:RegularBasesInterference}-(i), each element $(\bm{\Omega}\bm{B})_{kl}$ in the $(2n\times2)$ matrix $\bm{\Omega}\bm{B}$ is bounded by
\begin{align*}
|(\bm{\Omega}\bm{B})_{kl}| \leq \overline{a}\overline{d}n^\alpha.
\end{align*}
By Assumption~\ref{asp:RegularBasesInterference}-(ii),
\begin{align*}
\big\| (\bm{B}^\top\bm{\Omega}\bm{B})^{-1} \big\|_2 \leq \frac{1}{\underline{\lambda}_{\Omega}n},
\end{align*}
and so every element of the inverse matrix $(\bm{B}^\top \bm{\Omega} \bm{B})^{-1}$ is upper bounded by $\frac{1}{\underline{\lambda}_{\Omega}n}$. 
Now define $\overline{z} = \dfrac{2\overline{y}\overline{a}\overline{d}}{\underline{\pi}^2\underline{\lambda}_{\Omega}} > 0$.
We have for all $i \in [n]$ and $l \in \{1,2\}$,
\begin{align*}
|Z_{i,l}(1)|\leq \overline{z} n^{\alpha+2\beta-1}, \qquad |Z_{i,l}(0)|\leq \overline{z} n^{\alpha+2\beta-1}.
\end{align*}

Next, it is easy to see that
\begin{align*}
\bE_{\cW}[\widehat{\gamma}^{\HT}(1)] = \gamma^*(1),
\qquad
\bE_{\cW}[\widehat{\gamma}^{\HT}(0)] = \gamma^*(0).
\end{align*}
We calculate $\Var_{\cW}(\widehat{\gamma}^{\HT}(1))$ and $\Var_{\cW}(\widehat{\gamma}^{\HT}(0))$ as follows.
Note that,
\begin{align*}
\Var_{\cW}(\widehat{\gamma}^{\HT}(1)) \leq \overline{z}^2 n^{2\alpha+4\beta-2} \sum_{k=1}^{2n}\sum_{l=1}^{2n}|\Omega_{kl}| \leq 2\overline{z}^2\overline{d} n^{3\alpha+4\beta-1},
\end{align*}
where the last inequality is due to Assumption~\ref{asp:ExpDesign:general}-(ii).
Using Chebyshev inequality, for any $\epsilon > 0$,
\begin{align*}
\Pr\Big( \big\vert\widehat{\gamma}^{\HT}(1) - \gamma^*(1)\big\vert \geq \epsilon \Big) \leq 2\overline{z}^2\overline{d} n^{3\alpha+4\beta-1} \frac{1}{\epsilon^2}.
\end{align*}
Similarly, we have
\begin{align*}
\Pr\Big( \big\vert\widehat{\gamma}^{\HT}(0)-\gamma^*(0)\big\vert\geq\epsilon \Big) \leq 2\overline{z}^2\overline{d} n^{3\alpha+4\beta-1} \frac{1}{\epsilon^2}.
\end{align*}

Next, we calculate the variance of the two control variates.
Note that
\begin{align*}
\Var_{\cW}(\widehat{X}(1)) = \frac{1}{n^2} \Var_{\cW} \bigg( \sum_{i=1}^n a_i(1)\bI\{\cE_i(1)\} \bigg) \leq \overline{a}^2\overline{d} n^{\alpha-1}.
\end{align*}
Using Chebyshev inequality, for any $\epsilon > 0$,
\begin{align*}
\Pr\Big( \big\vert \widehat{X}(1) - \bE_{\cW}[\widehat{X}(1)] \big\vert \geq \epsilon \Big) \leq \overline{a}^2\overline{d} n^{\alpha-1} \frac{1}{\epsilon^2}.
\end{align*}
Similarly,
\begin{align*}
\Pr\Big( \big\vert \widehat{X}(0) - \bE_{\cW}[\widehat{X}(0)] \big\vert \geq \epsilon \Big) \leq \overline{a}^2\overline{d} n^{\alpha-1} \frac{1}{\epsilon^2}.
\end{align*}

Finally, define
\begin{multline*}
\Delta_n = \widehat{\tau}^\CV(\widehat{\gamma}^\HT(1),\widehat{\gamma}^\HT(0)) - \widehat{\tau}^\CV(\gamma^*(1),\gamma^*(0)) \\
= - \big(\widehat{\gamma}^\HT(1)-\gamma^*(1)\big) \big(\widehat{X}(1)-\bE_{\cW}[\widehat{X}(1)]\big) - \big(\widehat{\gamma}^\HT(0)-\gamma^*(0)\big) \big(\widehat{X}(0)-\bE_{\cW}[\widehat{X}(0)]\big).
\end{multline*}
For any $\epsilon > 0$, we have
\begin{align*}
\lim_{n \to +\infty} \Pr\big( |\Delta_n| \geq \epsilon \big) \leq & \lim_{n \to +\infty}
\bigg( \Pr\Big( \big\vert \widehat{\gamma}^\HT(1) - \gamma^*(1) \big\vert \geq \Big(\frac{\epsilon}{2}\Big)^{\frac{1}{2}} \Big) + \Pr\Big( \big\vert \widehat{X}(1) - \bE[\widehat{X}(1)] \big\vert \geq \Big(\frac{\epsilon}{2}\Big)^{\frac{1}{2}} \Big) \\
& \qquad + \Pr\Big( \big\vert \widehat{\gamma}^\HT(0) - \gamma^*(0) \big\vert \geq \Big(\frac{\epsilon}{2}\Big)^{\frac{1}{2}} \Big) + \Pr\Big( \big\vert \widehat{X}(0) - \bE[\widehat{X}(0)] \big\vert \geq \Big(\frac{\epsilon}{2}\Big)^{\frac{1}{2}} \Big) \bigg) \\
\leq & \lim_{n \to +\infty} \bigg( 4 \overline{z}^2 \overline{d} n^{3\alpha+4\beta-1} \cdot \frac{2}{\epsilon} + 2 \overline{a}^2 \overline{d} n^{\alpha-1} \cdot \frac{2}{\epsilon} \bigg) \\
= & 0,
\end{align*}
where the first inequality holds because, if $|\Delta_n| \geq \epsilon$ holds, then at least one of the four conditions must hold;
the last and only equality holds because $3 \alpha + 4 \beta < 1$.
So we have shown that
\begin{align*}
\Delta_n \xrightarrow{p} 0.
\end{align*}

Second, we show that $\widehat{\tau}^{\CV}(\gamma^*(1),\gamma^*(0))$ is a consistent estimator of $\tau(\eT,\eC)$.
Note that
\begin{multline*}
\widehat{\tau}^{\HT} - \tau(\eT,\eC) = \frac{1}{n} \sum_{i=1}^n \frac{Y_i(\eT)}{\Pr\nolimits_{\cW}(\cE_i(1))} \big( \bI\{\cE_i(1)\} - \Pr\nolimits_{\cW}(\cE_i(1)) \big) \\
- \frac{1}{n} \sum_{i=1}^n \frac{Y_i(\eC)}{\Pr\nolimits_{\cW}(\cE_i(0))} \big( \bI\{\cE_i(0)\} - \Pr\nolimits_{\cW}(\cE_i(0)) \big).
\end{multline*}
So we have
\begin{align*}
\Var_{\cW}(\widehat{\tau}^{\HT}) \leq \frac{1}{n^2} \frac{\overline{y}^2}{\underline{\pi}^2} n^{2\beta} \sum_{k=1}^{2n}\sum_{l=1}^{2n}|\Omega_{kl}| \leq \frac{2\overline{y}^2\overline{d}}{\underline{\pi}^2} n^{\alpha+2\beta-1}.
\end{align*}
Using Chebyshev inequality, for any $\epsilon > 0$,
\begin{align*}
\Pr\Big( \big\vert\widehat{\tau}^{\HT}-\tau(\eT,\eC)\big\vert \geq \epsilon \Big) \leq \frac{2\overline{y}^2\overline{d}}{\underline{\pi}^2} n^{\alpha+2\beta-1} \frac{1}{\epsilon^2}.
\end{align*}

Next, note that
\begin{multline*}
\widehat{\tau}^{\CV}(\gamma^*(1),\gamma^*(0)) - \tau(\eT,\eC) \\
= \big(\widehat{\tau}^{\HT}-\tau(\eT,\eC)\big) - \gamma^*(1) \big(\widehat{X}(1)-\bE_{\cW}[\widehat{X}(1)]\big) - \gamma^*(0) \big(\widehat{X}(0)-\bE_{\cW}[\widehat{X}(0)]\big).
\end{multline*}
So, for any $\epsilon > 0$,
\begin{align*}
&\lim_{n\to+\infty}
\Pr\Big( \big\vert \widehat{\tau}^{\CV}(\gamma^*(1),\gamma^*(0)) - \tau(\eT,\eC) \big\vert \geq \epsilon \Big) \\
\leq & \lim_{n\to+\infty} \bigg( \Pr\Big( \big\vert \widehat{\tau}^{\HT} - \tau(\eT,\eC) \big\vert \geq \frac{\epsilon}{3} \Big) + \Pr\Big( |\gamma^*(1)| \big\vert \widehat{X}(1) - \bE_{\cW}[\widehat{X}(1)] \big\vert \geq \frac{\epsilon}{3} \Big) \\
& \qquad + \Pr\Big( |\gamma^*(0)| \big\vert \widehat{X}(0) - \bE_{\cW}[\widehat{X}(0)] \big\vert \geq \frac{\epsilon}{3} \Big) \bigg)\\
\leq & \lim_{n\to+\infty} \bigg( \frac{2\overline{y}^2\overline{d}}{\underline{\pi}^2} n^{\alpha+2\beta-1}\frac{9}{\epsilon^2} + 2\overline{a}^2\overline{d}n^{\alpha-1} \frac{9}{\epsilon^2} \Big( \frac{2^{\frac{1}{2}}\overline{y}\overline{d}^{\frac{1}{2}}} {\underline{\pi}\underline{\lambda}_{\Omega}^{\frac{1}{2}}} n^{\frac{\alpha}{2}+\beta} \Big)^2 \bigg) \\ 
= & \ 0,
\end{align*}
where the first inequality holds because if $\big\vert \widehat{\tau}^{\CV}(\gamma^*(1),\gamma^*(0)) - \tau(\eT,\eC) \big\vert \geq \epsilon$ holds, then at least one of the three conditions must hold;
the second inequality is due to Lemma~\ref{lem:UBOptimalgammasInterference};
the last equality holds because $\alpha + 2 \beta < 2 \alpha + 2 \beta < 3 \alpha + 4 \beta < 1$.
Therefore,
\begin{align*}
\widehat{\tau}^{\CV}(\gamma^*(1),\gamma^*(0))
\xrightarrow{p}
\tau(\eT,\eC).
\end{align*}
Combining this with $\Delta_n\xrightarrow{p}0$, we have
\begin{align*}
\widehat{\tau}^{\CV}(\widehat{\gamma}^{\HT}(1),\widehat{\gamma}^{\HT}(0))
\xrightarrow{p}
\tau(\eT,\eC).
\end{align*}

\noindent \textbf{Claim 2: Asymptotic Variance}. 
In the proof of this claim, we assume $4 \alpha + 4 \beta < 1$.

Recall that we have defined
\begin{align*}
\Delta_n = \widehat{\tau}^\CV(\widehat{\gamma}^\HT(1),\widehat{\gamma}^\HT(0)) - \widehat{\tau}^\CV(\gamma^*(1),\gamma^*(0)).
\end{align*}
Then
\begin{multline}
\Var_{\cW}\Big( \widehat{\tau}^\CV(\widehat{\gamma}^\HT(1),\widehat{\gamma}^\HT(0))
\Big) \\
= \Var_{\cW}\Big( \widehat{\tau}^\CV(\gamma^*(1),\gamma^*(0)) \Big) + \Var_{\cW}\Big(\Delta_n\Big) + 2\Cov_{\cW}\Big( \widehat{\tau}^\CV(\gamma^*(1),\gamma^*(0)), \Delta_n \Big). \label{eqn:VarianceDecompositionInterference}
\end{multline}
It suffices to prove that both
\begin{align*}
n \Var_{\cW}\big(\Delta_n\big) \to 0
\end{align*}
and
\begin{align*}
n\Cov_{\cW}\Big( \widehat{\tau}^\CV(\gamma^*(1),\gamma^*(0)), \Delta_n \Big) \to 0.
\end{align*}

We first show that $n \Var_{\cW}(\Delta_n) \to 0$.
Because $0 \leq \Var_{\cW}(\Delta_n) \leq \bE_{\cW}[\Delta_n^2]$, it suffices to show that $n \bE_{\cW}[\Delta_n^2] \to 0$.
Recall that
\begin{align*}
\Delta_n = - \big(\widehat{\gamma}^\HT(1)-\gamma^*(1)\big) \big(\widehat{X}(1)-\bE[\widehat{X}(1)]\big) - \big(\widehat{\gamma}^\HT(0)-\gamma^*(0)\big) \big(\widehat{X}(0)-\bE[\widehat{X}(0)]\big).
\end{align*}
So we have,
\begin{align*}
\bE_{\cW}[\Delta_n^2] \leq & 2\bE_{\cW}\Big[ (\widehat{\gamma}^\HT(1)-\gamma^*(1))^2 (\widehat{X}(1)-\bE[\widehat{X}(1)])^2 \Big] + 2\bE_{\cW}\Big[ (\widehat{\gamma}^\HT(0)-\gamma^*(0))^2 (\widehat{X}(0)-\bE[\widehat{X}(0)])^2 \Big] \\
\leq & 2\bE_{\cW}\Big[ (\widehat{\gamma}^\HT(1)-\gamma^*(1))^4 \Big]^{\frac{1}{2}} \bE_{\cW}\left[ (\widehat{X}(1)-\bE[\widehat{X}(1)])^4 \right]^{\frac{1}{2}} \\
& \qquad \qquad + 2\bE_{\cW}\Big[ (\widehat{\gamma}^\HT(0)-\gamma^*(0))^4 \Big]^{\frac{1}{2}} \bE_{\cW}\left[ (\widehat{X}(0)-\bE[\widehat{X}(0)])^4 \right]^{\frac{1}{2}},
\end{align*}
where the second inequality is due to Cauchy-Schwarz inequality.

We next upper bound the four terms separately.
Note that
\begin{align*}
& \bE_{\cW}\big[(\widehat{\gamma}^\HT(1) - \gamma^*(1))^4\big] \\
= & \bE_{\cW}\bigg[\Big( \sum_{i=1}^{n} Z_{i,1}(1) \big( \bI\{\cE_i(1)\} - \Pr\nolimits_{\cW}(\cE_i(1)) \big) - \sum_{i=1}^{n} Z_{i,1}(0) \big( \bI\{\cE_i(0)\} - \Pr\nolimits_{\cW}(\cE_i(0)) \big) \Big)^4\bigg] \\
\leq & \overline{z}^4 n^{4 \alpha + 8 \beta - 4} \ \bE_{\cW}\bigg[\Big( \sum_{i=1}^n \big( \bI\{\cE_i(1)\} - \Pr\nolimits_{\cW}(\cE_i(1)) \big) - \sum_{i=1}^{n} \big( \bI\{\cE_i(0)\} - \Pr\nolimits_{\cW}(\cE_i(0)) \big) \Big)^4\bigg] \\
\leq & \overline{z}^4 n^{4 \alpha + 8 \beta - 4} \ 500 \ \overline{d}^2 n^{2+2 \alpha} \\
= & 500 \ \overline{z}^4 \overline{d}^2 \ n^{6\alpha + 8\beta - 2},
\end{align*}
where the second inequality is doing the combinatorial counting, upper bounding $|\bI\{\cE_i(1)\} - \Pr\nolimits_{\cW}(\cE_i(1))| \leq 1$ and $|\bI\{\cE_i(0)\} - \Pr\nolimits_{\cW}(\cE_i(0))| \leq 1$, and using Assumption~\ref{asp:ExpDesign:general}-(ii).
Similarly, we have
\begin{align*}
\bE_{\cW}\big[(\widehat{\gamma}^\HT(0) - \gamma^*(0))^4\big] \leq 500 \ \overline{z}^4 \overline{d}^2 \ n^{6\alpha + 8\beta - 2}.
\end{align*}
Additionally, we have
\begin{align*}
\bE_{\cW}\big[(\widehat{X}(1) - \bE_{\cW}[\widehat{X}(1)])^4\big] = & \bE_{\cW}\bigg[\Big( \frac{1}{n} \sum_{i=1}^n a_i(1) (\bI\{\cE_i(1)\} - \Pr\nolimits_{\cW}(\cE_i(1))) \Big)^4\bigg] \\
\leq & \overline{a}^4n^{-4} \bE_{\cW}\bigg[\Big( \sum_{i=1}^n \big( \bI\{\cE_i(1)\} - \Pr\nolimits_{\cW}(\cE_i(1)) \big) \Big)^4\bigg] \\
\leq & 125 \overline{a}^4 \overline{d}^2 \ n^{2\alpha - 2}, 
\end{align*}
where the second inequality is doing the combinatorial counting, upper bounding $|\bI\{\cE_i(1)\} - \Pr\nolimits_{\cW}(\cE_i(1))| \leq 1$, and using Assumption~\ref{asp:ExpDesign:general}-(ii).
Similarly, we have
\begin{align*}
\bE_{\cW}\big[(\widehat{X}(0) - \bE_{\cW}[\widehat{X}(0)])^4\big] \leq 125 \overline{a}^4 \overline{d}^2 \ n^{2\alpha - 2}. 
\end{align*}
Combining these four bounds, we have
\begin{align}
\lim_{n \to +\infty} n \bE_{\cW}[\Delta_n^2] \leq \lim_{n \to +\infty} 1000 \overline{z}^2 \overline{a}^2 \overline{d}^2 \ n^{4\alpha + 4\beta - 1} = 0, \label{eqn:VarConvergeTo0Interference}
\end{align}
where the equality holds because $4 \alpha + 4 \beta < 1$.

Next, we show that $n\Cov_{\cW}\big( \widehat{\tau}^\CV(\gamma^*(1),\gamma^*(0)), \Delta_n \big) \to 0$.
Note that
\begin{align*}
& \Big\vert \Cov_{\cW}\Big( \widehat{\tau}^\CV(\gamma^*(1),\gamma^*(0)), \Delta_n \Big) \Big\vert \\
= & \Big\vert \bE_{\cW}\Big[ \big( \widehat{\tau}^\CV(\gamma^*(1),\gamma^*(0)) - \tau(\eT,\eC) \big) \Delta_n \Big] \Big\vert \\
= & \Big\vert \bE_{\cW}\Big[ \big( \widehat{\tau}^\CV(\gamma^*(1),\gamma^*(0)) - \tau(\eT,\eC) \big) \big( \widehat{\gamma}^\HT(1) - \gamma^*(1) \big) \big( \widehat{X}(1)-\bE_{\cW}[\widehat{X}(1)] \big) \Big] \Big\vert \\
& \qquad + \Big\vert \bE_{\cW}\Big[ \big( \widehat{\tau}^\CV(\gamma^*(1),\gamma^*(0)) - \tau(\eT,\eC) \big) \big( \widehat{\gamma}^\HT(0) - \gamma^*(0) \big) \big( \widehat{X}(0)-\bE_{\cW}[\widehat{X}(0)] \big) \Big] \Big\vert \\
\leq & \bE_{\cW}\Big[ \big\vert \widehat{\tau}^\CV(\gamma^*(1),\gamma^*(0)) - \tau(\eT,\eC) \big\vert^3 \Big]^{\frac{1}{3}} \bE_{\cW}\Big[ \big\vert \widehat{\gamma}^\HT(1) - \gamma^*(1) \big\vert^3 \Big]^{\frac{1}{3}} \bE_{\cW}\Big[ \big\vert \widehat{X}(1)-\bE_{\cW}[\widehat{X}(1)] \big\vert^3 \Big]^{\frac{1}{3}} \\
& \qquad + \bE_{\cW}\Big[ \big\vert \widehat{\tau}^\CV(\gamma^*(1),\gamma^*(0)) - \tau(\eT,\eC) \big\vert^3 \Big]^{\frac{1}{3}} \bE_{\cW}\Big[ \big\vert \widehat{\gamma}^\HT(0) - \gamma^*(0) \big\vert^3 \Big]^{\frac{1}{3}} \bE_{\cW}\Big[ \big\vert \widehat{X}(0)-\bE_{\cW}[\widehat{X}(0)] \big\vert^3 \Big]^{\frac{1}{3}}
\end{align*}
where the first equality is because $\widehat{\tau}^\CV(\gamma^*(1),\gamma^*(0))$ is unbiased;
the second equality is using the definition of $\Delta_n$;
the last inequality is using Holder's inequality.

We next upper bound the three moments above.
First,
\begin{align*}
& \Big\vert \widehat{\tau}^\CV(\gamma^*(1),\gamma^*(0)) - \tau(\eT,\eC) \Big\vert \\
\leq & \frac{1}{n} \sum_{i=1}^n \Big\vert \frac{Y_i(\eT)}{\Pr_{\cW}(\cE_i(1))} - \gamma^*(1) a_i(1) \Big\vert \cdot \big\vert \bI\{\cE_i(1)\} - \Pr\nolimits_{\cW}(\cE_i(1)) \big\vert \\
& \qquad + \frac{1}{n} \sum_{i=1}^n \Big\vert \frac{Y_i(\eC)}{\Pr_{\cW}(\cE_i(0))} - \gamma^*(0) a_i(0) \Big\vert \cdot \big\vert \bI\{\cE_i(0)\} - \Pr\nolimits_{\cW}(\cE_i(0)) \big\vert \\
\leq & \frac{1}{n} \sum_{i=1}^n \Big( \frac{\overline{y}}{\underline{\pi}} n^{\beta} + \frac{2^{\frac{1}{2}} \overline{a} \overline{y} \overline{d}^{\frac{1}{2}}}{\underline{\pi} \ \underline{\lambda}_{\Omega}^{\frac{1}{2}}} n^{\frac{\alpha}{2}+\beta} \Big) \cdot \big\vert \bI\{\cE_i(1)\} - \Pr\nolimits_{\cW}(\cE_i(1)) \big\vert \\
& \qquad + \frac{1}{n} \sum_{i=1}^n \Big( \frac{\overline{y}}{\underline{\pi}} n^{\beta} + \frac{2^{\frac{1}{2}} \overline{a} \overline{y} \overline{d}^{\frac{1}{2}}}{\underline{\pi} \ \underline{\lambda}_{\Omega}^{\frac{1}{2}}} n^{\frac{\alpha}{2}+\beta} \Big) \cdot \big\vert \bI\{\cE_i(0)\} - \Pr\nolimits_{\cW}(\cE_i(0)) \big\vert \\
\leq & \frac{3 \overline{a} \overline{y} \overline{d}^{\frac{1}{2}}}{\underline{\pi} \ \underline{\lambda}_{\Omega}^{\frac{1}{2}}} n^{\frac{\alpha}{2} + \beta - 1} \bigg( \sum_{i=1}^n \big\vert \bI\{\cE_i(1)\} - \Pr\nolimits_{\cW}(\cE_i(1)) \big\vert + \sum_{i=1}^n \big\vert \bI\{\cE_i(0)\} - \Pr\nolimits_{\cW}(\cE_i(0)) \big\vert \bigg)
\end{align*}
where the first inequality is by definition and upper bounding each component by its absolute value;
the second inequality is due to Lemma~\ref{lem:UBOptimalgammasInterference};
the third inequality holds because $1 + 2^{\frac{1}{2}} \leq 3$ and when $n$ is sufficiently large. 
So we have
\begin{align*}
& \bE_{\cW}\Big[ \big\vert \widehat{\tau}^\CV(\gamma^*(1),\gamma^*(0)) - \tau(\eT,\eC) \big\vert^3 \Big] \\
\leq & \frac{27 \overline{a}^3 \overline{y}^3 \overline{d}^{\frac{3}{2}}}{\underline{\pi}^3 \underline{\lambda}_{\Omega}^{\frac{3}{2}}} n^{\frac{3\alpha}{2} + 3\beta - 3} \bE_{\cW}\bigg[ \Big( \sum_{i=1}^n \big\vert \bI\{\cE_i(1)\} - \Pr\nolimits_{\cW}(\cE_i(1)) \big\vert + \sum_{i=1}^n \big\vert \bI\{\cE_i(0)\} - \Pr\nolimits_{\cW}(\cE_i(0)) \big\vert \Big)^3 \bigg] \\
\leq & \frac{27 \overline{a}^3 \overline{y}^3 \overline{d}^{\frac{3}{2}}}{\underline{\pi}^3 \underline{\lambda}_{\Omega}^{\frac{3}{2}}} n^{\frac{3\alpha}{2} + 3\beta - 3} \cdot 10 \overline{d}^2 n^{1+2\alpha} \\
= & \frac{270 \overline{a}^3 \overline{y}^3 \overline{d}^{\frac{7}{2}}}{\underline{\pi}^3 \underline{\lambda}_{\Omega}^{\frac{3}{2}}} n^{\frac{7\alpha}{2} + 3\beta - 2},
\end{align*}
where the second inequality is doing the combinatorial counting, upper bounding $|\bI\{\cE_i(1)\} - \Pr\nolimits_{\cW}(\cE_i(1))| \leq 1$ and $|\bI\{\cE_i(0)\} - \Pr\nolimits_{\cW}(\cE_i(0))| \leq 1$, and using Assumption~\ref{asp:ExpDesign:general}-(ii).

Second, we have
\begin{align*}
& \bE_{\cW}\Big[ \big\vert \widehat{\gamma}^\HT(1) - \gamma^*(1) \big\vert^3 \Big] \\
= & \bE_{\cW}\bigg[ \Big\vert \sum_{i=1}^{n} Z_{i,1}(1) \big( \bI\{\cE_i(1)\} - \Pr\nolimits_{\cW}(\cE_i(1)) \big) - \sum_{i=1}^{n} Z_{i,1}(0) \big( \bI\{\cE_i(0)\} - \Pr\nolimits_{\cW}(\cE_i(0)) \big) \Big\vert^3 \bigg] \\
\leq & \overline{z}^3 n^{3\alpha+6\beta-3} \ 10 \overline{d}^2 n^{1+2\alpha} \\
= & \ 10 \overline{z}^3 \overline{d}^2 n^{5\alpha+6\beta-2},
\end{align*}
where the second inequality is using the upper bound $\vert Z_{i,1}(1) \vert \leq \overline{z} n^{\alpha+2\beta-1}$ and $\vert Z_{i,1}(0) \vert \leq \overline{z} n^{\alpha+2\beta-1}$, doing the combinatorial counting, upper bounding $|\bI\{\cE_i(1)\} - \Pr\nolimits_{\cW}(\cE_i(1))| \leq 1$ and $|\bI\{\cE_i(0)\} - \Pr\nolimits_{\cW}(\cE_i(0))| \leq 1$, and using Assumption~\ref{asp:ExpDesign:general}-(ii).
Similarly,
\begin{align*}
\bE_{\cW}\Big[ \big\vert \widehat{\gamma}^\HT(0) - \gamma^*(0) \big\vert^3 \Big] \leq 10 \overline{z}^3 \overline{d}^2 n^{5\alpha+6\beta-2}.
\end{align*}

Third, we have
\begin{align*}
\bE_{\cW}\Big[ \big\vert \widehat{X}(1)-\bE_{\cW}[\widehat{X}(1)] \big\vert^3 \Big] = & \bE_{\cW}\bigg[ \Big\vert \frac{1}{n} \sum_{i=1}^n a_i(1) \big(\bI\{\cE_i(1)\} - \Pr\nolimits_{\cW}(\cE_i(1))\big) \Big\vert^3 \bigg] \\
\leq & \overline{a}^3 n^{-3} \ 10 \overline{d}^2 n^{1+2\alpha} \\
= & \ 10 \overline{a}^3 \overline{d}^2 n^{2\alpha-2},
\end{align*}
where the second inequality is upper bounding $\vert a_i(1) \vert \leq \overline{a}$ and doing the combinatorial counting, upper bounding $|\bI\{\cE_i(1)\} - \Pr\nolimits_{\cW}(\cE_i(1))| \leq 1$, and using Assumption~\ref{asp:ExpDesign:general}-(ii).
Similarly,
\begin{align*}
\bE_{\cW}\Big[ \big\vert \widehat{X}(0)-\bE_{\cW}[\widehat{X}(0)] \big\vert^3 \Big] \leq 10 \overline{a}^3 \overline{d}^2 n^{2\alpha-2}.
\end{align*}

Combining these bounds, we have
\begin{align}
\lim_{n \to +\infty} n \Big\vert \Cov_{\cW}\Big( \widehat{\tau}^\CV(\gamma^*(1),\gamma^*(0)), \Delta_n \Big) \Big\vert \leq \lim_{n \to +\infty} \frac{30 \overline{a}^2 \overline{y} \overline{z} \overline{d}^{\frac{5}{2}}}{\underline{\pi} \ \underline{\lambda}_{\Omega}^{\frac{1}{2}}} n^{\frac{7\alpha}{2} + 3\beta - 1} = 0, \label{eqn:CovConvergeTo0Interference}
\end{align}
where the equality holds because $\frac{7\alpha}{2} + 3\beta < 4\alpha + 4\beta < 1$.

Putting \eqref{eqn:VarConvergeTo0Interference} and \eqref{eqn:CovConvergeTo0Interference} in \eqref{eqn:VarianceDecompositionInterference}, we have
\begin{align*}
\lim_{n \to +\infty} n \Big( \Var_{\cW}\big( \widehat{\tau}^\CV(\widehat{\gamma}^\HT(1), \widehat{\gamma}^\HT(0)) \big) - \Var_{\cW}\big( \widehat{\tau}^\CV(\gamma^*(1), \gamma^*(0)) \big) \Big) \to 0.
\end{align*}

\noindent \textbf{Claim 3: Asymptotic Normality}. 
In the proof of this claim, we assume $7\alpha + 6\beta < 1$.

First, we show that
\begin{align*}
\frac{\widehat{\tau}^{\CV}(\gamma^*(1),\gamma^*(0)) - \tau(\eT,\eC)}{\sqrt{\Var_{\cW}\big(
\widehat{\tau}^{\CV}(\gamma^*(1),\gamma^*(0)) \big)}}
\xrightarrow{d} \cN(0,1).
\end{align*}

To show this, we introduce some notations. 
Define
\begin{multline*}
\xi_i = \Big( \frac{Y_i(\eT)}{\Pr_{\cW}(\cE_i(1))} - \gamma^*(1)a_i(1) \Big) \big(\bI\{\cE_i(1)\} - \Pr\nolimits_{\cW}(\cE_i(1))\big) \\
- \Big( \frac{Y_i(\eC)}{\Pr_{\cW}(\cE_i(0))} - \gamma^*(0)a_i(0) \Big) \big(\bI\{\cE_i(0)\} - \Pr\nolimits_{\cW}(\cE_i(0))\big).
\end{multline*}
Using this notation, we have
\begin{align*}
\widehat{\tau}^\CV(\gamma^*(1),\gamma^*(0)) - \tau(\eT,\eC) = \frac{1}{n}\sum_{i=1}^n \xi_i.
\end{align*}
We can upper bound each $\xi_i$ for any $i \in [n]$ by
\begin{multline*}
\vert \xi_i \vert \leq \Big( \frac{\overline{y}}{\underline{\pi}} n^{\beta} + \frac{2^{\frac{1}{2}} \overline{a} \overline{y} \overline{d}^{\frac{1}{2}}}{\underline{\pi} \ \underline{\lambda}_{\Omega}^{\frac{1}{2}}} n^{\frac{\alpha}{2}+\beta} \Big) \cdot \Big( \big\vert \bI\{\cE_i(1)\} - \Pr\nolimits_{\cW}(\cE_i(1)) \big\vert + \big\vert \bI\{\cE_i(0)\} - \Pr\nolimits_{\cW}(\cE_i(0)) \big\vert \Big) \\
\leq \frac{3 \overline{a} \overline{y} \overline{d}^{\frac{1}{2}}}{\underline{\pi} \ \underline{\lambda}_{\Omega}^{\frac{1}{2}}} n^{\frac{\alpha}{2}+\beta} \cdot \Big( \big\vert \bI\{\cE_i(1)\} - \Pr\nolimits_{\cW}(\cE_i(1)) \big\vert + \big\vert \bI\{\cE_i(0)\} - \Pr\nolimits_{\cW}(\cE_i(0)) \big\vert \Big)
\leq \frac{6 \overline{a} \overline{y} \overline{d}^{\frac{1}{2}}}{\underline{\pi} \ \underline{\lambda}_{\Omega}^{\frac{1}{2}}} n^{\frac{\alpha}{2}+\beta}
\end{multline*}
the first inequality is due to Lemma~\ref{lem:UBOptimalgammasInterference};
the second inequality holds when $n$ is sufficiently large;
the third inequality is because $\big\vert \bI\{\cE_i(1)\} - \Pr\nolimits_{\cW}(\cE_i(1)) \big\vert \leq 1$ and $\big\vert \bI\{\cE_i(0)\} - \Pr\nolimits_{\cW}(\cE_i(0)) \big\vert \leq 1$.
From here, denote $\overline{\xi} = \dfrac{6 \overline{a} \overline{y} \overline{d}^{\frac{1}{2}}}{\underline{\pi} \ \underline{\lambda}_{\Omega}^{\frac{1}{2}}}$.
Next, denote
\begin{align*}
\sigma_n = \sqrt{ \Var_{\cW}\Big( \sum_{i=1}^n \xi_i \Big) }, \qquad S_n = \frac{\sum_{i=1}^n \xi_i}{\sigma_n}.
\end{align*}
We will use Lemma~\ref{lem:DependencyNeighborhoodCLT} to show the central limit theorem. 
To do so, we bound each term in Lemma~\ref{lem:DependencyNeighborhoodCLT}.

First, recall that we additionally assume that for sufficiently large $n$, $n \Var\big( \widehat{\tau}^{\CV}(\gamma^*(1), \gamma^*(0)) \big) \geq \underline{c}$.
So we have
\begin{align*}
\sigma_n^2 = \Var_{\cW}\Big(\sum_{i=1}^n \xi_i\Big) = n^2 \Var_{\cW}\Big(\frac{1}{n} \sum_{i=1}^n \xi_i\Big) = n^2 \Var_{\cW}\big( \widehat{\tau}^{\CV}(\gamma^*(1), \gamma^*(0)) \big) \geq \underline{c} n.
\end{align*}
We also have
\begin{align*}
\sigma_n^3 \geq \underline{c}^\frac{3}{2} n^\frac{3}{2}.
\end{align*}
Next, we have
\begin{align*}
\sum_{i=1}^n \bE[\vert\xi_i\vert^3] & \leq \sum_{i=1}^n \overline{\xi}^3 n^{\frac{3}{2}\alpha + 3\beta} = \overline{\xi}^3 n^{1 + \frac{3}{2}\alpha + 3\beta}, \\
\sum_{i=1}^n \bE[\xi_i^4] & \leq \sum_{i=1}^n \overline{\xi}^4 n^{2\alpha + 4\beta} = \overline{\xi}^4 n^{1 + 2\alpha + 4\beta}.
\end{align*}
Putting all together into Lemma~\ref{lem:DependencyNeighborhoodCLT}, for $Z$ a standard normal random variable, 
\begin{align*}
\lim_{n \to +\infty} d_W(S_n, Z) \leq & \lim_{n \to +\infty} \ \frac{\overline{d}^2 n^{2\alpha}}{\underline{c}^\frac{3}{2} n^\frac{3}{2}} \overline{\xi}^3 n^{1+\frac{3}{2}\alpha+3\beta} + \frac{26^\frac{1}{2} \overline{d}^\frac{3}{2} n^{\frac{3}{2} \alpha}}{\pi^\frac{1}{2} \underline{c} n} \Big(\overline{\xi}^4 n^{1+2\alpha+4\beta}\Big)^\frac{1}{2} \\
= & \lim_{n \to +\infty} \ \frac{\overline{d}^2 \overline{\xi}^3}{\underline{c}^\frac{3}{2}} n^{\frac{7}{2}\alpha +3\beta - \frac{1}{2}} + \frac{26^\frac{1}{2} \overline{d}^\frac{3}{2} \overline{\xi}^2}{\pi^\frac{1}{2} \underline{c}} n^{\frac{5}{2}\alpha +2\beta - \frac{1}{2}} \\
= & \ 0,
\end{align*}
where $D_W(S_n, Z)$ stands for the Wasserstein distance between $S_n$ and $Z$;
the first inequality holds using Lemma~\ref{lem:DependencyNeighborhoodCLT};
the last equality holds because $\frac{5}{2}\alpha +2\beta < \frac{7}{2}\alpha +3\beta < \frac{1}{2}$.
So we have
\begin{align*}
\lim_{n \to +\infty} S_n = \lim_{n \to +\infty} \frac{\widehat{\tau}^\CV(\gamma^*(1), \gamma^*(0)) - \tau(\eT,\eC)}{\sqrt{\Var\big(\widehat{\tau}^{\CV}(\gamma^*(1), \gamma^*(0))\big)}} \xrightarrow{d} \cN(0,1).
\end{align*}

Next, we establish the connection between $\widehat{\tau}^\CV(\widehat{\gamma}^\HT(1), \widehat{\gamma}^\HT(0))$ and $\widehat{\tau}^\CV(\gamma^*(1), \gamma^*(0))$.
Note that,
\begin{multline}
\frac{\widehat{\tau}^\CV(\widehat{\gamma}^\HT(1), \widehat{\gamma}^\HT(0)) - \tau(\eT,\eC)}{\sqrt{\Var\big(\widehat{\tau}^\CV(\widehat{\gamma}^\HT(1), \widehat{\gamma}^\HT(0))\big)}} 
\\
= \frac{\widehat{\tau}^\CV(\gamma^*(1), \gamma^*(0)) - \tau(\eT,\eC)}{\sqrt{\Var\big(\widehat{\tau}^\CV(\gamma^*(1), \gamma^*(0))\big)}} \frac{\sqrt{\Var\big(\widehat{\tau}^\CV(\gamma^*(1), \gamma^*(0))\big)}}{\sqrt{\Var\big(\widehat{\tau}^\CV(\widehat{\gamma}^\HT(1), \widehat{\gamma}^\HT(0))\big)}} \\
+ \frac{\Delta_n}{\sqrt{\Var\big(\widehat{\tau}^\CV(\gamma^*(1), \gamma^*(0))\big)}} \frac{\sqrt{\Var\big(\widehat{\tau}^\CV(\gamma^*(1), \gamma^*(0))\big)}}{\sqrt{\Var\big(\widehat{\tau}^\CV(\widehat{\gamma}^\HT(1), \widehat{\gamma}^\HT(0))\big)}}. \label{eqn:CLTDecompositionInterference}
\end{multline}
From here, note that
\begin{align*}
\bE\Bigg[\frac{\Delta_n}{\sqrt{\Var\big(\widehat{\tau}^\CV(\gamma^*(1), \gamma^*(0))\big)}} \Bigg] \leq \bE\Bigg[\frac{\Delta_n^2}{\Var\big(\widehat{\tau}^\CV(\gamma^*(1), \gamma^*(0))\big)} \Bigg]^\frac{1}{2} \leq \sqrt{ \frac{500 \overline{z}^2 \overline{a}^2 \overline{d}^2}{\underline{c}} \ n^{4\alpha +4\beta - 1}},
\end{align*}
where the first inequality is Cauchy-Schwarz inequality;
and the second inequality is using \eqref{eqn:VarConvergeTo0Interference} and assuming $\Var\big(\widehat{\tau}^\CV(\gamma^*(1), \gamma^*(0))\big) \geq \underline{c} n^{-1}$.
So for any $\epsilon > 0$,
\begin{align*}
\lim_{n \to +\infty} \Pr\Bigg( \Bigg\vert \frac{\Delta_n}{\sqrt{\Var\big(\widehat{\tau}^\CV(\gamma^*(1), \gamma^*(0))\big)}} \Bigg\vert \geq \epsilon\Bigg) 
\leq \lim_{n \to +\infty} \sqrt{\frac{500 \overline{z}^2 \overline{a}^4 \overline{d}^4}{\underline{c}} \cdot n^{4\alpha +4\beta - 1}} \cdot \frac{1}{\epsilon}
= 0,
\end{align*}
where the first inequality is due to Markov inequality;
the last and only equality holds because $4\alpha + 4\beta < 7\alpha + 6\beta < 1$.
So we have
\begin{align*}
\lim_{n \to +\infty} \frac{\Delta_n}{\sqrt{\Var\big(\widehat{\tau}^\CV(\gamma^*(1), \gamma^*(0))\big)}} \xrightarrow{p} 0.
\end{align*}
Note also that $\lim_{n \to +\infty} n \Var\big(\widehat{\tau}^\CV(\widehat{\gamma}^\HT(1), \widehat{\gamma}^\HT(0))\big) = \lim_{n \to +\infty} n \Var\big(\widehat{\tau}^\CV(\gamma^*(1), \gamma^*(0))\big)$ and for sufficiently large $n$ we assume $n \Var\big(\widehat{\tau}^\CV(\gamma^*(1), \gamma^*(0))\big) \geq \underline{c}$.
This ensures that $n \Var\big(\widehat{\tau}^\CV(\widehat{\gamma}^\HT(1), \widehat{\gamma}^\HT(0))\big) \geq \underline{c} > 0$ for sufficiently large $n$.
So we have
\begin{align*}
& \lim_{n \to +\infty} \Bigg\vert \frac{\Var\big(\widehat{\tau}^\CV(\gamma^*(1), \gamma^*(0))\big)}{\Var\big(\widehat{\tau}^\CV(\widehat{\gamma}^\HT(1), \widehat{\gamma}^\HT(0))\big)} - 1 \Bigg\vert \\ 
= & \lim_{n \to +\infty} \frac{\big\vert n \Var\big(\widehat{\tau}^\CV(\widehat{\gamma}^\HT(1), \widehat{\gamma}^\HT(0))\big) - n \Var\big(\widehat{\tau}^\CV(\gamma^*(1), \gamma^*(0))\big) \big\vert}{n \Var\big(\widehat{\tau}^\CV(\widehat{\gamma}^\HT(1), \widehat{\gamma}^\HT(0))\big)} \\
\leq & \lim_{n \to +\infty} \frac{\big\vert n \Var\big(\widehat{\tau}^\CV(\widehat{\gamma}^\HT(1), \widehat{\gamma}^\HT(0))\big) - n \Var\big(\widehat{\tau}^\CV(\gamma^*(1), \gamma^*(0))\big) \big\vert}{\underline{c}} \\
= & 0,
\end{align*}
where the last equality is Claim 2 in Theorem~\ref{thm:AsymptoticInterference}.

Putting the above together into \eqref{eqn:CLTDecompositionInterference} and using the Slutsky theorem we have
\begin{align*}
\lim_{n \to +\infty} \frac{\widehat{\tau}^\CV(\widehat{\gamma}^\HT(1), \widehat{\gamma}^\HT(0)) - \tau(\eT,\eC)}{\sqrt{\Var\big(\widehat{\tau}^\CV(\widehat{\gamma}^\HT(1), \widehat{\gamma}^\HT(0))\big)}} \xrightarrow{d} \cN(0,1).
\end{align*}
This finishes the proof.
\hfill \halmos
\endproof

\subsection{Proof of Theorem~\ref{thm:OPTInterference}}

\proof{Proof of Theorem~\ref{thm:OPTInterference}.}
Define the following block diagonal constraint
\begin{align*}
\cB = \bigg\{ \bm{B} \in \bR^{2n \times 2} \bigg\vert \bm{B} = 
\begin{bmatrix} 
\bm{a}(1) & \bm{0}_n  \\ 
\bm{0}_n  & \bm{a}(0) 
\end{bmatrix} 
\bigg\}.
\end{align*}

Recall that the objective function of \eqref{eqn:FormulationCausal:general} can be written as
\begin{align*}
\min_{\bm{B} \in \cB} \ \bE_{\bm{Y}(\eT), \bm{Y}(\eC) \sim \cY_{\delta}^n} \Big[\bm{Y}^\top \bm{\Pi}^{-1} \bm{\Omega} \bm{\Pi}^{-1} \bm{Y} - \bm{Y}^\top \bm{\Pi}^{-1} \bm{\Omega} \bm{B} \big(\bm{B}^\top \bm{\Omega} \bm{B}\big)^{-1} \bm{B}^\top \bm{\Omega} \bm{\Pi}^{-1} \bm{Y}\Big].
\end{align*}
Note that $\bE_{\bm{Y}(\eT), \bm{Y}(\eC) \sim \cY_{\delta}^n} \big[\bm{Y}^\top \bm{\Pi}^{-1} \bm{\Omega} \bm{\Pi}^{-1} \bm{Y}\big]$ does not depend on $\bm{B}$.
So this minimization problem can be written as a maximization problem
\begin{align}
\max_{\bm{B} \in \cB} \ \bE_{\bm{Y}(\eT), \bm{Y}(\eC) \sim \cY_{\delta}^n} \Big[\bm{Y}^\top \bm{\Pi}^{-1} \bm{\Omega} \bm{B} \big(\bm{B}^\top \bm{\Omega} \bm{B}\big)^{-1} \bm{B}^\top \bm{\Omega} \bm{\Pi}^{-1} \bm{Y}\Big]. \label{eqn:CIIntermediate2:general}
\end{align}
Next, note that
\begin{align*}
& \bE_{\bm{Y}(\eT), \bm{Y}(\eC) \sim \cY_{\delta}^n} \Big[\bm{Y}^\top \bm{\Pi}^{-1} \bm{\Omega} \bm{B} \big(\bm{B}^\top \bm{\Omega} \bm{B}\big)^{-1} \bm{B}^\top \bm{\Omega} \bm{\Pi}^{-1} \bm{Y}\Big] \\
= & \bE_{\bm{Y}(\eT), \bm{Y}(\eC) \sim \cY_{\delta}^n} \Big[ \Tr\big( \bm{Y}^\top \bm{\Pi}^{-1} \bm{\Omega} \bm{B} \big(\bm{B}^\top \bm{\Omega} \bm{B}\big)^{-1} \bm{B}^\top \bm{\Omega} \bm{\Pi}^{-1} \bm{Y} \big) \Big] \\
= & \bE_{\bm{Y}(\eT), \bm{Y}(\eC) \sim \cY_{\delta}^n} \Big[ \Tr\big( \big(\bm{B}^\top \bm{\Omega} \bm{B}\big)^{-1} \bm{B}^\top \bm{\Omega} \bm{\Pi}^{-1} \bm{Y} \bm{Y}^\top \bm{\Pi}^{-1} \bm{\Omega} \bm{B} \big) \Big] \\
= & \Tr\Big( \big(\bm{B}^\top \bm{\Omega} \bm{B}\big)^{-1} \bm{B}^\top \bm{\Omega} \bm{\Pi}^{-1} \bE_{\bm{Y}(\eT), \bm{Y}(\eC) \sim \cY_{\delta}^n} \big[ \bm{Y} \bm{Y}^\top \big] \bm{\Pi}^{-1} \bm{\Omega} \bm{B} \Big),
\end{align*}
where the first equality is because a scalar is equal to its trace;
the second equality is because of the cyclic property of a trace;
the third equality is because of linearity of expectations.
Now we define 
\begin{align*}
\bm{M} = & \bm{\Omega}^{\frac{1}{2}} \bm{\Pi}^{-1} \bE_{\bm{Y}(\eT), \bm{Y}(\eC) \sim \cY_{\delta}^n} \Big[ \bm{Y} \bm{Y}^\top\Big] \bm{\Pi}^{-1} \bm{\Omega}^{\frac{1}{2}} \\
= & \bm{\Omega}^{\frac{1}{2}} \bm{\Pi}^{-1} 
\begin{bmatrix}
\mu^2(1) \bm{1}_n \bm{1}_n^\top + \sigma^2(1) \bm{I}_n & - \mu(1)\mu(0) \bm{1}_n \bm{1}_n^\top - \sigma(1,0) \bm{I}_n \\
- \mu(1)\mu(0) \bm{1}_n \bm{1}_n^\top - \sigma(1,0) \bm{I}_n & \mu^2(0) \bm{1}_n \bm{1}_n^\top + \sigma^2(0) \bm{I}_n
\end{bmatrix}
\bm{\Pi}^{-1} \bm{\Omega}^{\frac{1}{2}},
\end{align*}
where $\bm{I}_n$ is a $n \times n$ identity matrix and $\bm{1}_n$ is a $n$-dimensional vector with each element equal to $1$. 
So \eqref{eqn:CIIntermediate2:general} can be written as
\begin{align}
\max_{\bm{B} \in \cB} \ \Tr\Big( \big(\bm{B}^\top \bm{\Omega} \bm{B}\big)^{-1} \bm{B}^\top \bm{\Omega}^{\frac{1}{2}} \bm{M} \bm{\Omega}^{\frac{1}{2}} \bm{B} \Big). \label{eqn:CIIntermediate4-0:general}
\end{align}

Now we consider an unconstrained version of \eqref{eqn:CIIntermediate4-0:general}.
\begin{align}
\max_{\bm{B}} \ \Tr\Big( \big(\bm{B}^\top \bm{\Omega} \bm{B}\big)^{-1} \bm{B}^\top \bm{\Omega}^{\frac{1}{2}} \bm{M} \bm{\Omega}^{\frac{1}{2}} \bm{B} \Big). \label{eqn:CIIntermediate4:general}
\end{align}
Below, we drop $\bm{B} \in \bR^{2n \times 2}$ from the maximization above to stand for unconstrained maximization. 

We introduce a variable transformation
Note that $\bm{\Omega}$ is a symmetric and positive semidefinite matrix.
So $\bm{B}^\top \bm{\Omega} \bm{B}$ is also a symmetric and positive semidefinite matrix.
So we can define $\big(\bm{B}^\top \bm{\Omega} \bm{B}\big)^{-\frac{1}{2}}$ as the pseudo inverse square root matrix of $\bm{B}^\top \bm{\Omega} \bm{B}$.
Now define
\begin{align*}
\bm{V} = \bm{\Omega}^{\frac{1}{2}} \bm{B} \big(\bm{B}^\top \bm{\Omega} \bm{B}\big)^{-\frac{1}{2}}.
\end{align*}
Using the definition of $\bm{V}$, \eqref{eqn:CIIntermediate4:general} can be written as
\begin{align*}
\max_{\bm{B}} \ \Tr\Big( \bm{V}^\top \bm{M} \bm{V} \Big). 
\end{align*}

Now we consider the rank of $\bm{V}$.
There are two cases. 
First, if $\mathrm{rank}(\bm{V}) = 2$, then 
\begin{align*}
\bm{V}^\top \bm{V} =  \big(\bm{B}^\top \bm{\Omega} \bm{B}\big)^{-\frac{1}{2}} \bm{B}^\top \bm{\Omega}^{\frac{1}{2}} \bm{\Omega}^{\frac{1}{2}} \bm{B} \big(\bm{B}^\top \bm{\Omega} \bm{B}\big)^{-\frac{1}{2}} = \bm{I}_2.
\end{align*}
Using Lemma~\ref{lem:FanPrinciple} the Fan's Principle (\citet{fan1949theorem} Theorem 1), the optimal objective value of \eqref{eqn:CIIntermediate4:general} is the sum of the two largest eigenvalues of $\bm{M}$, denoted as $\lambda_1(\bm{M})$ and $\lambda_2(\bm{M})$.
So the optimal variance reduction is given by 
\begin{align*}
\bE_{\bm{Y}(\eT), \bm{Y}(\eC) \sim \cY_{\delta}^n} \Big[ \Var\big( \widehat{\tau}^{\HT} \big) - \Var\big( \widehat{\tau}^{\CV}(\gamma^*(1), \gamma^*(0)) \big) \Big] = \frac{\lambda_1(\bm{M}) + \lambda_2(\bm{M})}{n^2}.
\end{align*}

Second, if $\mathrm{rank}(\bm{V}) = 1$, then  
\begin{align*}
\bm{V}^\top \bm{V} =  \big(\bm{B}^\top \bm{\Omega} \bm{B}\big)^{-\frac{1}{2}} \bm{B}^\top \bm{\Omega}^{\frac{1}{2}} \bm{\Omega}^{\frac{1}{2}} \bm{B} \big(\bm{B}^\top \bm{\Omega} \bm{B}\big)^{-\frac{1}{2}} = 
\begin{bmatrix}
1 & 0 \\
0 & 0
\end{bmatrix}.
\end{align*}
In this case, using Lemma~\ref{lem:FanPrinciple} the Fan's Principle (\citet{fan1949theorem} Theorem 1), the optimal objective value of \eqref{eqn:CIIntermediate4} is the largest eigenvalue of $\bm{M}$, denoted as $\lambda_1(\bm{M})$.
So the optimal variance reduction is given by 
\begin{align*}
\bE_{\bm{Y}(\eT), \bm{Y}(\eC) \sim \cY_{\delta}^n} \Big[ \Var\big( \widehat{\tau}^{\HT} \big) - \Var\big( \widehat{\tau}^{\CV}(\gamma^*(1), \gamma^*(0)) \big) \Big] = \frac{\lambda_1(\bm{M})}{n^2} \leq \frac{\lambda_1(\bm{M}) + \lambda_2(\bm{M})}{n^2}.
\end{align*}
So the variance reduction in Case 2 is always smaller or equal to the variance reduction in Case 1. 

Finally, we compare the constrained maximization problem \eqref{eqn:CIIntermediate4-0:general} and the unconstrained maximization problem \eqref{eqn:CIIntermediate4:general}.
Since the constrained maximization problem is always smaller than or equal to the unconstrained problem, we conclude that 
\begin{align*}
\bE_{\bm{Y}(\eT), \bm{Y}(\eC) \sim \cY_{\delta}^n} \Big[ \Var\big( \widehat{\tau}^{\HT} \big) - \Var\big( \widehat{\tau}^{\CV}(\gamma^*(1), \gamma^*(0)) \big) \Big] \leq \frac{\lambda_1(\bm{M}) + \lambda_2(\bm{M})}{n^2},
\end{align*}
for the constrained maximization problem \eqref{eqn:CIIntermediate4-0:general}.
\hfill \halmos
\endproof

\subsection{Proof of Corollary~\ref{coro:SUTVASpecialCase}}

\proof{Proof of Corollary~\ref{coro:SUTVASpecialCase}.}
Consider the special case of causal inference under SUTVA.
In this special case, Example~\ref{exa:SUTVA} shows that $\bm{\Omega}$ has a special block structure 
\begin{align*}
\bm{\Omega} =
\begin{bmatrix}
\bm{\Sigma} & -\bm{\Sigma}\\
-\bm{\Sigma} & \bm{\Sigma}
\end{bmatrix}.
\end{align*} 
Using this special block structure, we show that the optimal bases matrix has a closed form solution.
Consider an unconstrained solution of \eqref{eqn:ConstrainedKyFan},
\begin{align*}
\bm{B}^{\circ} \in \argmax_{\bm{B} \in \bR^{2n \times 2}} \ \Tr\Big( \big(\bm{B}^\top \bm{\Omega} \bm{B}\big)^{-1} \bm{B}^\top \bm{\Omega}^{\frac{1}{2}} \bm{M} \bm{\Omega}^{\frac{1}{2}} \bm{B} \Big). 
\end{align*}
The solution is explicitly given by
\begin{align*}
\bm{B}^{\circ} = \bm{\Omega}^{-\frac{1}{2}} \big[\bm{u}_1(\bm{M}), \bm{u}_2(\bm{M})\big] \bm{C}
\end{align*}
where $\bm{C}$ is any invertible, symmetric and semidefinite $2 \times 2$ matrix. 
Now denote
\begin{align*}
\bm{B}^{\circ} =
\begin{bmatrix}
\bm{u}_1 & \bm{u}_2 \\
\bm{v}_1 & \bm{v}_2
\end{bmatrix}
\in \bR^{2n \times 2}.
\end{align*}
Using the solution $\bm{B}^{\circ}$, we construct an optimal solution
\begin{align*}
\bm{B}^* =
\begin{bmatrix}
\bm{u}_1 - \bm{v}_1 & \bm{0}_n            \\
\bm{0}_n            & \bm{v}_2 - \bm{u}_2 
\end{bmatrix}
\in \cB.
\end{align*}

Next, we show that $\bm{B}^{\circ}$ and $\bm{B}^*$ attain exactly the same objective value, that is, 
\begin{align*}
\Tr\Big( \big((\bm{B}^{\circ})^\top \bm{\Omega} \bm{B}^{\circ}\big)^{-1} (\bm{B}^{\circ})^\top \bm{\Omega}^{\frac{1}{2}} \bm{M} \bm{\Omega}^{\frac{1}{2}} \bm{B}^{\circ} \Big) = \Tr\Big( \big((\bm{B}^*)^\top \bm{\Omega} \bm{B}^*\big)^{-1} (\bm{B}^*)^\top \bm{\Omega}^{\frac{1}{2}} \bm{M} \bm{\Omega}^{\frac{1}{2}} \bm{B}^* \Big).
\end{align*}

To see this, we consider $\bm{B}^{\circ} - \bm{B}^*$. 
For the first column,
\begin{align*}
\begin{bmatrix}
\bm{u}_1\\
\bm{v}_1
\end{bmatrix}
-
\begin{bmatrix}
\bm{u}_1 - \bm{v}_1\\
\bm{0}_n
\end{bmatrix}
=
\begin{bmatrix}
\bm{v}_1\\
\bm{v}_1
\end{bmatrix},
\end{align*}
and for the second column,
\begin{align*}
\begin{bmatrix}
\bm{u}_2\\
\bm{v}_2
\end{bmatrix}
-
\begin{bmatrix}
\bm{0}_n \\
\bm{v}_2 - \bm{u}_2
\end{bmatrix}
=
\begin{bmatrix}
\bm{u}_2\\
\bm{u}_2
\end{bmatrix}.
\end{align*}
Given the special block structure under SUTVA, both vectors are in the null space of $\bm{\Omega}$.
Therefore,
\begin{align*}
\bm{\Omega}^{\frac{1}{2}}\bm{B}^{\circ} = \bm{\Omega}^{\frac{1}{2}} \bm{B}^*.
\end{align*}
Then we have
\begin{align*}
(\bm{B}^{\circ})^\top \bm{\Omega} \bm{B}^{\circ} = (\bm{B}^*)^\top \bm{\Omega} \bm{B}^*
\end{align*}
and
\begin{align*}
(\bm{B}^{\circ})^\top \bm{\Omega}^{\frac{1}{2}} \bm{M} \bm{\Omega}^{\frac{1}{2}} \bm{B}^{\circ} = (\bm{B}^*)^\top \bm{\Omega}^{\frac{1}{2}} \bm{M} \bm{\Omega}^{\frac{1}{2}} \bm{B}^*.
\end{align*}
So $\bm{B}^*$ and $\bm{B}^{\circ}$ attain the same objective value. 
Since every $2n \times 2$ bases matrix corresponds to a block diagonal bases matrix with the same objective value, the constrained and unconstrained optimization problems have the same optimal value under SUTVA. 
This shows that, in expectation, the optimal variance reduction is exactly equal to $(\lambda_1(\bm{M}) + \lambda_2(\bm{M})) n^{-2}$.
\hfill \halmos
\endproof

\subsection{Proof of Corollary~\ref{coro:NoiselessPOExposures}}

\proof{Proof of Corollary~\ref{coro:NoiselessPOExposures}.}
The proof of Corollary~\ref{coro:NoiselessPOExposures} is very straightforward. 
Recall from Section~\ref{sec:OptimalControlVariates:general} that if the potential outcomes $\bm{Y}(\eT)$ and $\bm{Y}(\eC)$ were known to take values $Y_i(\eT) = y_i(\eT)$ and $Y_i(\eC) = y_i(\eC)$ for any $i \in [n]$, the problem 
\begin{align*}
\min_{\bm{B} \in \cB} \ \Var\big( \widehat{\tau}^{\CV}(\gamma^*(1), \gamma^*(0)) \big) \ = \ \min_{\bm{B} \in \cB} \frac{1}{n^2} \bigg( \bm{y}^\top \bm{\Pi}^{-1} \bm{\Omega} \bm{\Pi}^{-1} \bm{y} - \bm{y}^\top \bm{\Pi}^{-1} \bm{\Omega} \bm{B} \big(\bm{B}^\top \bm{\Omega} \bm{B}\big)^{-1} \bm{B}^\top \bm{\Omega} \bm{\Pi}^{-1} \bm{y} \bigg).
\end{align*}
can be easily solved by choosing $\bm{B}$ such that $\bm{\Pi}^{-1} \bm{y}$ is in the column space of $\bm{B}$, such as $\bm{a}(1) = \bm{\Pi}(1)^{-1} \bm{y}(\eT)$ and $\bm{a}(0) = -\bm{\Pi}(0)^{-1}\bm{y}(\eC)$.
Now because we assume that the potential outcomes $Y_i(\eT) = y(\eT)$ and $Y_i(\eC) = y(\eC)$ take the same unknown constant, then one optimal solution is given by $\bm{a}(1) = \bm{\Pi}(1)^{-1} \bm{1}_n \cdot y(\eT)$ and $\bm{a}(0) = - \bm{\Pi}(0)^{-1} \bm{1}_n \cdot y(\eC)$.
\hfill \halmos
\endproof

\subsection{Proof of Theorem~\ref{thm:1/2approx}}

\proof{Proof of Theorem~\ref{thm:1/2approx}.}
The proof proceeds in three steps. 
First, we solve the problem 
\begin{align*}
\max_{\bm{b} \in \cB} \frac{\bm{b}^\top \bm{\Omega}^{\frac{1}{2}} \bm{M} \bm{\Omega}^{\frac{1}{2}} \bm{b}}{\bm{b}^\top \bm{\Omega} \bm{b}}.
\end{align*}
This is a Rayleigh quotient problem and the optimal objective value is given as $\lambda_1(\bm{M})$, the largest eigenvalue of $\bm{M}$, and the optimal solution is explicitly given as 
\begin{align*}
\widetilde{\bm{b}} = c \bm{\Omega}^{-\frac{1}{2}} \bm{u}_1(\bm{M}),
\end{align*}
where $c \ne 0$ is any constant, and $\bm{u}_1(\bm{M})$ is the eigenvector corresponding to $\lambda_1(\bm{M})$. 
The computation of $\widetilde{\bm{b}}$ takes polynomial time.
Now denote $\widetilde{\bm{b}} = (\widetilde{\bm{b}}_1^\top, \widetilde{\bm{b}}_2^\top)^\top$ and we construct
\begin{align*}
\widetilde{\bm{B}} = 
\begin{bmatrix}
\widetilde{\bm{b}}_1 & \bm{0}_n             \\
\bm{0}_n             & \widetilde{\bm{b}}_2 
\end{bmatrix}.
\end{align*} 
Apparently, $\widetilde{\bm{B}} \in \cB$ is a feasible solution.

Second, we evaluate the feasible solution $\widetilde{\bm{B}}$.
Denote the following $2 \times 2$ symmetric and positive semidefinite matrix
\begin{align*}
\bm{S} = \big(\widetilde{\bm{B}}^\top \bm{\Omega} \widetilde{\bm{B}}\big)^{-\frac{1}{2}} \widetilde{\bm{B}}^\top \bm{\Omega}^{\frac{1}{2}} \bm{M} \bm{\Omega}^{\frac{1}{2}} \widetilde{\bm{B}} \big(\widetilde{\bm{B}}^\top \bm{\Omega} \widetilde{\bm{B}}\big)^{-\frac{1}{2}}.
\end{align*}
Using the cyclic property of a trace, we have 
\begin{align*} 
\Tr\Big( \big(\widetilde{\bm{B}}^\top \bm{\Omega} \widetilde{\bm{B}}\big)^{-1} \widetilde{\bm{B}}^\top \bm{\Omega}^{\frac{1}{2}} \bm{M} \bm{\Omega}^{\frac{1}{2}} \widetilde{\bm{B}} \Big) = \Tr\big( \bm{S} \big).
\end{align*}
Next, note that $\widetilde{\bm{b}} = \widetilde{\bm{B}} \bm{1}_2$, where $\bm{1}_2 = (1,1)^\top$ is a two-dimensional vector with both elements equal to $1$. 
We denote the $2$-dimensional vector $\bm{x} = \big(\widetilde{\bm{B}}^\top \bm{\Omega} \widetilde{\bm{B}}\big)^{-\frac{1}{2}} \bm{1}_2$.
Using this variable transformation, we have
\begin{align*}
\frac{\widetilde{\bm{b}}^\top \bm{\Omega}^{\frac{1}{2}} \bm{M} \bm{\Omega}^{\frac{1}{2}} \widetilde{\bm{b}}}{\widetilde{\bm{b}}^\top \bm{\Omega} \widetilde{\bm{b}}} = \frac{\bm{x}^\top \bm{S} \bm{x}}{\bm{x}^\top \bm{x}}.
\end{align*}
So we have
\begin{align*}
\frac{\widetilde{\bm{b}}^\top \bm{\Omega}^{\frac{1}{2}} \bm{M} \bm{\Omega}^{\frac{1}{2}} \widetilde{\bm{b}}}{\widetilde{\bm{b}}^\top \bm{\Omega} \widetilde{\bm{b}}} = \frac{\bm{x}^\top \bm{S} \bm{x}}{\bm{x}^\top \bm{x}} \leq \lambda_1(\bm{S}) \leq \sum_{l=1}^2 \lambda_l(\bm{S}) = \Tr(\bm{S}) = \Tr\Big( \big(\widetilde{\bm{B}}^\top \bm{\Omega} \widetilde{\bm{B}}\big)^{-1} \widetilde{\bm{B}}^\top \bm{\Omega}^{\frac{1}{2}} \bm{M} \bm{\Omega}^{\frac{1}{2}} \widetilde{\bm{B}} \Big),
\end{align*}
where the first inequality is due to Lemma~\ref{lem:RayleighQuotient};
the second inequality is because $\bm{S}$ is a positive semidefinite matrix so all the eigenvalues are non-negative;
the second equality is because the trace of a matrix is the sum of all its eigenvalues.
The above inequality shows that
\begin{align} 
\Tr\Big( \big(\widetilde{\bm{B}}^\top \bm{\Omega} \widetilde{\bm{B}}\big)^{-1} \widetilde{\bm{B}}^\top \bm{\Omega}^{\frac{1}{2}} \bm{M} \bm{\Omega}^{\frac{1}{2}} \widetilde{\bm{B}} \Big) \geq \frac{\widetilde{\bm{b}}^\top \bm{\Omega}^{\frac{1}{2}} \bm{M} \bm{\Omega}^{\frac{1}{2}} \widetilde{\bm{b}}}{\widetilde{\bm{b}}^\top \bm{\Omega} \widetilde{\bm{b}}} = \lambda_1(\bm{M}). \label{eqn:1/2approx1}
\end{align}

Third, we compare the feasible solution $\widetilde{\bm{B}}$ with the optimal solution.
Note that
\begin{multline}
\Tr\Big( \big(\widetilde{\bm{B}}^\top \bm{\Omega} \widetilde{\bm{B}}\big)^{-1} \widetilde{\bm{B}}^\top \bm{\Omega}^{\frac{1}{2}} \bm{M} \bm{\Omega}^{\frac{1}{2}} \widetilde{\bm{B}} \Big) \leq \max_{\bm{B} \in \cB} \Tr\Big( \big(\bm{B}^\top \bm{\Omega} \bm{B}\big)^{-1} \bm{B}^\top \bm{\Omega}^{\frac{1}{2}} \bm{M} \bm{\Omega}^{\frac{1}{2}} \bm{B} \Big) \\
\leq \max_{\bm{B} \in \bR^{2n \times 2}} \Tr\Big( \big(\bm{B}^\top \bm{\Omega} \bm{B}\big)^{-1} \bm{B}^\top \bm{\Omega}^{\frac{1}{2}} \bm{M} \bm{\Omega}^{\frac{1}{2}} \bm{B} \Big) 
= \lambda_1(\bm{M}) + \lambda_2(\bm{M})
\leq 2 \lambda_1(\bm{M}), \label{eqn:1/2approx2}
\end{multline}
where the first inequality is because we compare the feasible solution with the optimal solution;
the second inequality is because we relax the constrained maximization problem \eqref{eqn:ConstrainedKyFan} to the unconstrained maximization problem;
the equality is due to Lemma~\ref{lem:FanPrinciple}.

Combining \eqref{eqn:1/2approx1} and \eqref{eqn:1/2approx2} we have
\begin{align*}
\Tr\Big( \big(\widetilde{\bm{B}}^\top \bm{\Omega} \widetilde{\bm{B}}\big)^{-1} \widetilde{\bm{B}}^\top \bm{\Omega}^{\frac{1}{2}} \bm{M} \bm{\Omega}^{\frac{1}{2}} \widetilde{\bm{B}} \Big) \geq \frac{1}{2} \max_{\bm{B} \in \cB} \Tr\Big( \big(\bm{B}^\top \bm{\Omega} \bm{B}\big)^{-1} \bm{B}^\top \bm{\Omega}^{\frac{1}{2}} \bm{M} \bm{\Omega}^{\frac{1}{2}} \bm{B} \Big),
\end{align*}
which finishes the proof.
\hfill \halmos
\endproof

\subsection{Proof of Theorem~\ref{thm:AlternatingLocalSearch}}

\proof{Proof of Theorem~\ref{thm:AlternatingLocalSearch}.}
The proof proceeds in two steps. 
We first prove monotonicity.
For any $k$, because $\bm{a}^k(1) \in \argmax_{\bm{a}(1)} \psi(\bm{a}(1), \bm{a}^{k-1}(0))$ whereas $\bm{a}^{k-1}(1)$ is a feasible solution, 
\begin{align*}
\psi(\bm{a}^k(1), \bm{a}^{k-1}(0)) \geq \psi(\bm{a}^{k-1}(1), \bm{a}^{k-1}(0)).
\end{align*}
Next, because $\bm{a}^k(0) \in \argmax_{\bm{a}(0)} \psi(\bm{a}^{k}(1), \bm{a}(0))$ whereas $\bm{a}^{k-1}(0)$ is a feasible solution, 
\begin{align*}
\psi(\bm{a}^k(1), \bm{a}^k(0)) \geq \psi(\bm{a}^k(1), \bm{a}^{k-1}(0)).
\end{align*}
Combining both inequalities, we have
\begin{align*}
\psi^k = \psi(\bm{a}^k(1), \bm{a}^k(0)) \geq \psi(\bm{a}^{k-1}(1), \bm{a}^{k-1}(0)) = \psi^{k-1}.
\end{align*}

Next, we see from Theorem~\ref{thm:OPTInterference} that even the optimal objective value of \eqref{eqn:ConstrainedKyFan} is upper bounded by 
\begin{align*}
\psi(\bm{a}^*(1), \bm{a}^*(0)) \leq \lambda_1(\bm{M}) + \lambda_2(\bm{M}).
\end{align*}
This means that, for any $k$,
\begin{align*}
\psi(\bm{a}^k(1), \bm{a}^k(0)) \leq \lambda_1(\bm{M}) + \lambda_2(\bm{M}).
\end{align*}
A monotone increasing and upper bounded sequence $\{\psi^k\}_{k\geq 0}$ must converge. 
\hfill \halmos
\endproof

\subsection{Derivation of Expression \eqref{eqn:OPTgammas:CI}}

We state the expressions we want to derive again. 

The variance minimizing coefficients $\gamma^*(1)$ and $\gamma^*(0)$ for any generic control variates $\widehat{X}(1)$ and $\widehat{X}(0)$ are given by
\begin{align*}
\begin{bmatrix}
\gamma^*(1) \\
\gamma^*(0) 
\end{bmatrix}
= 
\begin{bmatrix}
\Var\big(\widehat{X}(1)\big) & \Cov\big(\widehat{X}(1), \widehat{X}(0)\big) \\
\Cov\big(\widehat{X}(1), \widehat{X}(0)\big) & \Var\big(\widehat{X}(0)\big) 
\end{bmatrix}^{-1}
\cdot
\begin{bmatrix}
\Cov\big( \widehat{\tau}^\HT, \widehat{X}(1) \big) \\
\Cov\big( \widehat{\tau}^\HT, \widehat{X}(0) \big)
\end{bmatrix}.
\end{align*}
And the smallest variance is given by
\begin{multline}
\Var\big(\widehat{\tau}^\CV(\gamma^*(1), \gamma^*(0))\big) = \Var\big(\widehat{\tau}^\HT\big) - \\
{\footnotesize \begin{bmatrix}
\Cov\big( \widehat{\tau}^\HT, \widehat{X}(1) \big) \\
\Cov\big( \widehat{\tau}^\HT, \widehat{X}(0) \big)
\end{bmatrix}^\top
\cdot
\begin{bmatrix}
\Var\big(\widehat{X}(1)\big) & \Cov\big(\widehat{X}(1), \widehat{X}(0)\big) \\
\Cov\big(\widehat{X}(1), \widehat{X}(0)\big) & \Var\big(\widehat{X}(0)\big) 
\end{bmatrix}^{-1}
\cdot
\begin{bmatrix}
\Cov\big( \widehat{\tau}^\HT, \widehat{X}(1) \big) \\
\Cov\big( \widehat{\tau}^\HT, \widehat{X}(0) \big)
\end{bmatrix}.}
\label{eqn:OPTvariance}
\end{multline}

To derive expression \eqref{eqn:OPTgammas:CI}, note that
\begin{multline}
\Var\big(\widehat{\tau}^\CV(\gamma(1), \gamma(0))\big) = \Var\big(\widehat{\tau}^\HT\big) - 2 \begin{bmatrix}
\gamma(1) \\
\gamma(0)
\end{bmatrix}^\top
\begin{bmatrix}
\Cov\big(\widehat{\tau}^\HT, \widehat{X}(1)\big) \\
\Cov\big(\widehat{\tau}^\HT, \widehat{X}(0)\big)
\end{bmatrix} \\
+ \begin{bmatrix}
\gamma(1) \\
\gamma(0)
\end{bmatrix}^\top
\cdot
\begin{bmatrix}
\Var\big(\widehat{X}(1)\big) & \Cov\big(\widehat{X}(1), \widehat{X}(0)\big) \\
\Cov\big(\widehat{X}(1), \widehat{X}(0)\big) & \Var\big(\widehat{X}(0)\big) 
\end{bmatrix}
\cdot
\begin{bmatrix}
\gamma(1) \\
\gamma(0)
\end{bmatrix}. \label{eqn:VarianceCVCausalInference}
\end{multline}
Because it is a quadratic function with respect to $\gamma(1)$ and $\gamma(0)$, first order condition ensures optimality.
The first order condition is given by
\begin{align*}
- 2 
\begin{bmatrix}
\Cov\big(\widehat{\tau}^\HT, \widehat{X}(1)\big) \\
\Cov\big(\widehat{\tau}^\HT, \widehat{X}(0)\big)
\end{bmatrix} 
+ 2
\begin{bmatrix}
\Var\big(\widehat{X}(1)\big) & \Cov\big(\widehat{X}(1), \widehat{X}(0)\big) \\
\Cov\big(\widehat{X}(1), \widehat{X}(0)\big) & \Var\big(\widehat{X}(0)\big) 
\end{bmatrix}
\cdot
\begin{bmatrix}
\gamma^*(1) \\
\gamma^*(0)
\end{bmatrix} = 0.
\end{align*}
Solving this system of linear equations yields
\begin{align*}
\begin{bmatrix}
\gamma^*(1) \\
\gamma^*(0) 
\end{bmatrix}
= 
\begin{bmatrix}
\Var\big(\widehat{X}(1)\big) & \Cov\big(\widehat{X}(1), \widehat{X}(0)\big) \\
\Cov\big(\widehat{X}(1), \widehat{X}(0)\big) & \Var\big(\widehat{X}(0)\big) 
\end{bmatrix}^{-1}
\cdot
\begin{bmatrix}
\Cov\big( \widehat{\tau}^\HT, \widehat{X}(1) \big) \\
\Cov\big( \widehat{\tau}^\HT, \widehat{X}(0) \big)
\end{bmatrix}.
\end{align*}
Putting the above into \eqref{eqn:VarianceCVCausalInference} we derive expression \eqref{eqn:OPTvariance}, that is,
\begin{multline*}
\Var\big(\widehat{\tau}^\CV(\gamma^*(1), \gamma^*(0))\big) = \Var\big(\widehat{\tau}^\HT\big) - \\
\begin{bmatrix}
\Cov\big( \widehat{\tau}^\HT, \widehat{X}(1) \big) \\
\Cov\big( \widehat{\tau}^\HT, \widehat{X}(0) \big)
\end{bmatrix}^\top
\cdot
\begin{bmatrix}
\Var\big(\widehat{X}(1)\big) & \Cov\big(\widehat{X}(1), \widehat{X}(0)\big) \\
\Cov\big(\widehat{X}(1), \widehat{X}(0)\big) & \Var\big(\widehat{X}(0)\big) 
\end{bmatrix}^{-1}
\cdot
\begin{bmatrix}
\Cov\big( \widehat{\tau}^\HT, \widehat{X}(1) \big) \\
\Cov\big( \widehat{\tau}^\HT, \widehat{X}(0) \big)
\end{bmatrix}.
\end{multline*}

\subsection{Proof of Lemma~\ref{lem:OPTgammas}}
\proof{Proof of Lemma~\ref{lem:OPTgammas}.}
Recall that we have defined $\bm{\Sigma}$ as 
\begin{align*}
\bm{\Sigma} = 
\begin{bmatrix}
\pi_{11} - \pi_1 \pi_1 & \pi_{12} - \pi_1 \pi_2 & \dots & \pi_{1n} - \pi_1 \pi_n \\
\pi_{12} - \pi_1 \pi_2 & \pi_{22} - \pi_2 \pi_2 & \dots & \pi_{2n} - \pi_2 \pi_n \\
\vdots & \vdots & \ddots & \vdots \\
\pi_{1n} - \pi_1 \pi_n & \pi_{2n} - \pi_2 \pi_n & \dots & \pi_{nn} - \pi_n \pi_n
\end{bmatrix}.
\end{align*}
Similarly, we define the following matrices,
\begin{align*}
\bm{\Sigma}(1,1) = {\footnotesize
\begin{bmatrix}
\Var(\bI\{W_1=1\})               & \Cov(\bI\{W_1=1\}, \bI\{W_2=1\}) & \dots  & \Cov(\bI\{W_1=1\}, \bI\{W_n=1\}) \\
\Cov(\bI\{W_1=1\}, \bI\{W_2=1\}) & \Var(\bI\{W_2=1\})               & \dots  & \Cov(\bI\{W_2=1\}, \bI\{W_n=1\}) \\
\vdots                           & \vdots                           & \ddots & \vdots                           \\
\Cov(\bI\{W_1=1\}, \bI\{W_n=1\}) & \Cov(\bI\{W_2=1\}, \bI\{W_n=1\}) & \dots  & \Var(\bI\{W_n=1\})               
\end{bmatrix},
}
\end{align*}
\begin{align*}
\bm{\Sigma}(1,0) = {\footnotesize
\begin{bmatrix}
\Cov(\bI\{W_1=1\}, \bI\{W_1=0\}) & \Cov(\bI\{W_1=1\}, \bI\{W_2=0\}) & \dots  & \Cov(\bI\{W_1=1\}, \bI\{W_n=0\}) \\
\Cov(\bI\{W_1=1\}, \bI\{W_2=0\}) & \Cov(\bI\{W_2=1\}, \bI\{W_2=0\}) & \dots  & \Cov(\bI\{W_2=1\}, \bI\{W_n=0\}) \\
\vdots                           & \vdots                           & \ddots & \vdots                           \\
\Cov(\bI\{W_1=1\}, \bI\{W_n=0\}) & \Cov(\bI\{W_2=1\}, \bI\{W_n=0\}) & \dots  & \Cov(\bI\{W_n=1\}, \bI\{W_n=0\}) 
\end{bmatrix},
}
\end{align*}
and
\begin{align*}
\bm{\Sigma}(0,0) = {\footnotesize
\begin{bmatrix}
\Var(\bI\{W_1=0\})               & \Cov(\bI\{W_1=0\}, \bI\{W_2=0\}) & \dots  & \Cov(\bI\{W_1=0\}, \bI\{W_n=0\}) \\
\Cov(\bI\{W_1=0\}, \bI\{W_2=0\}) & \Var(\bI\{W_2=0\})               & \dots  & \Cov(\bI\{W_2=0\}, \bI\{W_n=0\}) \\
\vdots                           & \vdots                           & \ddots & \vdots                           \\
\Cov(\bI\{W_1=0\}, \bI\{W_n=0\}) & \Cov(\bI\{W_2=0\}, \bI\{W_n=0\}) & \dots  & \Var(\bI\{W_n=0\})               
\end{bmatrix}.
}
\end{align*}
Under SUTVA, we show that 
\begin{align*}
\bm{\Sigma}(1,1) = \bm{\Sigma}(0,0) = - \bm{\Sigma}(1,0) = \bm{\Sigma}.
\end{align*}
To see this, note that because $\bI\{W_i=1\} = 1 - \bI\{W_i=0\}$, we have
\begin{align*}
\Var(\bI\{W_i=0\}) = \Var(1-\bI\{W_i=1\}) = \Var(\bI\{W_i=1\}),
\end{align*}
and we also have
\begin{align*}
\Cov(\bI\{W_i=1\}, \bI\{W_i=0\}) = \Cov(\bI\{W_i=1\}, 1-\bI\{W_i=1\}) = - \Var(\bI\{W_i=1\}).
\end{align*}
Additionally, we have
\begin{align*}
\Cov(\bI\{W_i=0\}, \bI\{W_j=0\}) = \Cov(1-\bI\{W_i=1\}, 1-\bI\{W_j=1\}) = \Cov(\bI\{W_i=1\}, \bI\{W_j=1\}),
\end{align*}
and we also have
\begin{align*}
\Cov(\bI\{W_i=1\}, \bI\{W_j=0\}) = \Cov(\bI\{W_i=1\}, 1-\bI\{W_j=1\}) = - \Cov(\bI\{W_i=1\}, \bI\{W_j=1\}).
\end{align*}
Combining above, we have
\begin{align*}
\bm{\Sigma}(1,1) = \bm{\Sigma}(0,0) = - \bm{\Sigma}(1,0) = \bm{\Sigma}.
\end{align*}

Next, we use the above notations and follow \eqref{eqn:OPTgammas:CI} to find the expressions of $\gamma^*(1)$ and $\gamma^*(0)$.
We start with $\Cov\big(\widehat{\tau}^{\HT}, \widehat{X}(1)\big)$.
\begin{align*}
& \Cov\big(\widehat{\tau}^{\HT}, \widehat{X}(1)\big) \\
= & \bE\big[ \widehat{\tau}^{\HT} \widehat{X}(1) \big] - \tau \bE\big[ \widehat{X}(1) \big] \\
= & \frac{1}{n^2} \sum_{i=1}^n \sum_{j=1}^n \bE\bigg[ \Big( \frac{Y_i(1) \bI\{W_i=1\}}{\pi_i} - \frac{Y_i(0) \bI\{W_i=0\}}{1-\pi_i} \Big) a_j(1) \bI\{W_j=1\} \bigg] - \frac{1}{n^2} \sum_{i=1}^n \sum_{j=1}^n \big(Y_i(1) - Y_i(0)\big) a_j(1) \pi_j \\
= & \frac{1}{n^2} \sum_{i=1}^n \sum_{j=1}^n \frac{Y_i(1)}{\pi_i} \Cov(\bI\{W_i=1\}, \bI\{W_j=1\}) a_j(1) - \frac{1}{n^2} \sum_{i=1}^n \sum_{j=1}^n \frac{Y_i(0)}{1-\pi_i} \Cov(\bI\{W_i=0\}, \bI\{W_j=1\}) a_j(1) \\
= & \frac{1}{n^2} \bm{Y}(1)^\top \bm{\Pi}(1)^{-1} \bm{\Sigma}(1,1) \bm{a}(1) - \frac{1}{n^2} \bm{Y}(0)^\top \bm{\Pi}(0)^{-1} \bm{\Sigma}(1,0) \bm{a}(1) \\
= & \frac{1}{n^2} \bm{Y}(1)^\top \bm{\Pi}(1)^{-1} \bm{\Sigma} \bm{a}(1) + \frac{1}{n^2} \bm{Y}(0)^\top \bm{\Pi}(0)^{-1} \bm{\Sigma} \bm{a}(1).
\end{align*}
Similarly, we derive the expression for $\Cov\big(\widehat{\tau}^{\HT}, \widehat{X}(0)\big)$.
\begin{align*}
& \Cov\big(\widehat{\tau}^{\HT}, \widehat{X}(0)\big) \\ 
= & \bE\big[ \widehat{\tau}^{\HT} \widehat{X}(0) \big] - \tau \bE\big[ \widehat{X}(0) \big] \\ 
= & \frac{1}{n^2} \sum_{i=1}^n \sum_{j=1}^n \bE\bigg[ \Big( \frac{Y_i(1) \bI\{W_i=1\}}{\pi_i} - \frac{Y_i(0) \bI\{W_i=0\}}{1-\pi_i} \Big) a_j(0) \bI\{W_j=0\} \bigg] - \frac{1}{n^2} \sum_{i=1}^n \sum_{j=1}^n \big(Y_i(1) - Y_i(0)\big) a_j(0) (1-\pi_j) \\ 
= & \frac{1}{n^2} \sum_{i=1}^n \sum_{j=1}^n \frac{Y_i(1)}{\pi_i} \Cov(\bI\{W_i=1\}, \bI\{W_j=0\}) a_j(0) - \frac{1}{n^2} \sum_{i=1}^n \sum_{j=1}^n \frac{Y_i(0)}{1-\pi_i} \Cov(\bI\{W_i=0\}, \bI\{W_j=0\}) a_j(0) \\ 
= & \frac{1}{n^2} \bm{Y}(1)^\top \bm{\Pi}(1)^{-1} \bm{\Sigma}(1,0) \bm{a}(0) - \frac{1}{n^2} \bm{Y}(0)^\top \bm{\Pi}(0)^{-1} \bm{\Sigma}(0,0) \bm{a}(0) \\
= & - \frac{1}{n^2} \bm{Y}(1)^\top \bm{\Pi}(1)^{-1} \bm{\Sigma} \bm{a}(0) - \frac{1}{n^2} \bm{Y}(0)^\top \bm{\Pi}(0)^{-1} \bm{\Sigma} \bm{a}(0).
\end{align*}

Next, we derive the expression for $\Cov\big(\widehat{X}(1), \widehat{X}(0)\big)$.
\begin{align*}
\Cov\big(\widehat{X}(1), \widehat{X}(0)\big) = & \bE\big[ \widehat{X}(1) \widehat{X}(0) \big] - \bE\big[ \widehat{X}(1) \big] \bE\big[ \widehat{X}(0) \big] \\ 
= & \frac{1}{n^2} \sum_{i=1}^n \sum_{j=1}^n a_i(1) \bE\big[\bI\{W_i=1, W_j=0\}\big] a_j(0) - \frac{1}{n^2} \sum_{i=1}^n \sum_{j=1}^n a_i(1) \pi_i (1-\pi_j) a_j(0) \\
= & \frac{1}{n^2} \sum_{i=1}^n \sum_{j=1}^n a_i(1) \Cov\big(\bI\{W_i=1\}, \bI\{W_j=0\}\big) a_j(0) \\
= & \frac{1}{n^2} \bm{a}(1)^\top \bm{\Sigma}(1,0) \bm{a}(0) \\
= & - \frac{1}{n^2} \bm{a}(1)^\top \bm{\Sigma} \bm{a}(0).
\end{align*}
Similarly, we derive the expression for $\Var\big(\widehat{X}(1)\big)$.
\begin{align*}
\Var\big(\widehat{X}(1)\big) = & \bE\big[ \widehat{X}(1)^2 \big] - \bE\big[ \widehat{X}(1) \big]^2 \\ 
= & \frac{1}{n^2} \sum_{i=1}^n \sum_{j=1}^n a_i(1) \bE\big[\bI\{W_i=1, W_j=1\}\big] a_j(1) - \frac{1}{n^2} \sum_{i=1}^n \sum_{j=1}^n a_i(1) \pi_i \pi_j a_j(1) \\
= & \frac{1}{n^2} \sum_{i=1}^n \sum_{j=1}^n a_i(1) \Cov\big(\bI\{W_i=1\}, \bI\{W_j=1\}\big) a_j(1) \\
= & \frac{1}{n^2} \bm{a}(1)^\top \bm{\Sigma}(1,1) \bm{a}(1) \\
= & \frac{1}{n^2} \bm{a}(1)^\top \bm{\Sigma} \bm{a}(1).
\end{align*}
Finally, we derive the expression for $\Var\big(\widehat{X}(0)\big)$.
\begin{align*}
\Var\big(\widehat{X}(0)\big) = & \bE\big[ \widehat{X}(0)^2 \big] - \bE\big[ \widehat{X}(0) \big]^2 \\ 
= & \frac{1}{n^2} \sum_{i=1}^n \sum_{j=1}^n a_i(0) \bE\big[\bI\{W_i=0, W_j=0\}\big] a_j(0) - \frac{1}{n^2} \sum_{i=1}^n \sum_{j=1}^n a_i(0) (1-\pi_i) (1-\pi_j) a_j(0) \\
= & \frac{1}{n^2} \sum_{i=1}^n \sum_{j=1}^n a_i(0) \Cov\big(\bI\{W_i=0\}, \bI\{W_j=0\}\big) a_j(0) \\
= & \frac{1}{n^2} \bm{a}(0)^\top \bm{\Sigma}(0,0) \bm{a}(0) \\
= & \frac{1}{n^2} \bm{a}(0)^\top \bm{\Sigma} \bm{a}(0).
\end{align*}

Putting all the above into expression \eqref{eqn:OPTgammas:CI}, we have
\begin{align*}
\begin{bmatrix}
\gamma^*(1) \\
\gamma^*(0) 
\end{bmatrix}
= 
\begin{bmatrix}
 \bm{a}(1)^\top \bm{\Sigma} \bm{a}(1) & -\bm{a}(1)^\top \bm{\Sigma} \bm{a}(0) \\
-\bm{a}(1)^\top \bm{\Sigma} \bm{a}(0) &  \bm{a}(0)^\top \bm{\Sigma} \bm{a}(0) 
\end{bmatrix}^{-1}
\begin{bmatrix}
 \big(\bm{Y}(1)^\top \bm{\Pi}(1)^{-1} + \bm{Y}(0)^\top \bm{\Pi}(0)^{-1}\big) \bm{\Sigma} \bm{a}(1) \\
-\big(\bm{Y}(1)^\top \bm{\Pi}(1)^{-1} + \bm{Y}(0)^\top \bm{\Pi}(0)^{-1}\big) \bm{\Sigma} \bm{a}(0)
\end{bmatrix}.
\end{align*}
Note that
\begin{align*}
\begin{bmatrix}
 \bm{a}(1)^\top \bm{\Sigma} \bm{a}(1) & -\bm{a}(1)^\top \bm{\Sigma} \bm{a}(0) \\
-\bm{a}(1)^\top \bm{\Sigma} \bm{a}(0) &  \bm{a}(0)^\top \bm{\Sigma} \bm{a}(0) 
\end{bmatrix}
= \bm{A}^\top \bm{\Sigma} \bm{A},
\end{align*}
and that
\begin{align*}
\begin{bmatrix}
 \big(\bm{Y}(1)^\top \bm{\Pi}(1)^{-1} + \bm{Y}(0)^\top \bm{\Pi}(0)^{-1}\big) \bm{\Sigma} \bm{a}(1) \\
-\big(\bm{Y}(1)^\top \bm{\Pi}(1)^{-1} + \bm{Y}(0)^\top \bm{\Pi}(0)^{-1}\big) \bm{\Sigma} \bm{a}(0)
\end{bmatrix}
= \bm{A}^\top \bm{\Sigma} \big(\bm{\Pi}(1)^{-1} \bm{Y}(1) + \bm{\Pi}(0)^{-1} \bm{Y}(0)\big).
\end{align*}
These succinct notations leads to 
\begin{align*}
\big(\gamma^*(1), \gamma^*(0) \big)^\top = \big(\bm{A}^\top \bm{\Sigma} \bm{A}\big)^{-1} \bm{A}^\top \bm{\Sigma} \big(\bm{\Pi}(1)^{-1} \bm{Y}(1) + \bm{\Pi}(0)^{-1} \bm{Y}(0)\big).
\end{align*}
Using the definition of baseline outcomes $\bm{G} = \bm{\Pi}(1)^{-1} \bm{Y}(1) + \bm{\Pi}(0)^{-1} \bm{Y}(0)$ we finish the proof for the coefficients $(\gamma^*(1), \gamma^*(0))^\top$. 

Finally, we focus on the expression for $\Var\big(\widehat{\tau}^\HT\big)$.
\begin{align*}
& \Var\big(\widehat{\tau}^\HT\big) \\
= & \bE\big[ \big( \widehat{\tau}^{\HT} \big)^2 \big] - \tau^2 \\
= & \frac{1}{n^2} \sum_{i=1}^n \sum_{j=1}^n \bE\bigg[ \Big( \frac{Y_i(1) \bI\{W_i=1\}}{\pi_i} - \frac{Y_i(0) \bI\{W_i=0\}}{1-\pi_i} \Big) \Big( \frac{Y_j(1) \bI\{W_j=1\}}{\pi_j} - \frac{Y_j(0) \bI\{W_j=0\}}{1-\pi_j} \Big) \bigg] \\
& - \frac{1}{n^2} \sum_{i=1}^n \sum_{j=1}^n \big(Y_i(1) - Y_i(0)\big) \big(Y_j(1) - Y_j(0)\big) \\
= & \frac{1}{n^2} \sum_{i=1}^n \sum_{j=1}^n \frac{Y_i(1)}{\pi_i} \Cov(\bI\{W_i=1\}, \bI\{W_j=1\}) \frac{Y_j(1)}{\pi_j} - \frac{1}{n^2} \sum_{i=1}^n \sum_{j=1}^n \frac{Y_i(1)}{\pi_i} \Cov(\bI\{W_i=1\}, \bI\{W_j=0\}) \frac{Y_j(0)}{1-\pi_j} \\
& - \frac{1}{n^2} \sum_{i=1}^n \sum_{j=1}^n \frac{Y_i(0)}{1-\pi_i} \Cov(\bI\{W_i=0\}, \bI\{W_j=1\}) \frac{Y_j(1)}{\pi_j} + \frac{1}{n^2} \sum_{i=1}^n \sum_{j=1}^n \frac{Y_i(0)}{1-\pi_i} \Cov(\bI\{W_i=0\}, \bI\{W_j=0\}) \frac{Y_j(0)}{1-\pi_j} \\
= & \frac{1}{n^2} \bm{Y}(1)^\top \bm{\Pi}(1)^{-1} \bm{\Sigma}(1,1) \bm{\Pi}(1)^{-1} \bm{Y}(1) - \frac{1}{n^2} \bm{Y}(1)^\top \bm{\Pi}(1)^{-1} \bm{\Sigma}(1,0) \bm{\Pi}(0)^{-1} \bm{Y}(0) \\
& - \frac{1}{n^2} \bm{Y}(0)^\top \bm{\Pi}(0)^{-1} \bm{\Sigma}(1,0) \bm{\Pi}(1)^{-1} \bm{Y}(1) + \frac{1}{n^2} \bm{Y}(0)^\top \bm{\Pi}(0)^{-1} \bm{\Sigma}(0,0) \bm{\Pi}(0)^{-1} \bm{Y}(0) \\
= & \frac{1}{n^2} \big( \bm{Y}(1)^\top \bm{\Pi}(1)^{-1} + \bm{Y}(0)^\top \bm{\Pi}(0)^{-1} \big) \bm{\Sigma} \big( \bm{\Pi}(1)^{-1} \bm{Y}(1) + \bm{\Pi}(0)^{-1} \bm{Y}(0) \big).
\end{align*}

Putting all the above into expression \eqref{eqn:OPTvariance}, we have
\begin{multline*}
\Var\big( \widehat{\tau}^{\CV}(\gamma^*(1), \gamma^*(0)) \big) = \frac{1}{n^2} \bigg( \big(\bm{Y}(1)^\top \bm{\Pi}(1)^{-1} + \bm{Y}(0)^\top \bm{\Pi}(0)^{-1}\big) \bm{\Sigma} \big(\bm{\Pi}(1)^{-1} \bm{Y}(1) + \bm{\Pi}(0)^{-1} \bm{Y}(0)\big) \\
- \big(\bm{Y}(1)^\top \bm{\Pi}(1)^{-1} + \bm{Y}(0)^\top \bm{\Pi}(0)^{-1}\big) \bm{\Sigma} \bm{A} \big(\bm{A}^\top \bm{\Sigma} \bm{A}\big)^{-1} \bm{A}^\top \bm{\Sigma} \big(\bm{\Pi}(1)^{-1} \bm{Y}(1) + \bm{\Pi}(0)^{-1} \bm{Y}(0)\big) \bigg).
\end{multline*}
Using the definition of baseline outcomes $\bm{G} = \bm{\Pi}(1)^{-1} \bm{Y}(1) + \bm{\Pi}(0)^{-1} \bm{Y}(0)$ we finish the proof for the variance $\Var\big( \widehat{\tau}^{\CV}(\gamma^*(1), \gamma^*(0)) \big)$. 
\hfill \halmos
\endproof

\subsection{Proof of Theorem~\ref{thm:AsymptoticCI}}

We present the following Lemma~\ref{lem:UBOptimalgammasCI} which will be useful in the proof of Theorem~\ref{thm:AsymptoticCI}. 
We first introduce a vector notation $\bm{\gamma}^* = (\gamma^*(1), \gamma^*(0))^\top$. 

\begin{lemma}
\label{lem:UBOptimalgammasCI}
Under Assumptions~\ref{asp:ExpDesign} and~\ref{asp:RegularBasesCI}, and assuming the potential outcomes $|Y_i(1)|\leq \overline{y}$ and $|Y_i(0)|\leq \overline{y}$ are all bounded, the optimal coefficients $\gamma^*(1)$ and $\gamma^*(0)$ are bounded, that is,
\begin{align*}
|\gamma^*(1)|\leq \frac{2\overline{y}\overline{d}^{\frac{1}{2}}}{\underline{\pi} \ \underline{\lambda}_\Sigma^{\frac{1}{2}}} n^{\frac{\alpha}{2}+\beta}, \qquad |\gamma^*(0)|\leq \frac{2\overline{y}\overline{d}^{\frac{1}{2}}}{\underline{\pi} \ \underline{\lambda}_\Sigma^{\frac{1}{2}}} n^{\frac{\alpha}{2}+\beta}.
\end{align*}
\end{lemma}

\proof{Proof of Lemma~\ref{lem:UBOptimalgammasCI}.}
Recall that the variance minimizing coefficients are given by
\begin{align*}
\bm{G} = \bm{\Pi}(1)^{-1}\bm{Y}(1) + \bm{\Pi}(0)^{-1}\bm{Y}(0).
\end{align*}
By Assumption~\ref{asp:ExpDesign}-(i) and the boundedness of potential outcomes, for each $i\in[n]$, $|G_i| \leq \frac{2\overline{y}}{\underline{\pi}} n^{\beta}$.
So we have
\begin{align*}
\bm{G}^\top \bm{\Sigma} \bm{G} \leq \frac{4\overline{y}^2}{\underline{\pi}^2} n^{2\beta} \sum_{i=1}^n\sum_{j=1}^n |\pi_{ij} - \pi_i \pi_j| \leq \frac{4\overline{y}^2\overline{d}}{\underline{\pi}^2} n^{1+\alpha+2\beta},
\end{align*}
where the last inequality is due to Assumption~\ref{asp:ExpDesign}-(ii).

Next, note that
\begin{align*}
\big(\bm{\gamma}^*\big)^\top \bm{A}^\top\bm{\Sigma}\bm{A} \bm{\gamma}^* 
= \bm{G}^\top\bm{\Sigma}\bm{A} (\bm{A}^\top\bm{\Sigma}\bm{A})^{-1} \bm{A}^\top\bm{\Sigma}\bm{G} 
\leq \bm{G}^\top\bm{\Sigma}\bm{G},
\end{align*}
where the inequality follows because the middle term is the variance reduction from projecting \(\bm{G}\) onto the column space of \(\bm{A}\), and hence cannot exceed the total variance \(\bm{G}^\top\bm{\Sigma}\bm{G}\).

By Assumption~\ref{asp:RegularBasesCI}-(ii),
\begin{align*}
\big(\bm{\gamma}^*\big)^\top \bm{A}^\top\bm{\Sigma}\bm{A} \bm{\gamma}^* \geq \underline{\lambda}_\Sigma n \Big( (\gamma^*(1))^2+(\gamma^*(0))^2 \Big).
\end{align*}
Putting the above inequalities together, we have
\begin{align*}
(\gamma^*(1))^2+(\gamma^*(0))^2 \leq \frac{4\overline{y}^2\overline{d}}{\underline{\pi}^2\underline{\lambda}_\Sigma} n^{\alpha+2\beta}.
\end{align*}
Taking square root we finish the proof.
\hfill \halmos
\endproof

\

\noindent Now we prove Theorem~\ref{thm:AsymptoticCI} as follows.

\proof{Proof of Theorem~\ref{thm:AsymptoticCI}.}
We prove the three claims (consistency, asymptotic variance, and asymptotic normality) in Theorem~\ref{thm:AsymptoticCI} one by one.

\noindent \textbf{Claim 1: Consistency}. 
In the proof of this claim, we assume $3 \alpha + 4 \beta < 1$.
We prove consistency through two steps.
First, we show that
\begin{align*}
\widehat{\tau}^\CV(\widehat{\gamma}^\HT(1),\widehat{\gamma}^\HT(0)) - \widehat{\tau}^\CV(\gamma^*(1),\gamma^*(0)) \xrightarrow{p} 0.
\end{align*}

For each $i\in[n]$, we have
\begin{align*}
\widehat{G}_i-G_i = \big(\bI\{W_i=1\}-\pi_i\big) \bigg( \frac{Y_i(1)}{\pi_i^2} + \frac{Y_i(0)}{(1-\pi_i)^2} \bigg).
\end{align*}
For each $i\in[n]$, define
\begin{align*}
Z_i(1) = \bigg( \frac{Y_i(1)}{\pi_i^2} + \frac{Y_i(0)}{(1-\pi_i)^2} \bigg) \big( \bm{\Sigma}\bm{A}(\bm{A}^\top\bm{\Sigma}\bm{A})^{-1} \big)_{i1},
\end{align*}
and
\begin{align*}
Z_i(0) = \bigg( \frac{Y_i(1)}{\pi_i^2} + \frac{Y_i(0)}{(1-\pi_i)^2} \bigg) \big( \bm{\Sigma}\bm{A}(\bm{A}^\top\bm{\Sigma}\bm{A})^{-1} \big)_{i2},
\end{align*}
where we use $\big( \bm{\Sigma}\bm{A}(\bm{A}^\top\bm{\Sigma}\bm{A})^{-1} \big)_{i1}$ to stand for the element in the $i$-th row and the $1$-st column, and $\big( \bm{\Sigma}\bm{A}(\bm{A}^\top\bm{\Sigma}\bm{A})^{-1} \big)_{i2}$ to stand for the element in the $i$-th row and the $2$-nd column.
Using the above definitions, we have
\begin{align*}
\widehat{\gamma}^\HT(1)-\gamma^*(1) = \sum_{i=1}^n Z_i(1)\big(\bI\{W_i=1\}-\pi_i\big),
\end{align*}
and
\begin{align*}
\widehat{\gamma}^\HT(0)-\gamma^*(0) = \sum_{i=1}^n Z_i(0)\big(\bI\{W_i=1\}-\pi_i\big),
\end{align*}

We next upper bound $|Z_i(1)|$ and $|Z_i(0)|$.
By Assumption~\ref{asp:ExpDesign}-(i) and boundedness of the potential outcomes,
\begin{align*}
\bigg\vert \frac{Y_i(1)}{\pi_i^2} + \frac{Y_i(0)}{(1-\pi_i)^2} \bigg\vert \leq \frac{2\overline{y}}{\underline{\pi}^2} n^{2\beta}.
\end{align*}
By Assumptions~\ref{asp:ExpDesign}-(ii) and~\ref{asp:RegularBasesCI}-(i), each element $\big(\bm{\Sigma}\bm{A}\big)_{il}$ in the $(n \times 2)$ matrix $\bm{\Sigma}\bm{A}$ is bounded by
\begin{align*}
\big\vert \big(\bm{\Sigma}\bm{A}\big)_{il} \big\vert = \bigg\vert \sum_{j=1}^n (\pi_{ij} - \pi_i \pi_j) a_j(l) \bigg\vert \leq \overline{a}\overline{d} n^\alpha.
\end{align*}
By Assumption~\ref{asp:RegularBasesCI}-(ii),
\begin{align*}
\big\| (\bm{A}^\top\bm{\Sigma}\bm{A})^{-1} \big\|_2 \leq \frac{1}{\underline{\lambda}_\Sigma n}.
\end{align*}
and so every element of the inverse matrix $(\bm{A}^\top \bm{\Sigma} \bm{A})^{-1}$ is upper bounded by $\frac{1}{\underline{\lambda}_{\Sigma}n}$.
Now define $\overline{z} = \dfrac{4\overline{y}\overline{a}\overline{d}}{\underline{\pi}^2 \underline{\lambda}_{\Sigma}} > 0$ such that
\begin{align*}
|Z_i(1)|\leq \overline{z} n^{\alpha+2\beta-1},
\qquad
|Z_i(0)|\leq \overline{z} n^{\alpha+2\beta-1}.
\end{align*}

Next, it is easy to see that
\begin{align*}
\bE_{\cW}[\widehat{\gamma}^\HT(1)]=\gamma^*(1),
\qquad
\bE_{\cW}[\widehat{\gamma}^\HT(0)]=\gamma^*(0).
\end{align*}
Next, we calculate $\Var(\widehat{\gamma}^\HT(1))$ as
\begin{align*}
\Var(\widehat{\gamma}^\HT(1)) = \sum_{i=1}^n\sum_{j=1}^n Z_i(1)Z_j(1)(\pi_{ij}-\pi_i\pi_j) \leq \overline{z}^2 n^{2\alpha+4\beta-2} \sum_{i=1}^n \sum_{j=1}^n |\pi_{ij}-\pi_i\pi_j| \leq \overline{z}^2 \overline{d} n^{3\alpha+4\beta-1},
\end{align*}
where the last inequality is due to Assumption~\ref{asp:ExpDesign}-(ii).
Using Chebyshev inequality, for any $\epsilon > 0$,
\begin{align*}
\Pr\Big( \big\vert \widehat{\gamma}^\HT(1) - \gamma^*(1) \big\vert \geq \epsilon \Big) \leq \frac{\Var(\widehat{\gamma}^\HT(1))}{\epsilon^2} \leq \overline{z}^2\overline{d} n^{3\alpha+4\beta-1} \frac{1}{\epsilon^2}.
\end{align*}
Similarly, we have
\begin{align*}
\Pr\Big( \big\vert \widehat{\gamma}^\HT(0) - \gamma^*(0) \big\vert \geq \epsilon \Big) \leq \frac{\Var(\widehat{\gamma}^\HT(0))}{\epsilon^2} \leq \overline{z}^2\overline{d} n^{3\alpha+4\beta-1} \frac{1}{\epsilon^2},
\end{align*}

Next, we calculate the variance of the two control variates. 
Note that
\begin{align*}
\Var(\widehat{X}(1)) = \frac{1}{n^2} \Var\bigg( \sum_{i=1}^n a_i(1)\bI\{W_i=1\} \bigg) = \frac{1}{n^2} \sum_{i=1}^n\sum_{j=1}^n a_i(1)a_j(1)(\pi_{ij}-\pi_i\pi_j) \leq \overline{a}^2\overline{d} n^{\alpha-1}.
\end{align*}
Using Chebyshev inequality, for any $\epsilon > 0$,
\begin{align*}
\Pr\Big( \big\vert \widehat{X}(1) - \bE[\widehat{X}(1)] \big\vert \geq \epsilon \Big) \leq \overline{a}^2 \overline{d} n^{\alpha-1} \frac{1}{\epsilon^2},
\end{align*}
Similarly, because $\bI\{W_i=0\}-(1-\pi_i) = - \big(\bI\{W_i=1\}-\pi_i\big)$, we have
\begin{align*}
\Pr\Big( \big\vert \widehat{X}(0) - \bE[\widehat{X}(0)] \big\vert \geq \epsilon \Big) \leq \overline{a}^2 \overline{d} n^{\alpha-1} \frac{1}{\epsilon^2},
\end{align*}

Finally, define
\begin{multline*}
\Delta_n = \widehat{\tau}^\CV(\widehat{\gamma}^\HT(1),\widehat{\gamma}^\HT(0)) - \widehat{\tau}^\CV(\gamma^*(1),\gamma^*(0)) \\
= - \big(\widehat{\gamma}^\HT(1)-\gamma^*(1)\big) \big(\widehat{X}(1)-\bE[\widehat{X}(1)]\big) - \big(\widehat{\gamma}^\HT(0)-\gamma^*(0)\big) \big(\widehat{X}(0)-\bE[\widehat{X}(0)]\big).
\end{multline*}
For any $\epsilon > 0$, we have
\begin{align*}
\lim_{n \to +\infty} \Pr\big( |\Delta_n| \geq \epsilon \big) \leq & \lim_{n \to +\infty}
\bigg( \Pr\Big( \big\vert \widehat{\gamma}^\HT(1) - \gamma^*(1) \big\vert \geq \Big(\frac{\epsilon}{2}\Big)^{\frac{1}{2}} \Big) + \Pr\Big( \big\vert \widehat{X}(1) - \bE[\widehat{X}(1)] \big\vert \geq \Big(\frac{\epsilon}{2}\Big)^{\frac{1}{2}} \Big) \\
& \qquad + \Pr\Big( \big\vert \widehat{\gamma}^\HT(0) - \gamma^*(0) \big\vert \geq \Big(\frac{\epsilon}{2}\Big)^{\frac{1}{2}} \Big) + \Pr\Big( \big\vert \widehat{X}(0) - \bE[\widehat{X}(0)] \big\vert \geq \Big(\frac{\epsilon}{2}\Big)^{\frac{1}{2}} \Big) \bigg) \\
\leq & \lim_{n \to +\infty} \bigg( 2 \overline{z}^2 \overline{d} n^{3\alpha+4\beta-1} \frac{2}{\epsilon} + 2 \overline{a}^2 \overline{d} n^{\alpha-1} \frac{2}{\epsilon} \bigg) \\
= & 0,
\end{align*}
where the first inequality holds because, if $|\Delta_n| \geq \epsilon$ holds, then at least one of the four conditions must hold;
the last and only equality holds because $3\alpha + 4\beta < 1$.
So we have shown that
\begin{align*}
\Delta_n \xrightarrow{p} 0.
\end{align*}

Second, we show that $\widehat{\tau}^\CV(\gamma^*(1),\gamma^*(0))$ is a consistent estimator of $\tau$.
Note that
\begin{align*}
\widehat{\tau}^\HT-\tau = \frac{1}{n} \sum_{i=1}^n \bigg( \frac{Y_i(1)}{\pi_i} + \frac{Y_i(0)}{1-\pi_i} \bigg) \big(\bI\{W_i=1\} - \pi_i\big).
\end{align*}
So we have,
\begin{align*}
\Var(\widehat{\tau}^\HT) \leq \frac{1}{n^2} \frac{4 \overline{y}^2}{\underline{\pi}^2} n^{2\beta} \Var\Big( \sum_{i=1}^n \big(\bI\{W_i=1\} - \pi_i\big) \Big) \leq \frac{4\overline{y}^2\overline{d}}{\underline{\pi}^2} n^{\alpha+2\beta-1}.
\end{align*}
Using Chebyshev inequality, for any $\epsilon > 0$,
\begin{align*}
\Pr\Big( \big\vert \widehat{\tau}^\HT - \tau \big\vert \geq \epsilon \Big) \leq \frac{\Var(\widehat{\tau}^\HT)}{\epsilon^2} \leq \frac{4\overline{y}^2\overline{d}}{\underline{\pi}^2} n^{\alpha+2\beta-1}\frac{1}{\epsilon^2}.
\end{align*}

Next, note that
\begin{align*}
\widehat{\tau}^\CV(\gamma^*(1),\gamma^*(0)) - \tau = (\widehat{\tau}^\HT - \tau) - \gamma^*(1)\big(\widehat{X}(1)-\bE[\widehat{X}(1)]\big) - \gamma^*(0)\big(\widehat{X}(0)-\bE[\widehat{X}(0)]\big).
\end{align*}
So we have, for any $\epsilon > 0$,
\begin{align*}
& \lim_{n \to +\infty} \Pr\Big( \big\vert \widehat{\tau}^\CV(\gamma^*(1),\gamma^*(0)) - \tau \big\vert \geq \epsilon \Big)\\
\leq & \lim_{n \to +\infty} \bigg( \Pr\Big( \big\vert \widehat{\tau}^\HT - \tau \big\vert \geq \frac{\epsilon}{3} \Big) + \Pr\Big( \big\vert\gamma^*(1) (\widehat{X}(1)-\bE[\widehat{X}(1)]) \big\vert \geq \frac{\epsilon}{3} \Big) + \Pr\Big( \big\vert\gamma^*(0) (\widehat{X}(0)-\bE[\widehat{X}(0)])\big\vert \geq \frac{\epsilon}{3} \Big) \bigg) \\
\leq & \lim_{n \to +\infty} \bigg( \frac{4\overline{y}^2\overline{d}}{\underline{\pi}^2} n^{\alpha+2\beta-1}\frac{9}{\epsilon^2} + 2\overline{a}^2\overline{d} n^{\alpha-1} \frac{9}{\epsilon^2} \Big( \frac{2\overline{y}\overline{d}^{\frac{1}{2}}}{\underline{\pi}\underline{\lambda}_\Sigma^{\frac{1}{2}}} n^{\frac{\alpha}{2}+\beta} \Big)^2 \bigg) \\
= & \ 0,
\end{align*}
where the first inequality holds because, if $\big\vert \widehat{\tau}^\CV(\gamma^*(1),\gamma^*(0)) - \tau \big\vert \geq \epsilon$ holds, then at least one of the three conditions must hold;
the second inequality holds due to Lemma~\ref{lem:UBOptimalgammasCI};
the last equality holds because $\alpha + 2\beta < 2\alpha + 2\beta < 3\alpha + 4\beta < 1$.

Therefore,
\begin{align*}
\widehat{\tau}^\CV(\gamma^*(1),\gamma^*(0)) \xrightarrow{p} \tau.
\end{align*}
Combining this with $\Delta_n \xrightarrow{p} 0$ we have
\begin{align*}
\widehat{\tau}^\CV(\widehat{\gamma}^\HT(1),\widehat{\gamma}^\HT(0)) \xrightarrow{p} \tau.
\end{align*}

\noindent \textbf{Claim 2: Asymptotic Variance}. 
In the proof of this claim, we assume $4\alpha + 4\beta < 1$.

Recall that we have defined
\begin{align*}
\Delta_n = \widehat{\tau}^\CV(\widehat{\gamma}^\HT(1),\widehat{\gamma}^\HT(0)) - \widehat{\tau}^\CV(\gamma^*(1),\gamma^*(0)).
\end{align*}
Then
\begin{multline}
\Var_{\cW}\Big( \widehat{\tau}^\CV(\widehat{\gamma}^\HT(1),\widehat{\gamma}^\HT(0))
\Big) \\
= \Var_{\cW}\Big( \widehat{\tau}^\CV(\gamma^*(1),\gamma^*(0)) \Big) + \Var_{\cW}\Big(\Delta_n\Big) + 2\Cov_{\cW}\Big( \widehat{\tau}^\CV(\gamma^*(1),\gamma^*(0)), \Delta_n \Big). \label{eqn:VarianceDecompositionCI}
\end{multline}
It suffices to prove that both
\begin{align*}
n \Var_{\cW}\big(\Delta_n\big) \to 0
\end{align*}
and
\begin{align*}
n\Cov_{\cW}\Big( \widehat{\tau}^\CV(\gamma^*(1),\gamma^*(0)), \Delta_n \Big) \to 0.
\end{align*}

We first show that $n \Var_{\cW}(\Delta_n) \to 0$.
Because $0 \leq \Var_{\cW}(\Delta_n) \leq \bE_{\cW}[\Delta_n^2]$, it suffices to show that $n \bE_{\cW}[\Delta_n^2] \to 0$.
Recall that
\begin{align*}
\Delta_n = - \big(\widehat{\gamma}^\HT(1)-\gamma^*(1)\big) \big(\widehat{X}(1)-\bE[\widehat{X}(1)]\big) - \big(\widehat{\gamma}^\HT(0)-\gamma^*(0)\big) \big(\widehat{X}(0)-\bE[\widehat{X}(0)]\big).
\end{align*}
So we have,
\begin{align*}
\bE_{\cW}[\Delta_n^2] \leq & 2\bE_{\cW}\Big[ (\widehat{\gamma}^\HT(1)-\gamma^*(1))^2 (\widehat{X}(1)-\bE[\widehat{X}(1)])^2 \Big] + 2\bE_{\cW}\Big[ (\widehat{\gamma}^\HT(0)-\gamma^*(0))^2 (\widehat{X}(0)-\bE[\widehat{X}(0)])^2 \Big] \\
\leq & 2\bE_{\cW}\Big[ (\widehat{\gamma}^\HT(1)-\gamma^*(1))^4 \Big]^{\frac{1}{2}} \bE_{\cW}\left[ (\widehat{X}(1)-\bE[\widehat{X}(1)])^4 \right]^{\frac{1}{2}} \\
& \qquad \qquad + 2\bE_{\cW}\Big[ (\widehat{\gamma}^\HT(0)-\gamma^*(0))^4 \Big]^{\frac{1}{2}} \bE_{\cW}\left[ (\widehat{X}(0)-\bE[\widehat{X}(0)])^4 \right]^{\frac{1}{2}},
\end{align*}
where the second inequality is due to Cauchy-Schwarz inequality.

We next upper bound the four terms separately.
Note that
\begin{align*}
\bE_{\cW}\big[(\widehat{\gamma}^\HT(1) - \gamma^*(1))^4\big] = & \bE_{\cW}\bigg[\Big( \sum_{i=1}^n Z_i(1)\big(\bI\{W_i=1\}-\pi_i\big) \Big)^4\bigg] \\
\leq & \overline{z}^4 n^{4\alpha + 8\beta - 4} \ \bE_{\cW}\bigg[\Big( \sum_{i=1}^n \big( \bI\{W_i=1\} - \pi_i \big) \Big)^4\bigg] \\
\leq & \overline{z}^4 n^{4\alpha + 8\beta - 4} \ 125 \ \overline{d}^2 n^{2+2\alpha} \\
= & 125 \ \overline{z}^4 \overline{d}^2 \ n^{6\alpha + 8\beta - 2},
\end{align*}
where the second inequality is doing the combinatorial counting, upper bounding $|\bI\{W_i=1\} - \pi_i| \leq 1$, and using Assumption~\ref{asp:ExpDesign}-(ii).
Similarly, we have
\begin{align*}
\bE_{\cW}\big[(\widehat{\gamma}^\HT(0) - \gamma^*(0))^4\big] = & \bE_{\cW}\bigg[\Big( \sum_{i=1}^n Z_i(0)\big(\bI\{W_i=1\}-\pi_i\big) \Big)^4\bigg] \\
\leq & \overline{z}^4 n^{4 \alpha - 4} \ \bE_{\cW}\bigg[\Big( \sum_{i=1}^n \big( \bI\{W_i=1\} - \pi_i \big) \Big)^4\bigg] \\
\leq & \overline{z}^4 n^{4 \alpha - 4} \ 125 \ \overline{d}^2 n^{2+2 \alpha} \\
= & 125 \ \overline{z}^4 \overline{d}^2 \ n^{6\alpha + 8\beta - 2}, 
\end{align*}
where the first equality is because $\bI\{W_i=0\}-(1-\pi_i) = - \big(\bI\{W_i=1\}-\pi_i\big)$;
the second inequality is doing the combinatorial counting, upper bounding $|\bI\{W_i=1\} - \pi_i| \leq 1$, and using Assumption~\ref{asp:ExpDesign}-(ii).
Additionally, we have
\begin{align*}
\bE_{\cW}\big[(\widehat{X}(1) - \bE_{\cW}[\widehat{X}(1)])^4\big] = & \bE_{\cW}\bigg[\Big( \frac{1}{n} \sum_{i=1}^n a_i(1) (\bI\{W_i=1\} - \pi_i) \Big)^4\bigg] \nonumber \\
\leq & \overline{a}^4n^{-4} \bE_{\cW}\bigg[\Big( \sum_{i=1}^n \big( \bI\{W_i=1\} - \pi_i \big) \Big)^4\bigg] \nonumber \\
\leq & 125 \overline{a}^4 \overline{d}^2 \ n^{2\alpha - 2}, 
\end{align*}
where the second inequality is doing the combinatorial counting, upper bounding $|\bI\{W_i=1\} - \pi_i| \leq 1$, and using Assumption~\ref{asp:ExpDesign}-(ii).
Similarly, we have
\begin{align*}
\bE_{\cW}\big[(\widehat{X}(0) - \bE_{\cW}[\widehat{X}(0)])^4\big] \leq 125 \overline{a}^4 \overline{d}^2 \ n^{2\alpha - 2}.
\end{align*}
Combining these four bounds, we have
\begin{align}
\lim_{n \to +\infty} n \bE_{\cW}[\Delta_n^2] \leq \lim_{n \to +\infty} 500 \overline{z}^2 \overline{a}^2 \overline{d}^2 \ n^{4\alpha + 4\beta - 1} = 0, \label{eqn:VarConvergeTo0}
\end{align}
where the equality holds because $4\alpha + 4\beta < 1$.

Next, we show that $n\Cov_{\cW}\big( \widehat{\tau}^\CV(\gamma^*(1),\gamma^*(0)), \Delta_n \big) \to 0$.
Note that
\begin{align*}
& \Big\vert \Cov_{\cW}\Big( \widehat{\tau}^\CV(\gamma^*(1),\gamma^*(0)), \Delta_n \Big) \Big\vert \\
= & \Big\vert \bE_{\cW}\Big[ \big( \widehat{\tau}^\CV(\gamma^*(1),\gamma^*(0))-\tau \big) \Delta_n \Big] \Big\vert \\
= & \Big\vert \bE_{\cW}\Big[ \big( \widehat{\tau}^\CV(\gamma^*(1),\gamma^*(0)) - \tau \big) \big( \widehat{\gamma}^\HT(1) - \gamma^*(1) \big) \big( \widehat{X}(1)-\bE_{\cW}[\widehat{X}(1)] \big) \Big] \Big\vert \\
& \qquad + \Big\vert \bE_{\cW}\Big[ \big( \widehat{\tau}^\CV(\gamma^*(1),\gamma^*(0)) - \tau \big) \big( \widehat{\gamma}^\HT(0) - \gamma^*(0) \big) \big( \widehat{X}(0)-\bE_{\cW}[\widehat{X}(0)] \big) \Big] \Big\vert \\
\leq & \bE_{\cW}\Big[ \big\vert \widehat{\tau}^\CV(\gamma^*(1),\gamma^*(0)) - \tau \big\vert^3 \Big]^{\frac{1}{3}} \bE_{\cW}\Big[ \big\vert \widehat{\gamma}^\HT(1) - \gamma^*(1) \big\vert^3 \Big]^{\frac{1}{3}} \bE_{\cW}\Big[ \big\vert \widehat{X}(1)-\bE_{\cW}[\widehat{X}(1)] \big\vert^3 \Big]^{\frac{1}{3}} \\
& \qquad + \bE_{\cW}\Big[ \big\vert \widehat{\tau}^\CV(\gamma^*(1),\gamma^*(0)) - \tau \big\vert^3 \Big]^{\frac{1}{3}} \bE_{\cW}\Big[ \big\vert \widehat{\gamma}^\HT(0) - \gamma^*(0) \big\vert^3 \Big]^{\frac{1}{3}} \bE_{\cW}\Big[ \big\vert \widehat{X}(0)-\bE_{\cW}[\widehat{X}(0)] \big\vert^3 \Big]^{\frac{1}{3}}
\end{align*}
where the first equality is because $\widehat{\tau}^\CV(\gamma^*(1),\gamma^*(0))$ is unbiased;
the second equality is using the definition of $\Delta_n$;
the last inequality is using Holder's inequality.

We next upper bound the three moments above.
First,
\begin{align*}
\Big\vert \widehat{\tau}^\CV(\gamma^*(1),\gamma^*(0)) - \tau \Big\vert \leq & \frac{1}{n} \sum_{i=1}^n \Big\vert \frac{Y_i(1)}{\pi_i} - \frac{Y_i(0)}{1-\pi_i} - \gamma^*(1) a_i(1) + \gamma^*(0) a_i(0) \Big\vert \cdot \big\vert \bI\{W_i=1\} - \pi_i \big\vert \\
\leq & \frac{1}{n} \sum_{i=1}^n \Big( \frac{2\overline{y}}{\underline{\pi}} n^{\beta}+ \frac{4 \overline{a} \overline{y} \overline{d}^{\frac{1}{2}}}{\underline{\pi} \ \underline{\lambda}_{\Sigma}^{\frac{1}{2}}} n^{\frac{\alpha}{2}+\beta} \Big) \cdot \big\vert \bI\{W_i=1\} - \pi_i \big\vert \\
\leq & \frac{1}{n} \sum_{i=1}^n \frac{6 \overline{a} \overline{y} \overline{d}^{\frac{1}{2}}}{\underline{\pi} \ \underline{\lambda}_{\Sigma}^{\frac{1}{2}}} n^{\frac{\alpha}{2}+\beta} \cdot \big\vert \bI\{W_i=1\} - \pi_i \big\vert
\end{align*}
where the first inequality is by definition and upper bounding each component by its absolute value;
the second inequality is due to Lemma~\ref{lem:UBOptimalgammasCI};
the third inequality holds when $n$ is sufficiently large. 
So we have
\begin{align*}
\bE_{\cW}\Big[ \big\vert \widehat{\tau}^\CV(\gamma^*(1),\gamma^*(0)) - \tau \big\vert^3 \Big] \leq & \frac{216 \overline{a}^3 \overline{y}^3 \overline{d}^{\frac{3}{2}}}{\underline{\pi}^3 \underline{\lambda}_{\Sigma}^{\frac{3}{2}}} n^{\frac{3\alpha}{2} + 3\beta - 3} \bE_{\cW}\Big[ \Big\vert \sum_{i=1}^n \big\vert \bI\{W_i=1\} - \pi_i \big\vert \Big\vert^3 \Big] \\
\leq & \frac{216 \overline{a}^3 \overline{y}^3 \overline{d}^{\frac{3}{2}}}{\underline{\pi}^3 \underline{\lambda}_{\Sigma}^{\frac{3}{2}}} n^{\frac{3\alpha}{2} + 3\beta - 3} \ 5 \overline{d}^2 n^{1+2\alpha} \\
= & \frac{1080 \overline{a}^3 \overline{y}^3 \overline{d}^{\frac{7}{2}}}{\underline{\pi}^3 \underline{\lambda}_{\Sigma}^{\frac{3}{2}}} n^{\frac{7\alpha}{2} + 3\beta - 2},
\end{align*}
where the second inequality is doing the combinatorial counting, upper bounding $|\bI\{W_i=1\} - \pi_i| \leq 1$, and using Assumption~\ref{asp:ExpDesign}-(ii).

Second, we have
\begin{align*}
\bE_{\cW}\Big[ \big\vert \widehat{\gamma}^\HT(1) - \gamma^*(1) \big\vert^3 \Big] = & \bE_{\cW}\bigg[ \Big\vert \sum_{i=1}^n Z_i(1) \big(\bI\{W_i=1\} - \pi_i\big) \Big\vert^3 \bigg] \\
\leq & \overline{z}^3 n^{3\alpha+6\beta-3} \ 5 \overline{d}^2 n^{1+2\alpha} \\
= & 5 \overline{z}^3 \overline{d}^2 n^{5\alpha+6\beta-2},
\end{align*}
where the second inequality is using the upper bound $\vert Z_i(1) \vert \leq \overline{z} n^{\alpha+2\beta-1}$ and doing the combinatorial counting, upper bounding $|\bI\{W_i=1\} - \pi_i| \leq 1$, and using Assumption~\ref{asp:ExpDesign}-(ii).
Similarly,
\begin{align*}
\bE_{\cW}\Big[ \big\vert \widehat{\gamma}^\HT(0) - \gamma^*(0) \big\vert^3 \Big] \leq 5 \overline{z}^3 \overline{d}^2 n^{5\alpha+6\beta-2}.
\end{align*}

Third, we have
\begin{align*}
\bE_{\cW}\Big[ \big\vert \widehat{X}(1)-\bE_{\cW}[\widehat{X}(1)] \big\vert^3 \Big] = & \bE_{\cW}\bigg[ \Big\vert \frac{1}{n} \sum_{i=1}^n a_i(1) \big(\bI\{W_i=1\} - \pi_i\big) \Big\vert^3 \bigg] \\
\leq & \overline{a}^3 n^{-3} \ 5 \overline{d}^2 n^{1+2\alpha} \\
= & 5 \overline{a}^3 \overline{d}^2 n^{2\alpha-2},
\end{align*}
where the second inequality is upper bounding $\vert a_i(1) \vert \leq \overline{a}$ and doing the combinatorial counting, upper bounding $|\bI\{W_i=1\} - \pi_i| \leq 1$, and using Assumption~\ref{asp:ExpDesign}-(ii).
Similarly,
\begin{align*}
\bE_{\cW}\Big[ \big\vert \widehat{X}(0)-\bE_{\cW}[\widehat{X}(0)] \big\vert^3 \Big] \leq 5 \overline{a}^3 \overline{d}^2 n^{2\alpha-2}.
\end{align*}

Combining these bounds, we have
\begin{align}
\lim_{n \to +\infty} n \Big\vert \Cov_{\cW}\Big( \widehat{\tau}^\CV(\gamma^*(1),\gamma^*(0)), \Delta_n \Big) \Big\vert \leq \lim_{n \to +\infty} \frac{30 \overline{a}^2 \overline{y} \overline{z} \overline{d}^{\frac{5}{2}}}{\underline{\pi} \ \underline{\lambda}_{\Sigma}^{\frac{1}{2}}} n^{\frac{7\alpha}{2} + 3\beta - 1} = 0, \label{eqn:CovConvergeTo0}
\end{align}
where the equality holds because $\frac{7\alpha}{2} + 3\beta < 4\alpha + 4\beta < 1$.

Putting \eqref{eqn:VarConvergeTo0} and \eqref{eqn:CovConvergeTo0} in \eqref{eqn:VarianceDecompositionCI}, we have
\begin{align*}
\lim_{n \to +\infty} n \Big( \Var_{\cW}\big( \widehat{\tau}^\CV(\widehat{\gamma}^\HT(1), \widehat{\gamma}^\HT(0)) \big) - \Var_{\cW}\big( \widehat{\tau}^\CV(\gamma^*(1), \gamma^*(0)) \big) \Big) \to 0.
\end{align*}

\noindent \textbf{Claim 3: Asymptotic Normality}. 
In the proof of this claim, we assume $7\alpha + 6\beta < 1$.

First, we show that
\begin{align*}
\frac{\widehat{\tau}^\CV(\gamma^*(1),\gamma^*(0)) - \tau}{\sqrt{ \Var_{\cW}\big(\widehat{\tau}^\CV(\gamma^*(1),\gamma^*(0))\big)}} 
\xrightarrow{d} \cN(0,1).
\end{align*}

To show this, we introduce some notations. 
Define
\begin{align*}
\xi_i = \Big( \frac{Y_i(1)}{\pi_i} - \frac{Y_i(0)}{1-\pi_i} - \gamma^*(1)a_i(1) + \gamma^*(0)a_i(0) \Big) \big(\bI\{W_i=1\} - \pi_i\big).
\end{align*}
Using this notation, we have
\begin{align*}
\widehat{\tau}^\CV(\gamma^*(1),\gamma^*(0)) - \tau = \frac{1}{n}\sum_{i=1}^n \xi_i.
\end{align*}
We can upper bound each $\xi_i$ for any $i \in [n]$ by
\begin{align*}
\vert \xi_i \vert \leq \Big( \frac{2\overline{y}}{\underline{\pi}} n^{\beta} + \frac{4 \overline{a} \overline{y} \overline{d}^{\frac{1}{2}}}{\underline{\pi} \ \underline{\lambda}_{\Sigma}^{\frac{1}{2}}} n^{\frac{\alpha}{2}+\beta} \Big) \cdot \big\vert \bI\{W_i=1\} - \pi_i \big\vert 
\leq \frac{4 \overline{a} \overline{y} \overline{d}^{\frac{1}{2}}}{\underline{\pi} \ \underline{\lambda}_{\Sigma}^{\frac{1}{2}}} n^{\frac{\alpha}{2}+\beta} \cdot \big\vert \bI\{W_i=1\} - \pi_i \big\vert
\leq \frac{4 \overline{a} \overline{y} \overline{d}^{\frac{1}{2}}}{\underline{\pi} \ \underline{\lambda}_{\Sigma}^{\frac{1}{2}}} n^{\frac{\alpha}{2}+\beta}
\end{align*}
the first inequality is due to Lemma~\ref{lem:UBOptimalgammasCI};
the second inequality holds when $n$ is sufficiently large;
the third inequality is because $\big\vert \bI\{W_i=1\} - \pi_i \big\vert \leq 1$.
From here, denote the constant $\overline{\xi} = \dfrac{4 \overline{a} \overline{y} \overline{d}^{\frac{1}{2}}}{\underline{\pi} \ \underline{\lambda}_{\Sigma}^{\frac{1}{2}}}$.
Next, denote
\begin{align*}
\sigma_n = \sqrt{ \Var_{\cW}\Big( \sum_{i=1}^n \xi_i \Big) }, \qquad S_n = \frac{\sum_{i=1}^n \xi_i}{\sigma_n}.
\end{align*}
We will use Lemma~\ref{lem:DependencyNeighborhoodCLT} to show the central limit theorem. 
To do so, we bound each term in Lemma~\ref{lem:DependencyNeighborhoodCLT}.

First, recall that we additionally assume that for sufficiently large $n$, $n \Var\big( \widehat{\tau}^{\CV}(\gamma^*(1), \gamma^*(0)) \big) \geq \underline{c}$.
So we have
\begin{align*}
\sigma_n^2 = \Var_{\cW}\Big(\sum_{i=1}^n \xi_i\Big) = n^2 \Var_{\cW}\Big(\frac{1}{n} \sum_{i=1}^n \xi_i\Big) = n^2 \Var_{\cW}\big( \widehat{\tau}^{\CV}(\gamma^*(1), \gamma^*(0)) \big) \geq \underline{c} n.
\end{align*}
We also have
\begin{align*}
\sigma_n^3 \geq \underline{c}^\frac{3}{2} n^\frac{3}{2}.
\end{align*}
Next, we have
\begin{align*}
\sum_{i=1}^n \bE[\vert\xi_i\vert^3] & \leq \sum_{i=1}^n \overline{\xi}^3 n^{\frac{3}{2}\alpha + 3\beta} = \overline{\xi}^3 n^{1 + \frac{3}{2}\alpha + 3\beta}, \\
\sum_{i=1}^n \bE[\xi_i^4] & \leq \sum_{i=1}^n \overline{\xi}^4 n^{2\alpha + 4\beta} = \overline{\xi}^4 n^{1 + 2\alpha + 4\beta}.
\end{align*}
Putting all together into Lemma~\ref{lem:DependencyNeighborhoodCLT}, for $Z$ a standard normal random variable, 
\begin{align*}
\lim_{n \to +\infty} d_W(S_n, Z) \leq & \lim_{n \to +\infty} \ \frac{\overline{d}^2 n^{2\alpha}}{\underline{c}^\frac{3}{2} n^\frac{3}{2}} \overline{\xi}^3 n^{1 + \frac{3}{2} \alpha} + \frac{26^\frac{1}{2} \overline{d}^\frac{3}{2} n^{\frac{3}{2}\alpha+3\beta}}{\pi^\frac{1}{2} \underline{c} n} \Big(\overline{\xi}^4 n^{1+2\alpha+4\beta}\Big)^\frac{1}{2} \\
= & \lim_{n \to +\infty} \ \frac{\overline{d}^2 \overline{\xi}^3}{\underline{c}^\frac{3}{2}} n^{\frac{7}{2}\alpha + 3\beta - \frac{1}{2}} + \frac{26^\frac{1}{2} \overline{d}^\frac{3}{2} \overline{\xi}^2}{\pi^\frac{1}{2} \underline{c}} n^{\frac{5}{2}\alpha + 2\beta - \frac{1}{2}} \\
= & \ 0,
\end{align*}
where $D_W(S_n, Z)$ stands for the Wasserstein distance between $S_n$ and $Z$;
the first inequality holds using Lemma~\ref{lem:DependencyNeighborhoodCLT};
the last equality holds because $\frac{5}{2}\alpha + 2\beta < \frac{7}{2}\alpha + 3\beta < \frac{1}{2}$.
So we have
\begin{align*}
\lim_{n \to +\infty} S_n = \lim_{n \to +\infty} \frac{\widehat{\tau}^\CV(\gamma^*(1), \gamma^*(0)) - \tau}{\sqrt{\Var\big(\widehat{\tau}^{\CV}(\gamma^*(1), \gamma^*(0))\big)}} \xrightarrow{d} \cN(0,1).
\end{align*}

Next, we establish the connection between $\widehat{\tau}^\CV(\widehat{\gamma}^\HT(1), \widehat{\gamma}^\HT(0))$ and $\widehat{\tau}^\CV(\gamma^*(1), \gamma^*(0))$.
Note that,
\begin{multline}
\frac{\widehat{\tau}^\CV(\widehat{\gamma}^\HT(1), \widehat{\gamma}^\HT(0)) - \tau}{\sqrt{\Var\big(\widehat{\tau}^\CV(\widehat{\gamma}^\HT(1), \widehat{\gamma}^\HT(0))\big)}} 
= \frac{\widehat{\tau}^\CV(\gamma^*(1), \gamma^*(0)) - \tau}{\sqrt{\Var\big(\widehat{\tau}^\CV(\gamma^*(1), \gamma^*(0))\big)}} \frac{\sqrt{\Var\big(\widehat{\tau}^\CV(\gamma^*(1), \gamma^*(0))\big)}}{\sqrt{\Var\big(\widehat{\tau}^\CV(\widehat{\gamma}^\HT(1), \widehat{\gamma}^\HT(0))\big)}} \\
+ \frac{\Delta_n}{\sqrt{\Var\big(\widehat{\tau}^\CV(\gamma^*(1), \gamma^*(0))\big)}} \frac{\sqrt{\Var\big(\widehat{\tau}^\CV(\gamma^*(1), \gamma^*(0))\big)}}{\sqrt{\Var\big(\widehat{\tau}^\CV(\widehat{\gamma}^\HT(1), \widehat{\gamma}^\HT(0))\big)}}. \label{eqn:CLTDecompositionCI}
\end{multline}
From here, note that
\begin{align*}
\bE\Bigg[\frac{\Delta_n}{\sqrt{\Var\big(\widehat{\tau}^\CV(\gamma^*(1), \gamma^*(0))\big)}} \Bigg] \leq \bE\Bigg[\frac{\Delta_n^2}{\Var\big(\widehat{\tau}^\CV(\gamma^*(1), \gamma^*(0))\big)} \Bigg]^\frac{1}{2} \leq \sqrt{ \frac{500 \overline{z}^2 \overline{a}^2 \overline{d}^2}{\underline{c}} \ n^{4\alpha + 4\beta - 1}},
\end{align*}
where the first inequality is Cauchy-Schwarz inequality;
and the second inequality is using \eqref{eqn:VarConvergeTo0} and assuming $\Var\big(\widehat{\tau}^\CV(\gamma^*(1), \gamma^*(0))\big) \geq \underline{c} n^{-1}$.
So we have, for any $\epsilon > 0$,
\begin{align*}
\lim_{n \to +\infty} \Pr\Bigg( \Bigg\vert \frac{\Delta_n}{\sqrt{\Var\big(\widehat{\tau}^\CV(\gamma^*(1), \gamma^*(0))\big)}} \Bigg\vert \geq \epsilon\Bigg) 
\leq \lim_{n \to +\infty} \sqrt{\frac{500 \overline{z}^2 \overline{a}^4 \overline{d}^4}{\underline{c}} \cdot n^{4\alpha + 4\beta - 1}} \cdot \frac{1}{\epsilon}
= 0,
\end{align*}
where the first inequality is due to Markov inequality;
the last and only equality holds because $4\alpha + 4\beta < 7\alpha + 6\beta < 1$.
So we have
\begin{align*}
\lim_{n \to +\infty} \frac{\Delta_n}{\sqrt{\Var\big(\widehat{\tau}^\CV(\gamma^*(1), \gamma^*(0))\big)}} \xrightarrow{p} 0.
\end{align*}
Note also that $\lim_{n \to +\infty} n \Var\big(\widehat{\tau}^\CV(\widehat{\gamma}^\HT(1), \widehat{\gamma}^\HT(0))\big) = \lim_{n \to +\infty} n \Var\big(\widehat{\tau}^\CV(\gamma^*(1), \gamma^*(0))\big)$ and for sufficiently large $n$ we assume $n \Var\big(\widehat{\tau}^\CV(\gamma^*(1), \gamma^*(0))\big) \geq \underline{c}$.
This ensures that $n \Var\big(\widehat{\tau}^\CV(\widehat{\gamma}^\HT(1), \widehat{\gamma}^\HT(0))\big) \geq \underline{c} > 0$ for sufficiently large $n$.
So we have
\begin{align*}
& \lim_{n \to +\infty} \Bigg\vert \frac{\Var\big(\widehat{\tau}^\CV(\gamma^*(1), \gamma^*(0))\big)}{\Var\big(\widehat{\tau}^\CV(\widehat{\gamma}^\HT(1), \widehat{\gamma}^\HT(0))\big)} - 1 \Bigg\vert \\ 
= & \lim_{n \to +\infty} \frac{\big\vert n \Var\big(\widehat{\tau}^\CV(\widehat{\gamma}^\HT(1), \widehat{\gamma}^\HT(0))\big) - n \Var\big(\widehat{\tau}^\CV(\gamma^*(1), \gamma^*(0))\big) \big\vert}{n \Var\big(\widehat{\tau}^\CV(\widehat{\gamma}^\HT(1), \widehat{\gamma}^\HT(0))\big)} \\
\leq & \lim_{n \to +\infty} \frac{\big\vert n \Var\big(\widehat{\tau}^\CV(\widehat{\gamma}^\HT(1), \widehat{\gamma}^\HT(0))\big) - n \Var\big(\widehat{\tau}^\CV(\gamma^*(1), \gamma^*(0))\big) \big\vert}{\underline{c}} \\
= & 0,
\end{align*}
where the last equality is Claim 2 in Theorem~\ref{thm:AsymptoticCI}.

Putting the above together into \eqref{eqn:CLTDecompositionCI} and using the Slutsky theorem we have
\begin{align*}
\lim_{n \to +\infty} \frac{\widehat{\tau}^\CV(\widehat{\gamma}^\HT(1), \widehat{\gamma}^\HT(0)) - \tau}{\sqrt{\Var\big(\widehat{\tau}^\CV(\widehat{\gamma}^\HT(1), \widehat{\gamma}^\HT(0))\big)}} \xrightarrow{d} \cN(0,1).
\end{align*}
This finishes the proof.
\hfill \halmos
\endproof

\subsection{Proof of Theorem~\ref{thm:OptimalControlVariates}}

\proof{Proof of Theorem~\ref{thm:OptimalControlVariates}.}
Recall that the objective function of \eqref{eqn:FormulationCausal} can be written as
\begin{align*}
\min_{\bm{A}} \ \bE_{\bm{Y}(1), \bm{Y}(0) \sim \cY_{\PO}^n} \Big[\bm{G}^\top \bm{\Sigma} \bm{G} - \bm{G}^\top \bm{\Sigma} \bm{A} \big(\bm{A}^\top \bm{\Sigma} \bm{A}\big)^{-1} \bm{A}^\top \bm{\Sigma} \bm{G}\Big].
\end{align*}
Note that $\bE_{\bm{Y}(1), \bm{Y}(0) \sim \cY_{\PO}^n} \big[\bm{G}^\top \bm{\Sigma} \bm{G}\big]$ does not depend on $\bm{A}$.
So this minimization problem can be written as a maximization problem
\begin{align}
\max_{\bm{A}} \ \bE_{\bm{Y}(1), \bm{Y}(0) \sim \cY_{\PO}^n} \Big[\bm{G}^\top \bm{\Sigma} \bm{A} \big(\bm{A}^\top \bm{\Sigma} \bm{A}\big)^{-1} \bm{A}^\top \bm{\Sigma} \bm{G}\Big]. \label{eqn:CIIntermediate2}
\end{align}

Now we introduce a variable transformation.
Note that $\bm{\Sigma}$ is a symmetric and positive semidefinite matrix.
So $\bm{A}^\top \bm{\Sigma} \bm{A}$ is also a symmetric and positive semidefinite matrix.
So we can define $\big(\bm{A}^\top \bm{\Sigma} \bm{A}\big)^{-1}$ as the pseudo inverse matrix of $\bm{A}^\top \bm{\Sigma} \bm{A}$, and $\big(\bm{A}^\top \bm{\Sigma} \bm{A}\big)^{-\frac{1}{2}}$ as the pseudo inverse square root matrix of $\bm{A}^\top \bm{\Sigma} \bm{A}$.
Now define
\begin{align*}
\bm{V} = \bm{\Sigma}^{\frac{1}{2}} \bm{A} \big(\bm{A}^\top \bm{\Sigma} \bm{A}\big)^{-\frac{1}{2}}. 
\end{align*}
Using the definition of $\bm{V}$, \eqref{eqn:CIIntermediate2} can be written as
\begin{align}
\max_{\bm{A}} \ \bE_{\bm{Y}(1), \bm{Y}(0) \sim \cY_{\PO}^n} \Big[\bm{G}^\top \bm{\Sigma}^{\frac{1}{2}} \bm{V} \bm{V}^\top \bm{\Sigma}^{\frac{1}{2}} \bm{G} \Big]. \label{eqn:CIIntermediate3}
\end{align}
Next, note that
\begin{align*}
\bE_{\bm{Y}(1), \bm{Y}(0) \sim \cY_{\PO}^n} \Big[\bm{G}^\top \bm{\Sigma}^{\frac{1}{2}} \bm{V} \bm{V}^\top \bm{\Sigma}^{\frac{1}{2}} \bm{G}\Big] 
& = \bE_{\bm{Y}(1), \bm{Y}(0) \sim \cY_{\PO}^n} \Big[ \Tr\big( \bm{G}^\top \bm{\Sigma}^{\frac{1}{2}} \bm{V} \bm{V}^\top \bm{\Sigma}^{\frac{1}{2}} \bm{G} \big) \Big] \\
& = \bE_{\bm{Y}(1), \bm{Y}(0) \sim \cY_{\PO}^n} \Big[ \Tr\big( \bm{V}^\top \bm{\Sigma}^{\frac{1}{2}} \bm{G} \bm{G}^\top \bm{\Sigma}^{\frac{1}{2}} \bm{V} \big) \Big] \\
& = \Tr\bigg( \bm{V}^\top \bm{\Sigma}^{\frac{1}{2}} \bE_{\bm{Y}(1), \bm{Y}(0) \sim \cY_{\PO}^n} \Big[ \bm{G} \bm{G}^\top\Big] \bm{\Sigma}^{\frac{1}{2}} \bm{V} \bigg),
\end{align*}
where the first equality is because a scalar is equal to its trace;
the second equality is because of the cyclic property of a trace;
the third equality is because of linearity of expectations.
Now we define 
\begin{align*}
\bm{M} = \bm{\Sigma}^{\frac{1}{2}} \bE_{\bm{Y}(1), \bm{Y}(0) \sim \cY_{\PO}^n} \Big[ \bm{G} \bm{G}^\top\Big] \bm{\Sigma}^{\frac{1}{2}} = & \bm{\Sigma}^{\frac{1}{2}} \bm{\Pi}(1)^{-1} \big( \mu^2(1) \bm{1}_n \bm{1}_n^\top + \sigma^2(1) \bm{I}_n \big) \bm{\Pi}(1)^{-1} \bm{\Sigma}^{\frac{1}{2}} \\
& \quad + \bm{\Sigma}^{\frac{1}{2}} \bm{\Pi}(1)^{-1} \big( \mu(1) \mu(0) \bm{1}_n \bm{1}_n^\top + \sigma(1,0) \bm{I}_n \big) \bm{\Pi}(0)^{-1} \bm{\Sigma}^{\frac{1}{2}} \\
& \quad + \bm{\Sigma}^{\frac{1}{2}} \bm{\Pi}(0)^{-1} \big( \mu(1) \mu(0) \bm{1}_n \bm{1}_n^\top + \sigma(1,0) \bm{I}_n \big) \bm{\Pi}(1)^{-1} \bm{\Sigma}^{\frac{1}{2}} \\
& \quad + \bm{\Sigma}^{\frac{1}{2}} \bm{\Pi}(0)^{-1} \big( \mu^2(0) \bm{1}_n \bm{1}_n^\top + \sigma^2(0) \bm{I}_n \big) \bm{\Pi}(0)^{-1} \bm{\Sigma}^{\frac{1}{2}}.
\end{align*}
So \eqref{eqn:CIIntermediate3} can be written as
\begin{align}
\max_{\bm{V}} \ \Tr\bigg( \bm{V}^\top \bm{M} \bm{V} \bigg). \label{eqn:CIIntermediate4}
\end{align}

Now we consider the rank of $\bm{V}$.
There are two cases. 
First, if $\mathrm{rank}(\bm{V}) = 2$, then 
\begin{align*}
\bm{V}^\top \bm{V} =  \big(\bm{A}^\top \bm{\Sigma} \bm{A}\big)^{-\frac{1}{2}} \bm{A}^\top \bm{\Sigma}^{\frac{1}{2}} \bm{\Sigma}^{\frac{1}{2}} \bm{A} \big(\bm{A}^\top \bm{\Sigma} \bm{A}\big)^{-\frac{1}{2}} = \bm{I}_2.
\end{align*}
Using Lemma~\ref{lem:FanPrinciple} the Fan's Principle (\citet{fan1949theorem} Theorem 1), the optimal objective value of \eqref{eqn:CIIntermediate4} is the sum of the two largest eigenvalues of $\bm{M}$, denoted as $\lambda_1(\bm{M})$ and $\lambda_2(\bm{M})$.
So the optimal variance reduction is given by 
\begin{align*}
\bE_{\bm{Y}(1), \bm{Y}(0) \sim \cY_{\PO}^n} \Big[ \Var\big( \widehat{\tau}^{\HT} \big) - \Var\big( \widehat{\tau}^{\CV}(\gamma^*(1), \gamma^*(0)) \big) \Big] = \frac{\lambda_1(\bm{M}) + \lambda_2(\bm{M})}{n^2}.
\end{align*}

Additionally, using Lemma~\ref{lem:FanPrinciple} the Fan's Principle (\citet{fan1949theorem} Theorem 1), the optimal solution to \eqref{eqn:CIIntermediate4} is the two eigenvectors that correspond to $\lambda_1(\bm{M})$ and $\lambda_2(\bm{M})$, denoted as $\bm{u}_1\big(\bm{M}\big)$ and $\bm{u}_2(\bm{M})$.
And we have the optimal solution being $\bm{V} = [\bm{u}_1(\bm{M}), \bm{u}_2(\bm{M})]$.
Substituting this back to the definition of $\bm{V}$ we have that the optimal bases are
\begin{align*}
\bm{A}^* = \bm{\Sigma}^{-\frac{1}{2}} \big[\bm{u}_1(\bm{M}), \bm{u}_2(\bm{M})\big] \bm{C}
\end{align*}
where $\bm{C}$ is any invertible, symmetric and semidefinite matrix. 

Second, if $\mathrm{rank}(\bm{V}) = 1$, then 
\begin{align*}
\bm{V}^\top \bm{V} =  \big(\bm{A}^\top \bm{\Sigma} \bm{A}\big)^{-\frac{1}{2}} \bm{A}^\top \bm{\Sigma}^{\frac{1}{2}} \bm{\Sigma}^{\frac{1}{2}} \bm{A} \big(\bm{A}^\top \bm{\Sigma} \bm{A}\big)^{-\frac{1}{2}} = 
\begin{bmatrix}
1 & 0 \\
0 & 0
\end{bmatrix}.
\end{align*}
In this case, using Lemma~\ref{lem:FanPrinciple} the Fan's Principle (\citet{fan1949theorem} Theorem 1), the optimal objective value of \eqref{eqn:CIIntermediate4} is the largest eigenvalue of $\bm{M}$, denoted as $\lambda_1(\bm{M})$.
So the optimal variance reduction is given by 
\begin{align*}
\bE_{\bm{Y}(1), \bm{Y}(0) \sim \cY_{\PO}^n} \Big[ \Var\big( \widehat{\tau}^{\HT} \big) - \Var\big( \widehat{\tau}^{\CV}(\gamma^*(1), \gamma^*(0)) \big) \Big] = \frac{\lambda_1(\bm{M})}{n^2} \leq \frac{\lambda_1(\bm{M}) + \lambda_2(\bm{M})}{n^2}.
\end{align*}
So the variance reduction in Case 2 is always smaller or equal to the variance reduction in Case 1. 
This means that we should choose the bases $\bm{A}$ to have $\mathrm{rank}(\bm{A}) = 2$ full column rank, and the optimal solution comes from Case 1. 
\hfill \halmos
\endproof

\subsection{Proof of Example~\ref{exa:SUTVAnInterference}}

We show that, under SUTVA, for any $2n \times 2$ bases matrix $\bm{B} \in \bR^{2n \times 2}$, we can explicitly construct a block diagonal bases matrix $\bm{B}^\circ \in \cB$ that attains exactly the same objective value. 

\proof{Proof of Example~\ref{exa:SUTVAnInterference}}
We start with any $2n \times 2$ bases matrix
\begin{align*}
\bm{B} =
\begin{bmatrix}
\bm{u}_1 & \bm{u}_2 \\
\bm{v}_1 & \bm{v}_2 
\end{bmatrix}
\in \bR^{2n \times 2},
\end{align*}
where $\bm{u}_1, \bm{u}_2, \bm{v}_1, \bm{v}_2$ are all $n$-dimensional vectors. 
From here, we construct the block-diagonal matrix
\begin{align*}
\bm{B}^\circ =
\begin{bmatrix}
\bm u_1-\bm v_1 & \bm 0_n\\
\bm 0_n & \bm v_2-\bm u_2
\end{bmatrix}
\in \cB.
\end{align*}

Now we consider $\bm{B} - \bm{B}^\circ$. For the first column,
\begin{align*}
\begin{bmatrix}
\bm{u}_1\\
\bm{v}_1
\end{bmatrix}
-
\begin{bmatrix}
\bm{u}_1 - \bm{v}_1\\
\bm{0}_n
\end{bmatrix}
=
\begin{bmatrix}
\bm{v}_1\\
\bm{v}_1
\end{bmatrix},
\end{align*}
and for the second column,
\begin{align*}
\begin{bmatrix}
\bm{u}_2\\
\bm{v}_2
\end{bmatrix}
-
\begin{bmatrix}
\bm{0}_n \\
\bm{v}_2 - \bm{u}_2
\end{bmatrix}
=
\begin{bmatrix}
\bm{u}_2\\
\bm{u}_2
\end{bmatrix}.
\end{align*}
Recall the special block structure \eqref{eqn:SpecialBlockStructure} that
\begin{align*}
\bm{\Omega} =
\begin{bmatrix}
\bm{\Sigma} & -\bm{\Sigma}\\
-\bm{\Sigma} & \bm{\Sigma}
\end{bmatrix}. 
\end{align*}
Both vectors are in the null space of $\bm{\Omega}$.
Therefore,
\begin{align*}
\bm{\Omega}^{\frac{1}{2}}\bm{B} = \bm{\Omega}^{\frac{1}{2}} \bm{B}^\circ.
\end{align*}
Then we have
\begin{align*}
\bm{B}^\top \bm{\Omega} \bm{B} = (\bm{B}^\circ)^\top \bm{\Omega} \bm{B}^\circ
\end{align*}
and
\begin{align*}
\bm{B}^\top \bm{\Omega}^{\frac{1}{2}} \bm{M} \bm{\Omega}^{\frac{1}{2}} \bm{B} = (\bm{B}^\circ)^\top \bm{\Omega}^{\frac{1}{2}} \bm{M} \bm{\Omega}^{\frac{1}{2}} \bm{B}^\circ.
\end{align*}
So $\bm{B}$ and $\bm{B}^\circ$ attain the same objective value. 
Since every $2n \times 2$ bases matrix corresponds to a block diagonal bases matrix with the same objective value, the constrained and unconstrained optimization problems have the same optimal value under SUTVA. 
This shows that Theorem~\ref{thm:OPTInterference} in the network interference setting nests Theorem~\ref{thm:OptimalControlVariates} in the SUTVA setting as a special case. 
\hfill \halmos
\endproof

\subsection{Proof of Proposition~\ref{prop:OrthonormalOptA}}

We present the following Lemma~\ref{lem:ColumnSpaceCI} which will be useful in the proof of Proposition~\ref{prop:OrthonormalOptA}.
To introduce Lemma~\ref{lem:ColumnSpaceCI}, we denote the following notation.
For any matrix $\bm{E} \in \bR^{n \times n}$, denote $\mathrm{range}(\bm{E})$ to be the column space of matrix $\bm{E}$, defined as
\begin{align*}
\mathrm{range}(\bm{E}) = \big\{ \bm{E} \bm{x} \ \vert \ \bm{x} \in \bR^{n} \big\}.
\end{align*}

\begin{lemma}
\label{lem:ColumnSpaceCI}
Let $\bm{E} \in \bR^{n \times n}$ be any symmetric matrix. 
For any $l \leq n$, let $\lambda_l(\bm{M})$ be the $l$-th largest eigenvalue of $\bm{M}$ and $\bm{u}_l(\bm{M})$ be the eigenvector that corresponds to $\lambda_l(\bm{M})$, where $\bm{M}$ is a symmetric matrix defined as $\bm{M} = \bm{\Sigma}^{\frac{1}{2}} \bm{E} \bm{\Sigma}^{\frac{1}{2}}$. 
When $\lambda_l(\bm{M}) > 0$, we have
\begin{align*}
\bm{u}_l(\bm{M}) \in \mathrm{range}(\bm{\Sigma}).
\end{align*}
\end{lemma}

\proof{Proof of Lemma~\ref{lem:ColumnSpaceCI}.}
Note that $\bm{u}_l(\bm{M}) \in \mathrm{range}(\bm{M})$ because $\bm{u}_l(\bm{M})$ is an eigenvector of the symmetric matrix $\bm{M}$.
Next, we show that $\mathrm{range}(\bm{M}) \subseteq \mathrm{range}(\bm{\Sigma})$.
To see this, note that for any vector $\bm{x}$,
\begin{align*}
\bm{M} \bm{x} = \bm{\Sigma}^{\frac{1}{2}} \Big( \bm{E} \bm{\Sigma}^{\frac{1}{2}} \bm{x} \Big).
\end{align*}
Because $\big( \bm{E} \bm{\Sigma}^{\frac{1}{2}} \bm{x} \big)$ may not span the entire space of $\bR^n$, we have that
\begin{align*}
\mathrm{range}(\bm{M}) \subseteq \mathrm{range}(\bm{\Sigma}^\frac{1}{2}).
\end{align*}
Next, because $\bm{\Sigma} = \bm{U} \bm{\Lambda} \bm{U}^{-1}$ and $\bm{\Sigma}^\frac{1}{2} = \bm{U} \bm{\Lambda}^\frac{1}{2} \bm{U}^{-1}$ have the same directions along the eigenvectors corresponding to positive eigenvalues, we have 
\begin{align*}
\mathrm{range}(\bm{\Sigma}^\frac{1}{2}) = \mathrm{range}(\bm{\Sigma}).
\end{align*}

Putting all above together, we have
\begin{align*}
\bm{u}_l(\bm{M}) \in \mathrm{range}(\bm{M}) \subseteq \mathrm{range}(\bm{\Sigma}^\frac{1}{2}) = \mathrm{range}(\bm{\Sigma}).
\end{align*}
This finishes the proof.
\hfill \halmos
\endproof

\

\noindent Now we prove Proposition~\ref{prop:OrthonormalOptA} as follows. 

\proof{Proof of Proposition~\ref{prop:OrthonormalOptA}.}
If we choose $\bm{C} = n^{\frac{1}{2}} \bm{I}_2$ where $\bm{I}_2$ stands for a $2 \times 2$ identity matrix, Theorem~\ref{thm:OptimalControlVariates} suggests that 
\begin{align*}
\bm{a}(1) = n^{\frac{1}{2}} \bm{\Sigma}^{-\frac{1}{2}} \bm{u}_1(\bm{M}), \qquad \text{and} \qquad \bm{a}(0) = - n^{\frac{1}{2}} \bm{\Sigma}^{-\frac{1}{2}} \bm{u}_2(\bm{M}).
\end{align*}
Using the above expressions,
\begin{align*}
\bm{a}(1)^\top \bm{\Sigma} \bm{a}(0) = \big( \bm{\Sigma}^{-\frac{1}{2}} \bm{u}_1(\bm{M}) \big)^\top \bm{\Sigma} \big( - \bm{\Sigma}^{-\frac{1}{2}} \bm{u}_2(\bm{M}) \big) = - \bm{u}_1(\bm{M})^\top \bm{\Sigma}^{-\frac{1}{2}} \bm{\Sigma} \bm{\Sigma}^{-\frac{1}{2}} \bm{u}_2(\bm{M}).
\end{align*}

Denote the following diagonal matrix $\bI^+ = \bm{\mathrm{diag}}(\bI\{\lambda_1 > 0\}, \bI\{\lambda_2 > 0\}, ..., \bI\{\lambda_n > 0\})$, where $\lambda_i$ stands for the $i$-th largest eigenvalue of matrix $\bm{\Sigma}$.
Denote $\bm{P}_\Sigma = \bm{\Sigma}^{-\frac{1}{2}} \bm{\Sigma} \bm{\Sigma}^{-\frac{1}{2}}$.
Recall the eigen-decomposition $\bm{\Sigma} = \bm{U} \bm{\Lambda} \bm{U}^{-1}$.
So we know that
\begin{align*}
\bm{P}_\Sigma = \bm{\Sigma}^{-\frac{1}{2}} \bm{\Sigma} \bm{\Sigma}^{-\frac{1}{2}} = \bm{U} \bI^+ \bm{U}^{-1}
\end{align*}
is a projection matrix, that is, if a vector $\bm{x} \in \mathrm{range}(\bm{\Sigma})$, we have $\bm{P}_\Sigma \bm{x} = \bm{x}$.

Using Lemma~\ref{lem:ColumnSpaceCI}, we have 
\begin{align*}
\bm{u}_2(\bm{M}) \in \mathrm{range}(\bm{\Sigma}).
\end{align*}
Consequently, $\bm{P}_\Sigma \bm{u}_2(\bm{M}) = \bm{u}_2(\bm{M})$.
This leads to 
\begin{align*}
\bm{a}(1)^\top \bm{\Sigma} \bm{a}(0) = - \bm{u}_1(\bm{M})^\top \bm{u}_2(\bm{M}) = 0,
\end{align*}
where the last equality is because $\bm{u}_1(\bm{M})$ and $\bm{u}_2(\bm{M})$ are orthogonal.

Moreover, recall that using Lemma~\ref{lem:ColumnSpaceCI} we have
\begin{align*}
\bm{u}_1(\bm{M}), \bm{u}_2(\bm{M}) \in \mathrm{range}(\bm{\Sigma}).
\end{align*}
Consequently, $\bm{P}_\Sigma \bm{u}_1(\bm{M}) = \bm{u}_1(\bm{M})$ and $\bm{P}_\Sigma \bm{u}_2(\bm{M}) = \bm{u}_2(\bm{M})$.
Next, we have
\begin{align*}
\bm{a}(1)^\top \bm{\Sigma} \bm{a}(1) = n \bm{u}_1(\bm{M})^\top \bm{\Sigma}^{-\frac{1}{2}} \bm{\Sigma} \bm{\Sigma}^{-\frac{1}{2}} \bm{u}_1(\bm{M}) = n \bm{u}_1(\bm{M})^\top \bm{u}_1(\bm{M}) = n,
\end{align*}
and 
\begin{align*}
\bm{a}(0)^\top \bm{\Sigma} \bm{a}(0) = n \bm{u}_2(\bm{M})^\top \bm{\Sigma}^{-\frac{1}{2}} \bm{\Sigma} \bm{\Sigma}^{-\frac{1}{2}} \bm{u}_2(\bm{M}) = n \bm{u}_2(\bm{M})^\top \bm{u}_2(\bm{M}) = n,
\end{align*}
where the last equality in the above expressions holds because $\bm{u}_1(\bm{M})$ and $\bm{u}_2(\bm{M})$ are eigenvectors so they are unit length.
\hfill \halmos
\endproof

\subsection{Proof of Corollary~\ref{coro:NoiselessPotentialOutcomes}}

\proof{Proof of Corollary~\ref{coro:NoiselessPotentialOutcomes}.}
The proof of Corollary~\ref{coro:NoiselessPotentialOutcomes} is very straightforward. 
Recall from Section~\ref{sec:OptimalControlVariates} that if the potential outcomes $\bm{Y}(1)$ and $\bm{Y}(0)$ were known to take values $Y_i(1) = y_i(1)$ and $Y_i(0) = y_i(0)$ for any $i \in [n]$, the problem 
\begin{align*}
\min_{\bm{A}} \ \Var\big( \widehat{\tau}^{\CV}(\gamma^*(1), \gamma^*(0)) \big) \ = \ \min_{\bm{A}} \frac{1}{n^2} \bigg( \bm{g}^\top \bm{\Sigma} \bm{g} - \bm{g}^\top \bm{\Sigma} \bm{A} \big(\bm{A}^\top \bm{\Sigma} \bm{A}\big)^{-1} \bm{A}^\top \bm{\Sigma} \bm{g} \bigg)
\end{align*}
can be easily solved by choosing $\bm{A}$ such that $\bm{g}$ is in the column space of $\bm{A}$, such as $\bm{a}(1) = \bm{\Pi}(1)^{-1} \bm{y}(1)$ and $\bm{a}(0) = -\bm{\Pi}(0)^{-1}\bm{y}(0)$. 
Now because we assume that the potential outcomes $Y_i(1) = y(1)$ and $Y_i(0) = y(0)$ take the same unknown constant, then one optimal solution is given by $\bm{a}(1) = \bm{\Pi}(1)^{-1} \bm{1}_n \cdot y(1)$ and $\bm{a}(0) = - \bm{\Pi}(0)^{-1} \bm{1}_n \cdot y(0)$.
\hfill \halmos
\endproof

\subsection{Proof of Example~\ref{exa:NaiveBenchmarkOrthonormalBases}}

\proof{Proof of Example~\ref{exa:NaiveBenchmarkOrthonormalBases}.}
Note that, for any $\bm{\gamma} = [\gamma(1), \gamma(0)]^\top$, we have
\begin{align*}
\Var\big( \widehat{\tau}^{\CV}(\gamma(1), \gamma(0)) \big) = \Var\big( \widehat{\tau}^{\HT} \big) - \frac{2}{n^2} \bm{\gamma}^\top \bm{A}^\top \bm{\Sigma} \bm{G} + \frac{1}{n^2} \bm{\gamma}^\top \bm{A}^\top \bm{\Sigma} \bm{A} \bm{\gamma}.
\end{align*}
So we have
\begin{align}
& n^2 \Big( \Var\big( \widehat{\tau}^{\CV}(\gamma^{\Naive}(1), \gamma^{\Naive}(0)) \big) - \Var\big( \widehat{\tau}^{\CV}(\gamma^*(1), \gamma^*(0)) \big) \Big)\nonumber \\
= & \ 2 \big(\bm{\gamma}^*\big)^\top \bm{A}^\top \bm{\Sigma} \bm{G} - \big(\bm{\gamma}^*\big)^\top \bm{A}^\top \bm{\Sigma} \bm{A} \bm{\gamma}^* - 2 \big(\bm{\gamma}^{\Naive}\big)^\top \bm{A}^\top \bm{\Sigma} \bm{G} + \big(\bm{\gamma}^{\Naive}\big)^\top \bm{A}^\top \bm{\Sigma} \bm{A} \bm{\gamma}^{\Naive} \nonumber \\
= & \ 2 \big(\bm{\gamma}^*\big)^\top \bm{A}^\top \bm{\Sigma} \bm{A} \bm{\gamma}^* - \big(\bm{\gamma}^*\big)^\top \bm{A}^\top \bm{\Sigma} \bm{A} \bm{\gamma}^* - 2 \big(\bm{\gamma}^{\Naive}\big)^\top \bm{A}^\top \bm{\Sigma} \bm{A} \bm{\gamma}^* + \big(\bm{\gamma}^{\Naive}\big)^\top \bm{A}^\top \bm{\Sigma} \bm{A} \bm{\gamma}^{\Naive} \nonumber \\
= & \ \big(\bm{\gamma}^*\big)^\top \bm{A}^\top \bm{\Sigma} \bm{A} \bm{\gamma}^* - 2 \big(\bm{\gamma}^{\Naive}\big)^\top \bm{A}^\top \bm{\Sigma} \bm{A} \bm{\gamma}^* + \big(\bm{\gamma}^{\Naive}\big)^\top \bm{A}^\top \bm{\Sigma} \bm{A} \bm{\gamma}^{\Naive} \nonumber \\
= & \ \big(\bm{\gamma}^{\Naive} - \bm{\gamma}^*\big)^\top \bm{A}^\top \bm{\Sigma} \bm{A} \big(\bm{\gamma}^{\Naive} - \bm{\gamma}^*\big), \label{eqn:VarComparison}
\end{align}
where the second equality is because $\bm{\gamma}^* = \big(\bm{A}^\top \bm{\Sigma} \bm{A}\big)^{-1} \bm{A}^\top \bm{\Sigma} \bm{G}$ so $\bm{A}^\top \bm{\Sigma} \bm{G} = \bm{A}^\top \bm{\Sigma} \bm{A} \bm{\gamma}^*$.

Following Proposition~\ref{prop:OrthonormalOptA}, if we simply choose
\begin{align*}
\bm{a}(1) = n^{\frac{1}{2}} \bm{\Sigma}^{-\frac{1}{2}} \bm{u}_1(\bm{M}), \qquad \text{and} \qquad \bm{a}(0) = - n^{\frac{1}{2}} \bm{\Sigma}^{-\frac{1}{2}} \bm{u}_2(\bm{M}), 
\end{align*}
then $\bm{a}(1)^\top \bm{\Sigma} \bm{a}(0) = 0$ and $\bm{a}(1)^\top \bm{\Sigma} \bm{a}(1) = \bm{a}(0)^\top \bm{\Sigma} \bm{a}(0) = n$ and we have that the variance minimizing coefficients in Lemma~\ref{lem:OPTgammas} can be simplified by
\begin{align*}
\gamma^*(1) = n^{-\frac{1}{2}} \bm{G}^\top \bm{\Sigma}^{\frac{1}{2}} \bm{u}_1(\bm{M}), \qquad \text{and} \qquad \gamma^*(0) = n^{-\frac{1}{2}} \bm{G}^\top \bm{\Sigma}^{\frac{1}{2}} \bm{u}_2(\bm{M}).
\end{align*}
Next, note that
\begin{align*}
\gamma^{\Naive}(1) = \frac{\Cov\big(\widehat{\mu}^{\HT}(1), \widehat{X}(1)\big)}{\Var\big(\widehat{X}(1)\big)} = \frac{\bm{Y}(1)^\top \bm{\Pi}(1)^{-1} \bm{\Sigma} \bm{a}(1)}{\bm{a}(1)^\top \bm{\Sigma} \bm{a}(1)} = n^{-\frac{1}{2}} \bm{Y}(1)^\top \bm{\Pi}(1)^{-1} \bm{\Sigma}^{\frac{1}{2}} \bm{u}_1(\bm{M})
\end{align*}
and
\begin{align*}
\gamma^{\Naive}(0) = - \frac{\Cov\big(\widehat{\mu}^{\HT}(0), \widehat{X}(0)\big)}{\Var\big(\widehat{X}(0)\big)} = - \frac{\bm{Y}(0)^\top \bm{\Pi}(0)^{-1} \bm{\Sigma} \bm{a}(0)}{\bm{a}(0)^\top \bm{\Sigma} \bm{a}(0)} = n^{-\frac{1}{2}} \bm{Y}(0)^\top \bm{\Pi}(0)^{-1} \bm{\Sigma}^{\frac{1}{2}} \bm{u}_2(\bm{M}).
\end{align*}
So we have
\begin{align*}
\gamma^{\Naive}(1) - \gamma^*(1) = - n^{-\frac{1}{2}} \bm{Y}(0)^\top \bm{\Pi}(0)^{-1} \bm{\Sigma}^{\frac{1}{2}} \bm{u}_1(\bm{M})
\end{align*}
and
\begin{align*}
\gamma^{\Naive}(0) - \gamma^*(0) = - n^{-\frac{1}{2}} \bm{Y}(1)^\top \bm{\Pi}(1)^{-1} \bm{\Sigma}^{\frac{1}{2}} \bm{u}_2(\bm{M}).
\end{align*}
Putting the above into \eqref{eqn:VarComparison} and recalling $\bm{A}^\top \bm{\Sigma} \bm{A} = n \bm{I}_2$ where $\bm{I}_2$ is a $2 \times 2$ identity matrix, we have
\begin{multline*}
\Var\big( \widehat{\tau}^{\CV}(\gamma^{\Naive}(1), \gamma^{\Naive}(0)) \big) - \Var\big( \widehat{\tau}^{\CV}(\gamma^*(1), \gamma^*(0)) \big) \\
= \ \frac{1}{n^2} \bigg( \big(\bm{Y}(0)^\top \bm{\Pi}(0)^{-1} \bm{\Sigma}^{\frac{1}{2}} \bm{u}_1(\bm{M})\big)^2 + \big(\bm{Y}(1)^\top \bm{\Pi}(1)^{-1} \bm{\Sigma}^{\frac{1}{2}} \bm{u}_2(\bm{M})\big)^2 \bigg) \geq 0,
\end{multline*}
where the inequality takes equality if and only if when both $\bm{Y}(0)^\top \bm{\Pi}(0)^{-1} \bm{\Sigma}^{\frac{1}{2}} \bm{u}_1(\bm{M}) = 0$ and $\bm{Y}(1)^\top \bm{\Pi}(1)^{-1} \bm{\Sigma}^{\frac{1}{2}} \bm{u}_2(\bm{M}) = 0$. 
\hfill \halmos
\endproof

\subsection{Proof of Example~\ref{exa:NaiveBenchmarkIPWBases}}

\proof{Proof of Example~\ref{exa:NaiveBenchmarkIPWBases}.}
Using the bases $a_i(1) = \frac{1}{\pi_i}$ and $a_i(0) = - \frac{1}{1-\pi_i}$ for any $i \in [n]$, we have
\begin{align*}
\bm{A}^\top \bm{\Sigma} \bm{A} = 
\begin{bmatrix}
\sum_{i=1}^n \frac{1-\pi_i}{\pi_i} & n \vspace{5pt} \\
n & \sum_{i=1}^n \frac{\pi_i}{1-\pi_i}
\end{bmatrix},
\end{align*}
and
\begin{align*}
\bm{A}^\top \bm{\Sigma} \bm{G} = 
\begin{bmatrix}
\sum_{i=1}^n (1-\pi_i) \big(\frac{Y_i(1)}{\pi_i} + \frac{Y_i(0)}{1-\pi_i}\big) \vspace{5pt} \\
\sum_{i=1}^n \pi_i \big(\frac{Y_i(1)}{\pi_i} + \frac{Y_i(0)}{1-\pi_i}\big)
\end{bmatrix}.
\end{align*}
We further derive
\begin{align*}
\gamma^{\Naive}(1) = \frac{\sum_{i=1}^n \frac{1-\pi_i}{\pi_i} Y_i(1)}{\sum_{i=1}^n \frac{1-\pi_i}{\pi_i}}, \qquad \text{and} \qquad \gamma^{\Naive}(0) = \frac{\sum_{i=1}^n \frac{\pi_i}{1-\pi_i} Y_i(0)}{\sum_{i=1}^n \frac{\pi_i}{1-\pi_i}}.
\end{align*}
Finally, we can show that
\begin{align*}
\det(\bm{A}^\top \bm{\Sigma} \bm{A}) = \Big( \sum_{i=1}^n \frac{1-\pi_i}{\pi_i} \Big) \Big( \sum_{i=1}^n \frac{\pi_i}{1-\pi_i} \Big) - n^2 \geq \Big( \sum_{i=1}^n \sqrt{\frac{1-\pi_i}{\pi_i}} \cdot \sqrt{\frac{\pi_i}{1-\pi_i}} \Big)^2 - n^2 = 0,
\end{align*}
where the inequality is due to Cauchy-Schwarz inequality, and the inequality takes equality if and only if $\pi_i = \pi$ for any $i \in [n]$.

Note that,
\begin{align}
& \Var\big( \widehat{\tau}^{\CV}(\gamma^{\Naive}(1), \gamma^{\Naive}(0)) \big) - \Var\big( \widehat{\tau}^{\CV}(\gamma^*(1), \gamma^*(0)) \big) \nonumber \\
= & \ \frac{1}{n^2} \big(\bm{\gamma}^{\Naive} - \bm{\gamma}^*\big)^\top \bm{A}^\top \bm{\Sigma} \bm{A} \big(\bm{\gamma}^{\Naive} - \bm{\gamma}^*\big) \nonumber \\
= & \ \frac{1}{n^2} \big( \bm{A}^\top \bm{\Sigma} \bm{A} \bm{\gamma}^{\Naive} - \bm{A}^\top \bm{\Sigma} \bm{G} \big)^\top \big( \bm{A}^\top \bm{\Sigma} \bm{A} \big)^{-1} \big( \bm{A}^\top \bm{\Sigma} \bm{A} \bm{\gamma}^{\Naive} - \bm{A}^\top \bm{\Sigma} \bm{G} \big), \label{eqn:VarComparison2}
\end{align}
where the first equality is due to \eqref{eqn:VarComparison} in the proof of Example~\ref{exa:NaiveBenchmarkOrthonormalBases};
the second equality is because $\bm{\gamma}^* = (\bm{A}^\top \bm{\Sigma} \bm{A})^{-1} \bm{A}^\top \bm{\Sigma} \bm{G}$.
From here, we focus on
\begin{align*}
\bm{A}^\top \bm{\Sigma} \bm{A} \bm{\gamma}^{\Naive} - \bm{A}^\top \bm{\Sigma} \bm{G} = & 
\begin{bmatrix}
\sum_{i=1}^n \frac{1-\pi_i}{\pi_i} & n \vspace{5pt} \\
n & \sum_{i=1}^n \frac{\pi_i}{1-\pi_i}
\end{bmatrix} 
\cdot
\begin{bmatrix}
\gamma^{\Naive}(1) \vspace{5pt} \\
\gamma^{\Naive}(0) 
\end{bmatrix}
-
\begin{bmatrix}
\sum_{i=1}^n (1-\pi_i) \big(\frac{Y_i(1)}{\pi_i} + \frac{Y_i(0)}{1-\pi_i}\big) \vspace{5pt} \\
\sum_{i=1}^n \pi_i \big(\frac{Y_i(1)}{\pi_i} + \frac{Y_i(0)}{1-\pi_i}\big) 
\end{bmatrix} \\
= & 
\begin{bmatrix}
n \gamma^{\Naive}(0) - \sum_{i=1}^n Y_i(0) \vspace{5pt} \\
n \gamma^{\Naive}(1) - \sum_{i=1}^n Y_i(1)
\end{bmatrix}.
\end{align*}

Now we distinguish two cases.

\noindent \textbf{Case 1}: $\det(\bm{A}^\top \bm{\Sigma} \bm{A}) > 0$,
\begin{align*}
\big( \bm{A}^\top \bm{\Sigma} \bm{A} \big)^{-1} = \frac{1}{\det(\bm{A}^\top \bm{\Sigma} \bm{A})} \begin{bmatrix}
\sum_{i=1}^n \frac{\pi_i}{1-\pi_i} & -n \vspace{5pt} \\
-n & \sum_{i=1}^n \frac{1-\pi_i}{\pi_i}
\end{bmatrix}.
\end{align*}
Putting the above into \eqref{eqn:VarComparison2} we have
\begin{multline*}
\Var\big( \widehat{\tau}^{\CV}(\gamma^{\Naive}(1), \gamma^{\Naive}(0)) \big) - \Var\big( \widehat{\tau}^{\CV}(\gamma^*(1), \gamma^*(0)) \big) \\
= \frac{1}{n^2 \big(\sum_{i=1}^n \frac{1-\pi_i}{\pi_i} \sum_{i=1}^n \frac{\pi_i}{1-\pi_i} - n^2\big)} \bigg( \sum_{i=1}^n \frac{1-\pi_i}{\pi_i} \Theta(1)^2 - 2 n \Theta(1) \Theta(0) + \sum_{i=1}^n \frac{\pi_i}{1-\pi_i} \Theta(0)^2 \bigg),
\end{multline*}
where we denote
\begin{align*}
\Theta(1) & = \frac{n \sum_{i=1}^n \frac{1-\pi_i}{\pi_i} Y_i(1) - \big(\sum_{i=1}^n \frac{1-\pi_i}{\pi_i}\big) \big(\sum_{i=1}^n Y_i(1)\big)}{\sum_{i=1}^n \frac{1-\pi_i}{\pi_i}}; \\
\Theta(0) & = \frac{n \sum_{i=1}^n \frac{\pi_i}{1-\pi_i} Y_i(0) - \big(\sum_{i=1}^n \frac{\pi_i}{1-\pi_i}\big) \big(\sum_{i=1}^n Y_i(0)\big)}{\sum_{i=1}^n \frac{\pi_i}{1-\pi_i}}.
\end{align*}

Note that, 
\begin{align*}
& \sum_{i=1}^n \frac{1-\pi_i}{\pi_i} \Theta(1)^2 - 2 n \Theta(1) \Theta(0) + \sum_{i=1}^n \frac{\pi_i}{1-\pi_i} \Theta(0)^2 \\
= & \sum_{i=1}^n \bigg( \frac{1-\pi_i}{\pi_i} \Theta(1)^2 - 2 \Theta(1) \Theta(0) + \frac{\pi_i}{1-\pi_i} \Theta(0)^2 \bigg) \\
= & \sum_{i=1}^n \bigg( \sqrt{\frac{1-\pi_i}{\pi_i}} \Theta(1) - \sqrt{\frac{\pi_i}{1-\pi_i}} \Theta(0) \bigg)^2 \\
\geq & 0,
\end{align*}
where inequality takes equality if and only if one of the following two conditions hold:
\begin{enumerate}
\item $\pi_i = \pi$ for any $i \in [n]$;
\item $Y_i(1) = Y(1)$ and $Y_i(0) = Y(0)$ for any $i \in [n]$.
\end{enumerate}
But in this case, $\det(\bm{A}^\top \bm{\Sigma} \bm{A}) > 0$ so the first condition does not hold. 
So the inequality takes equality if and only if the second condition holds. 

\noindent \textbf{Case 2}: $\det(\bm{A}^\top \bm{\Sigma} \bm{A}) = 0$.
In this case, $\pi_i = \pi$ for any $i \in [n]$, and we have
\begin{align*}
\gamma^{\Naive}(1)=\frac{1}{n}\sum_{i=1}^n Y_i(1), \qquad \text{and} \qquad \gamma^{\Naive}(0)=\frac{1}{n}\sum_{i=1}^n Y_i(0).
\end{align*}
So we know that
\begin{align*}
\bm{A}^\top \bm{\Sigma} \bm{A} \bm{\gamma}^{\Naive} - \bm{A}^\top \bm{\Sigma} \bm{G} = 
\begin{bmatrix}
0 \vspace{5pt} \\
0
\end{bmatrix}.
\end{align*}
Putting the above into \eqref{eqn:VarComparison2} we have
\begin{align*}
\Var\big( \widehat{\tau}^{\CV}(\gamma^{\Naive}(1), \gamma^{\Naive}(0)) \big) - \Var\big( \widehat{\tau}^{\CV}(\gamma^*(1), \gamma^*(0)) \big) = 0.
\end{align*}
Combining both cases we finish the proof. 
\hfill \halmos
\endproof

\section{Additional Details of the Swiss Environmental Panel Data}
\label{sec:DataSource}

\subsubsection*{Empirical setup.}
To start, the data comes from two files downloaded from the cumulative data release of the Swiss Environmental Panel \citep{quoss2021swiss}.
They surveyed the same group of people three times between 2018 and 2019. 
Wave~1 was fielded from February~1 to August~20, 2018; Wave~2 from November~29,
2018 to May~17, 2019; and Wave~3 from June~4 to October~21, 2019.
Each wave asked about a different topic. Wave~1 asked general questions about environmental attitudes, energy and climate policy, mobility, and politics, and did not ask about food waste. 
Wave~3 introduced a new module called ``Food in Switzerland,'' which asked
about food purchase habits, personal food disposal, and how respondents think food waste is divided across different parts of the economy. 

\begin{table}[!tb]
\centering
\TABLE{Invitation and cumulative response counts by region in the Swiss Environmental Panel survey
\label{tbl:UnequalProbabilities}}
{\begin{tabular}{l>{\centering}p{3.6cm}>{\centering}p{3.6cm}>{\centering}p{3.2cm}c}
Region                   & Invited (Wave 1) & Responded (Wave 3) & Response Rate & \\ \hline
Lake Geneva region       & \ 2712           & \ 480              & 17.7\%        & \\
Espace Mittelland        & \ 3140           & \ 692              & 22.0\%        & \\
Northwestern Switzerland & \ 1998           & \ 447              & 22.4\%        & \\
Zurich                   & \ 2569           & \ 543              & 21.1\%        & \\
Eastern Switzerland      & \ 1998           & \ 441              & 22.1\%        & \\
Central Switzerland      & \ 1285           & \ 307              & 23.9\%        & \\
Ticino                   & \ 1259           & \ 275              & 21.8\%        & \\ \hline
Total                    & 14961            & 3185               &               & 
\end{tabular}
}
{\textit{Note:} The ``Responded'' column reports direct counts; the ``Invited'' column is reconstructed. We explain how each number is obtained in Section~\ref{sec:DataSource}. The data are from \citet{quoss2021swiss}.}
\end{table}

Only Wave~1 drew a fresh random sample from the Swiss population; Waves~2 and~3 went back to the same people. 
In Wave~1, a total of $14961$ people were invited and $4813$ responded. 
Before Wave~3, $224$ of these $4813$ people explicitly asked to stop participating, leaving $4589$ people to be contacted again. 
Among them, $3228$ responded in Wave~3, with $3185$ useful responses without missing value.
We focus on these $3185$ responses as the observed outcomes. 
See Table~\ref{tbl:UnequalProbabilities} for details.

\subsubsection*{Data files and data pre-processing.}
Now we explain where each number in Table~\ref{tbl:UnequalProbabilities} comes from. 
In the survey, a total of $14961$ people were invited in Wave 1, and a total of $3185$ people responded with a useful answer in Wave 3. 
These two numbers come from Table~1 of the documentation file \texttt{1220\_SEP\_Doc\_Documentation\_W1-W6\_EN.pdf}. 

The Responded column in Table~\ref{tbl:UnequalProbabilities} is exact.
We use the file \texttt{1220\_SEP\_Data\_W3\_v2.0.0.dta}, which has one row for each of the $3228$ people who responded in Wave~3.
Each row carries a variable named \texttt{w3\_bigreg} that records which of the seven regions the respondent lives in, and a variable \texttt{w3\_q8x1} that records the respondent's answer to the household food waste share question; $43$ respondents have a missing value code for \texttt{w3\_q8x1} and are dropped, leaving $3185$ usable respondents.
We count how many of these $3185$ rows fall into each region.
These seven counts add up to $3185$.

The Invited column in Table~\ref{tbl:UnequalProbabilities} is reconstructed. 
The documentation file does not give the number of invited people in each region directly. 
The documentation file says that ``All sampling was conducted as a simple random sample on the level of NUTS-2 regions. The region of Ticino was oversampled by a factor of two to guarantee enough respondents for the Italian speaking part of the country.''
Table~3 of the documentation file reports, for the group it calls ``Population'' (meaning the people who were invited), the share of each region as a fraction rounded to two decimal places: $0.19$, $0.22$, $0.14$, $0.18$, $0.14$, $0.09$, and $0.04$, for the seven regions in the order shown in Table~\ref{tbl:UnequalProbabilities}. 
These shares cannot be multiplied by $14961$ directly, because the documentation file states that the table is reported ``including sampling weights'': each reported share $s_r$ is proportional to $w_r I_r$, where $I_r$ is the number invited in region $r$ and $w_r$ is the design weight recorded in \texttt{w3\_weight}, equal to $0.5$ for Ticino and $1.1024$ for the other six regions. 
We therefore invert the weights and find $I_r = 14961 \cdot (s_r/w_r) / \sum_{j=1}^7 (s_j/w_j)$, which gives the Invited column. This raises Ticino from a $4.0\%$ to an $8.4\%$ share of the invitations, that is, about twice its $4\%$ share of the Swiss population, which is the oversampling ``by a factor of two'' quoted in the documentation file. 

Because the shares are only given to two decimal places, the Invited numbers carry some rounding error; for Ticino, a reported share anywhere in $[0.035, 0.045]$ gives between $1113$ and $1401$ invitations and a response rate between $19.6\%$ and $24.7\%$, which stays comparable to the range spanned by the other six regions. 

It is worth mentioning that, the weight variable \texttt{w3\_weight} in the Wave~3 data only corrects for the fact that Ticino residents were invited at roughly twice the rate of the other six regions. 
The codebook \texttt{1220\_SEP\_Doc\_Codebooks\_w3\_codebook\_SEP.pdf} says directly that it ``does not take into account response behavior.'' 
The weights therefore describe who was invited but not who answered, which is why we use the realized response rates below.

Given above, in our empirical section, we analyze the data again by using the different response rates in Table~\ref{tbl:UnequalProbabilities} as the sampling probabilities $\pi_i$ for $i \in [n]$.

\subsubsection*{Population and outcomes.}
We take the finite population to be the $n = 14961$ people invited in Wave~1.
The realized survey sample consists of the $3185$ Wave~3 respondents with a usable answer.
Thus, for each unit $i$ in the Wave~1 invited population, we define $W_i=1$ if this unit appears among the $3185$ Wave~3 samples, and $W_i=0$ otherwise.
The remaining units are part of the finite population, but their outcomes are not observed.

We use the region level cumulative response rate as the sampling probability.
Specifically, if unit $i$ belongs to region $r(i)$, then
\begin{align}
\pi_i = \frac{\text{number of Wave~3 responded units in region } r(i)}{\text{number of Wave~1 invited units in region } r(i)}. \label{eqn:ResponseRates}
\end{align}
The numerator and denominator are reported in Table~\ref{tbl:UnequalProbabilities}. 
The sampling probabilities are constant within each region but different across regions, ranging from $17.7\%$ in Lake Geneva region to $23.9\%$ in Central Switzerland. 
In the analysis below, we assume a Bernoulli sampling design as in Example~\ref{exa:Bernoulli}. 

We focus on a direct survey outcome from Question~8 of Wave~3.
This question asks each respondent to distribute $100\%$ of national food waste across five sectors.
We use the share attributed to households as the outcome, denoted by $Y_i$, which is a number between $0$ and $100$.
Because Ticino was deliberately invited at about twice the rate of the other six regions to ensure enough respondents for the Italian speaking part of the country, the Wave~1 population itself over-represents Ticino residents. 
The survey uses a weight $w_i$ to correct for this: $w_i = 0.5$ for Ticino and $w_i = 1.1024$ for the other six regions. 
We therefore wish to estimate the weighted mean
\begin{align*}
\mu^w_n = \frac{\sum_{i=1}^n w_i Y_i}{\sum_{i=1}^n w_i} = \sum_{i=1}^n Y_i^w
\end{align*}
rather than the plain population mean of $Y_i$ among the $n = 14961$ people, where $Y_i^{w} = Y_i \cdot \frac{w_i}{\sum_{i=1}^n w_i}$.

\section{Additional Details of the Insurance Network Experiment Data}
\label{sec:InsuranceNetworkDataSource}

\subsubsection*{Data files and data pre-processing.}
Our analysis uses two data files. 
The household file \texttt{0422survey.dta} contains $4902$ households and the assignment, outcome, and village variables. 
Using \texttt{address} as the natural village identifier, it contains $173$ natural villages.
These counts are smaller than the $5332$ households and $185$ natural villages
reported in \citet{cai2015social}. 
We cannot reconstruct households absent from the data files.

The file \texttt{0422allinforawnet.dta} contains the pre-experimental social network. 
A few days before the experiment, each household was asked to name up to five friends, inside or outside the village, with whom the household may discuss rice production or financial matters.
A row with \texttt{id}$=i$ and \texttt{network\_id}$=j$ means that household $i$ named household $j$; it is therefore a directed nomination $i \to j$, instead of undirected.

The raw file contains $23243$ rows. 
Removing $331$ placeholder rows marked as missing names leaves $22912$ rows. Removing the remaining missing friend identifiers and the sentinel identifier $99$ leaves $22620$ rows. 
Requiring both endpoints to occur in the released household file leaves $17069$ rows.
Removing $54$ self loops and $18$ duplicated ordered pairs leaves $16997$ valid directed nominations. 
We keep the $877$ edges that cross the natural village identifier because the questionnaire permitted friends inside or outside the village.

\subsubsection*{Directed and undirected edges.}
Let $A$ be the cleaned adjacency matrix. 
The directed edges are defined as $A_{ij} = \bI\{i\text{ named }j\}$.
Hence $A_{ij}=1$ need not imply $A_{ji}=1$. 
Among the $16997$ directed nominations, there are $2196$ pairs of households that nominate each other, contributing to $4392$ directed nominations; the other $12605$ nominations are only in one direction. 
The out-degree is at most five by the questionnaire.
The in-degree can be as big as $18$ in the data.

\citet{cai2015social} consider the directed edges when studying the effects of having one additional friend in a first round intensive session.
In the companion codes \texttt{rawnet.do}, \citet{cai2015social} construct the network measures by aggregating the original \texttt{id}$\to$\texttt{network\_id} rows within household \texttt{id}.
We follow this implementation and use directed outgoing nominations in our analysis. 
Under this implementation, a friend belongs to household $i$'s exposure set if and only if $i$ named that friend.

\subsubsection*{Summary Statistics.}
The experiment involved 4902 households.
There are four combinations of treatment conditions: first round versus second round, and simple information session versus intensive information session. 
The simple information sessions introduced the insurance contract; 
the intensive information sessions introduced the same information together with additional financial education about how insurance works and its expected benefits.
The second round information sessions happened 3 days after the first round sessions. 
The 3 days in time was intended to allow communication and information diffusion before the second round.
In the second round, households were also assigned to 3 different information conditions: (i) no additional information about first round purchases, (ii) first round attendance and overall purchase rates, and (iii) a detailed list of first round purchases.
See Table~\ref{tbl:InsuranceSessionGroups} for a summary of all the treatment conditions.

\begin{table}[!tb]
\centering
\TABLE{Household counts by session round, session type, and second round information condition
\label{tbl:InsuranceSessionGroups}}
{
\begin{tabular}{l>{\centering}p{2.7cm}>{\centering}p{3cm}>{\centering}p{3cm}>{\centering}p{3cm}c}
                  & First round & \multicolumn{3}{c}{Second round}                     & \\ \cline{3-5}
                  &             & No information & Purchase rates & Detailed purchases & \\ \hline
Simple session    & T1: $1079$  & U1: $657$      & U2: $355$      & U3: $362$          & \\
Intensive session & T2: $1096$  & U4: $660$      & U5: $350$      & U6: $343$          & \\
\end{tabular}
}
{\textit{Note:} T1, T2, U1, ..., U6 are the group names. T3 refers to the union of U1, U2, and U3; T4 refers to the union of U4, U5, and U6. The numbers following the group names are the numbers of households in the samples.}
\end{table}

\subsubsection*{Population and outcomes.}
Following \citet{cai2015social}, the outcome population contains both U1 and U4 households in the second round.
The first round households are not the outcome population.
Instead, their session assignments, either simple session or intensive session, determine the friend exposures of the second round households.

For each household $i$ in U1 and U4, we define $W_i=1$ if this household is assigned to the intensive session and let $W_i=0$ if this household is assigned to the simple session.
For each household $i$ in U1 and U4, we denote $m_i$ to the number of its friends who were assigned to a first round session (T1 and T2), and $k_i$ to be the number of its friends who were assigned to a first round intensive session (T2). 
The observed outcome is a binary variable that takes value $1$ if and only if household $i$ purchased the insurance after its information session. 
Let $Y_i(w_i, k_i)$ denote its purchase outcome when its own session is $w_i \in \{0,1\}$ and exactly $k_i$ of its friends were assigned to a first round intensive session.

We wish to estimate the marginal effect of having one more friend who received intensive session in the first round.
For $k\in\{1,2,3\}$, define the target population as
\begin{align*}
\mathcal I_k = \{i\text{ in U1 or U4} \ \vert \ m_i \geq k\},
\end{align*}
and denote $n_k = |\mathcal I_k|$.
The target populations contain $n_1=881$, $n_2=481$, and $n_3=168$ households, respectively.
In these three populations, we consider the following three effects
\begin{align*}
\small 
\tau_{10} = \frac{1}{n_1} \sum_{i\in\mathcal I_1} \big(Y_i(0,1) - Y_i(0,0)\big), 
\
\tau_{21} = \frac{1}{n_2} \sum_{i\in\mathcal I_2} \big(Y_i(0,2) - Y_i(0,1)\big), 
\
\tau_{32} = \frac{1}{n_3} \sum_{i\in\mathcal I_3} \big(Y_i(0,3) - Y_i(0,2)\big).
\end{align*}
We do not consider $\tau_{43}$ or $\tau_{54}$, because their target populations contain too few households.
The target population for estimating $\tau_{32}$ is quite small with $n_3=168$, and, as we will see below, this may already cause extreme estimates. 

\citet{cai2015social} mention that the randomized experiment was a stratified experiment within villages by household size and rice-production area per capita, but the released files do not identify the randomization strata.
We therefore make the assumption that the assignment of households into U1 and U4 follows independent Bernoulli$(\frac{1}{2})$ distributions.

\subsubsection*{Sample splitting.}
We now provide more details on sample splitting.
For each household in U1 or U4, we say that another household is its neighbor if the first round friends of the two households overlap.
We then construct an undirected edge among two neighbors. 
We then find the connected components of this network and keep every component entirely within one fold. 
For each candidate split, we independently assign every connected component to Fold 1 or Fold 2 with equal probability. 
We accept the split if both folds contain at least $45\%$ households of the population; otherwise, we randomly redraw all component assignments. 
We stop if we have redrawn $10000$ times and returns an error if no sufficiently balanced split is found.
The sample splitting uses only the graph and population size information, and never uses any information about the realized exposures or outcomes.
For the estimation of $\tau_{10}$, the two folds have sizes $401$ and $480$.
For the estimation of $\tau_{21}$, the two folds have sizes $229$ and $252$.
For the estimation of $\tau_{32}$, the two folds have sizes $90$ and $78$.

For the estimation of $\tau_{32}$, the sample splitting results require caution.
There are only $4$ observations satisfying the exposure condition that there are $3$ first round friends receiving intensive sessions. 
The two folds have $3$ and $1$ such observations, respectively. 
In the fold with only $1$ observation, the estimated second moment equals
the squared estimated mean, and so the plug-in variance is zero.
We take the maximum of zero and $10^{-8}$ in our later calculation. 
This is one reason why the corresponding sample splitting results require caution.

Finally, our estimates are related to, but do not replicate, the nonlinear regression in column~(4) of Table~4 of \citet{cai2015social}. 
The results in \citet{cai2015social} control for each household's covariates, include village fixed effects, and use clustered regression standard errors.

\end{document}